\documentclass[a4,10pt,twoside,openright,italian,english]{book}

\usepackage[twoside=true]{geometry}
\usepackage{graphicx}

\usepackage[cam,center,a4,pdflatex,axes]{crop}

\usepackage{makecell}

\usepackage{phdthesis}

\newglossary[not]{notation}{not}{ntn}{Notation}

\usepackage{adjustbox}
\usepackage{fancyhdr}
\usepackage{color}
\usepackage{array}
\usepackage{mdwtab}
\usepackage{amsmath,amssymb}
\usepackage{cite}
\usepackage{listings}
\usepackage{booktabs}
\usepackage{latexsym}
\usepackage{url}
\usepackage{bnf}
\usepackage{rotating}
\usepackage{multirow}
\usepackage{phdtitle}
\usepackage{paralist}
\usepackage{physics}
\newcommand{\INSTR}{\mathrm{INSTR}}
\usepackage[bookmarks=true,
bookmarksopen=true,hidelinks]{hyperref}

\usepackage{lscape}
\usepackage{algorithmic}
\usepackage{algorithm}
\usepackage{longtable}
\usepackage[T1]{fontenc} 
\usepackage[utf8]{inputenc}
\usepackage{subcaption}

\usepackage{bookmark}   
\usepackage[nottoc]{tocbibind} 

\usepackage{microtype}        
\usepackage{libertinus}       
\usepackage[libertine]{newtxmath} 

\usepackage{amsthm}
\theoremstyle{definition}
\newtheorem{definition}{Definition}[chapter]
\newtheorem{assumption}{Assumption} [chapter]
\newtheorem{remark}{Remark} [chapter]

\theoremstyle{plain}
\newtheorem{theorem}{Theorem}[chapter]
\newtheorem{lemma}{Lemma}[chapter]
\newtheorem{corollary}{Corollary}[chapter]
\newtheorem{proposition}{Proposition}[chapter]
\newtheorem{problem}{Problem}[chapter]

\usepackage{chngcntr}
\makeatletter
\counterwithout{footnote}{chapter}
\makeatother

\usepackage[english,italian]{babel}
\usepackage{ulem}
\hypersetup{pdftitle={Title}, pdfauthor={Author}}

\department {Dottorato di ricerca in Ingegneria dell'Informazione}

\author{Vinay Kumar} 

\title{Making Quantum Networks Work\\
\large Routing, Calibration, and Programmable Quantum Repeaters}

\tutor{Dr. Raffaele Bruno \\
Dr. Claudio Cicconetti \\
Dr. Andrea Passarella}

\titleimage{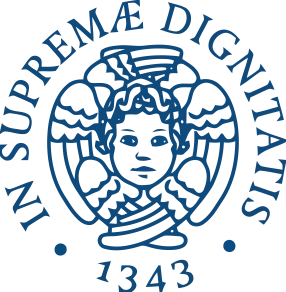} 
\phdyear{2026}
\phdmonth{January}
\phdcycle{Cycle XXXVIII}

\makenoidxglossaries
\newacronym{qkdLabel}{QKD}{Quantum Key Distribution}

\newacronym{eprLabel}{EPR}{Einstein--Podolsky--Rosen}
\newacronym{ghzLabel}{GHZ}{Greenberger--Horne--Zeilinger}
\newacronym{bbpsswLabel}{BBPSSW}{Bennett--Brassard--Popescu--Schumacher--Smolin--Wootters}

\newacronym{qloLabel}{QLO}{Quantum Link Orchestration}
\newacronym{isaLabel}{ISA}{Instruction Set Architecture}
\newacronym{kaLabel}{KA}{Knowledge Aware}
\newacronym{kspLabel}{KSP}{K-Shortest Path}
\newacronym{spLabel}{SP}{Shortest Path}

\newacronym{bpLabel}{BP}{Blocking Probability}
\newacronym{pmfLabel}{PMF}{Probability Mass Function}

\newacronym{nvLabel}{NV}{Nitrogen Vacancy}
\newacronym{mwLabel}{MW}{Microwave}
\newacronym{rfLabel}{RF}{Radio Frequency}

\newacronym{hqLavel}{HQ}{High Quality}

\newacronym{lqLavel}{LQ}{Low Quality}

\newacronym{sdLabel}{SD}{Source--Destination}
\newacronym{e2eLabel}{E2E}{End-to-End}

\newacronym{sdnLabel}{SDN}{Software Defined Networking}

\newacronym{lcuLabel}{LCU}{Linear Combination of Unitaries}
\makeatletter

\newglossaryentry{identity}{
  name={$I$},
  description={Identity matrix or identity operator}
}

\newglossaryentry{ket}{
  name={$|\psi\rangle$},
  description={Ket vector representing a quantum state in Hilbert space}
}

\newglossaryentry{bra}{
  name={$\langle\psi|$},
  description={Bra vector corresponding to the dual of a ket}
}

\newglossaryentry{hilbert}{
  name={$\mathcal{H}$},
  description={Hilbert space associated with a quantum system}
}

\newglossaryentry{fidelity}{
  name={$F$},
  description={Fidelity of a quantum state or entangled pair}
}

\newglossaryentry{fidelity_threshold}{
  name={$F_{\mathrm{th}} \ \text{or} \ \bar{F}$},
  description={Minimum required fidelity threshold}
}

\newglossaryentry{end_to_end_fidelity}{
  name={$F_{\mathrm{ete}}$},
  description={End-to-end fidelity of an entangled path}
}

\newglossaryentry{nv center throughput}{
  name={$R$},
  description={Per-electron round throughput}
}

\newglossaryentry{nsd}{
  name={$n_{\mathrm{sd}}$},
  description={Number of source-destination pairs}
}

\newglossaryentry{path_order}{
  name={$\theta$},
  description={Path establishment order for a set of requests}
}

\newglossaryentry{efficiency}{
  name={$\eta$},
  description={Efficiency of the qubit measurement in the basis states}
}

\newglossaryentry{efficiency_class}{
  name={$\eta_g$},
  description={Efficiency associated with repeater class $g$}
}

\newglossaryentry{activation_time}{
  name={$a_e$},
  description={Activation time allocated to link $e$}
}

\newglossaryentry{calibration_time}{
  name={$c_e \ \text{or} \ \tau_c$},
  description={Calibration time required for link $e$ (constanst for each link)}
}

\newglossaryentry{gamma}{
  name={$\Gamma_e$},
  description={Fidelity decay or drift parameter of link $e$}
}

\newglossaryentry{tf}{
  name={$T_f$},
  description={Fidelity decay time window}
}

\newglossaryentry{lagrange}{
  name={$\lambda$},
  description={Robustness parameter}
}

\newglossaryentry{instruction_vector}{
  name={$I(t)$},
  description={Instruction vector broadcast by the classical controller}
}

\newglossaryentry{reset_time}{
  name={$\tau_{\mathrm{reset}}$},
  description={Nuclear spin re-initialisation time}
}

\newglossaryentry{xi}{
  name={$\xi$},
  description={Fraction of high-quality nodes in the network}
}

\newglossaryentry{average_degree}{
  name={$\beta$},
  description={Average node degree parameter of the network topology}
}

\newglossaryentry{k_paths}{
  name={$k$},
  description={Number of candidate paths considered in $k$-shortest path routing}
}

\newglossaryentry{initial_fidelity}{
  name={$F_M^e$},
  description={Initial maximum fidelity achievable on link $e$ after calibration}
}

\newglossaryentry{link_budget}{
  name={$L_e$},
  description={Fidelity allocation in L-space for Link $e$}
}

\newglossaryentry{omega}{
  name={$\Omega_L$},
  description={Equal activation optimum for a path (from end to end constraint) in L-space}
}

\newglossaryentry{omega1}{
  name={$\Omega_{\pi}$},
  description={Optimal point for a specific path}
}
\makeatother
\makeatletter
\newcommand{\@reviewer}{Prof. Michele Amoretti \newline Prof. Alfredo Grieco}
\makeatother

\makeatletter
\renewcommand*{\glsnumberformat}[1]{} 
\makeatother

\begin{document}
\selectlanguage{english}

\maketitle

\pagestyle{empty}

\cleardoublepage
\newpage

\thispagestyle{empty}
    \null\vspace{\stretch {1}}
        \begin{flushright}
                \emph{This thesis is dedicated to my late grandfather.}
        \end{flushright}
\vspace{\stretch{2}}\null

\cleardoublepage
\newpage

\pagestyle{empty}

\thispagestyle{empty}
    \null\vspace{\stretch {1}}
        \begin{flushright}
    
    ``The history of the universe is, in effect, a huge and ongoing quantum computation. The universe is a quantum computer.'' \\ -- Seth Lloyd
        \end{flushright}
\vspace{\stretch{2}}\null

\cleardoublepage
\newpage

\pagestyle{empty}
\setcounter{page}{1}
\pagenumbering{Roman}

\tableofcontents
\chapter*{Acknowledgements}
\addcontentsline{toc}{chapter}{Acknowledgements}
\markboth{Acknowledgements}{Acknowledgements}

\lettrine{I} would like to express my sincere gratitude to all those who have contributed, directly and indirectly, to the completion of this doctoral thesis. First and foremost, I would like to thank my supervisors and research group at the Institute for Informatics and Telematics of the National Research Council (IIT-CNR), Dr. Claudio Cicconetti, Dr. Marco Conti, and Dr. Andrea Passarella. Their constant feedback, insights, and guidance played a fundamental role in shaping this work and, more importantly, in shaping me as a researcher. From the very beginning of my PhD, they fostered an environment of intellectual freedom and a stimulating research culture that encouraged curiosity, independence, and the confidence to think beyond immediate problems.

When I started my PhD, this was my first experience living away from my home country and moving to Europe. During this period, my supervisors became a strong reference point for me, both professionally and personally. Our regular Wednesday meetings quickly became something I looked forward to every week, to the point that Wednesday has since become my favourite day of the week. I am also grateful to them for handling logistical and administrative aspects smoothly throughout the PhD, allowing me to focus fully on research.

Beyond this collective guidance, I would like to thank Dr. Andrea Passarella and Dr. Marco Conti individually for their complementary roles throughout my doctoral journey. Dr. Andrea Passarella acted as a guiding reference within the group, frequently drawing on his broad experience in networking research to provide insightful perspectives, particularly valuable given my background in physics. His ability to identify subtle points and raise considerations that might otherwise have been overlooked significantly strengthened the depth and direction of my work. Dr. Marco Conti offered a steady and reassuring presence, often helping to validate ideas and discussions through his long-standing experience. His thoughtful oversight and calm guidance contributed to maintaining coherence and confidence across different stages of the research.

In particular, I would like to thank Dr. Claudio Cicconetti, who has been more of a research companion than a supervisor in the traditional sense. He has been a constant presence throughout my PhD, patiently engaging with both my well-formed questions and my less well-thought-out ones, and providing steady guidance at every stage. If I were to go back and choose again, I would still choose the same research group for my PhD without hesitation.

I gratefully acknowledge the institutional and financial support provided by the University of Pisa and IIT-CNR, which made this doctoral work possible.

Alongside my work at IIT-CNR, I had the opportunity to undertake a research visit at the Deutsche Telekom Chair of Communication Networks, Technische Universität Dresden, in Germany. I would like to thank Dr. Riccardo Bassoli for the invitation and for creating a calm, open, and intellectually generous environment throughout the visit. He gave me the freedom to explore ideas independently while always being available for discussion, and his thoughtful feedback and ability to highlight key conceptual points greatly enriched the scientific quality of the work carried out during this period.

I would also like to thank Dr. Viktor Krutianskii for the many insightful discussions and for generously sharing his experimental perspective. His expertise in quantum networks played an important role in grounding the calibration model developed in this work in experimentally realistic conditions. In particular, his guidance and the experimental data he provided were crucial for shaping the calibration aspect of this work, helping to bridge the gap between theoretical abstraction and experimentally relevant performance considerations.

I am grateful to the staff at the Deutsche Telekom Chair of Communication Networks, Technische Universität Dresden for their support and for handling the visit logistics smoothly. In particular, I would like to thank Leon R\"{o}scher for his exceptional help during my stay, from assisting with settling in to engaging in technical discussions and providing a genuinely hands-on research environment. His openness, practical support, and willingness to share insights, including direct access to and experience with the quantum computing infrastructure available in the group, made the visit both productive and intellectually stimulating.

I would also like to thank Vignesh Raman for bringing a unique collaborator-and-friend energy to our discussions. I greatly enjoyed our long and detailed technical conversations, which were both engaging and motivating, and often helped refine ideas through open and rigorous exchange. I would also like to thank Benedikt Baier for his collaboration and for the many insightful ideas shared during our discussions, often in connection with ongoing work with Vignesh. I am further grateful to Krishan Joshi for the many deep technical exchanges and parallel lines of work we pursued together, which, while not directly reflected in this thesis, were intellectually important. His experimental mindset and openness to engaging in theory-driven projects made these interactions particularly valuable. I am also grateful to Hilal Sultan Duranoglu Tunc for her enthusiasm and initiative in approaching me for collaboration. Her energy and openness to joint work are sincerely appreciated, and I hope that our collaboration will continue and develop further in the future.

I would also like to acknowledge the valuable exchanges and interactions that took place during the Quantum Internet Alliance (QIA) hackathon in Delft. In particular, I would like to thank Anna Beata Kalisz Hedegaard, Mateo Maximiliano Blanco Rodríguez, David Pérez Castro, and Arthur Witt for the stimulating discussions and idea exchanges during the event. I am especially grateful to Anna Beata Kalisz Hedegaard for the invitation to give a talk at Quantum Security Defence, which I greatly appreciated as an opportunity to share and discuss my work with a broader audience.

Finally, I would like to thank my close friends and my family for being there whenever I needed support and encouragement. Their presence has been invaluable throughout this journey.

\selectlanguage{english}
\selectlanguage{english}

\cleardoublepage
\newpage

\pagestyle{fancy}

\selectlanguage{english}
\chapter*{Summary}
\addcontentsline{toc}{chapter}{Summary}
\markboth{Summary}{Summary}

\lettrine{Q}{uantum} networks promise the distribution of entanglement across long distances, enabling applications such as secure communication, distributed quantum computation, and networked quantum sensing. While significant progress has been made in individual components of the quantum internet, including physical platforms, entanglement generation protocols, and theoretical performance bounds, the design of scalable quantum networks remains an open challenge. In particular, existing approaches often treat routing, link operation, and node design as largely independent problems, overlooking the coupling that arises when these aspects must operate together to enable a functioning quantum network or a global quantum internet.

This thesis addresses the operation of quantum networks under realistic physical and technological assumptions by adopting a cross-layer perspective. Rather than assuming full and static knowledge of network properties, the proposed framework embraces partial knowledge, time-dependent link behaviour, and programmable hardware as some of the fundamental characteristics of near-term and medium-scale quantum networks. The central thesis is that scalable quantum networking requires coordinated design across the network, link, and node layers, with explicit mechanisms to expose and manage their interactions.

The thesis begins by establishing the conceptual and architectural foundations of quantum networking across Chapter~\ref{ch:foundations} and Chapter~\ref{ch:applications_architecture}. After introducing the relevant quantum information primitives and communication mechanisms, the work situates quantum networks as engineered systems whose global behaviour emerges from the interaction of hardware, local control, and network-level coordination. This perspective motivates the need for new abstractions and optimisation strategies that go beyond classical networking models.

The first technical part of the thesis, spanning Chapter~\ref{ch:routing} and Chapter~\ref{ch:performance_evaluation}, focuses on routing in quantum networks under partial knowledge and heterogeneity. Routing problems are formulated for networks composed of repeaters with differing efficiency figures, showing how heterogeneity directly impacts blocking probability and fairness. Building on this, a grey-box routing paradigm is introduced in which path selection relies only on network topology and end-to-end performance estimates. Through analytical reasoning and extensive simulation, these routing strategies are shown to achieve low blocking probability while maintaining fidelity constraints, without requiring detailed or perfectly accurate network properties.

The thesis then addresses the impact of calibration and hardware drift on quantum link operation in Chapter~\ref{ch:calibration}. Motivated by experimental evidence of fidelity degradation during sustained operation, a calibration-aware model of quantum links is developed that explicitly distinguishes between activation and calibration phases. For linear quantum repeater chains, optimal calibration schedules are derived analytically under fidelity constraints. The problem is subsequently extended to general network topologies with shared links, where a greedy quantum link orchestration strategy is proposed to coordinate activation across competing paths. These results demonstrate how calibration policies dynamically reshape the effective network topology and emerge as a central limiting factor for throughput and reliability.

The final technical part of the thesis in Chapter~\ref{ch:programmable_repeaters} bridges network-level protocols and hardware implementation by introducing an instruction set architecture for programmable quantum repeater nodes based on nitrogen-vacancy centers in diamond. A controller-driven abstraction is proposed in which local quantum operations are selected via a nuclear-spin program register, enabling both deterministic and coherent programmability. This model provides a concrete interface between physical operations and network-level decisions, and enables diagnostics and calibration techniques that have no classical analogue.

The thesis concludes by synthesising insights across routing, calibration, and hardware programmability in Chapter~\ref{ch:crosslayer}, highlighting their mutual dependencies and design trade-offs. These contributions provide a unified framework for designing and operating quantum networks as software-defined, calibration-aware systems. The results suggest that future quantum networks should prioritise partial knowledge, explicit management of calibration and drift, and programmable node architectures, providing a foundation for both experimental testbeds and scalable quantum internet deployments.

\selectlanguage{english}


\selectlanguage{italian}
\chapter*{Sommario}
\addcontentsline{toc}{chapter}{Sommario}
\markboth{Sommario}{Sommario}

\lettrine{L}{e} reti quantistiche (``quantum networks'') promettono la distribuzione di entanglement su lunghe distanze, abilitando applicazioni quali la comunicazione sicura, il calcolo quantistico distribuito e il quantum sensing in rete. Sebbene siano stati compiuti progressi significativi sui singoli componenti dell’internet quantistico, inclusi le piattaforme fisiche, i protocolli di generazione di entanglement e i limiti teorici di prestazione, la progettazione di reti quantistiche scalabili rimane una sfida aperta. In particolare, gli approcci esistenti tendono a trattare il routing, l’operazione dei collegamenti e la progettazione dei nodi come problemi in larga misura indipendenti, trascurando il forte accoppiamento che emerge quando tali aspetti devono operare congiuntamente per consentire il funzionamento di una rete quantistica o di un internet quantistico globale.

Questa tesi affronta il problema dell’operazione delle reti quantistiche sotto assunzioni fisiche e tecnologiche realistiche adottando una prospettiva cross layer. Anziché assumere una conoscenza completa e statica delle proprietà della rete, il quadro proposto abbraccia la conoscenza parziale, il comportamento temporale variabile dei collegamenti e l’hardware programmabile come caratteristiche fondamentali delle reti quantistiche di breve e medio termine. La tesi centrale è che la scalabilità delle reti quantistiche richiede una progettazione coordinata tra i livelli di rete, collegamento e nodo, con meccanismi espliciti in grado di esporre e gestire le loro interazioni.

La tesi inizia stabilendo le basi concettuali e architetturali delle reti quantistiche nei Capitoli~\ref{ch:foundations} e~\ref{ch:applications_architecture}. Dopo aver introdotto i principali primitivi dell’informazione quantistica e i meccanismi di comunicazione, il lavoro inquadra le reti quantistiche come sistemi ingegnerizzati il cui comportamento globale emerge dall’interazione tra hardware, controllo locale e coordinamento a livello di rete. Questa prospettiva motiva la necessità di nuove astrazioni e strategie di ottimizzazione che vadano oltre i modelli del networking classico.

La prima parte tecnica della tesi, che comprende i Capitoli~\ref{ch:routing} e~\ref{ch:performance_evaluation}, si concentra sul routing nelle reti quantistiche in presenza di conoscenza parziale ed eterogeneità. I problemi di routing sono formulati per reti composte da ripetitori con differenti figure di efficienza, mostrando come l’eterogeneità influisca direttamente sulla probabilità di blocco e sull’equità. Su queste basi viene introdotto un paradigma di routing grey box, in cui la selezione dei cammini si basa esclusivamente sulla topologia di rete e su stime delle prestazioni end to end. Attraverso analisi teoriche e simulazioni estensive, tali strategie di routing dimostrano di ottenere una bassa probabilità di blocco mantenendo i vincoli di fedeltà, senza richiedere una conoscenza dettagliata o perfettamente accurata delle proprietà della rete.

La tesi affronta quindi l’impatto della calibrazione e della deriva hardware sull’operazione dei collegamenti quantistici nel Capitolo~\ref{ch:calibration}. Motivato da evidenze sperimentali di degradazione della fedeltà durante l’operazione prolungata, viene sviluppato un modello di collegamento quantistico calibration aware che distingue esplicitamente tra fasi di attivazione e fasi di calibrazione. Per catene lineari di ripetitori quantistici, vengono derivate analiticamente politiche di calibrazione ottimali sotto vincoli di fedeltà. Il problema viene successivamente esteso a topologie di rete generali con collegamenti condivisi, dove viene proposta una strategia greedy di orchestrazione dei collegamenti quantistici per coordinare l’attivazione tra cammini concorrenti. Questi risultati mostrano come le politiche di calibrazione rimodellino dinamicamente la topologia effettiva della rete ed emergano come un fattore limitante centrale per throughput e affidabilità.

L’ultima parte tecnica della tesi, presentata nel Capitolo~\ref{ch:programmable_repeaters}, collega i protocolli a livello di rete con l’implementazione hardware introducendo un’architettura a insieme di istruzioni per nodi ripetitori quantistici programmabili basati su centri di azoto vacanza nel diamante. Viene proposta un’astrazione guidata da un controller, in cui le operazioni quantistiche locali sono selezionate tramite un registro di programma a spin nucleari, consentendo sia la programmabilità deterministica sia quella coerente. Questo modello fornisce un’interfaccia concreta tra le operazioni fisiche e le decisioni a livello di rete, e abilita tecniche di diagnostica e calibrazione prive di analoghi classici.

La tesi si conclude sintetizzando, nel Capitolo~\ref{ch:crosslayer}, gli approfondimenti ottenuti su routing, calibrazione e programmabilità hardware, mettendone in evidenza le dipendenze reciproche e i compromessi di progettazione. I contributi presentati forniscono un quadro unificato per la progettazione e l’operazione delle reti quantistiche come sistemi software defined e calibration aware. I risultati suggeriscono che le future reti quantistiche dovrebbero dare priorità alla conoscenza parziale, alla gestione esplicita della calibrazione e della deriva, e ad architetture di nodo programmabili, fornendo una base solida sia per testbed sperimentali sia per implementazioni scalabili dell’internet quantistico.

\selectlanguage{english}


\selectlanguage{english}
\chapter*{List of publications}
\addcontentsline{toc}{chapter}{List of Publications}
\markboth{List of Publications}{List of Publications}

\section*{International Journals}
\begin{enumerate}
    \item Kumar, V., Cicconetti, C., Conti, M., \& Passarella, A. (2026). Optimal calibration of quantum network links. (In preparation).

    \item Kumar, V., Cicconetti, C., Conti, M., \& Passarella, A. (2025). Quantum internet: Technologies, protocols, and research challenges. International Journal of Networked and Distributed Computing, 13(2), 22.
            
    \item Kumar, V., Cicconetti, C., Conti, M., \& Passarella, A. (2025). Routing in quantum networks with end‐to‐end knowledge. IET Quantum Communication, 6(1), e70000.

\end{enumerate}

\section*{International Conferences with Peer Review}
\begin{enumerate}
    \item Kumar, V., Cicconetti, C., Bassoli, R., Conti, M., \& Passarella, A. (2026). Instruction-Set Architecture for Programmable NV-Center Quantum Repeater Nodes. In Proceedings of the IEEE QCNC 2026 Workshops (WQNC). IEEE. (Accepted). arXiv:2602.14995.
    
    \item Kumar, V., Cicconetti, C., Conti, M., \& Passarella, A. (2024, October). Routing in Quantum Repeater Networks with Mixed Efficiency Figures. In 2024 IEEE Future Networks World Forum (FNWF) (pp. 198-203). IEEE.

\end{enumerate}
\selectlanguage{english}




\glsaddall

\printnoidxglossary[type=\acronymtype,title=List of Abbreviations]
\let\cleardoublepage\clearpage
\printnoidxglossary[type=notation,sort=def,title=List of Notations]
\cleardoublepage

\newpage

\listoffigures
\cleardoublepage

\listoftables
\cleardoublepage
\newpage

\setcounter{page}{1}
\pagenumbering{arabic}

\cleardoublepage
\chapter{Introduction}


\section{Vision: making quantum networks work}
The quantum internet\footnote{
In current literature, ``quantum internet" and ``quantum networks" are often used interchangeably due to the nascent state of the field. In this work, however, when we use ``quantum internet" we refer to its broader vision, impact, and potential applications, whereas ``quantum networks" denotes the specific protocols and operational mechanisms. It is important to note that today, only the concept of a quantum internet exists, and almost all studies refer to quantum networks.} has garnered increasing attention as quantum computing technology starts emerging in the market. The quantum internet aspires to extend quantum information processing beyond isolated devices, enabling the distribution and manipulation of quantum states across geographically separated nodes by leveraging fundamental quantum phenomena. The concept of integrating the quantum internet alongside the classical internet has resonated with the research community, driving efforts to explore ways to capitalise on the principles of quantum mechanics. This interest has sparked research into the architecture of quantum networks, the identification of potential applications and use cases, efficient entanglement distribution\footnote{There are two methods for distributing entanglement. The first involves directly sending one of the qubits of an entangled pair to the target location. The second method utilises entanglement swapping to distribute entangled pairs to two endpoints. Directly sending entangled pairs to the target is generally not advisable. Therefore, in this work, when we refer to ``distribution," we specifically mean the use of the entanglement swapping procedure to distribute the entanglement.} within quantum networks. Such a network would support applications ranging from provably secure communication and distributed quantum computing to network assisted sensing and coordinated quantum control. Over the past two decades, foundational work has established the physical feasibility of entanglement distribution and quantum repeaters, and experimental demonstrations have progressed from laboratory scale links to metropolitan and inter city testbeds \cite{kimble2008quantum, wehner2018quantum, kumar2025quantum}.

While physical quantum computing devices are essential for the global quantum internet, some early applications do not require fault-tolerant quantum devices and are simpler to implement. Significant progress has been made in the early stages of quantum internet, particularly in the development of quantum key distribution (QKD)-based networks, both terrestrial and satellite-based. Steady advancements are evident in the enhancement and deployment of QKD-based networks. At the same time, later stages of the quantum internet, particularly those that rely on quantum entanglement, have also seen progress in various areas. Notably, one of the critical aspects that has become a focus of considerable research in recent years is end-to-end entanglement distribution. Yet despite this progress, the quantum internet today remains largely a conceptual construct. The limiting factor is no longer the absence of protocols, but rather the absence of operational frameworks that can translate physical capabilities into scalable network behavior. Unlike classical networks, quantum networks operate under fundamental constraints imposed by quantum mechanics, including probabilistic operations, decoherence, and the no cloning theorem. These constraints invalidate many assumptions that underpin classical networking abstractions and require fundamentally different design principles.

This thesis is driven by a \emph{central premise}: 

\textit{``Quantum networks will only scale if routing, resource allocation, and hardware operation are designed in a coordinated manner under realistic physical constraints."}

Rather than treating hardware imperfections as secondary effects, this work places them at the center of network design.

\section{Quantum networks as engineered systems}
A quantum network is not merely a collection of entangled links, but an \emph{engineered system} whose global behaviour emerges from the interaction of physical hardware, local control protocols, and network-wide coordination. Entanglement distribution across a multi-hop network requires the successful execution of entanglement generation, storage, swapping, and possibly purification, each of which is probabilistic and sensitive to noise \cite{briegel1998quantum, dur1999quantum}.

From a systems perspective, certain fundamental properties distinguish quantum networks from classical ones. First, operations are probabilistic. Entanglement generation
succeeds with a probability that depends on channel loss, detector efficiency, and
timing synchronisation. As a result, network resources cannot be deterministically reserved, and routing decisions must be made under uncertainty.

Second, quantum states degrade over time. Quantum memories are imperfect, and stored entanglement loses fidelity due to decoherence and repeated operations. This introduces a direct coupling between time, resource usage, and performance.

Third, hardware is heterogeneous and dynamic. Quantum repeater nodes differ in efficiency, memory lifetime, and gate fidelity. Even within a single node, performance drifts over time due to calibration errors and environmental fluctuations.

Fourth, classical and quantum control are inseparable. Measurement outcomes must be communicated reliably and promptly, and delays in the classical control plane directly impact quantum performance.

These properties imply that abstractions separating network protocols from hardware behaviour, which were highly successful in classical networking, cannot be adopted wholesale in the quantum domain.

\section{Limitations of existing quantum networking approaches}

A substantial body of work has explored routing \cite{van2013path, caleffi2017optimal, pant2019routing, li2021effective, shi2020concurrent, zhang2021fragmentation, zhao2021redundant} and resource allocation \cite{cicconetti2022resource, wang2025quantum, kaewpuang2023entangled, gauthier2025demand} in quantum networks, often adapting classical shortest-path or flow-based formulations to entanglement distribution. While these approaches provide valuable insights, many rely on assumptions that are unlikely to hold in near-term deployments.

Common assumptions include homogeneous repeater capabilities, static and perfectly calibrated links, and full knowledge of link-level performance across the network. In practice, however, quantum networks are expected to be heterogeneous, partially observable, and continuously evolving due to hardware drift and environmental fluctuations.

In our work on routing with mixed efficiency figures, it is shown that ignoring heterogeneity can lead to increased blocking probability\footnote{Blocking probability is a routing metric used in the literature and can be interpreted as the inverse of achievable throughput.} and unfair resource allocation. At the same time, even partial knowledge of node quality is shown to significantly improve end-to-end performance and reduce path blocking, particularly in scenarios with incremental hardware upgrades \cite{kumar2024routing,kumar2025routing}.

Similarly, many routing proposals assume access to accurate and stable link-level parameters, such as per-link fidelity or success probability. In realistic settings, however, such information may be unavailable, noisy, or outdated. In contrast, our end-to-end knowledge routing framework demonstrates that robust routing decisions can be made using only network topology and aggregate end-to-end performance estimates, thereby avoiding the fragility of white-box models that depend on detailed physical parameters \cite{kumar2025routing}.

At the link layer, calibration is typically treated as a static preprocessing step, and quantum links are assumed to be continuously available for entanglement generation. Experimental evidence, however, indicates that sustained operation leads to fidelity degradation over time, requiring periodic recalibration \cite{krutyanskiy2023entanglement, peranic2023study, dowling2023non, zhou2025efficient}. Neglecting this dynamic behaviour results in overly optimistic throughput estimates and schedules that are infeasible in practice. In our calibration-aware modelling of quantum links, we explicitly capture the trade-off between activation time and calibration overhead, and show that optimal operation requires jointly optimising link usage and recalibration under fidelity constraints \cite{kumar2026calibration}.

Finally, most existing quantum networking architectures treat repeater nodes as fixed-function devices executing predetermined protocols \cite{pant2019routing, chakraborty2019distributed}. This rigid view limits adaptability and hinders the ability of the network to respond to changing conditions, application requirements, or control-plane decisions. To address this limitation, we introduce an instruction-set architecture for programmable quantum repeater nodes, in which local quantum operations are exposed through a controller-driven abstraction. This approach enables both deterministic and coherent programmability at the node level, bridging hardware capabilities and network-level protocols in a manner analogous to programmability in classical networks \cite{kumar2026nvcenter}.

Collectively, these limitations highlight the need for a \emph{cross-layer approach} to quantum networking, in which routing, link operation, and node architecture are designed with explicit awareness of realistic physical and operational constraints.

\section{Scope and objectives}\label{sec:scope_objective}
The objective of this thesis is to develop a cross-layer \emph{design approach} for making quantum networks operational under realistic physical and technological constraints.

The focus is on bipartite entanglement distribution over quantum repeater networks, which constitutes the foundational service for most envisioned quantum internet applications. Rather than proposing a single protocol, the thesis develops models, algorithms, and abstractions that together enable scalable operation.

The scope spans three tightly interconnected layers:
\begin{enumerate}
    \item Network layer routing, under uncertainty, heterogeneity, and limited knowledge.
    \item Link layer operation, accounting for calibration overhead and fidelity drift.
    \item Node level programmability, bridging physical operations and network control.
\end{enumerate}

This holistic scope is deliberate: isolating any one layer leads to designs that could fail when deployed in realistic settings.

\section{Research questions}\label{sec:research_questions}
This thesis addresses the following research questions:

\begin{itemize}
    \item How can routing in quantum networks remain effective when detailed link-level information is unavailable or unreliable?
    \item How does heterogeneity in repeater efficiency affect fidelity, throughput, and fairness, and how can routing exploit partial upgrades?
    \item How should quantum links be operated and calibrated over time to maximise throughput under fidelity constraints?
    \item How can quantum repeater nodes be made programmable in a way that exposes physical capabilities to network controllers?
\end{itemize}

\section{Contributions}
The main contributions of this thesis are:

\begin{itemize}

\item A routing framework for quantum networks with heterogeneous repeater efficiencies, demonstrating performance gains from partial knowledge and incremental upgrades \cite{kumar2024routing, kumar2025routing}.

\item A grey box routing paradigm based on end-to-end performance estimates, shown to be robust to estimation errors while maintaining fairness \cite{kumar2025routing}.

\item A calibration-aware model of quantum link operation motivated by experimental observations of fidelity drift, enabling realistic throughput analysis \cite{kumar2026calibration}.

\item Analytical solutions for optimal calibration scheduling in linear quantum repeater chains under both initial and end-to-end fidelity constraints, formalised through the Quantum Link Orchestration theorem \cite{kumar2026calibration}.

\item A global optimisation framework for shared link orchestration, proposed as greedy quantum link orchestration heuristic while using minimum selection rule as benchmark \cite{kumar2026calibration}.

\item An instruction set architecture for programmable quantum repeater nodes based on NV centers, introducing deterministic and coherent programmability and linking hardware-level operations to network protocols \cite{kumar2026nvcenter}.
\end{itemize}

\section{Broader impact and outlook}
Beyond its immediate technical contributions, this work advances a \emph{design philosophy} for quantum networks that emphasises operability, adaptability, and realism. By treating hardware constraints as first-class design parameters rather than nuisances, the proposed design approach aligns closely with current experimental capabilities while remaining scalable toward future deployments.

The ideas developed here are directly relevant to emerging testbeds and provide a foundation for future work on multi-partite entanglement, fault-tolerant networking, and quantum native control planes.

\section{Organisation of the thesis}
The remainder of the thesis is organised as follows.

Chapter~\ref{ch:foundations} introduces the foundations of quantum networking and the key physical and protocol-level concepts required for the rest of the thesis.

Chapter~\ref{ch:applications_architecture} discusses the applications and architecture of the quantum internet, and the problem of entanglement distribution leading to the role of routing, forwarding, and scheduling in quantum networks.

Chapters~\ref{ch:routing} and \ref{ch:performance_evaluation} focus on routing in quantum networks, including heterogeneous repeater models, grey box routing, and extensive performance evaluation.

Chapter~\ref{ch:calibration} develops calibration-aware models for quantum link operation and presents optimal orchestration strategies for both linear and general network topologies.

Chapter~\ref{ch:programmable_repeaters} introduces programmable quantum repeater nodes and presents an instruction set architecture that connects hardware level control to network protocols.

Chapter~\ref{ch:crosslayer} discusses cross-layer insights and design principles derived from the preceding chapters.

Chapter~\ref{ch:conclusion} concludes the thesis and outlines open research directions toward large-scale programmable quantum networks. 
\chapter{Foundations o\label{ch:foundations}f Quantum Networking}
The quantum internet is fundamentally different from the classical internet due to the properties acquired from quantum mechanical phenomena. These acquired counterintuitive properties in the classical world provide the backbone of the quantum internet. In this chapter, we explain the basic concepts relevant to the quantum internet as quantum information preliminaries, quantum communication primitives, followed by a component-by-component explanation of the quantum internet.

\section{Quantum information preliminaries}\label{sec:prelim}
  \subsection{Qubits and quantum gates}
  \textbf{Qubit:} A qubit\footnote{While a qubit is a two-level quantum system, a generalised version with
n-levels is known as a \textit{qudit}. Despite having interesting properties that could be exploited for enhanced systems or applications, qudit technologies lag significantly behind their equivalent, that is, qubits \cite{wang2020qudits}.} is fundamental to quantum computing, similar to a bit in classical computing. While a classical bit on a normal computer can be in a state of 0 or 1, the state of a quantum bit is defined by
    \begin{equation}\label{eq:sample_psi}
        |\psi \rangle = \alpha |0 \rangle + \beta |1 \rangle.
    \end{equation}
    The state $|\psi \rangle $ is said to be in a linear combination of states or a \textit{superposition state} which collapses to 0 or 1 upon \textit{measurement} with a probability of $|\alpha|^2$ and $|\beta|^2$ respectively depending upon the \textit{probability amplitudes} that is, $\alpha$ and $\beta$, which are complex coefficients. Hence by definition $| \alpha |^2 + |\beta|^2 = 1$. 
    
The notation $|\psi \rangle$ used here is a standard for representing quantum mechanical states called \textit{bra-ket} notation introduced by Paul Dirac. Here the represented $|\psi \rangle$ is a \textit{ket} 
which is a vector in a complex vector space called \textit{Hilbert space} $\mathcal{H}$. A Hilbert space is a linear vector space with three additional properties that is, strictly positive scalar product, separability, and completeness \cite{zettili2009quantum}. For every ket $| \psi \rangle $, there exists a unique \textit{bra} $\langle \psi |$, which belongs to the corresponding \textit{dual-Hilbert space} $\mathcal{H}_d$. While the Eq.~\eqref{eq:sample_psi} above describes a single qubit system, multi-qubit systems can be represented by extending this notation. In general, an \( n \)-qubit system is described by
\begin{equation}
|\psi \rangle = \sum_{i=0}^{2^n-1} c_i\, |i\rangle,
\end{equation}
where \( |i\rangle \) represents the computational basis states corresponding to the binary representation of \( i \), and \( c_i \in \mathbb{C} \) are the complex probability amplitudes satisfying the normalisation condition
\begin{equation}
\sum_{i=0}^{2^n-1} |c_i|^2 = 1.
\end{equation}
For example, a two-qubit system is expressed as
\begin{equation}
    |\psi \rangle = \alpha |00 \rangle + \beta |01 \rangle + \gamma |10 \rangle + \delta |11 \rangle,
\end{equation} 
where $|\alpha|^2+|\beta|^2+|\gamma|^2+|\delta|^2=1$.

\vspace*{0.5cm}

\textbf{Quantum gates:} A quantum gate is a fundamental building block of quantum circuits that can manipulate a quantum state or qubit. Quantum gates, for example, $X$, Hadamard, Toffoli, $CZ$, etc., can be considered analogous to logic gates like AND, OR, XOR, etc, in classical computing. In quantum mechanical terms, a quantum gate is an \emph{operator} applied to a ket $|\psi \rangle = \alpha |0 \rangle + \beta |1 \rangle$ that changes it to a ket $|\psi^{'} \rangle = \alpha^{'} |0 \rangle + \beta^{'} |1 \rangle$. The quantum gates can be represented by square matrices which act linearly on the quantum states \cite{nielsen2010quantum}. For example, the quantum NOT gate, that is, $X$, is expressed as follows,
    \begin{equation}\label{eq:x_gate}
     X = \begin{pmatrix}
     0 & 1 \\
     1 & 0 
     \end{pmatrix},
    \end{equation}
and it has the same effect as its classical counterpart: the classical NOT gate flips a $0$ into a $1$ (and a $1$ into a $0$), whereas the quantum X gate flips a $|0\rangle$ into a $|1\rangle$ (and a $|1\rangle$ into a $|0\rangle$).
For a matrix $U$ to be a quantum gate, it must be unitary; that is, it should satisfy the condition $U^{\dagger}U = I$. Here, the notation $U^{\dagger}$ is called the \textit{adjoint}, that is, the transpose conjugate of the matrix $U$. This is necessary to maintain the sum of probabilities to be one after the operation of the quantum gate, that is, $| \alpha^{'} |^2 + |\beta^{'}|^2 = 1$.
    
A quantum gate can be categorised as an $n$-qubit gate where $n$ is the number of qubits the gate can operate on. The $X$ gate represented in Eq.~\eqref{eq:x_gate} is a 1-qubit gate. An example of a 2-qubit gate is $CZ$, which is represented as follows,
    \begin{equation}\label{eq:cz_gate}
        CZ = \begin{pmatrix}
         1 & 0 & 0 & 0 \\
         0 & 1 & 0 & 0 \\
         0 & 0 & 1 & 0 \\
         0 & 0 & 0 & -1
        \end{pmatrix}.
    \end{equation}
    In general, an $n$-qubit quantum gate corresponds to a matrix of dimension of $2^n$.

\subsection{Entanglement and Bell states}
Entanglement is an important quantum phenomenon that is core to the quantum internet. Till now, we have been introducing concepts with an analogy to their classical counterparts. But from now on, that approach would not be possible as the following concepts do not have a classical counterpart. Entanglement refers to a system prepared so that there are correlations among quantum bits or states independent of distance, famously dubbed by Einstein as \textit{spooky action at a distance}. In principle, entanglement can be established between any number of qubits known as Greenberger–Horne–Zeilinger (\textit{GHZ-state}) \cite{greenberger1989going, greenberger1990bell} and \textit{W-state} \cite{dur2000three} or even in multiple degrees of freedoms known as hyperentanglement\footnote{Due to the complexity of managing multiple degrees of freedom, the use of hyperentanglement for entanglement distribution within the quantum internet is rarely addressed. However, protocols involving teleportation, swapping, and purification—which could leverage hyperentanglement—are well-documented in the literature and remain an active area of research. For further reading on this topic, a relevant review article is available at \cite{deng2017quantum}.} \cite{kwiat1997hyper} as shown in Fig.~\ref{fig:entanglement_type}.
But for the sake of the introduction of the concept, let's consider the case of a two-qubit system with a single degree of freedom.
Throughout this work, our discussion will concentrate on the workings of the quantum internet using bipartite entanglement, which has been a primary focus in the literature. In particular, the quantum internet is concerned with the distribution of a special two-qubit entangled system, named \textit{EPR pair} \cite{einstein1935can} or \textit{Bell state} \cite{bell1964einstein}. An example of a Bell state is as follows,

\begin{figure}
\centering
\includegraphics[width=\columnwidth]{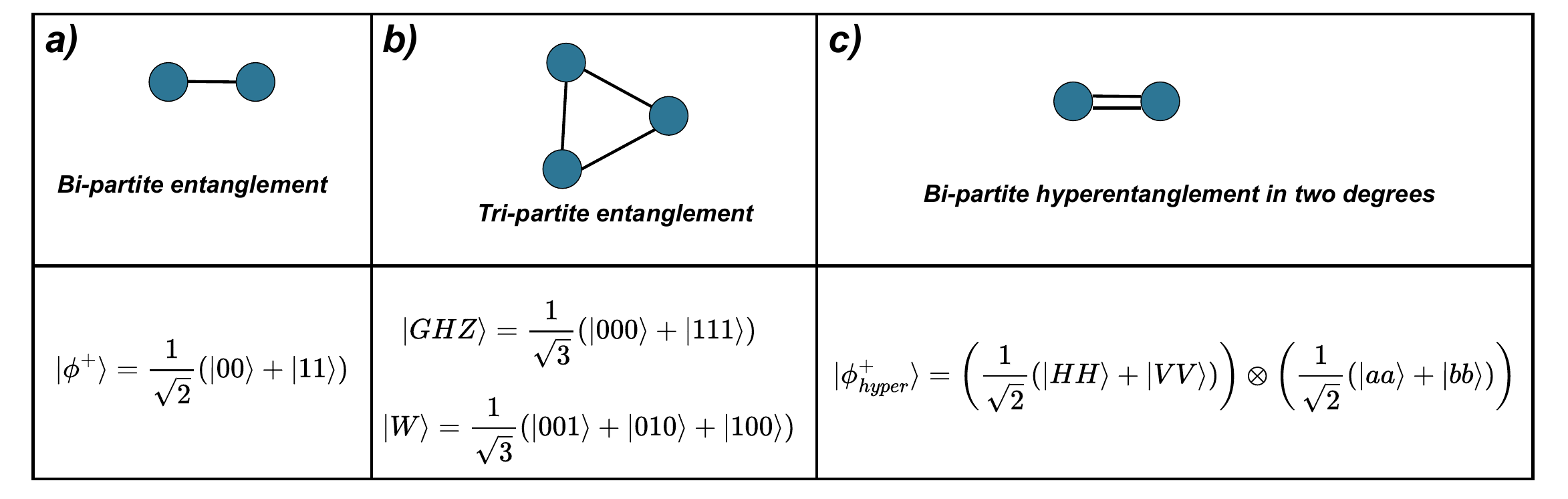}
\caption{Examples of different types of entanglement. \textbf{a)} a bipartite entangled state depicted by the Bell state; \textbf{b)} two forms of tripartite entanglement, one following the GHZ state and the other the W state; \textbf{c)} a bipartite hyperentanglement in two degrees of freedom, combining polarisation and spatial modes (taken from \cite{kumar2025quantum})}
    \label{fig:entanglement_type}
\end{figure}

    \begin{equation}\label{eq:bell_state}
        | \phi^+ \rangle = \frac{1}{\sqrt{2}} \left( |00 \rangle + |11 \rangle \right).
    \end{equation}
Upon measurement of the two-qubit state $|\phi^+ \rangle$, the whole system collapses to $|00 \rangle$ or $|11 \rangle$ with a 50\% probability. Independent of the state the system collapses to, the two-qubit measurement result is correlated: it is impossible to know beforehand the result of the measurement, but if the measurement corresponding to the first qubit is $0$, then the second qubit measurement will be $0$, and the same happens if the first qubit is measured as a $1$, in which case the second qubit will also be measured as $1$.
    
The other Bell states are as follows,
        \begin{equation}\label{eq:other_bell_states}
        \begin{split}
            | \phi^- \rangle = \frac{1}{\sqrt{2}} \left( |00 \rangle - |11 \rangle \right) \\
            | \psi^+ \rangle = \frac{1}{\sqrt{2}} \left( |01 \rangle + |10 \rangle \right) \\
            | \psi^- \rangle = \frac{1}{\sqrt{2}} \left( |01 \rangle - |10 \rangle \right).
        \end{split}
    \end{equation}

The \( + \) and \( - \) symbols in the Bell states indicate the relative phase between the two basis states that make up each Bell state. Bell states are important mainly for two reasons.
First, they are maximally entangled, which means informally that it is not possible to prepare any quantum state where the entanglement between two qubits is stronger.
The formal definition of this property would require introducing additional and unnecessary notation.
Therefore, we suggest interested readers seek further information in textbooks, such as~\cite{zettili2009quantum, nielsen2010quantum}.
Second, it is possible to transform a Bell state into any arbitrary state of choice via local operations only, that is, the application of quantum gates.
In summary, while one could distribute arbitrary quantum states in the quantum internet, doing so with Bell pairs only is an efficient and practical alternative, which is universally accepted by the research community.

  \subsection{Decoherence, Fidelity and No-cloning theorem}
\textbf{Decoherence:} In quantum information, \emph{decoherence} refers to the gradual loss of ``quantum-ness" of a quantum state, usually caused by its interaction with the environment. This disruption occurs because the quantum system loses information
to its surroundings. Over time or due to specific interactions, such as qubit measurements or quantum gates, this loss of coherence leads the quantum state to behave more like a classical system, losing the features that made it quantum \cite{schlosshauer2019quantum}.

For the quantum internet, we can think of decoherence as the natural degradation of quantum information as it travels or undergoes operations. This loss of coherence means that information can become noisy or unusable, and it places limits on how long and how far we can reliably use a quantum state for communication or computation. Addressing and minimising decoherence is essential to maintaining the quality of information in the quantum internet, as even small disturbances from the environment can lead to a breakdown in the system’s ability to preserve entanglement and other quantum correlations.


\textbf{Fidelity:} It is a metric which quantifies the closeness of two quantum states. Let $|\psi_i \rangle$ and $|\psi_j \rangle$ be two pure quantum states\footnote{A \emph{pure quantum state} is a state that can be described by a single state vector in the Hilbert space, representing a system with complete information and no statistical mixture. In contrast, a \emph{mixed state} must be described by a density matrix, representing a statistical ensemble of pure states.} then the fidelity $F_{ij}$ of these states is given by,
    \begin{equation}\label{eq:fidelity_bra_ket}
        F_{ij} = | \langle \psi_i | \psi_j \rangle|^2,
    \end{equation}
where $\langle\psi_i|\psi_j\rangle$ is the inner product between the states. Note that Eq.~\eqref{eq:fidelity_bra_ket} applies specifically to pure quantum states. For general quantum states (including mixed states), fidelity is defined using density matrices as $F(\rho, \sigma) = \left(\text{Tr}\sqrt{\sqrt{\rho}\sigma\sqrt{\rho}}\right)^2$.
 
The fidelity always satisfies $0 \leq F \leq 1$ regardless of whether the states are pure or mixed, where $F = 0$ indicates orthogonal states (maximally distinguishable) and $F = 1$ indicates identical states. For two identical mixed states, $F = 1$ still holds.


\textbf{No cloning theorem:} The \emph{no-cloning theorem} states that it is not possible to copy an unknown quantum state. More formally, there does not exist a universal quantum operation (unitary transformation) that takes any arbitrary quantum state $|\psi \rangle $ along with a standard ``blank'' state $|e \rangle$ and produces two copies of $|\psi \rangle$, that is, there is no unitary operation $U$ such that for all $|\psi \rangle$ \cite{wootters1982single},
    \begin{equation}
        U(|\psi\rangle \otimes |e\rangle) = |\psi\rangle \otimes |\psi\rangle.
    \end{equation}

This theorem is one of the fundamental concepts that distinguishes the quantum internet from the classical internet. In classical systems, copying information is extensively utilised for various purposes, including assisting data transmission, error correction, caching, load balancing, and ensuring backup \& redundancy.

\section{Quantum communication primitives}
  \subsection{Teleportation}
  Teleportation of a quantum state refers to sending an arbitrary quantum state to an arbitrary distance using quantum correlations \cite{bennett1993teleporting}. In Fig.~\ref{fig:teleport_swap_chain}a Alice wants to send an arbitrary quantum state $|\psi \rangle_0^{A_1} = \alpha |0 \rangle + \beta |1 \rangle$ to Bob. Let Alice and Bob initially have an EPR pair $| \phi^+ \rangle^{A_2 B}$ shared among them. Then, the state of the collective three-qubit system can be written as follows,

\begin{equation}
        | \psi \rangle_0^{A_1} | \phi^+ \rangle^{A_2 B} = \left( \alpha |0 \rangle + \beta |1 \rangle \right)^{A_1} \ \frac{1}{\sqrt{2}} \left( |00 \rangle + |11 \rangle \right)^{A_2 B},
    \end{equation}
    \begin{equation}
        \Rightarrow | \psi \rangle_0^{A_1} | \phi^+ \rangle^{A_2 B} = \frac{1}{\sqrt{2}} (\alpha |000 \rangle + \alpha |011 \rangle + \beta | 100 \rangle + \beta|111\rangle)^{A_1 A_2 B}.
    \end{equation}
Using the definitions of the Bell states in Eq.~\eqref{eq:bell_state} and \eqref{eq:other_bell_states} we have:
\begin{equation}
\begin{split}
     | \psi \rangle_0^{A_1} | \phi^+ \rangle^{A_2 B} = \frac{1}{2} (\alpha (|\phi^+ \rangle + |\phi^- \rangle) |0 \rangle + \alpha (|\psi^+ \rangle + |\psi^- \rangle) | 1 \rangle + \\
    \beta (|\psi^+\rangle -|\psi^- \rangle)|0\rangle + \beta (|\phi^+ \rangle - |\phi^-\rangle)|1\rangle)^{A_1 A_2 B},
    \end{split}
\end{equation}
    \begin{equation} \label{eq:teleport_expanded}
    \begin{split}
        \Rightarrow | \psi \rangle_0^{A_1} | \phi^+ \rangle^{A_2 B} = \frac{1}{2} | \phi^+ \rangle^{A_1 A_2} \left( \alpha |0 \rangle + \beta |1  \rangle \right)^B + \frac{1}{2} | \phi^- \rangle^{A_1 A_2} \left( \alpha |0 \rangle - \beta |1  \rangle \right)^B \\
        + \frac{1}{2} | \psi^+ \rangle^{A_1 A_2} \left( \alpha |1 \rangle + \beta |0  \rangle \right)^B + \frac{1}{2} | \psi^- \rangle^{A_1 A_2} \left( \alpha |1 \rangle - \beta |0  \rangle \right)^B,
    \end{split}
    \end{equation}
From Eq.~\eqref{eq:teleport_expanded}, a simultaneous Bell state measurement by Alice that is, on qubits $A_1$ and $A_2$ would result in either of four bell states that is, $|\phi^+ \rangle^{A_1 A_2}, |\phi^- \rangle^{A_1 A_2}, |\psi^+ \rangle^{A_1 A_2}$, and $|\psi^- \rangle^{A_1 A_2}$. Meanwhile the corresponding state of Bob's qubit that is, $B$ would be in the state: $\left( \alpha |0 \rangle + \beta |1 \rangle \right)^B, \left( \alpha |0 \rangle - \beta |1 \rangle \right)^B, \left( \alpha |1 \rangle + \beta |0 \rangle \right)^B$, and $\left( \alpha |1 \rangle - \beta |0 \rangle \right)^B$. The measurement results obtained by Alice are communicated to Bob via a classical channel. Depending upon these measurement results Bob would perform a single-qubit gate that is, $I$ for $| \phi^+ \rangle$, $\sigma_z$ for $|\phi^- \rangle$, $\sigma_x$ for $| \psi^+ \rangle$, and $\sigma_x \sigma_z$ for $|\psi^- \rangle$ which gets the state of Bob's qubit transformed to the state that was intended to teleport by Alice \cite{weinfurter1994experimental}. This concludes the teleportation procedure as summarised in table~\ref{tab:bell_state_analysis}.
The net effect is that the source quantum state has been transferred to a remote party through the consumption of a Bell state shared by the parties and with the collapse of the origin state.
Note that this procedure does not violate speed-of-light constraints on the transfer of matter or information, as it relies on both the pre-distribution of a Bell state between Alice and Bob and the transmission of measurement results via a classical communication channel.

\begin{table}[h!]
\centering
\caption{\textbf{Bell state analysis for Teleportation:} Bell state measurement at Alice's qubits $A_1$ and $A_2$ \textit{\textbf{(column: $A_1A_2$ results)}} collapses (projects) the Bob's qubit $B$ to several different states \textit{\textbf{(column: $B$ result)}}. Hence, a single-qubit gate(s) is applied at Bob's qubit $B$ \textit{\textbf{(column: applied single-qubit gate at Bob ($B$))}} to arrive at correct final state of Bob's qubit $B$ \textit{\textbf{(column: final state at Bob ($B$))}}.}
\label{tab:bell_state_analysis}
\begin{tabular}{llll}
\hline
\textbf{$A_1 A_2$ results} & \textbf{$B$ result} & \textbf{applied single-qubit gate at Bob ($B$)} & \textbf{final state at Bob ($B$)} \\
\hline
$|\phi^+ \rangle^{A_1 A_2}$ & $\left( \alpha |0 \rangle + \beta |1 \rangle \right)^B$ & $I$ & $\left( \alpha |0 \rangle + \beta |1 \rangle \right)^B$ \\
$|\phi^- \rangle^{A_1 A_2}$ & $\left( \alpha |0 \rangle - \beta |1 \rangle \right)^B$ & $\sigma_z$ & $\left( \alpha |0 \rangle + \beta |1 \rangle \right)^B$ \\
$|\psi^+ \rangle^{A_1 A_2}$ & $\left( \alpha |1 \rangle + \beta |0 \rangle \right)^B$ & $\sigma_x$ & $\left( \alpha |0 \rangle + \beta |1 \rangle \right)^B$ \\
$|\psi^- \rangle^{A_1 A_2}$ & $\left( \alpha |1 \rangle - \beta |0 \rangle \right)^B$ & $\sigma_x \sigma_z$ & $\left( \alpha |0 \rangle + \beta |1 \rangle \right)^B$ \\
\hline
\end{tabular}
\end{table}

\begin{figure}
\centering
\includegraphics[width=\columnwidth]{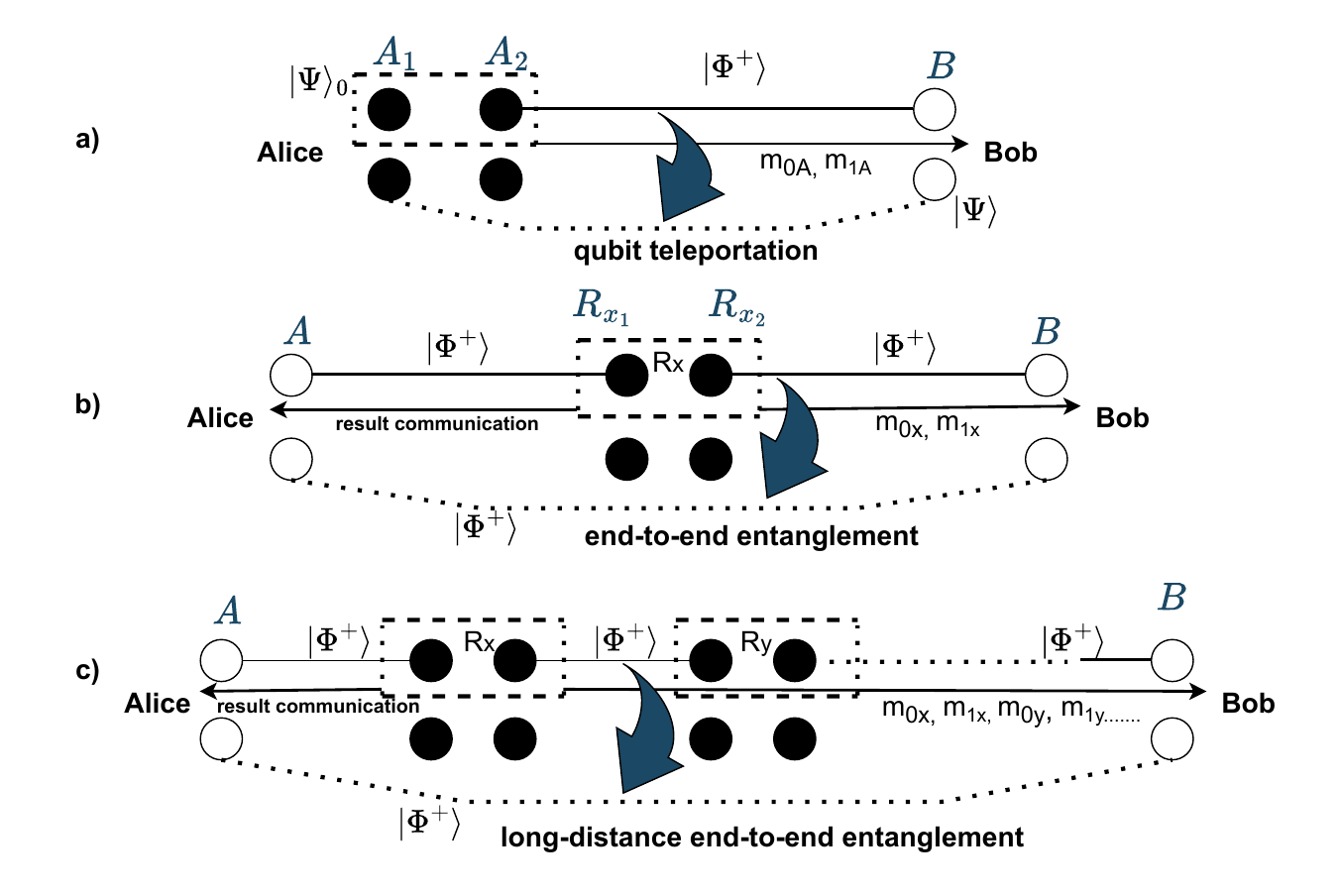}
\caption{\textbf{a) Teleportation:} Teleportation of Alice's qubit \(A_1\) from state \(|\Psi\rangle_0\) to Bob's qubit \(B\) (yielding state \(|\Psi\rangle\) with fidelity \( |\langle \Psi_0 | \Psi \rangle|^2 \)) is achieved via a Bell state measurement on qubits \(A_1\) and \(A_2\). Alice then communicates the resulting classical bits (\(m_{0A}\) and \(m_{1A}\)) to Bob, who applies the appropriate quantum operations on \(B\) based on these outcomes. \textbf{b) Entanglement Swapping:} Entanglement swapping between Alice and Bob is achieved by converting two bipartite entangled pairs—one between Alice and repeater \(R_x\) (qubits \(A\) and \(R_{x_1}\)) and one between repeater \(R_x\) and Bob (qubits \(R_{x_2}\) and \(B\))—into a direct (end-to-end) entanglement between Alice and Bob (qubits \(A\) and \(B\)). This process is executed by performing a Bell state measurement on \(R_x\)'s qubits \(R_{x_1}\) and \(R_{x_2}\), communicating the resulting classical bits (\(m_{0A}\) and \(m_{1A}\)) to Bob, who then applies the appropriate quantum operations on his qubit \(B\). \textbf{c) A linear network:} End-to-end bipartite entanglement between two distant end nodes, Alice and Bob (qubits \(A\) and \(B\)), is established via a cascade of entanglement swapping along a linear chain of quantum repeaters. At each repeater, a Bell state measurement is performed, and the corresponding classical bits (\(m_{0i}\) and \(m_{1i}\)) from the \(i^{th}\) quantum repeater are communicated to Bob, who applies the appropriate quantum operations on his qubit \(B\). (modified from \cite{kumar2025routing}).}
    \label{fig:teleport_swap_chain}
\end{figure}

\subsection{Entanglement swapping}\label{ssec:swapping}
Entanglement swapping can be considered as an extended version of the teleportation procedure above: entanglement swapping essentially \textit{swaps} two EPR pairs distributed at a shorter distance with a single EPR distributed at a longer distance \cite{zukowski1993event}. In contrast, the teleportation procedure is used to teleport an arbitrary quantum state, as shown above. In Fig.~\ref{fig:teleport_swap_chain}b, two EPR pairs $|\phi^+ \rangle$ are initially distributed over three stations, that is, Alice, Repeater $R_x$, and Bob. Then, the state of the collective four-qubit system can be written as follows,
    \begin{equation}
        |\phi^+ \rangle^{A R_{x_1}} |\phi^+ \rangle^{R_{x_2} B} = \frac{1}{\sqrt{2}} \left( |00 \rangle + |11 \rangle \right)^{A R_{x_1}} \ \frac{1}{\sqrt{2}} \left( |00 \rangle + |11 \rangle \right)^{R_{x_2} B},
    \end{equation}
    \begin{equation}
        \Rightarrow |\phi^+ \rangle^{A R_{x_1}} |\phi^+ \rangle^{R_{x_2} B} = \frac{1}{2} (|0000 \rangle + |0011 \rangle + |1100 \rangle + |1111 \rangle)^{A R_{x_1} R_{x_2} B}.
    \end{equation}
    Using the definitions of the Bell states in Eq.~\eqref{eq:bell_state} and \eqref{eq:other_bell_states}, we have:
    \begin{equation}
        \begin{split}
            \Rightarrow |\phi^+ \rangle^{A R_{x_1}} |\phi^+ \rangle^{R_{x_2} B} = \frac{1}{2 \sqrt{2}} \left[ |0\rangle (|\phi^+ \rangle + |\phi^- \rangle) |0 \rangle + |0 \rangle (| \psi^+ \rangle +| \psi^- \rangle)|1 \rangle + \right. \\
            \left. |1 \rangle (|\psi^+ \rangle - |\psi^- \rangle) |0 \rangle + |1 \rangle (|\phi^+ \rangle - |\phi^- \rangle)|1 \rangle \right]^{A R_{x_1} R_{x_2} B},
        \end{split}
    \end{equation}

        \begin{equation}
        \begin{split}
            \Rightarrow |\phi^+ \rangle^{A R_{x_1}} |\phi^+ \rangle^{R_{x_2} B} = \frac{1}{2 \sqrt{2}} \left[ (|\phi^+ \rangle + |\phi^- \rangle) |00 \rangle + (| \psi^+ \rangle +| \psi^- \rangle)|01 \rangle + \right. \\
            \left. (|\psi^+ \rangle - |\psi^- \rangle) |10 \rangle + (|\phi^+ \rangle - |\phi^- \rangle)|11 \rangle \right]^{ R_{x_1} R_{x_2} A B},
        \end{split}
    \end{equation}
\begin{equation}\label{eq:swap_expanded}
\begin{split}
    \Rightarrow |\phi^+ \rangle^{A R_{x_1}} |\phi^+ \rangle^{R_{x_2} B} = \frac{1}{2} \left[
    | \phi^+ \rangle^{R_{x_1} R_{x_2}} \frac{1}{\sqrt{2}}\left( |00 \rangle + |11 \rangle \right)^{A B} \right. \\
    \left. + | \phi^- \rangle^{R_{x_1} R_{x_2}} \frac{1}{\sqrt{2}}\left( |00 \rangle - |11 \rangle \right)^{A B} \right. \\
    \left. + | \psi^+ \rangle^{R_{x_1} R_{x_2}} \frac{1}{\sqrt{2}}\left( |01 \rangle + |10 \rangle \right)^{A B} \right. \\
    \left. + | \psi^- \rangle^{R_{x_1} R_{x_2}} \frac{1}{\sqrt{2}}\left( |01 \rangle - |10 \rangle \right)^{A B} \right],
\end{split}
\end{equation}
where the Bell basis are given by Eq.~\eqref{eq:bell_state} and Eq.~\eqref{eq:other_bell_states} similar to the teleportation procedure.

From Eq.~\eqref{eq:swap_expanded}, a simultaneous bell state measurement at repeater station $R_x$ that is, on qubits $R_{x_1}$ and $R_{x_2}$ would result in either of four Bell states that is, \\ $|\phi^+ \rangle^{R_{x_1} R_{x_2}}, |\phi^- \rangle^{R_{x_1} R_{x_2}}, |\psi^+ \rangle^{R_{x_1} R_{x_2}}$, and $|\psi^- \rangle^{R_{x_1} R_{x_2}}$. Meanwhile, the collective state of qubits of Alice and Bob that is, $A$ and $B$ would be in the state: \\ $ \frac{1}{\sqrt{2}}\left( \alpha |00 \rangle + \beta |11 \rangle \right)^{AB}, \frac{1}{\sqrt{2}}\left( \alpha |00 \rangle - \beta |11 \rangle \right)^{AB}, \frac{1}{\sqrt{2}}\left( \alpha |01 \rangle + \beta |10 \rangle \right)^{AB}$, and \\ $\frac{1}{\sqrt{2}}\left( \alpha |01 \rangle - \beta |10 \rangle \right)^{AB}$. The measurement results obtained at repeater station $R_x$ are communicated to Alice or Bob via a classical channel.
Depending upon these measurement results a single-qubit gate that is, $I$ for $| \phi^+ \rangle$, $\sigma_z$ for $|\phi^- \rangle$, $\sigma_x$ for $| \psi^+ \rangle$, and $\sigma_x \sigma_z$ for $|\psi^- \rangle$ which gets the collective state of Alice and Bob qubit transformed to an EPR pair $| \phi^+ \rangle$. This concludes the entanglement-swapping procedure as summarised in table~\ref{tab:bsa_swap}.

\begin{table}[h!]
\centering
\caption{\textbf{Bell state analysis for Entanglement swapping:} Bell state measurement at quantum repeater's ($R_x$) qubits $R_{x_1}$ and $R_{x_2}$ \textit{\textbf{(column: $R_{x_1}R_{x_2}$ results)}} collapses (projects) the combined Alice's qubit $A$ and Bob's qubit $B$ to several different bell states \textit{\textbf{(column: $AB$ result)}}. Hence, a single-qubit gate(s) is applied at Bob's qubit $B$ \textit{\textbf{(column: applied single-qubit gate at Bob (B))}} to arrive at correct end-to-end entanglement state between Alice's qubit $A$ and Bob's qubit $B$ \textit{\textbf{(column: final state ($AB$))}}.}
\label{tab:bsa_swap}
\begin{tabular}{llll}
\hline
\textbf{$R_{x_1} R_{x_2}$} \textbf{results} & \textbf{$AB$ result} & \textbf{applied single-qubit gate at Bob ($B$)} & \textbf{final state ($AB$)} \\
\hline
$|\phi^+ \rangle^{R_{x_1} R_{x_2}}$ & $\frac{1}{\sqrt{2}}\left( \alpha |00 \rangle + \beta |11 \rangle \right)^{AB}$ & $I$ & $\frac{1}{\sqrt{2}}\left( \alpha |00 \rangle + \beta |11 \rangle \right)^{AB}$ \\
$|\phi^- \rangle^{R_{x_1} R_{x_2}}$ & $\frac{1}{\sqrt{2}}\left( \alpha |00 \rangle - \beta |11 \rangle \right)^{AB}$ & $\sigma_z$ & $\frac{1}{\sqrt{2}}\left( \alpha |00 \rangle + \beta |11 \rangle \right)^{AB}$ \\
$|\psi^+ \rangle^{R_{x_1} R_{x_2}}$ & $\frac{1}{\sqrt{2}}\left( \alpha |01 \rangle + \beta |10 \rangle \right)^{AB}$ & $\sigma_x$ & $\frac{1}{\sqrt{2}}\left( \alpha |00 \rangle + \beta |11 \rangle \right)^{AB}$ \\
$|\psi^- \rangle^{R_{x_1} R_{x_2}}$ & $\frac{1}{\sqrt{2}}\left( \alpha |01 \rangle - \beta |10 \rangle \right)^{AB}$ & $\sigma_x \sigma_z$ & $\frac{1}{\sqrt{2}}\left( \alpha |00 \rangle + \beta |11 \rangle \right)^{AB}$ \\
\hline
\end{tabular}
\end{table}

Entanglement swapping is the basic concept implemented by quantum repeaters, introduced below, which are the fundamental building blocks of the quantum internet.

\subsection{Entanglement purification}\label{ssec:purification}
The entanglement purification procedure uses \textit{sacrificial EPR pairs} of low-fidelity to attain higher-fidelity EPR pairs or entanglement. For example, as shown in Fig.~\ref{fig:purification_errorcorrection}a, eight EPR pairs (fidelity $F_0$) are consumed over three layers of purification to extract a single higher fidelity EPR pair (fidelity $F_3$).
For simplicity, we assume that the 1 $\&$ 2-bit operations always result in the success of the purification protocol.

Let the initial EPR pair that is, $|\phi^+ \rangle = \frac{1}{\sqrt{2}} \left( |00 \rangle + |11 \rangle \right)$ with qubits 1 and 2 be distributed by two neighbouring nodes via a noisy quantum channel, which degrades the quality of the original \emph{pure} quantum states, hence diminishing their fidelity compared to the original Bell state and leading to a \emph{mixed}, that is, non-pure, state. The resulting mixed Bell state after exposure to this channel can be written using the following Werner's state \cite{werner1989quantum}:
\begin{equation}
    \rho = F | \phi^+ \rangle \langle \phi^+ | + \frac{1-F}{3} \left( | \phi^- \rangle \langle \phi^- | + | \psi^+ \rangle \langle \psi^+ | + | \psi^- \rangle \langle \psi^- |
    \right).
\end{equation}
This means that the probability of finding the initial EPR pair ($\rho$) shared on exposure to a noisy quantum channel intact with respect to bell state $| \phi^+ \rangle$ is $F$. Meanwhile, the probability of finding any other state is $1-F$.

To purify this EPR pair, we share another EPR pair with qubits 3 and 4 of the same Bell state as before. If node A has qubits 1 and 3 while node B has qubits 2 and 4, then CNOT quantum gates with qubits 1 and 2 as sources and qubits 3 and 4 as targets are operated. Node A measures qubit 3, and Node B measures qubit 4, and they exchange their measured results through the classical channel (\textit{two-way signalling}).
If the measurement results match, then the EPR pair with qubits 1 and 2 is kept, or else it is discarded. Due to the measurement, the EPR pair is no longer entangled as soon as qubits 3 and 4 are measured. However, by using this \textit{sacrificial EPR pair} the fidelity of retained EPR pair is given by \cite{bennett1996purification}:
\begin{equation}
    F_1 = \frac{F^2 + \frac{1}{9} (1-F)^2}{F^2 + \frac{2}{3}F(1-F) + \frac{5}{9}(1-F)^2}.
\end{equation}
It is to be noted that $F_1 > F$ for only $F > 0.5$.

The purification protocol discussed above represents one of the initial approaches to this technology. Subsequent protocols have introduced various enhancements. A recent survey detailing advancements in entanglement purification is available in \cite{yan2023advances}.

Entanglement purification has no equivalent in classical digital systems because the data stored or transferred are either fully correct (a 0 is a 0, a 1 is a 1) or flipped due to errors (a 0 is a 1, a 1 is a 0).

\begin{figure}
\centering
\includegraphics[width=0.9\columnwidth]{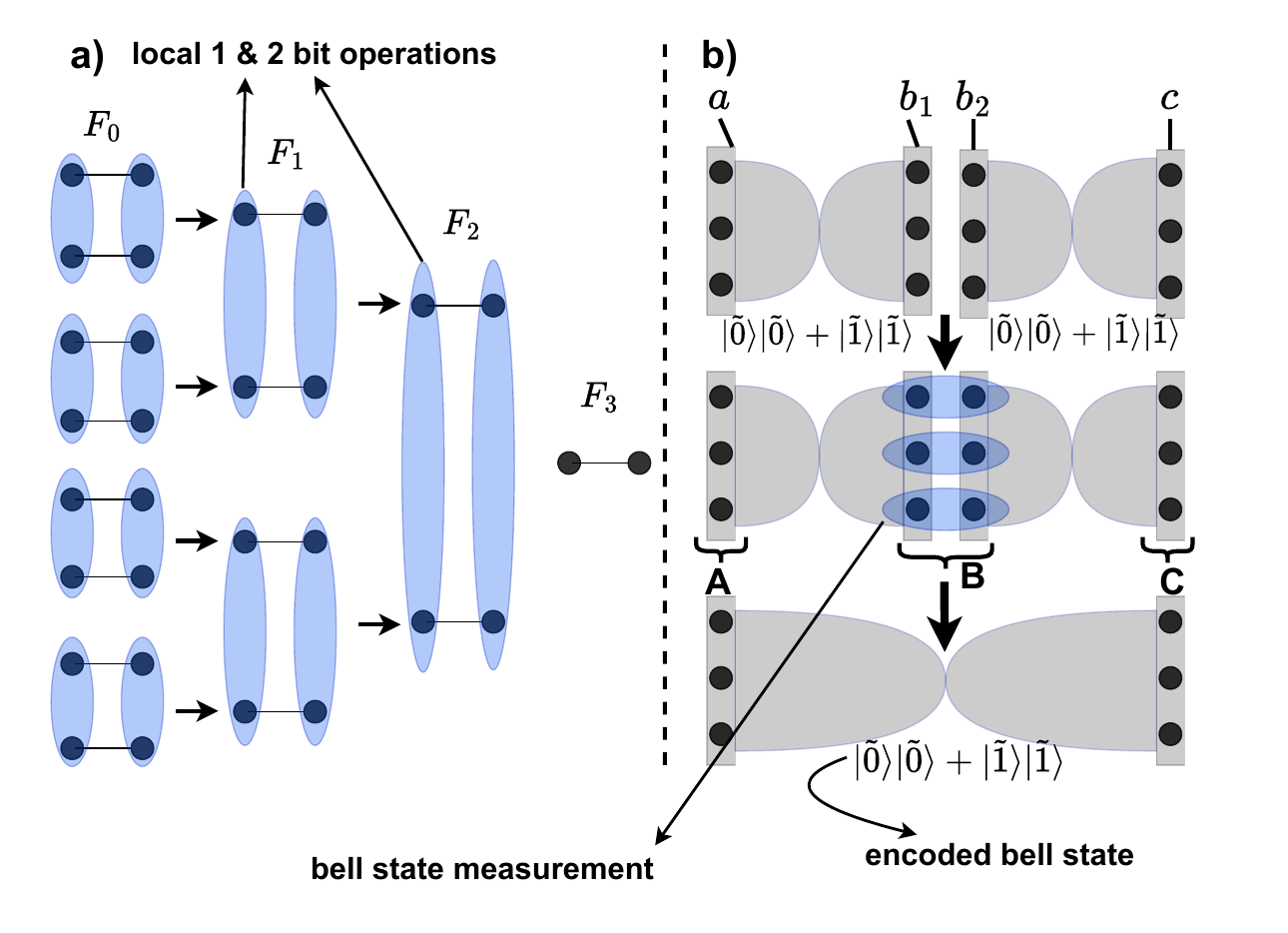}
\caption{
Using extra qubit resources to boost error resilience. \textbf{a) Entanglement purification:} Each EPR pair is paired with an additional, \textit{sacrificial} EPR pair to enhance its fidelity (modified from \cite{dur1999quantum}); \textbf{b) Quantum error correction:} Each EPR pair is supported by extra qubits arranged as \textit{encoded} EPR pairs to reduce the error rate (modified from \cite{jiang2009quantum})}
    \label{fig:purification_errorcorrection}
\end{figure}

\subsection{Quantum error correction}
Quantum error correction (QEC) is a critical technique used to encode quantum states in a way that makes them robust against errors, analogous to Forward Error Correction (FEC) in classical systems. In classical FEC, redundancy is introduced to recover information lost due to errors during storage (e.g., on hard drives) or transmission (e.g., over noisy wireless links). To illustrate quantum error correction in the context of quantum networks, consider the example of entanglement swapping using encoded Einstein-Podolsky-Rosen (EPR) pairs.

Imagine a three-node setup: \textbf{Node A} (with a qubit array \( a \)), \textbf{Node B} (with qubit arrays \( b_1 \) and \( b_2 \)), and \textbf{Node C} (with a qubit array \( c \)), as depicted in Fig.~\ref{fig:purification_errorcorrection}b. This scenario builds upon the previously discussed entanglement swapping procedure (Fig.~\ref{fig:teleport_swap_chain}b) but incorporates quantum error correction by using \textit{logical} qubits instead of \textit{physical} qubits.

A \textit{logical qubit} is an abstract representation of quantum information encoded within a set of physical qubits, making it resistant to errors through redundancy. In this setup, each qubit in a standard entanglement swapping protocol is replaced by an array of \( n \) physical qubits utilising an $n$-qubit repetition code. In our case (Fig.~\ref{fig:purification_errorcorrection}b), we have \( n = 3 \), which corresponds to the use of a three-qubit repetition code.

In the broader context of error correction, it is also instructive to introduce the general notation $[[n,k,d]]$, which succinctly characterises an error-correcting code. Here, $n$ denotes the total number of physical qubits used in the encoding, $k$ represents the number of logical qubits encoded, and $d$ is the code distance, the minimum number of physical qubit errors required to cause an undetectable error. This distance determines the error-correcting capability of the code; specifically, a code with distance $d$ can detect up to $d-1$ errors and correct up to $\lfloor (d-1)/2 \rfloor$ errors. This notation mirrors the classical code notation \([n,k,d]\). In our example utilising the three-qubit repetition code, the parameters can be expressed as $[[3,1,3]]$; a single logical qubit is redundantly encoded into three physical qubits, and the code is capable of correcting any single-qubit error.

Instead of standard EPR pairs, entanglement swapping with error correction employs \textit{encoded EPR pairs}, requiring \( n = 3 \) physical qubits for each logical qubit. For example, as shown in Fig.~\ref{fig:purification_errorcorrection}b, encoded EPR pairs are established between \textbf{Node A} and \textbf{Node B} (\( a \) and \( b_1 \)) and between \textbf{Node B} and \textbf{Node C} (\( b_2 \) and \( c \)). These encoded EPR pairs are represented as:

\[
| \Tilde{\phi}^+ \rangle_{ab_1} = \frac{1}{\sqrt{2}} \left( |\Tilde{0} \rangle |\Tilde{0} \rangle + |\Tilde{1} \rangle |\Tilde{1} \rangle \right)_{ab_1}
\]
and
\[
| \Tilde{\phi}^+ \rangle_{b_2c} = \frac{1}{\sqrt{2}} \left( |\Tilde{0} \rangle |\Tilde{0} \rangle + |\Tilde{1} \rangle |\Tilde{1} \rangle \right)_{b_2c}.
\]

Here, each logical qubit is encoded using the three-qubit repetition code, where \( |\Tilde{0} \rangle = |000 \rangle \) and \( |\Tilde{1} \rangle = |111 \rangle \). Establishing a single encoded EPR pair between two nodes involves a three-step process and requires an additional ancilla qubit for each physical qubit in the encoded EPR pair, as described in \cite{gottesman1999demonstrating, zhou2000methodology}. This process is known as the \textit{encoded generation} of EPR pairs.

Analogous to standard entanglement swapping, the combined state of the four logical qubits can be expressed as \cite{jiang2009quantum}:

\begin{equation}\label{eq:error_correction}
\begin{split}
    |\Tilde{\phi}^+ \rangle_{ab_1} |\Tilde{\phi}^+ \rangle_{b_2c} = \frac{1}{2} \Big( 
|\Tilde{\phi}^+ \rangle_{ac} |\Tilde{+} \rangle_{b_1} |\Tilde{0} \rangle_{b_2} +  |\Tilde{\phi}^- \rangle_{ac} |\Tilde{-} \rangle_{b_1} |\Tilde{0} \rangle_{b_2} \\ 
+  |\Tilde{\psi}^+ \rangle_{ac} |\Tilde{-} \rangle_{b_1} |\Tilde{0} \rangle_{b_2} + |\Tilde{\psi}^- \rangle_{ac} |\Tilde{-} \rangle_{b_1} |\Tilde{1} \rangle_{b_2} \Big).
\end{split}
\end{equation}

This formulation enables a modified form of entanglement swapping, where Bell-state measurements are performed on logical qubits (groups of physical qubits). Specifically, the logical qubit \( b_1 \) at \textbf{Node B} is measured in the \( \{ |\Tilde{+} \rangle, |\Tilde{-} \rangle \} \) basis, while \( b_2 \) is measured in the \( \{ |\Tilde{0} \rangle, |\Tilde{1} \rangle \} \) basis. This step is referred to as the \textit{encoded connection}.

The results of these measurements yield a two-bit classical message, which is used to establish the encoded EPR pair between \textbf{Node A} and \textbf{Node C}.

Notably, these measurement results do not need to be communicated to other nodes, as is required in standard entanglement-swapping protocols. Instead, they are utilised locally at \textbf{Node B} to determine the state of the encoded qubit.

\section{Components of Quantum Internet}\label{sec:components_of_qi}
While a fully functional quantum internet or network is not yet available, current understanding allows us to identify its key components. As the field advances, additional components may be introduced. In this discussion, we outline the essential components required for the operation of a quantum internet based on our present knowledge.

\textbf{Quantum Repeaters:} The quantum repeater\footnote{An alternative approach to entanglement distribution, utilising percolation theory, has been proposed \cite{perseguers2008entanglement, das2018robust}. However, the prevailing focus in current literature has been on the development of quantum repeaters rather than the percolation theory approach. Nonetheless, the percolation approach continues to be studied and may represent a potential area for future investigation.} is a fundamental component of the quantum internet, playing a crucial role in the entanglement swapping process, which extends entanglement over long distances, as illustrated in Fig.~\ref{fig:teleport_swap_chain}c. In a linear chain of repeaters connected by shared EPR pairs, simultaneous Bell state measurements at the quantum repeaters entangle the end nodes. This is accomplished by communicating the classical results to one of the end nodes and applying the appropriate single-qubit gate at that node. It should be noted that this process is a generalised version of entanglement swapping, involving multiple repeater nodes instead of just one. 

A quantum repeater may be equipped with a \emph{quantum memory}, that is, a device that can store quantum states as qubits for a limited amount of time, and may also include a \emph{quantum processor}, that is, equipment for the execution of local operations through quantum gates. The requirement for quantum memory in a quantum repeater depends on the intended application and the entanglement distribution protocols employed by the quantum network. For instance, in prepare-and-measure quantum key distribution (QKD) networks, quantum memories are not necessary, whereas more advanced quantum applications, such as distributed quantum computing, rely on quantum memory to function effectively.
From a protocol perspective, the need for quantum memory also varies. In networks where entanglement generation and swapping proceed sequentially or with minimal delay to directly entangle end nodes, quantum memory support may not be essential. However, if the protocol is designed such that, following initial entanglement generation, some links must wait for entanglement swapping to proceed, then quantum memory support becomes crucial. Quantum memory is also indispensable for protocols that involve purification or quantum error correction.
Similarly, the inclusion of a quantum processor depends on the network protocols. If certain protocols are implemented that leverage a quantum processor to optimise the performance of the quantum network, then it becomes necessary. In the absence of such protocols, a quantum processor may not be required. An \emph{entanglement generator} can be part of a quantum repeater if the entanglement generation scheme follows either the \textit{MeetInTheMiddle} or \textit{SenderReceiver} configuration \cite{jones2016design}. In the MeetInTheMiddle scheme, entanglement generation occurs at the outer nodes, with one qubit of the entangled pairs from each outer node sent to an intermediate station where entanglement swapping takes place. In contrast, in the SenderReceiver scheme, entanglement generation also occurs at the outer nodes, as in the previous case. However, one qubit of the entangled pair from one outer node is sent directly to the other outer node, eliminating the need for an intermediate station. For the \textit{MidpointSource} scheme, a separate entanglement generation component positioned at the midpoint of a quantum network link would be required, as shown in Fig.~\ref{fig:entanglement_generation}.

\begin{figure}
\centering
\includegraphics[width=\columnwidth]{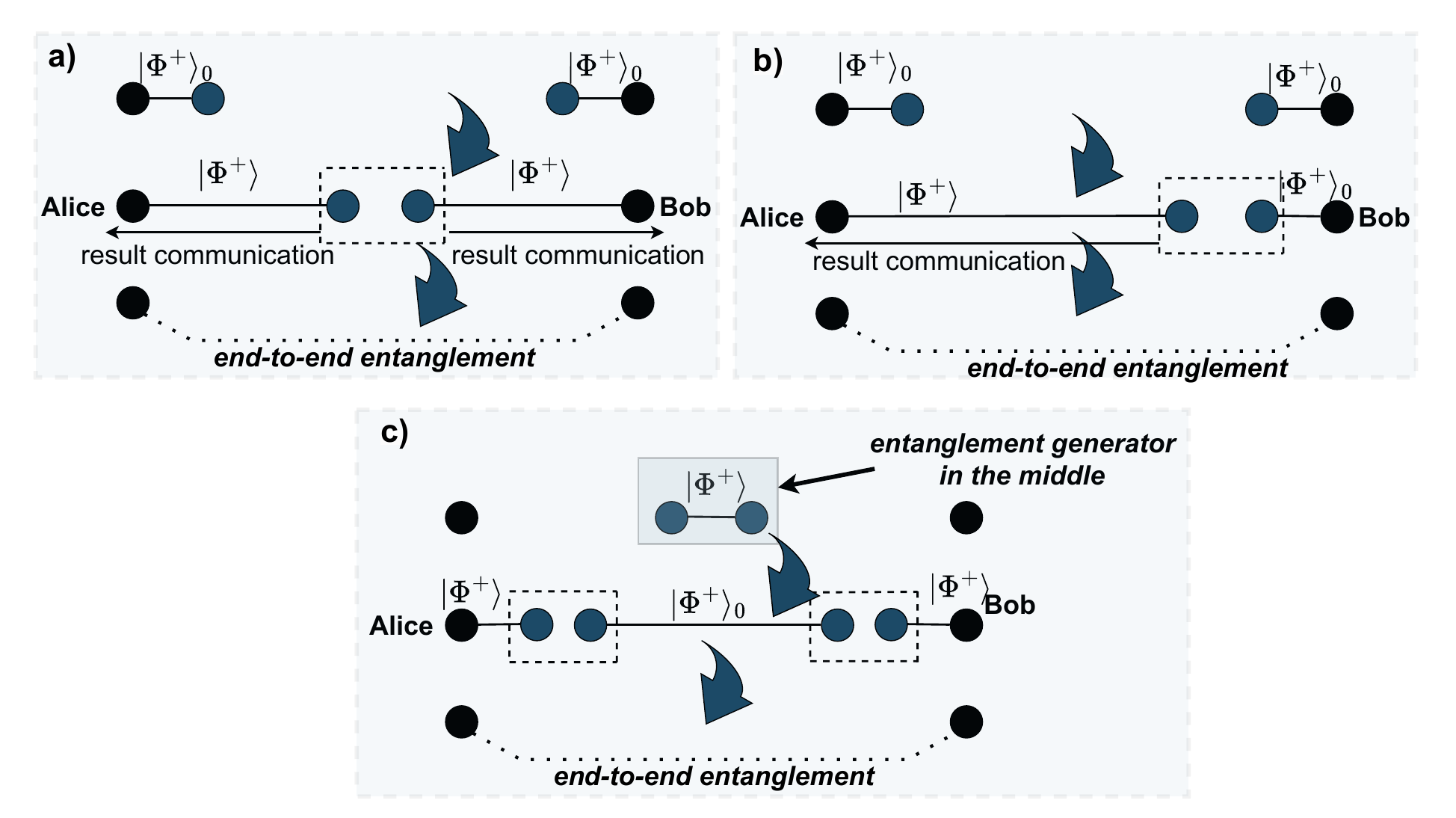}
\caption{
Entanglement generation schemes. \textit{\textbf{a) MeetInTheMiddle:}} Each of two distant stations holds one half of an EPR pair. The entangled qubits are sent to an intermediate station, where entanglement swapping establishes a direct end-to-end entanglement between the remote stations; \textit{\textbf{b) SenderRecever:}} An entangled qubit generated at the source station is transmitted to the destination station. Here, entanglement swapping between this qubit and one from the destination’s EPR pair creates a direct entanglement link between the stations; \textit{\textbf{c) MidPointSource:}} An intermediate station produces an EPR pair and sends one qubit to the source station and the other to the destination station. At each station, entanglement swapping between the received qubit and a qubit from the local EPR pair sets up end-to-end entanglement (taken from \cite{kumar2025quantum}). }
    \label{fig:entanglement_generation}
\end{figure}

\begin{figure}
\centering
\includegraphics[width=\columnwidth]{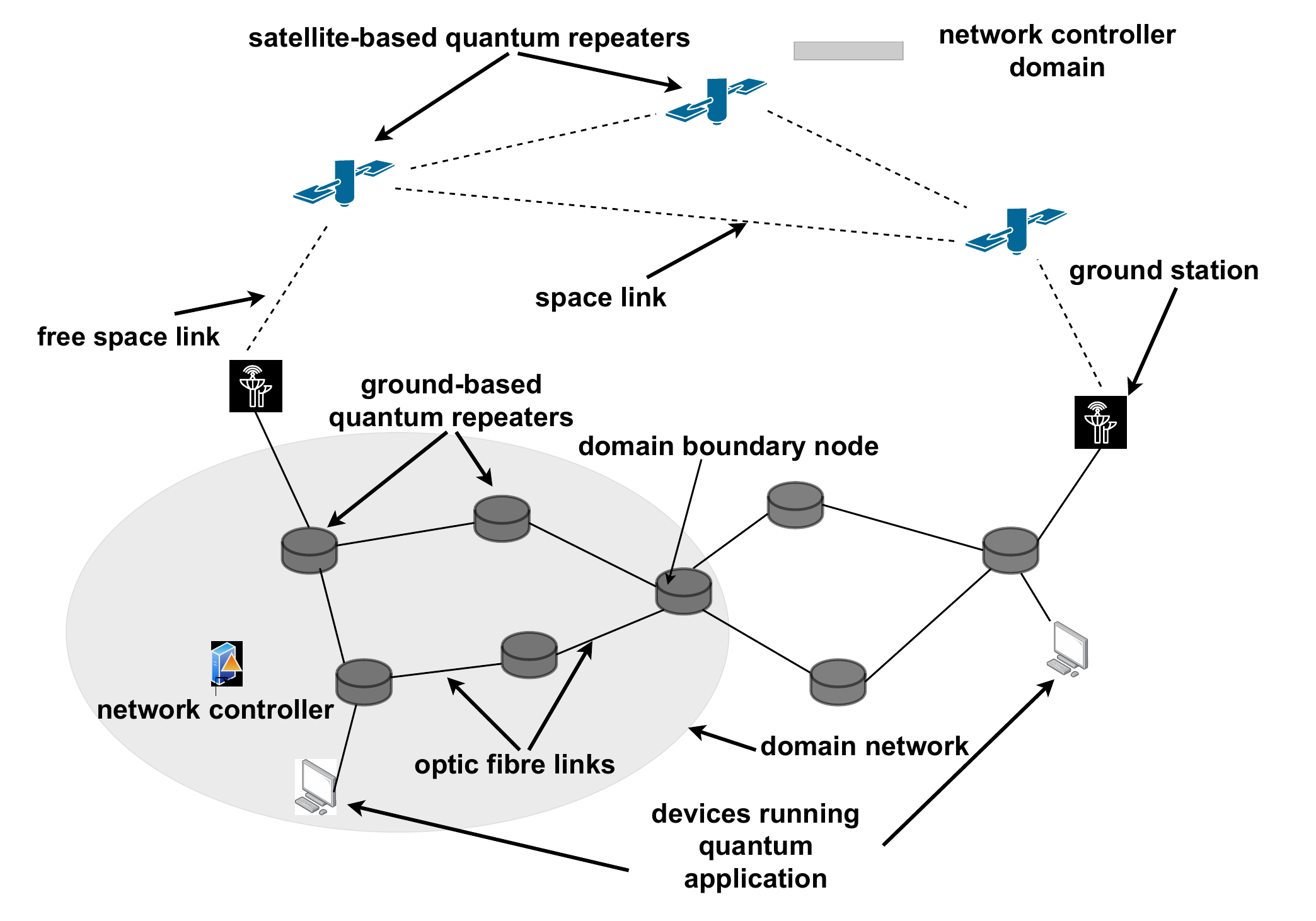}
\caption{
A schematic of the \textit{global} quantum internet.}
    \label{fig:network}
\end{figure}

\textbf{Quantum Device:} A quantum device, often referred to as an \textit{end node}, is essentially a quantum computer capable of running specific quantum applications while connected to the quantum internet. These devices are equipped with quantum processors to handle the execution of quantum applications and quantum memories to manage incoming entangled quantum states.


\textbf{Quantum links:} The quantum links are the quantum and classical channels that connect two neighbouring quantum repeaters or a quantum repeater and a quantum device. These links can be realised through fibre optic cables (represented as solid lines) or via free-space communication (represented as dotted lines), commonly used in ground-satellite networks, as shown in Fig.~\ref{fig:network}. Because of current technology limitations, the generation of EPR pairs between neighbouring components is a probabilistic process, whose success rate depends on the specific technology used and the distance between the nodes, but in general is rather small ($\ll 0.5$). This aspect highlights the inherent variability and complexity in the structure of quantum networks.


\textbf{Network Controller:} A network controller is a logically centralised entity responsible for overseeing a specific segment of the quantum network within its jurisdiction. It manages the communication of Bell state measurement results from quantum repeaters to the end nodes, facilitating the completion of the entanglement swapping process. Additionally, the network controller handles routing decisions for incoming requests, ensuring efficient network operation.
The network controller is commonly found in many proposed architectures of quantum networks, even if it may be an obstacle to scalability as the network grows, both in geographical size and the number of nodes.
Therefore, decentralised alternatives are also under study to cover such use cases, even if they are currently less explored due to the additional complexity required.
Only time will tell which approach will dominate the future of the quantum internet.

\chapter{Applications and Architecture of the Quantum Internet}\label{ch:applications_architecture}

\section{Applications and Use Cases}\label{sec:application}
Just as quantum computing provides advantages over classical computing, the potential applications of the quantum internet are expected to unfold as the technology advances. However, several applications have already been identified. These applications and use cases of the quantum internet can be categorised as follows.

\subsection{Enhanced security or privacy} The applications in this category essentially hinge on providing uncompromising security or privacy compared to their classical counterparts. Introduced in chapter~\ref{ch:foundations}, these applications use the fundamental principles of quantum mechanics, that is, the no-cloning theorem and quantum entanglement.

The advancements in quantum computing pose a significant threat to traditional public key cryptographic systems such as RSA \cite{rivest1978method}, Diffie-Hellman \cite{diffie2022new}, and ECC \cite{miller1985use, koblitz1987elliptic}. To address this challenge, two primary solutions are proposed. The first involves utilising quantum cryptography, specifically \textbf{quantum key distribution} (QKD) \cite{nurhadi2018quantum, mehic2020quantum, cao2022evolution}, which is one of its most successful applications. QKD protocols allow two parties, such as quantum devices within the quantum internet, to share symmetric secret keys securely. These keys can then be used with encryption algorithms to transmit messages securely. The second solution is the adoption of post-quantum cryptography \cite{bernstein2025post}, which employs classical systems designed to be secure against quantum attacks.

Once a quantum device connects to the quantum internet, it can utilise this network much like the classical internet we use today. In this environment, the device can offload quantum computation tasks to an untrusted device without risking the privacy or integrity of its data. This security method is termed \textbf{blind quantum computing} (BQC) \cite{broadbent2009universal, fitzsimons2017private}. Essentially, a client quantum device can employ one or more servers to perform computational tasks while concealing the nature of the computations from the servers themselves. Furthermore, some BQC protocols enhance security by incorporating the ability to verify the computations performed by the servers. This is achieved by embedding hidden tests within the computations, ensuring their correctness and integrity.

Following the theme of enhancing security and privacy, quantum technology offers various applications in distributed systems related to \textbf{consensus and verification}. A prime example is the quantum Byzantine agreement, developed by Ben-Or et al. \cite{ben2005fast}, which ensures that a group of participants can reach consensus on a bit value, robust against faulty or malicious behaviours. Building on this foundation of trust, quantum technology extends its utility to secure voting and surveying systems. As outlined by Vaccaro et al. \cite{vaccaro2007quantum}, these systems allow participants to cast votes or respond to surveys with guaranteed anonymity, enhancing the integrity of collective decision-making processes. Moreover, the field of quantum cryptography has made significant strides in enabling secure multiparty communications. Quantum Conference Key Agreement (CKA), as discussed in recent studies by Hahn et al. \cite{hahn2020anonymous} and Murta et al. \cite{murta2020quantum}, facilitates the establishment of a shared secret key among multiple parties, crucial for coordinating actions across different nodes without compromising security. Another innovative application is certified deletion, introduced by Broadbent et al. \cite{broadbent2020quantum}, which allows for the verifiable destruction of encrypted information, ensuring that once a ciphertext is deleted, it remains irrecoverable, even if the decryption key is compromised. Further expanding the horizon, quantum leader election protocols, such as those explored by Ganz et al. \cite{ganz2017quantum}, enable a group of distant, mutually distrustful entities to democratically elect a leader, ensuring fairness and transparency in critical decision-making scenarios. Additionally, the integration of quantum technologies into broader infrastructures is underway, with applications in the Internet of Things (IoT) being actively researched by Rahman et al. \cite{rahman2019quantum} and prospective uses in future 6G networks as envisioned by Rozenman et al. \cite{rozenman2023quantum}. These advancements illustrate the expanding role of quantum technology in shaping modern technological landscapes, promising unprecedented levels of security and efficiency in digital communications and beyond.

\subsection{Enhanced computing}
Applications within this category significantly enhance computing capabilities, enabling the completion of tasks that are challenging for classical computers or even a single quantum computer to perform. This augmentation is crucial for solving complex problems that exceed the processing power of conventional computing paradigms.

The leading concept in this category is \textbf{distributed quantum computing} \cite{cacciapuoti2019quantum, caleffi2024distributed}. Distributed quantum computing involves leveraging the computational resources of multiple quantum devices interconnected via the quantum internet to perform complex computational tasks. This approach is vital given the current presence of noise in practically realised qubits and the inherent limitations on the number of qubits per computer. In the absence of fault-tolerant qubits, distributed quantum computing offers a viable solution to scale up the number of qubits, enabling the execution of tasks that are too complex for single devices. Distributed quantum computing promises exponential speed-ups as compared to linear scaling in distributed classical computing \cite{cuomo2020towards}.

Another application of great potential interest is \textbf{quantum federated learning}\cite{sunkel2024towards}: each client trains their model on their local dataset using a quantum computer, thereby keeping their data private. Instead of sharing private data, clients only exchange their model's weights. The global model is then trained by aggregating these weights from all participating clients. While this is done today by exchanging only classical data, there are huge opportunities that can be unlocked by entanglement via the quantum internet, since this would remove the need for the data to be decoded and re-encoded at each iteration.

Finally, the quantum internet enables quantum-as-a-service (QaaS) models 
that extend beyond current classical-communication-based systems \cite{garcia2021quantum, moguel2022quantum}. 
The evolution of QaaS follows the staged development of quantum networks:

\textbf{Near-term QaaS (Stages 1-2):} Current quantum-as-a-service platforms 
provide remote access to quantum processors via classical channels, with a 
gateway selecting suitable quantum computers for specific tasks in real time 
\cite{garcia2021quantum, moguel2022quantum}. The immediate enhancement from quantum networks is secure 
quantum key distribution as a service (QKD-aaS), where users obtain 
cryptographic keys with information-theoretic security guarantees. Several 
metropolitan QKD networks already operate in this model, providing secure 
communication services to commercial and government clients \cite{chen2021integrated, 
liao2017satellite}.

\textbf{Medium-long-term QaaS (Stages 3-4):} As quantum repeater networks mature, 
distributed quantum computing becomes viable, enabling computational tasks 
to span multiple interconnected quantum processors.

The practical deployment timeline for quantum-enhanced QaaS depends 
critically on the maturity of quantum repeater technology (Chapter~\ref{sec:repeater}) 
and the development of robust network protocols (Chapter~\ref{sec:routing_forwarding_scheduling} and \ref{sec:qi_stack}). Current 
experimental demonstrations achieve point-to-point entanglement 
distribution \cite{pompili2021realization}, with multi-user QaaS requiring the 
routing and scheduling mechanisms addressed in subsequent chapters of 
this thesis.

\subsection{Specialised applications}

In addition to the general applications of the quantum internet, several specialised applications also play a crucial role in leveraging quantum mechanical capabilities, especially for scientific experiments.

One such specialised application is \textbf{quantum sensing} \cite{degen2017quantum}, which utilises the ability to establish quantum entanglement across networks to enhance the precision of measurements beyond the classical limits. This application exploits quantum properties to achieve superior measurement accuracy in various scientific and technological fields.

Another important application is \textbf{time synchronisation} in digital clocks \cite{chuang2000quantum, ilo2018remote}. Unlike classical methods that require $O{ (2^{2n})}$ messages to determine the $n$ digits of time difference $\Delta$ between two separated clocks in space, quantum algorithms can achieve the same with only $O (n)$ quantum messages. This reduction in message complexity makes quantum time synchronisation significantly more efficient.

Additionally, there is the application of \textbf{enhanced-baseline length for telescopes} \cite{gottesman2012longer, czupryniak2023optimal}. Traditional optical interferometers are limited in their resolution due to restricted baseline lengths, noise, and signal loss during photon transmission between telescopes. The use of quantum internet can potentially overcome these limitations, allowing for interferometers with arbitrarily long baselines and thus dramatically improving their resolving power.

\begin{table}[h!]
\centering
\caption{Stages of quantum internet development}
\label{tab:qi_stages}
\begin{tabular}{p{0.6cm}p{3.9cm}p{7.5cm}}
\toprule
Stage & Network & Functionality  \\
\midrule
\textit{S-1}  & QKD networks \newline
& Supports basic quantum key distribution between any nodes using trusted repeaters and post-selected prepare-and-measure techniques without end-to-end entanglement. \\

\textit{S-2} & Entanglement distribution networks & end-to-end entanglement between any nodes with no quantum memories \\
\textit{S-3} & Quantum memory networks &  end-to-end entanglement between any nodes with the capability of storing in quantum memories \\
\textit{S-4} & Few-qubits fault-tolerant networks & end-to-end entanglement between any nodes with few-qubits fault-tolerant capability on the quantum memory qubits \\
\textit{S-5} & Quantum computing networks & end-to-end entanglement between any nodes with full-fledged capability of using qubits in the quantum memory for computation and communication  \\
\bottomrule
\end{tabular}
\end{table}

\section{Stages of quantum internet}\label{sec:stages}
As with any developing technology, especially as sophisticated as the quantum internet, the implementation is only possible in stages. The different stages of the quantum internet have been categorised by the amount of incremental functionality available to the quantum devices or end nodes in~\cite{wehner2018quantum}. The summary of the stages of quantum internet development is given in table~\ref{tab:qi_stages}. 

Stage 1 of the quantum internet initiates the quantum internet with basic point-to-point quantum key distribution (QKD) setups that depend on trusted intermediate nodes for key relay and security \cite{scarani2009security, salvail2010security}. It primarily features networks where end-to-end quantum communication is absent, relying on secure key generation between adjacent nodes using trusted nodes. To bolster security against potentially untrusted nodes, measurement-device-independent QKD protocols are employed, enhancing security without relying on the trustworthiness of measurement devices \cite{lo2012measurement}. At this stage, the lack of end-to-end qubit entanglement prevents support for distributed quantum computing or quantum sensing. However, progress toward these goals begins with the ability of nodes to prepare and transmit a one-qubit state to any other node in the network. The transmission and measurement processes leverage post-selection, wherein only successful events—those where qubits are detected and measured correctly—are retained. Undetected or ``lost" qubits are disregarded.

This post-selected distribution of entanglement works as follows: a sender node prepares a pair of entangled qubits, retains one, and sends the other to a receiving node. If the receiving node successfully detects the transmitted qubit, the entanglement is confirmed and preserved. Although this method does not enable the deterministic transmission of arbitrary quantum states, it establishes a foundation for more sophisticated quantum operations in future stages.

Stage 2 of the quantum internet enables end-to-end entanglement without the need for post-selection during transmission or measurement, as in the case of Stage 1 above. However, due to the absence of quantum memory in the network, the entanglement must be used immediately after its creation. This stage allows for the successful distribution of end-to-end entanglement with a probability approaching 1. 

Stage 3 quantum internet upgrades to having the support of quantum memories at the local nodes in the network for application purposes. The quantum memories at network nodes allow more complex operations such as entanglement purification, quantum error correction and the creation of multi-partite states from bi-partite entanglement. However, due to the limited decoherence capabilities of the quantum memories, fault tolerance remains an issue in such networks. A functioning quantum memory network should have a decoherence time which encompasses the phase of entanglement generation and the time it takes for the classical signal to complete the entanglement distribution. This stage also provides the capability of deterministically sending arbitrary quantum states from one node to another.

Stage 4 of the quantum internet introduces fault-tolerance capabilities for quantum memory qubits, though these capabilities are limited to a finite number of qubits. Fault tolerance refers to the suppression of errors through the use of additional resources, namely, increased quantum memory. A group of error-corrected physical qubits is referred to as a \textit{logical qubit}. However, the suppression of errors through the combination of physical qubits is only feasible if the physical error rate remains below a critical threshold. Recent advancements in this area have demonstrated that surface code memories, even when operating below this threshold, can effectively suppress the logical error rate \cite{google2025quantum}.

Stage 5 of the quantum internet represents the realisation of a fully developed quantum internet, enabling the complete range of applications in both computation and communication. At this stage, fault tolerance is achieved for all available quantum memory qubits, unlocking the full potential of the quantum internet for diverse and robust applications.

The roadmap above was published in 2018~\cite{wehner2018quantum}, but at the time of writing, it can be considered still valid and useful to identify the upcoming milestones.
Today, stage~1 operational networks have been deployed in semi-commercial environments, while all other stages appear more elusive, in particular, due to the lack of commercial-grade products for components such as the quantum repeater.
However, significant progress has been demonstrated for many enabling technologies and quantum networks that can be categorised as stage~2 have been implemented in controlled environments, which gives us hope that technology will soon catch up with high expectations from the scientific community.

\section{Types of quantum network}\label{sec:physical_layer}

According to the current consensus in the field, the quantum internet is envisioned to operate alongside the classical internet, particularly during its initial phases \cite{garcia2024strategies}. Therefore, the categorisation of quantum networks that collectively constitute the quantum internet parallels that of the classical internet, based on the operational area's size. Broadly, quantum networks can be categorised into \textit{modular} networks, \textit{ground-based} networks, and \textit{satellite-based} networks. 

\textit{Modular networks} are distinct in their construction, often associated with a multi-core quantum architecture. They are predominantly used in distributed quantum computing \cite{caleffi2024distributed} and are typically implemented on a chip to interconnect various quantum computing modules \cite{rodrigo2021double, jnane2022multicore}. These networks are crucial for the scalability of quantum computing technologies, providing essential links within and between quantum processors.

On the other hand, \textit{ground-based networks} can be further subdivided according to their operational distances, similar to classical networks. The most frequently discussed type of ground-based networks in current research—and those that often have experimental test beds—are metropolitan quantum networks. Metropolitan quantum networks operate over metro-scale distances \cite{chen2021implementation, chung2022design}, facilitating regional connectivity within a more confined geographic area compared to their long-range counterparts.

The motivation for including \textit{satellite-based networks} in the global internet stems from the fact that free space links suffer polynomial loss compared to an exponential loss in optical fibre links \cite{wallnofer2022simulating}. Therefore, satellite-based quantum networks unlock quantum communication over continental and intercontinental ranges, which is otherwise challenging with purely ground-based quantum networks. However, the free space link (Forward $\&$ Return link) passes through the atmosphere, which induces several effects, such as diffraction, absorption, and scintillation effects \cite{chiti2024survey}. While research efforts are required on this front, the perseverance of entanglement in good weather conditions can be seen as a potential to include satellite-based networks in the quest for the quantum internet \cite{armengol2008quantum}. Given the vast distance of operation in satellite-based networks, latency also plays a huge part when considering the decoherence of quantum states. Quantum memories would play a vital role in such scenarios. Studies are ongoing on this front for the use of quantum memories in space \cite{gundougan2021topical}.

\section{Generations of quantum repeater}\label{sec:repeater}
Quantum repeaters, as introduced in Section~\ref{sec:components_of_qi}, are essential for the distribution of end-to-end entanglement, especially in a long-distance regime. However, two major challenges complicate this task: \textit{photon loss} and \textit{operational errors}. Photon loss occurs when photons, the carriers of quantum information, are absorbed or scattered during transmission, preventing successful entanglement. Operational errors, on the other hand, stem from imperfections in the devices that manipulate and measure quantum states, potentially leading to incorrect results.

To address these issues, quantum repeaters are classified into three generations based on the methods they employ to correct and mitigate errors \cite{munro2015inside, muralidharan2016optimal}.

\begin{figure}
\centering
\includegraphics[width=0.9\columnwidth]{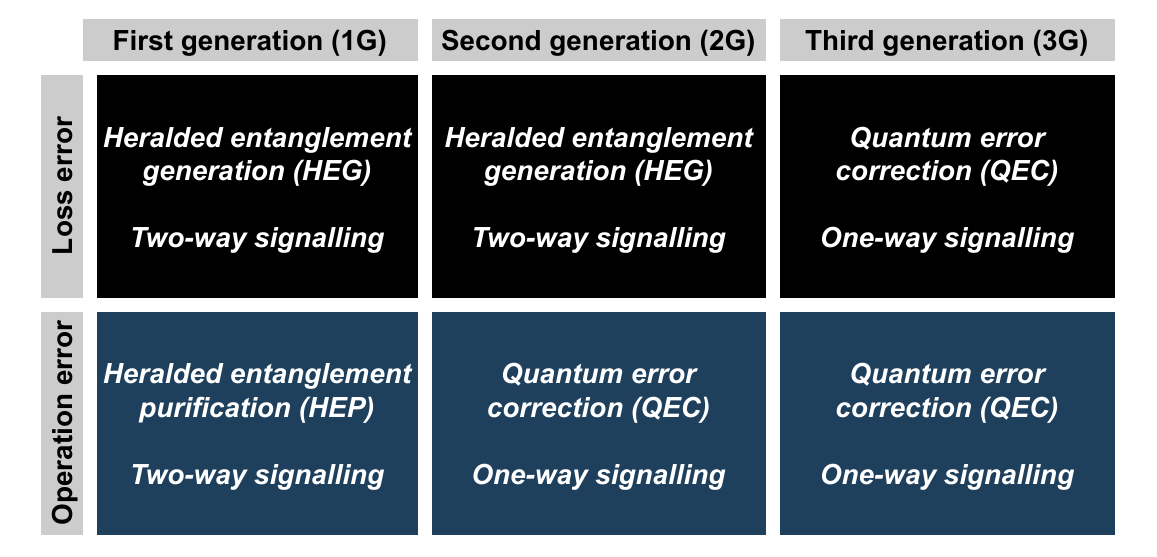}
\caption{Generations of quantum repeaters based on the method used to correct the error (modified from \cite{muralidharan2016optimal}).}
    \label{fig:repeater}
\end{figure}

The first generation of quantum repeaters employs two key techniques: \textit{heralded entanglement generation} (HEG) and \textit{heralded entanglement purification} (HEP). HEG is responsible for correcting loss errors through a deterministic process that ensures the reliable delivery of entanglement. It utilises two-way signalling to confirm the successful establishment of end-to-end entanglement. Similarly, HEP addresses operational errors by using two-way signalling to guarantee the effectiveness of the entanglement purification process. This technique has been further elaborated in Section~\ref{sec:prelim}. 

The second generation of quantum repeaters uses HEG to correct loss errors while \textit{quantum error correction} (QEC) to correct operational errors. HEG in the second generation of quantum repeaters still uses two-way signalling, while QEC only requires one-way signalling, which reduces the time consumption in the establishment of the end-to-end entanglement.

The third generation of quantum repeaters uses QEC to correct loss and operational errors. Subsequently, only one-way signalling is enough for the operation of the third generation of quantum repeaters, which greatly reduces the time consumption in the end-to-end entanglement establishment procedure. 

Today, the industry and academy are working towards realising stable 1G quantum repeaters, which are not yet ready for mass production and field deployment. For this reason, many of the scientific papers published, including those cited in this paper, focus on 1G repeaters only, while 2G/3G repeaters are currently mainly a matter of long-term speculation.

\begin{figure}
\centering
\includegraphics[width=0.9\columnwidth]{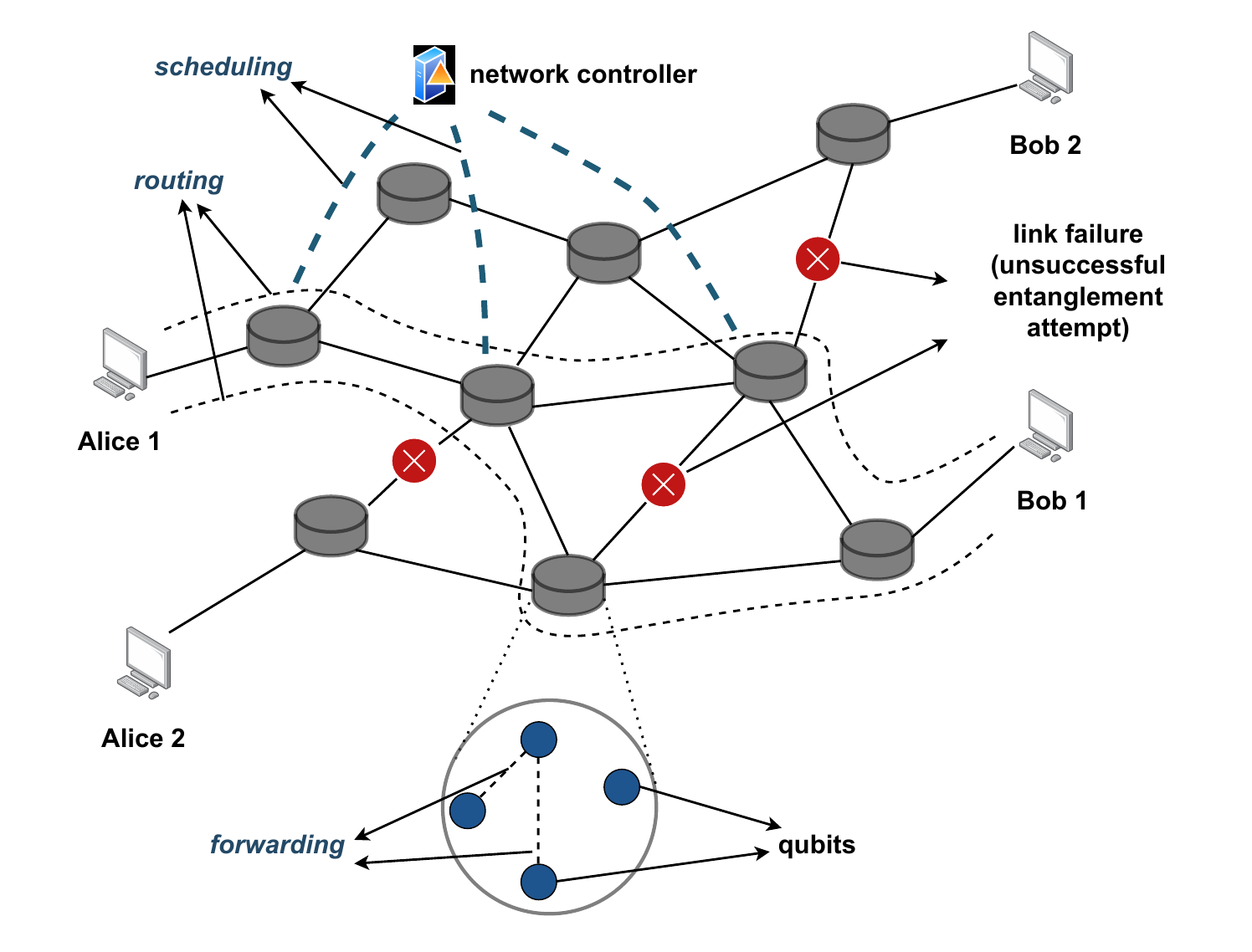}
\caption{\textbf{Routing:} selection of route. \textbf{Forwarding:} execution of entanglement swapping and purification procedure. \textbf{Scheduling:} selection of end-to-end entanglement in a time-slot (taken from \cite{kumar2025quantum}).}
    \label{fig:routing}
\end{figure}

\section{Routing, forwarding, and scheduling}\label{sec:routing_forwarding_scheduling}
As introduced in Section~\ref{sec:prelim}, the applications in Section~\ref{sec:application} utilise two processes. First, transporting an arbitrary quantum state using a quantum teleportation protocol. Second, swapping entanglement is distributed over smaller distances with an end-to-end entanglement between two nodes or devices using an entanglement-swapping protocol. While quantum teleportation is necessary for quantum communication, establishing end-to-end entanglement unlocks most applications. It is to be noted that the two protocols are almost similar and inter-convertible, with only the difference of an extra qubit and Bell state measurement involved, as depicted in Fig.~\ref{fig:teleport_swap_chain}.

Achieving end-to-end entanglement between two quantum devices within a network that includes quantum repeater nodes is a complex challenge under practical constraints. Let us delve into the intricacies of this problem. Imagine a network of quantum repeaters connected by quantum links, as depicted in Fig.~\ref{fig:routing}. The network controller receives requests to establish end-to-end entanglement between the connected quantum devices. The network operates using specific protocols for entanglement generation \cite{caleffi2017optimal}, entanglement purification \cite{victora2020purification, li2022fidelity, hu2024high}, and quantum error correction \cite{patil2024entanglement}. Additionally, the components of the network, including quantum links \cite{van2013path}, repeaters \cite{kumar2024routing, kumar2025routing}, and devices, could be heterogeneous; that is, they vary either in quality or involve different physical systems.

Given these protocols and assumptions, the \textit{problem} involves determining how to \textit{satisfy} the requests received by the network controller\footnote{For simplicity in explanation, we assume the presence of a network controller, as this is a common approach in the literature. However, it is worth noting that approaches not requiring a network controller are also feasible.} in a \textit{practical scenario} by selecting the most efficient routes, swapping orders, and scheduling policy within the network to effectively serve the requests referred to as \textit{routing}, \textit{forwarding}, and \textit{scheduling} respectively. To \textit{satisfy} means to serve these requests with the highest possible throughput \cite{zhao2021redundant} and fidelity—meeting or exceeding the threshold fidelity set by the requests—while minimising latency, ensuring fairness, and optimising resource \cite{zhang2021fragmentation}. The \textit{practical scenario} means taking into account operational challenges such as signal attenuation due to distance, noise during the application of quantum gates in the procedure, quantum state decoherence over time, and depolarisation due to the quantum channel. 

\begin{figure}
\centering
\includegraphics[width=0.7\columnwidth]{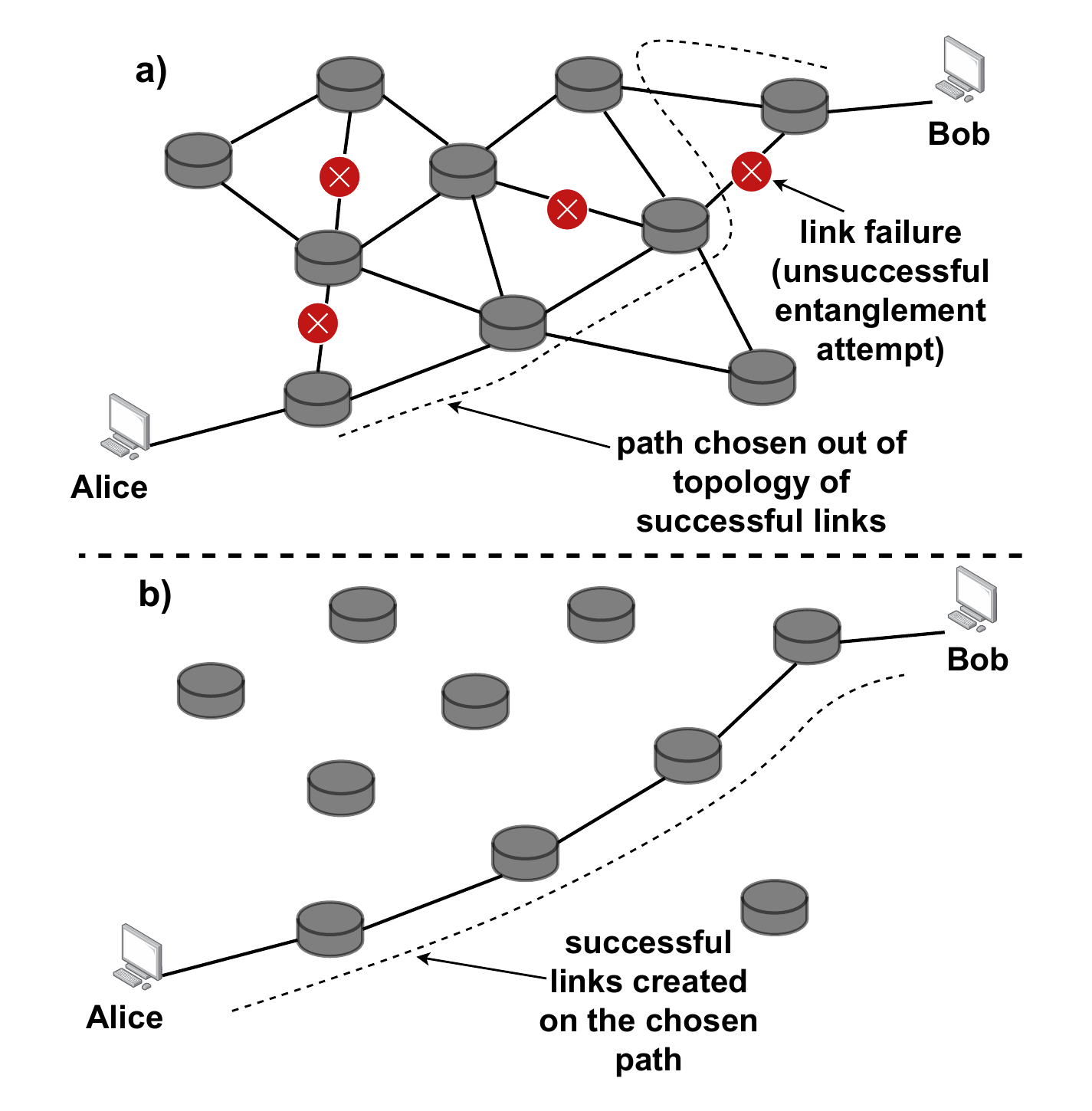}
\caption{\textbf{a) Advanced entanglement generation:} The Path is chosen out of a reduced topology constructed out of successful entanglement links. \textbf{b) On-demand entanglement generation:} Successful entanglement links are established on a chosen path (taken from \cite{kumar2025quantum}).}
    \label{fig:advanced_and_on_demand}
\end{figure}

For multiple end-to-end entanglement deliveries between source-destination pairs in a quantum network, finding an optimal path using some routing metric (discussed in Section~\ref{sec:routing_metric}) is referred to as \textit{routing}. Meanwhile, the actual execution of the entanglement swapping procedure and purification (depending upon the protocol), which considers the path found using the routing algorithm, is referred to as \textit{forwarding}. The routing phase can further be categorised based on entanglement generation utilised, that is, on-demand or advanced entanglement generation (depicted in Fig.~\ref{fig:advanced_and_on_demand}). On-demand entanglement generation calculates routing paths before initiating entanglement, whereas advanced entanglement generation bases its routes on the network topology that emerges following the probabilistic success or failure of initial link establishment. Each approach offers distinct benefits and limitations. On-demand generation allows routing decisions to be made before quantum states begin to decohere, which is advantageous. However, due to the inherently probabilistic nature of entanglement, it may require multiple attempts to secure all necessary links for a successful route. On the other hand, advanced entanglement generation relies on post-link-establishment topology data. While this method integrates more recent network information, it risks the routing decisions occurring within the decoherence window of the entangled photon pairs, which could degrade the quality of the entanglement. The choice of entanglement generation scheme significantly impacts the quantum network's requirements. For on-demand entanglement generation, a study \cite{khatri2019practical} has proposed two figures of merit. The first is the average connection time, which dictates the quantum memory requirements. The second is the average largest entanglement cluster size, which indicates the scalability of the quantum networks. In the literature, the routing phases used in on-demand and advanced entanglement generation are also referred to as proactive and reactive, respectively. 

The issue of \textit{scheduling} has recently begun to attract attention in the field of quantum networking. Scheduling involves selecting which end-to-end entangled EPR pairs to allocate to a current path request within a specific time-slot (discussed in Section~\ref{sec:slotted_model}) while reserving EPR pairs for future path requests \cite{cicconetti2021request, wang2023efficient, fittipaldi2022linear, gu2023esdi}. This process is crucial for managing latency, particularly due to the decoherence of quantum states. Effective scheduling ensures that the quantum states are utilised efficiently, minimising the impact of decoherence on network performance and maximising the fidelity of quantum communications.

Addressing all potential constraints to deliver entanglement in a multi-user quantum network presents substantial difficulties. For instance, optimising end-to-end entanglement for multiple source-destination pairs is an NP-hard problem, as identified in \cite{chakraborty2020entanglement}. Additionally, the specific type of application running on the network significantly influences the demands placed on it, which suggests that distinct routing problems could be formulated for each application class.

The routing problem described earlier, which focuses on delivering end-to-end entanglement between specified endpoints, is particularly relevant for applications that enhance security or privacy, as well as for other specialised applications discussed in Section~\ref{sec:application}. Conversely, application classes that primarily enhance computing capabilities, such as those found in distributed quantum computing, have distinct requirements that necessitate alternative routing solutions \cite{cicconetti2022resource, cicconetti2023service}. These variations underscore the complexity and specialised nature of routing in quantum networks.

\begin{table}[h!]
\centering
\caption{A general recipe for end-to-end entanglement distribution in a quantum network}
\label{tab:recipe_entanglement}
\begin{tabular}{p{3.6cm}p{8.5cm}}
\toprule
Stage & Choices  \\
\midrule
Ingredients & Bell (or GHZ) states  \\

Routing Algorithms & Dijkstra-based, optimisation programs (linear, integer, stochastic, etc), greedy-based, AI-based, etc.  \\
Evaluation metrics & Throughput, fidelity, hop count, etc. \\

Performance enhancers & Purification, error correction  \\

\bottomrule
\end{tabular}
\end{table}

\section{Quantum internet protocol stack}\label{sec:qi_stack}
For efficient and scalable network development, a protocol stack is indispensable, as it offers a structured framework enabling the independent development of each layer. Numerous research groups have put forth proposals for a quantum internet protocol stack, similar to the classical internet’s Open Systems Interconnection (OSI) model, which segregates network functions into distinct layers. In this paper, we will examine three quantum internet protocol stacks that are most frequently referenced in the scholarly literature. In Fig.~\ref{fig:stack}, we provide a graphical visualisation of the stacks, highlighting their differences. For more details on this specific topic, interested readers are invited to check recent relevant surveys, like~\cite{illiano2022quantum, li2024survey}.

\begin{figure}
\centering
\includegraphics[width=\columnwidth]{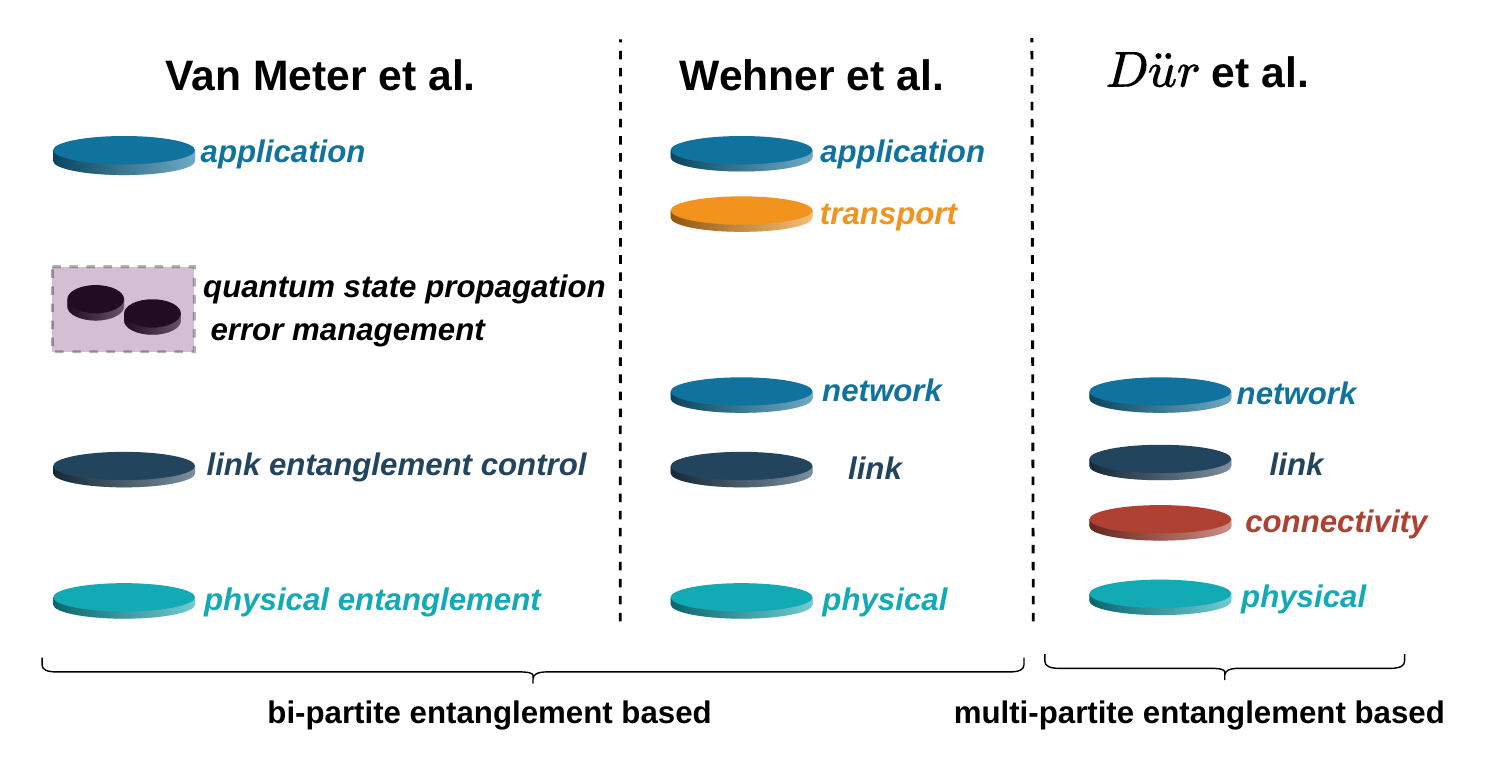}
\caption{State-of-the-art quantum internet protocol stacks. Van Meter et al. \cite{aparicio2011protocol}, Wehner et al. \cite{kozlowski2020designing}, and D\"{u}r et al. \cite{pirker2019quantum} (taken from \cite{kumar2025quantum}).}
    \label{fig:stack}
\end{figure}

The initial comprehensive proposal for a layered quantum internet protocol stack dates back to 2009, introduced by Van Meter et al. \cite{van2008system, aparicio2011protocol}. This protocol stack comprises five layers: \textit{physical entanglement}, \textit{link entanglement control}, \textit{error management}, \textit{quantum state propagation}, and \textit{application}. The \textit{physical entanglement} layer is tasked with generating EPR pairs between adjacent nodes, while the \textit{link entanglement control} layer monitors the success and failure of these pair establishment attempts. These two layers operate across every single hop and are applied recursively to maintain connectivity through shared EPR pairs. The \textit{error management} layer specifies the hops where purification is needed, which ensures the high fidelity of the EPR pairs and records the outcomes. The \textit{quantum state propagation} layer is responsible for establishing end-to-end entanglement by creating shared EPR pairs between the source and destination stations using the entanglement swapping procedure. This layer also communicates the results of the swapping to the end nodes, enabling them to perform the necessary single-qubit operations. Both the \textit{error management} and \textit{quantum state propagation} layers are similarly operated recursively across multiple hops to ensure high fidelity of the distributed entanglement. Finally, the \textit{application} layer primarily manages the applications run on the quantum internet using high-fidelity EPR pairs. Recent advancements in this protocol stack have focused on ensuring synchronicity between distant nodes \cite{matsuo2019quantum}. This study introduced a RuleSet-based quantum link bootstrapping protocol that assesses the fidelity of quantum links and their throughput. Another study implemented a quantum recursive network architecture (QRNA) alongside the RuleSet-based protocol to enhance scalability, specifically achieving multi-party entanglement and internetworking within quantum networks \cite{van2022quantum}.

The protocol stack proposed by Wehner et al. \cite{dahlberg2019link, kozlowski2019towards, kozlowski2020designing} features layers named \textit{physical}, \textit{link}, \textit{network}, \textit{transport}, and \textit{application}, drawing inspiration from classical protocol stacks. The \textit{physical} layer focuses on the hardware components required to generate initial entanglement within a predefined time slot. The \textit{link} layer enhances the robustness of connections between nodes in the quantum internet, leveraging a quantum entanglement generation protocol (QEGP). Notably, this model incorporates key hardware parameters, leading to the development of a \textit{hardware-abstraction} sub-layer that bridges the physical systems to the link layer. The \textit{network} layer utilises these connections to design network protocols that facilitate endpoint communication across the quantum internet while managing the network's entanglement resources. The \textit{transport} layer is responsible for managing quantum internet traffic, including congestion control \cite{leone2021cost, li2023swapping}, re-transmission, and monitoring quantum channel quality. Finally, the \textit{application} layer ensures the functionality of desired applications on quantum devices interconnected by the quantum internet. 

The protocol stack developed by D\"{u}r et al. \cite{pirker2018modular, pirker2019quantum} represents a significant departure from previously discussed stacks, as it employs multipartite states instead of the conventional Bell states. This stack includes the layers \textit{physical}, \textit{connectivity}, \textit{link}, and \textit{network}. The \textit{physical} layer in this framework undertakes multiple roles, including the generation and transmission of entangled states, as well as signal conversion between quantum channels. The \textit{connectivity} layer is dedicated to establishing long-distance entanglement through entanglement purification processes. Following this, the \textit{link} layer addresses incoming requests by providing the necessary quantum states. Lastly, the \textit{network} layer ensures the distribution and sharing of entangled states across various networks, thus facilitating broad quantum communication capabilities.
\chapter{Routing in Quantum Networks}\label{ch:routing}
As discussed in Section~\ref{sec:routing_forwarding_scheduling}, routing in quantum networks is one of the critical problems for the distribution of entanglement. Consequently, it has become an area of significant interest in recent years due to the obvious differences from the classical routing techniques. In this chapter, we delve into quantum routing from a perspective of considering realistic constraints that move beyond the idealised assumptions. To this, we will consider two aspects: routing with heterogeneous quantum repeaters and routing with end-to-end knowledge. 

With this chapter, we begin the quest into the first scope of this work (see Section~\ref{sec:scope_objective}). As a reminder, we will be focusing on the first two research questions (see Section~\ref{sec:research_questions}) defined for this work in this chapter and performance evaluation will be followed in Chapter~\ref {ch:performance_evaluation}.

\section{Metrics for quantum routing}\label{sec:routing_metric}
Given the complexity of the problem as described in Section~\ref{sec:routing_forwarding_scheduling}, a common way in the literature is to use a routing metric to allocate paths for entanglement distribution. In this regard, a wide range of routing metrics has been employed in the literature to tackle the defined problem statements, respectively. This includes throughput \cite{shi2020concurrent, chakraborty2020entanglement, zhang2021fragmentation, zhao2021redundant}, inverse throughput \cite{van2013path}, continuous fidelity curves (entanglement generation fidelity vs rate) \cite{coutinho2023entanglement}, end-to-end entanglement rate of the path \cite{caleffi2017optimal}, hop count \cite{li2021effective}, latency \cite{cicconetti2021request}, E2E fidelity \cite{zhao2022e2e, kumar2024routing, kumar2025routing}, and fairness \cite{yang2022online}. 

The selection of routing paths using the routing metric can also be categorised using different routing algorithms such as Dijkstra-based \cite{van2013path, kumar2024routing, kumar2025routing}, linear programs \cite{iacovelli2024probability}, stochastic programs \cite{kaewpuang2023entangled}, integer programming \cite{zeng2023entanglement}, greedy algorithms \cite{pant2019routing, pandey2023greedy}, and AI-based routing algorithms \cite{le2022dqra, chaudhary2023learning, islam2024reinforcement}. 

\section{Related works in quantum routing}\label{sec:relatedworks_routing} 

Among the initial investigations, Meter et al. \cite{van2013path} proposed the inverse of the throughput of the link as the routing metric, which is Bell pairs per second at a certain fidelity. The other possible candidates for routing metrics, such as the total number of laser pulses or the total number of measurement operations, were found to be unsuitable. The proposed metric (inverse of throughput) takes into account the quality of links, which arises from their length. Caleffi et al. \cite{caleffi2017optimal} put forward a routing protocol which used the end-to-end entanglement rate of the path as the routing metric. This metric takes into account physical mechanisms such as entanglement generation, swapping, Bell state measurement, and decoherence time. 

Pant et al. \cite{pant2019routing} and Li et al. \cite{li2021effective} proposed routing schemes that simply utilise the hop count as the routing metric, which we will use as a baseline reference in our performance evaluation in chapter~\ref{ch:performance_evaluation}.

Shi et al. \cite{shi2020concurrent} introduced the QCAST framework, which
adopts a time-slot structure, a modified version of the approach proposed by Pant et al. \cite{pant2019routing}. Each time-slot consists of four phases. During phase one, all nodes receive information regarding the current source-destination pair, intending to establish end-to-end entanglement. In phase two, routes are determined for the requests from phase one using the QCAST algorithm, and external entanglement links are established between neighbouring nodes, utilising the qubits and channels associated with these nodes. Phase three involves nodes exchanging link state information through classical channels. Finally, in phase four, each node establishes internal links to complete the route between the source and destination. This work has inspired the system defined in Section~\ref{sec:slotted_model}, which also involves multiple phases in a time-slotted system. Zhang et al. \cite{zhang2021fragmentation} introduced a novel fragmentation-aware algorithm designed to address resource contention issues across multiple paths. The core concept involves redefining network resources as `public' (shared among all paths) rather than `private' (dedicated to specific paths), contrasting with the approach taken by QCAST. Zhao et al. \cite{zhao2021redundant} introduced a method aimed at enhancing network throughput over QCAST by as much as 68.75\%, achieved through the provision of redundant external links established during phase two. Such optimisations are orthogonal to the problem addressed in this work and, in principle, they can be used in combination with the solutions proposed.

\section{System model}
In line with Section~\ref{sec:components_of_qi}, in this work, we consider the general quantum network structure shown in Figure~\ref{fig:network_architecture}. A summary of its key components is as follows:
\begin{enumerate}
    \item \textbf{Quantum Repeater:} Fundamental component in quantum networks, which allows for long-distance end-to-end entanglement.
    \item \textbf{Quantum Device:} Quantum computer capable of running quantum applications connected to the quantum network.
    \item \textbf{Storage and Controller:} A logically centralised entity that maintains the network topology and an estimate of the end-to-end fidelity between pairs of quantum devices (storage) and is in charge of making routing decisions to establish end-to-end entanglement between them. 
    \item \textbf{Quantum Network Link:} A link with quantum and classical channels between two neighbouring quantum repeaters or a quantum repeater and a quantum device.
\end{enumerate}
\begin{figure}[tb]
    \centering
    \includegraphics[width=0.8\columnwidth]{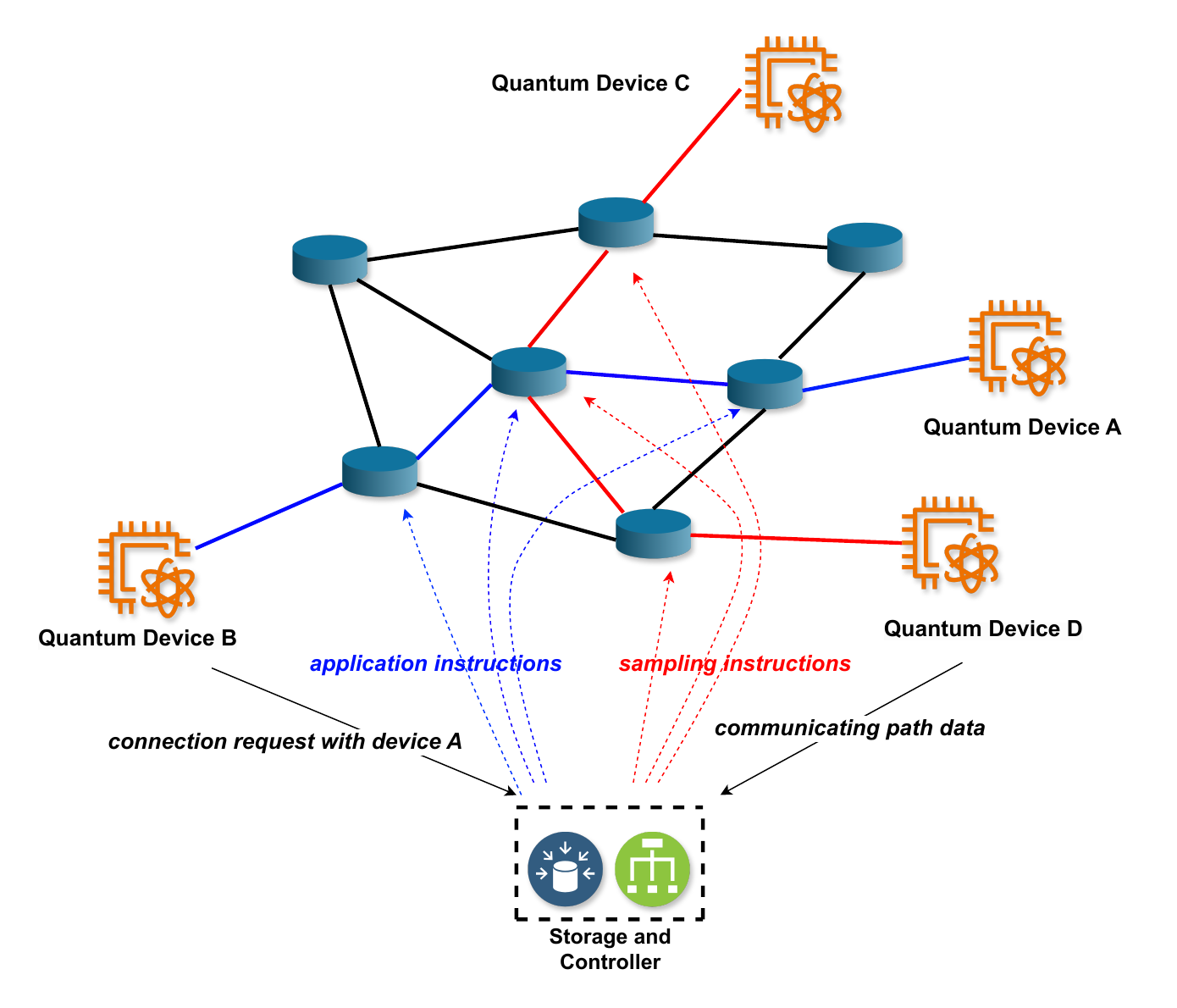}
    \caption{Quantum network architecture}
    \label{fig:network_architecture}
\end{figure}

It is assumed that for every link there is a heralded source of entangled photons, which are transmitted to each end of the link and enter a photon detector. Therefore, each node (quantum device or repeater) is assumed to be equipped with one photon detector for each of its links.  Furthermore, the quantum repeaters are required to have a sufficient number of Bell state analysers that can be used to perform entanglement swapping, as illustrated in Section~\ref{ssec:swapping}. Finally, for simplicity in design and analysis, it is assumed that no time multiplexing is performed; that is, for as long as two links of a quantum repeater are assigned to a specific end-to-end path, they cannot be used by other paths.
Two quantum devices, e.g., A and B, can make use of the quantum network to realise distributed quantum applications through end-to-end entanglement of (some of) their qubits.
We assume that this happens through a request to the controller, which is in charge of deciding the path between A and B among all the possible ones in the network.
Once such a path is selected, the controller is expected to notify all the quantum repeaters involved so that they perform appropriate operations to realise the end-to-end entanglement between quantum devices A and B until further notice.
Such notifications are referred to as ``application instructions'' in Fig.~\ref{fig:network_architecture}, and we assume here,
For simplicity, a link can be assigned a single path, i.e., it cannot be shared across applications between different pairs of quantum devices.

We now delve into the details of the establishment of end-to-end entanglement.
As exemplified in Fig.~\ref{fig:network_architecture}, end-to-end entanglement between quantum devices C and D or quantum devices B and A is carried out via entanglement swapping as shown in Fig.~\ref{fig:teleport_swap_chain}c). The process is as follows:
\begin{enumerate}
    \item All the EPR pairs between neighbouring quantum repeaters along the intended route are successfully shared via the respective quantum links.
    \item The quantum repeaters along the route perform a Bell state measurement on one qubit of the shared EPR pair.
    \item They send the classical results of Bell measurement to one of the end nodes (e.g., Bob) and inform the other one (Alice) about the procedure's success.
    \item Bob applies Pauli gates ($I$, $\sigma_x$, $\sigma_z$ or $\sigma_x \sigma_z$) depending upon the classical results of measurement received.
\end{enumerate}
Bell state measurement performed on each quantum repeater is central to establishing end-to-end entanglement, and it relies on the ability to project onto the basis states, which can be summarised through an $\eta$ parameter as follows \cite{helstrom1969quantum, dur1999quantum}:
\begin{equation}\label{eq:P0}
P_0 = \eta |0 \rangle \langle 0| + (1 - \eta) |1 \rangle \langle 1|,
\end{equation}
\begin{equation}\label{eq:P1}
P_1 = \eta |1 \rangle \langle 1| + (1 - \eta) |0 \rangle \langle 0|.
\end{equation}

where $\eta$ is the efficiency of the qubit measurement in the basis states.
Consider a scenario where the correct measurement basis for a qubit is $|1 \rangle \langle 1|$. Then, during entanglement swapping in quantum repeaters, the measurement is characterised by Eq.~\eqref{eq:P1}. For an ideal measurement device in a quantum repeater $\eta = 1$, the measurement results are entirely reliable. However, if $0 \leq \eta < 1$, there exists a non-zero probability of $1-\eta$ that the results will be incorrect. The same probability of error applies to measurement in the $|0 \rangle \langle 0|$ basis described by Eq.~\eqref{eq:P0}.

For a linear path emerging from $N$ entangled pairs, the end-to-end fidelity of the path is given by \cite{dur1999quantum}:
\begin{equation}\label{eq:linear_chain_routing_fidelity}
F_N = \frac{1}{4}\left\{
1 + 3 \left[\frac{4\eta^2 - 1}{3}\right]^{N-1}
\left[\frac{4F - 1}{3}\right]^{N}
\right\},
\end{equation}
where
$N$ is the number of entangled pairs,
$\eta$ is the efficiency of the entanglement swapping procedure, and
$F$ is the initial fidelity of the generated EPR pairs.

In this study, we consider the efficiency of the measurement devices integral to the quantum repeater's overall performance. Consequently, the errors introduced by the Bell state measurements influence the categorisation of quantum repeaters, i.e., a quantum repeater belongs to a certain class $g$ depending upon the efficiency figure $\eta_g$ of its measurement devices, which, in turn, affects the fidelity of the end-to-end entangled pairs used by the applications running in the quantum devices.
In fact, for a linear route connecting two quantum devices through $\sum_i^G N_g$ quantum repeaters belonging to $G$ classes, each class containing $N_g$ quantum repeaters, the fidelity is given by:
\begin{equation}\label{eq:main_routing_fidelity}
F_p
= \frac{1}{4}\left\{
1 + 3
\prod_{g=1}^G
\left[
\left( \frac{4\eta_g^2 - 1}{3} \right)^{N_g}
\left( \frac{4F - 1}{3} \right)^{N_g}
\right]
\left( \frac{4F - 1}{3} \right)
\right\},
\end{equation}
where each pair of qubits in a link is assumed to be in a Werner state with initial fidelity $F$ upon generation of the local entanglement. It is worth noting that this equation models the heterogeneity across quantum repeaters, where each class of quantum repeater has a distinct efficiency figure \( \eta_g \). The end-to-end fidelity \( F_p \) thus becomes a function of both the number of repeaters in each class and their respective efficiency figures. This heterogeneity influences the cumulative fidelity of entanglement over the entire route, thereby impacting the performance of algorithms designed to optimise routing paths in quantum networks. In addition, the initial fidelity of EPR pairs in quantum links can vary, leading to heterogeneity in quantum links that impacts the end-to-end fidelity calculation. This end-to-end fidelity can be directly utilised in the \emph{grey-box} algorithms (introduced in Section~\ref{ssec:grey_white_box}) to guide path selection and optimisation. However, we assume uniform initial fidelity of EPR pairs $F$ across all quantum links in our analysis for simplicity.

\section{The time-slotted model}\label{sec:slotted_model}
The general dynamics of delivering entanglement between sources and destinations in a quantum network are depicted in Fig.~\ref{fig:time}. For routing, a connection-oriented model is assumed, where the applications on quantum devices must reserve the use of the quantum network through the controller, specifying the peer end device with which they wish to establish end-to-end entanglement of qubits and a minimum fidelity, which depends on the application's characteristics.
It is commonly understood that the use of connections unlocks the potential for higher utilisation of the (scarce) resources in the quantum network, especially when coupled with link-layer protocols for managing the local entanglement between quantum repeaters sharing a link \cite{li2022connection}.

\begin{figure}[tb]
\centering
\includegraphics[width=0.8\columnwidth]{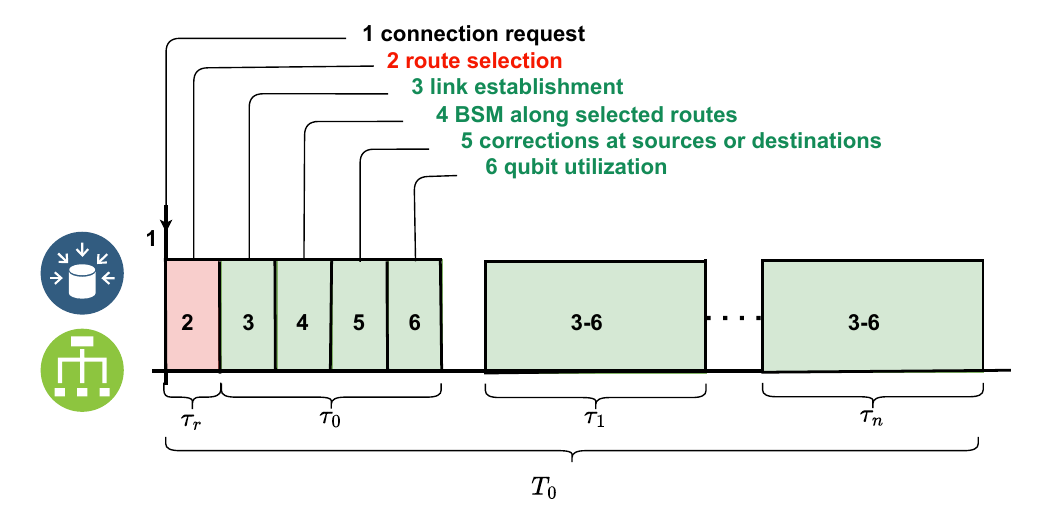}
    \caption{The timing diagram}
    \label{fig:time}
\end{figure}

In our system, we model the time evolution in \emph{periods}, each marked by a change in the set of active connections.
In other words, a new period begins ($T_0$ in Fig.~\ref{fig:time}), when either a new connection requests to be admitted by the controller or a previously active connection is terminated by the end devices because the application has been completed.

Upon starting a new period, the controller is responsible for assigning a path to the active (and new) connections, using a route selection algorithm such as that in Algorithm~\ref{alg:routing}. This phase is represented with duration \( \tau_{r} \) in Fig.~\ref{fig:time}. During this phase, all the quantum repeaters are also provided with a periodic schedule of instructions on how to perform entanglement swapping to materialise the routing decisions made by the controller, e.g., adopting the architecture and procedures illustrated in \cite{skrzypczyk2021architecture}.
After that, and until the end of the period, the routes remain statically assigned for the connections that have been accepted, while those rejected (or \emph{blocked}) will not be able to make use of the quantum network in the current period.

Typically, it is assumed that during the period the system will evolve as a sequence of evenly-spaced time slots, each including the phases from 3 to 6 in Fig.~\ref{fig:time}: establish local entanglement in a link (3), perform Bell state measurements (4), transfer the classical results to one of the quantum devices so that it can perform $\sigma_z / \sigma_x$ corrections (5), and consume the qubits on the quantum devices by the distributed applications running on them (6).
The duration of the time slot must be less than the decoherence time of the qubits, which, with today's state-of-the-art technology, is very challenging.
Furthermore, we note that local entanglements and Bell state measurements are stochastic processes, as they may fail with a given probability that depends on the specific technology used.
During the operation of the system, these failures may happen in Phases 3 and 4 and affect the long-term throughput, i.e., the number of end-to-end entangled qubits that can actually be used by the source and destination devices in a unit of time.

Instead, for the rest of this work, we will focus only on Phase 2, i.e., the route selection process, and we can consider the stochastic failures orthogonal to our work.

\begin{figure}[tb]
    \centering
    \includegraphics[width=0.8\columnwidth ]{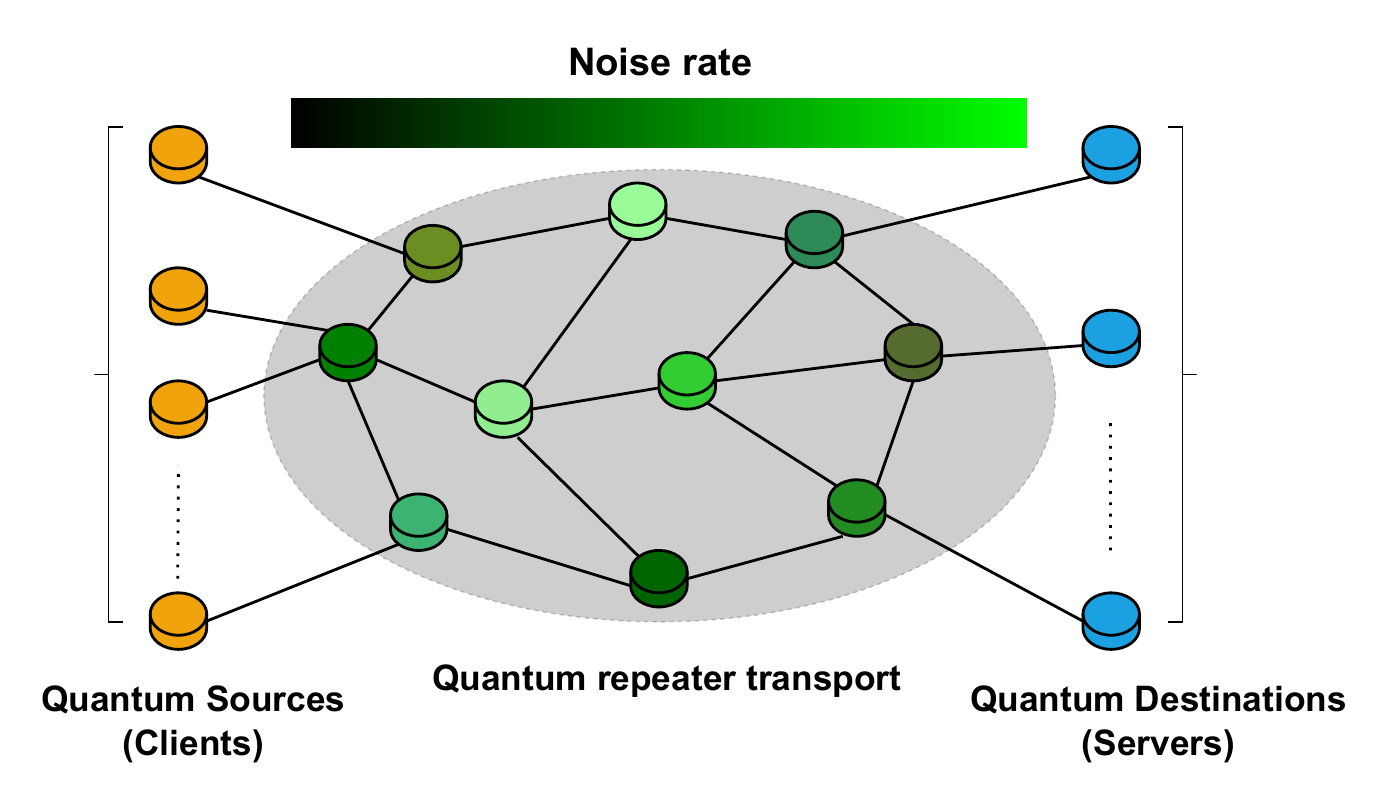}
    \caption{A general schematic of a quantum network with heterogeneous quantum repeaters.}
    \label{fig:schematic}
\end{figure}

\section{Routing with heterogeneous quantum repeaters}
We now consider the first aspect of the routing in quantum networks of this work. In this section, we will only focus on the heterogeneity of quantum repeaters. To this, the quantum network looks as shown in Fig.~\ref{fig:schematic}. Since the network is in the application phase, only application instructions are relevant to the general network architecture in Fig.~\ref{fig:network_architecture} for this aspect.

Across both aspects, in this work, the network topology is modelled as an undirected graph $G(V, E)$, which is static and known to the controller.
Such an assumption is reasonable for a local or metropolitan quantum network that is owned or operated by the same entity, and it is, in fact, valid for intra-domain routing protocols in classical networks, such as Open Shortest Path First (OSPF) and more recent solutions relying on Software Defined Network (SDN) concepts.

\subsection{Problem formulation}\label{ssec:problemformulation_hetero}
The \underline{Quantum Routing With Heterogeneous Quantum Repeaters Problem} is defined as follows:
\textit{Given a quantum network $G(V,E)$ where the repeaters belong to $G$ with characteristic $\eta_g$ and generate EPR pairs with fidelity $F$, a set of $P$ source-destination pairs, and a minimum fidelity threshold $\bar{F}$, find a set of non-overlapping paths such that the end-to-end fidelity computed according to Eq.~\eqref{eq:main_routing_fidelity} is above the threshold $\bar{F}$.}   

Depending on the network topology and the other requirements, it is possible that not all the source-destination pairs can be assigned a path.
When that happens, we say that the source-destination pair was \textit{blocked}, borrowing the term from call admission control schemes terminology.
Of course, during the system operation, it is desirable that blocking happens as sparingly as possible.
In principle, it is possible to convert the problem above to a minimisation problem, where the objective is, indeed, to minimise the number of blocked source-destination pairs, which would be very computationally challenging to solve due to the integrity of variables and the structure of constraints following Eq.~\eqref{eq:main_routing_fidelity}.
Instead, in the following, we adopt a simple and practical path selection procedure, described through the pseudo-code in Algorithm~\ref{alg:shortest-path-entanglement}, which is sufficient for our purposes, i.e., assessing how some of the key system factors affect the performance.
Since the procedure is greedy, i.e., it tries to assign a path to source-destination pairs following their ordering and without backtracking previous decisions, the ordering significantly affects the output of the procedure.
This bias will be removed in the analysis in Section~\ref{sec:performance_hetero_routing} by repeating multiple times the selection with random ordering, which allows us to gather performance in an average case.
The path selection algorithm can be easily extended to include different values of the initial fidelity $F$ of nodes, and a different fidelity threshold $\bar{F}$ for each source-destination pair, but we assume network-wide values for better readability and consistency with the performance evaluation below.

\begin{algorithm}
\caption{Vertex-weighted Dijkstra's algorithm}
\label{alg:shortest-path-entanglement}

\textbf{Input}:
Network as a graph $G(V,E)$, efficiency of all the nodes $\{ \eta_i \}, i \in V$, source-destination pairs $\{ s_k, d_k \}, k = 1..P$, where $P$ is the number of pairs, fidelity threshold $\bar{F}$, weight mapping function $f(\cdot) : \eta \rightarrow w$

\textbf{Output}:

List of end-to-end entangled paths and associated fidelity $\{ \pi_j, F_j \}$ and the number of source-destination pairs for which no path was allocated ($B$)

\begin{algorithmic}[1]
\STATE Paths = $\emptyset$
\STATE B = 0
\STATE Shuffle source-destination pairs
\STATE Derive weighted graph $G'(V',E')$ where $v \in V$ is assigned a weight based on its $\eta$ via $f(\cdot)$ and $E'=E$
\FOR {each pair $\{ s, d \}$}
\STATE Find the shortest path $\pi$ between $s$ and $d$ in $G'$
\STATE allocated=False
\IF {path $\pi$ is found}
\STATE Compute fidelity $F$ of the path $\pi$
\IF {$F \ge \bar{F}$}
\STATE Paths += $(\pi, F)$
\STATE Remove edges in $\pi$ from $E'$
\STATE allocated=True
\ENDIF
\ENDIF
\IF {allocated is False}
\STATE B++
\ENDIF
\ENDFOR
\STATE return \{ Paths, B \}
\end{algorithmic}
\vspace{-0.3em}
\end{algorithm}

\textit{By appropriately tuning $f(\cdot)$, it is possible to obtain a different path selection behaviour that is more or less sensitive to the nodes' efficiency figures (see Eq.~\ref{eq:f_sp_hetero} and \ref{eq:f_ea_hetero}).}

\newcommand{\colAsp}{4.0cm}
\newcommand{\colPol}{7.0cm}

\begin{table}[t]
\centering
\caption{Routing policies considered under each aspect of this work. Shortest-path approach (SP); Knowledge-aware (KA) aka. Efficiency-aware; $k$-shortest path approach (KSP).}
\label{tab:routing_policy_summary}
\begin{tabular}{p{\colAsp} p{\colPol}}
\hline
\textbf{Aspect} & \textbf{Routing policies considered} \\
\hline
Routing with heterogeneous quantum repeaters
& SP and KA \\
\hline
Routing with end-to-end knowledge
& SP, KA, KSP, and $kx$-path selection \\
\hline
\end{tabular}
\end{table}

\section{Routing with end-to-end knowledge}
We now consider the second aspect of the routing in quantum networks of this work. In this section, we introduce the white box and grey box models. This aspect directly builds upon the previous aspect, i.e., it considers the heterogeneous quantum repeater networks, and extends to include additional approaches on top of the shortest path and efficiency-aware (aka knowledge-aware) approaches considered in the first aspect as summarised in Table~\ref{tab:routing_policy_summary}.

\subsection{White box vs grey box models}\label{ssec:grey_white_box}
In the literature, as reviewed in Section~\ref{sec:relatedworks_routing}, it is common to assume that the controller will make use of detailed knowledge about all the quantum repeaters to select a path in a way that goes in the direction of maximising some performance objective, such as throughput or quality of entanglement.
In particular, it is typically assumed that the controller knows:

\begin{enumerate}
    \item The topology of the network, i.e., it knows which quantum repeaters share a link.
    \item Quantum repeater-specific information that can be used to estimate the nominal performance of end-to-end entanglement, usually measured through the \emph{fidelity}, for a given path.
\end{enumerate}

In this study, we question the latter assumption: estimating the fidelity of an end-to-end entanglement in real conditions would be an overly complex task since the latter is determined by a complex combination of time- and technology-dependent processes involving, among the others, the heralding of local entanglement on a link \cite{li2021heralding}, the purification of such an entanglement \cite{chen2024optimum}, the decoherence of quantum memories \cite{pompili2021realization}, the link-level protocols \cite{dahlberg2019link}.
We therefore advocate an alternative approach, in which such end-to-end fidelity is not estimated \emph{a~priori} but measured \emph{a~posteriori} through an in-band mechanism running in the network, which, e.g., assumes periodic activation of ``sampling instructions'' (Fig.~\ref{fig:network_architecture}) on unused links for the purpose of acquiring/refreshing the information on the measured fidelity between any two quantum devices in the network.
For instance, randomised benchmarking techniques could be used for this~\cite{helsen2023benchmarking}.
The specific procedure that can be used for this purpose is beyond the scope of this work, in which we merely assume that the so-called Storage component is made aware of the (average) fidelity that can be achieved for any given path, without knowing the details of how much each link/node contributes to the end-to-end result.

It is known that the number of possible paths in a generic network can be very large, which would render this approach practically unfeasible, but we make the following observation to advocate such a solution, at least in some scenarios. The topology of a network is expected to remain stable for a time that is much longer than the time scales relevant for execution of distributed quantum computing applications, as modifying the topology implies the installation of a device (quantum computer or repeater) and its physical interconnection with the quantum network or provisioning additional infrastructures (satellite or optical fibre).
Therefore, the information about end-to-end fidelity may be accumulated in the Storage over time through a background process, prioritising shorter paths between pairs of end nodes, which consume fewer resources and are more likely to exhibit a higher fidelity. However, it is crucial to consider that network topology information can be influenced by environmental factors beyond physical layout changes. Such environmental interactions may alter link performance, thereby impacting the accuracy of end-to-end fidelity predictions.

Consequently, developing efficient methods to manage the accumulation and updating of end-to-end fidelity knowledge remains an essential area for further research, as noted in recent studies (e.g., \cite{helsen2023benchmarking}). This topic, while outside the scope of this work, is complementary to our approach. In the meantime, given the early stages of development in this direction, we will introduce robustness tests to evaluate the methods under conditions with unreliable or outdated end-to-end knowledge, as will be discussed in the subsequent chapter.

\textbf{Implication for problem formulation:} While both the heterogeneous quantum repeater network (Section~\ref{ssec:problemformulation_hetero}) and the present grey-box approach operate on networks that physically contain nodes with different efficiency classes, \textbf{the key difference lies in what information is available to the routing algorithm.} In Section~\ref{ssec:problemformulation_hetero}, the controller knows the efficiency $\eta_i$ of each node $i$ and uses this to compute vertex weights. In the grey-box formulation that follows (Section~\ref{ssec:problemformulation_endtoend}), these efficiency values are \textbf{not available} to the routing controller. Instead, routing decisions are based solely on network topology $G(V,E)$ and end-to-end fidelity estimates obtained through sampling.

\subsection{Problem formulation}\label{ssec:problemformulation_endtoend}
For this aspect, the \underline{Routing Problem} is defined as follows: \textit{For a quantum network modelled as a graph $G(V, E)$ with $G$ classes of quantum repeaters, where each $v \in V$ is a quantum repeater with efficiency figure $\eta_g$ ($g \in G$) and $e \in E$ is a link between two quantum repeaters, a set of connections of cardinality $n_{sd}$ requesting to use the network to create end-to-end entanglement of qubits between a source and a destination with a fidelity threshold $F_{th}$, find a set of non-overlapping paths that minimises the blocking probability while satisfying the fidelity requirement}.

We address the problem in this aspect under the general framework illustrated by Algorithm~\ref{alg:routing}, by proposing several specific policies. The framework takes as input the quantum network topology represented as a graph \( G(V, E) \), the end-to-end (E2E) fidelity estimations, and a set of source-destination pairs \( \{s_k, d_k\} \) (where \( n_{sd} \) represents the total number of such pairs). Additionally, a path selection policy \( \mathcal{A} \) is employed to compute the optimal paths for the connections.

The algorithm begins by initialising the set of assigned paths \( P_a \) to an empty set (line 1). For each request to establish entanglement between source-destination pairs \( \{s, d\} \) (line 2), the policy \( \mathcal{A} \) is used to compute a feasible path \( \pi \) between the source and destination nodes in the network graph \( G(V, E) \), considering the network's current state and fidelity constraints (line 3). If a path \( \pi \) is successfully found, it is added to the set \( P_a \) (line 4). If no such path can be established, the connection between the source-destination pair \( \{s, d\} \) is marked as \emph{blocked} for that period (line 5).

The process repeats for all \( n_{sd} \) pairs, and at the end of the algorithm, the final set of assigned paths \( P_a \) is returned as the output (line 7).

\begin{algorithm}
\caption{Route selection algorithm}
\label{alg:routing}

\textbf{Input}:
Network as a graph $G(V,E)$, E2E fidelity estimations, source-destination pairs $\{ s_k, d_k \}, k = 1..n_{sd}$, where $n_{sd}$ is the number of connections, and $\mathcal{A}$ is the path selection policy

\textbf{Output}:

Set of assigned paths $P_a$

\begin{algorithmic}[1]
\STATE $P_a$ = $\emptyset$

\FOR {each pair $\{ s, d \}$}
\STATE compute the path $\pi$ using policy $\mathcal{A}$ between $s$ and $d$ in $G(V,E)$
\STATE if $\pi$ is found, add it to $P_a$
\STATE otherwise the connection with pair $\{s,d\}$ is \emph{blocked} in this period
\ENDFOR
\STATE return $P_a$
\end{algorithmic}
\vspace{-0.3em}
\end{algorithm}

\subsection{Routing policies}\label{ssec:routing_policies}
The policies are as follows:

\begin{enumerate}
    \item \textbf{Shortest-path approach:} This policy entails the straightforward use of Dijkstra's algorithm, i.e., selects one of the shortest paths between source and destination from the set of those with fidelity above the threshold. When more such shortest paths exist, one is picked randomly with uniform distribution. Considering the generic path $p$ characterised by a path length $L_p$, an end-to-end fidelity $F_p$, and considering a fidelity threshold $F_{th}$, the path selection under the shortest-path approach selects the best path $P_a$ as follows:
\begin{equation}\label{eq:sp} 
P_a = \arg\min_{p} \{L_p \mid F_p \geq F_{th}\}.
\end{equation}

\item \textbf{knowledge-aware approach:} This policy uses additional knowledge of quantum repeaters, i.e., their efficiency figures. Therefore, it is a \emph{white-box approach} solution to the quantum routing problem. Each node $v$ in the quantum network is assigned a cost $c(v)$ according to its efficiency figure. Then, a vertex-weighted Dijkstra's algorithm is employed for path-selection. This approach is also known as the efficiency-aware approach with respect to the first aspect~\cite{kumar2024routing}. 

Path selected:

\begin{equation}\label{eq:ka} 
P_a = \arg\min_{p} \{C_p = \sum_{v \in P} c(v) \mid F_p \geq F_{th}\}.
\end{equation}

\item \textbf{$k$-shortest path approach:} This policy selects the path with minimum end-to-end fidelity (yet above fidelity threshold) out of the first $k$ shortest-paths between the source and destination list, where $k$ is a configuration parameter of the algorithm. Note that those paths may be of different lengths (not all equal to the shortest one). The intuition behind this policy is that selecting paths with a lower fidelity, corresponding to quantum repeaters in low-quality classes, in the early rounds of the loop in Algorithm~\ref{alg:routing} should leave more resources available to connections under evaluation later on.
Since no per-node/per-link information is needed, the $k$ shortest-path is a \emph{grey-box approach} solution to the quantum routing problem. A more descriptive name for this algorithm could be ``min-k-shortest path", as it clarifies that the algorithm does not involve random selection among shortest paths or selection of multiple paths. However, for simplicity, we retain the term ``$k$-shortest path" throughout the paper. 

Path selected:
\begin{equation} \label{eq:ksp}
P_a = \arg\min_{p} \{F_p \mid F_p \geq F_{th}\}.
\end{equation}

\item \textbf{$kx$-path selection approach:} This policy is another \emph{grey-box approach} solution to the quantum routing problem and is a restrictive version of the $k$-shortest path approach: when selecting the path, it only considers candidate routes that do not deviate too much from the shortest possible one, depending on the configuration parameter $x$, which can be seen as an allowance factor. When $x = 0$, the path must be no longer than the shortest available; when $x = 1$ it is allowed to be one hop longer than the shortest one, and so on. Let us define the length of the shortest path as follows:

\begin{equation} \label{eq:kxb}
L_b = \min_{p} \{L_p \mid F_p \geq F_{th}\}.
\end{equation}

Path selected:
\begin{equation} \label{eq:kxa}
P_a = \arg\min_{p} \{F_p \mid L_p \leq (L_b + x) \wedge F_p \geq F_{th}\}.
\end{equation}
If more than one path satisfies Equation~\ref{eq:kxa}, the algorithm selects a random one with a uniform distribution.
\end{enumerate}

\chapter{Performance Evaluation and Simulation Framework}\label{ch:performance_evaluation}

In this chapter, the performance evaluation and the simulation framework for the system discussed in Chapter~\ref {ch:routing} are presented. While following a similar methodology, network models and topologies, we evaluate and discuss the results from both aspects of routing in quantum networks of this work, as categorised in Chapter~\ref{ch:routing}.

\section{Network models and topologies}\label{sec:topology}

\begin{figure}
\centering
    \begin{subfigure}{0.52\textwidth}
        \includegraphics[width=\linewidth, trim={3cm 3cm 3cm 3cm}, clip]{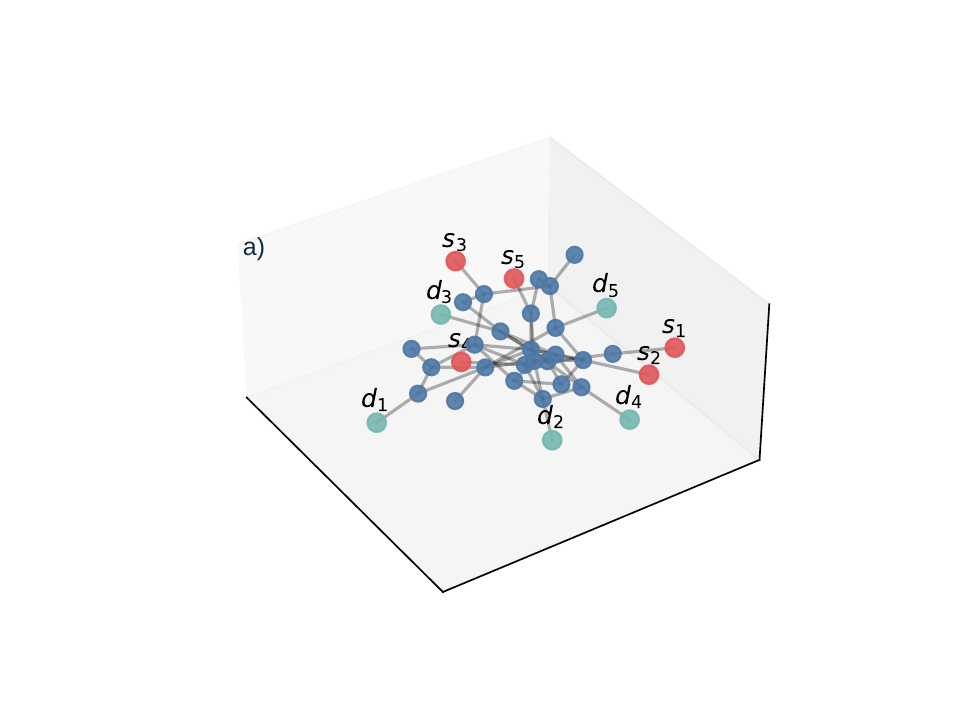}
        \caption{Random}
    \end{subfigure}
    \hspace{-1cm}
    \begin{subfigure}{0.52\textwidth}
        \includegraphics[width=\linewidth, trim={3cm 3cm 3cm 3cm}, clip]{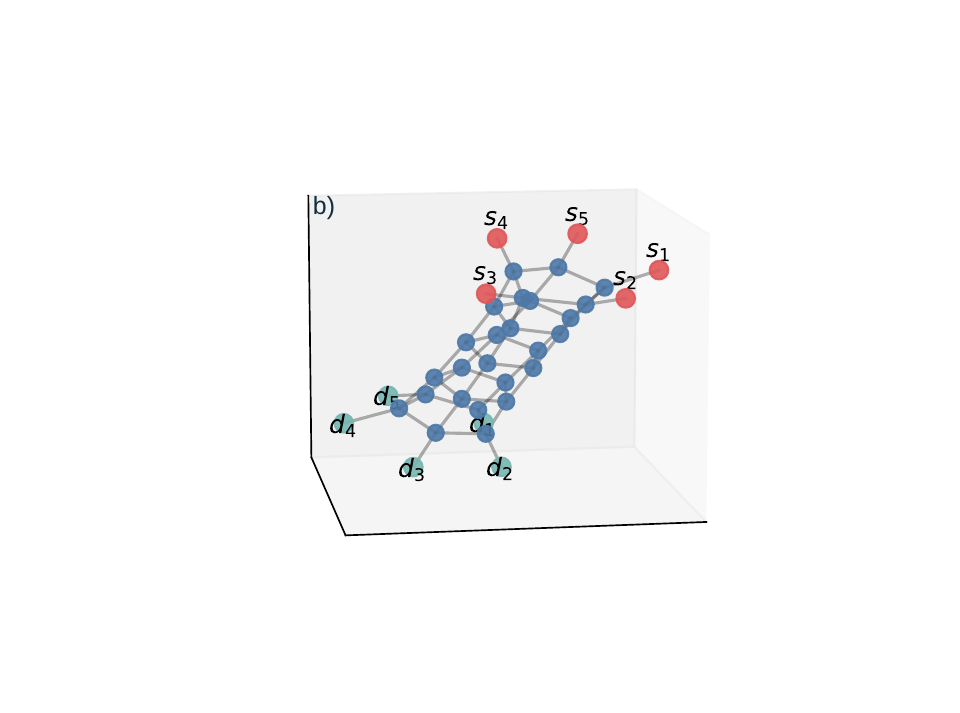}
        \caption{Regular}
    \end{subfigure}

    \caption{Topology of the quantum network for \( n_{sd} = 5 \) and 25 grid nodes. Source nodes \( s_n \) (in red), destination nodes \( d_n \) (in light blue), and grid nodes (in dark blue). In a regular network structure, the source-destination (SD) pairs are positioned farther apart, leading to a probability mass function (PMF) with less diversity in path lengths and a higher proportion of longer paths compared to a random network (taken from \cite{kumar2025routing}).
}
    \label{fig:topology}
\end{figure}

Given the absence of a functional quantum network and the limited experimental capability to manage only a few nodes, the design of quantum networks poses significant challenges, extending to the broader scope of the quantum internet. First, finding a ``realistic'' network topology is an issue \emph{per se}, which so far has been addressed inspired by classical networks.

A significant amount of current research utilises the Waxman model \cite{waxman2002routing} to acquire the random topology of the quantum network, as discussed later in this section. Others use lattices (grid and its variations) \cite{kumar2024routing, kumar2025routing}, rings \& spheres \cite{schoute2016shortcuts}, linear chains \cite{dai2020optimal}, and specialised topology, that is, SURFnet \cite{chakraborty2020entanglement, rabbie2022designing, nguyen2025maximizing} and US backbone network \cite{li2022fidelity, nguyen2025maximizing}.

\newcommand{\colA}{3.6cm}  
\newcommand{\colB}{2.4cm}  
\newcommand{\colC}{3.0cm}  
\newcommand{\colD}{3.2cm}  

\begin{table}[t]
\centering
\caption{Topology and quantum repeater class selection for the two routing aspects considered in this work.}
\label{tab:topology_summary}
\begin{tabular}{p{\colA} p{\colB} p{\colC} p{\colD}}
\hline
\textbf{Aspect} 
& \textbf{Topology class} 
& \textbf{Topology type} 
& \textbf{Number of quantum repeater classes} \\
\hline
\multirow{3}{\colA}{\parbox{\colA}{Routing with heterogeneous quantum repeaters}}
& \multirow{2}{\colB}{\parbox{\colB}{Regular}}
& Grid
& \multirow{3}{\colD}{\parbox{\colD}{2 (HQ, LQ)}} \\
\cline{3-3}
& 
& Cylindrical
&  \\
\cline{2-3}
& \parbox{\colB}{Random}
& Not considered
&  \\
\hline
\multirow{2}{\colA}{\parbox{\colA}{Routing with end-to-end knowledge}}
& \parbox{\colB}{Regular}
& Cylindrical
& \multirow{2}{\colD}{\parbox{\colD}{2 (HQ, LQ) and continuous}} \\
\cline{2-3}
& \parbox{\colB}{Random}
& Random graph
&  \\
\hline
\end{tabular}
\end{table}

With respect to the generic network shown in~Fig.~\ref{fig:network_architecture}, we consider two classes of network topologies, namely \emph{regular} and \emph{random}, as illustrated in Fig.~\ref{fig:topology} and summarised in Table~\ref{tab:topology_summary}. 

For the first aspect, namely routing with heterogeneous quantum repeaters, we restrict our analysis to regular topologies. In particular, we consider an $n \times n$ square \emph{grid} topology and its cylindrical variant, obtained by introducing wrap-around connections between the top and bottom transport nodes to eliminate boundary effects. The latter has the advantage that there
are no edge effects and it is always possible to find a path
between any two nodes in the source/destination tiers.

For the second aspect, namely routing with end-to-end knowledge, we retain the cylindrical topology from the class of regular networks and additionally introduce random topologies. The retention of cylindrical topology is guided by the insights obtained from the first aspect. 

The source and destination nodes are located on the opposite edge of the transport network grid. In the random topology, we use the well-known Waxman model~\cite{waxman2002routing} to generate the specific network topology. Specifically, we use the Waxman model to generate the network connecting the quantum repeaters among them. Then, we select $n_{sd}$ repeaters using a uniform random distribution and attach them to the source nodes, and another $n_{sd}$ repeaters to attach to the destination nodes.

When generating the network among repeaters according to the Waxman model, a link between any two nodes $u$ and $v$ is added with the following probability:

\begin{equation} \label{eq:wax}
P(u, v) = \beta e^{\left( \frac{-d(u, v)}{L \alpha} \right)},
\end{equation}

where, $d(u, v)$ represents the distance from node $u$ to node $v$, $L$ denotes the maximum distance between any two nodes, and $\alpha$ and $\beta$ are parameters that fall within the range $(0, 1)$. Higher values of $\beta$ lead to graphs with greater edge densities. Conversely, lower values of $\alpha$ tend to increase the proportion of shorter edges compared to longer ones.

We initially consider two classes of quantum repeaters for the first aspect and move towards continuous classes of quantum repeaters for the second aspect, as summarised in Table~\ref{tab:topology_summary}. Depending upon the efficiency of the quantum repeater, the node is categorised as high-quality (HQ), having efficiency $\eta_h$, or low-quality (LQ), with efficiency $\eta_l$. For what KA (Knowledge-aware approach) is concerned, we assign the costs so that $c(v_{\mathrm{HQ}}) \gg c(v_{\mathrm{LQ}})$.  In the first set of experiments, we vary the fraction of HQ and LQ repeaters. Specifically, a parameter $\xi$ is defined as the fraction of high-quality nodes among all quantum repeaters, i.e., $\xi = 1 (\xi = 0)$ means that all quantum repeaters are HQ (LQ). 

Moreover, to extend the analysis with respect to this standpoint, in another set of experiments, we allocate efficiency values over a continuous interval between the HQ and LQ values, according to the following considerations. If we selected efficiencies of repeaters according to a uniform distribution in this interval, as the end-to-end fidelity exponentially decays with respect to the efficiency figures (Eq.~\eqref{eq:main_routing_fidelity}), the fidelity of end-to-end paths would be predominantly influenced by the efficiency figures closer to LQ. Therefore, to obtain a better balance, we use the following algorithm to assign efficiency values. For a given repeater $i$, we use two ancillary parameters ($a$ and $x_i$) and write the efficiency $\eta_i$ as follows:

\begin{equation}
\eta_i = \frac{\ln \left( x_i  \right)}{a},
\end{equation}
where $a$ is a constant that tunes the bias towards higher efficiency values ($a=1$ means unbiased) and $x_i$ is drawn from a uniform distribution in the following range:

\begin{equation}\label{eq:range}
U\left( e^{a \eta_{l}}, e^{a \eta_{h}} \right).
\end{equation}

Finally, a parameter $\theta$, path establishment order, is defined that represents the order (from 1 to the number of source-destination (SD) pairs considered in the experiments, $n_{sd}$)  of a specific SD pair in the sequence used to serve all pairs.

\section{Evaluation metrics}
The performance of the proposed policies (Section~\ref{ssec:routing_policies}) is evaluated based on the following set of indices.

First, we consider the \emph{Blocking Probability} (BP) across the entire set of considered SD pairs, computed as the fraction of SD pairs to which the routing algorithm assigns a path, over the total number of SD pairs requesting connectivity. A lower blocking probability indicates better network performance. In addition, for some specific analyses, we also consider the Blocking Probability per edge (BP/E), defined as the Blocking Probability divided by the number of edges in the network. As with blocking probability, a lower BP/E reflects improved network performance. Additionally, we analyse the probability mass function (PMF) of path length, which describes the probability distribution of the path lengths selected to connect SD pairs. PMF varies with the routing approach, the network parameter $\xi$, and the network topology.

A second index we use is the \emph{fidelity} achieved by SD pairs over the paths allocated by the routing algorithm, computed as in Eq.~\eqref{eq:main_routing_fidelity}. We also consider the path establishment order ($\theta$), which represents the sequence in which SD pairs are served. In an optimised approach, the path order should minimally impact network performance. However, as shown in the results, this behaviour may not hold, particularly when the number of source-destination (SD) pairs increases, which can reduce network performance.

Finally, we characterise the \emph{fairness} achieved by the served SD pairs, as those served later may obtain ---if the routing algorithm is not cautious--- quite lower resources than those served first. To this end, we use Jain's fairness index~\cite{jain1984quantitative} computed over the number of paths successfully allocated to each SD pair across all replicas of a given experimental configuration, i.e.:

\begin{equation} \label{eq:jain}
J(x_1, x_2, ..., x_n) = \frac{\left( \sum_{i = 1}^{n} x_i \right)^{2}}{n \sum_{i = 1} n x_{i}^{2}},
\end{equation}
where, $x_i$ is the count of path delivered to the $i_{th}$ (source, destination) pair.

\section{Simulation and methodology}\label{sec:simulation}
To evaluate the system model and routing strategies introduced in Chapter~\ref{ch:routing} and above, we employ a specialised Python-based simulation framework, referred to as \emph{QuantaRoute}. The simulator was developed to model routing decisions under realistic quantum-network constraints and to enable controlled performance evaluation.

The implementation addressing the first aspect of routing in this work, namely routing with heterogeneous quantum repeaters, is provided in the original \emph{QuantaRoute} codebase\footnote{\url{https://github.com/vk9696/quantum-routing1}}. To study the second aspect, routing with end-to-end knowledge, we developed an extended simulator, \emph{QuantaRoute~2.0}, which builds upon the original framework and incorporates additional topology requirements, routing policies, and evaluation metrics\footnote{\url{https://github.com/vk9696/quantum-routing2}}.

 Throughout this study, a quantum network consisting of 25 quantum repeaters is used, irrespective of topology and the number of source-destination pairs involved. A total of 100 random combinations of source and destination nodes are considered in addition to 100 random assignments of HQ and LQ nodes to the transport network. This leads to a total of $10,000 \times n_{sd}$ potential paths i.e., $50,000$ for $n_{sd} = 5$ for each $\xi$ value. The performance results for each $\xi$ are obtained by averaging over random assignment of HQ and LQ along with t-Student's 95\% confidence intervals wherever required. While complete performance results are presented for the first aspect of this work, namely routing with heterogeneous quantum repeaters, the evaluation of the second aspect is restricted to the range $0.5 < \xi < 1$. Considering network configurations with $\xi < 0.5$ would further degrade performance, as such scenarios correspond to networks dominated by lower-quality nodes, offering limited additional insight. This restriction is adopted to improve clarity and facilitate a more focused analysis of the second aspect, whose behaviour and interactions are inherently more complex than those of the first. The fidelity of the initially shared EPR pairs should be as high as possible, ideally close to 1, to tolerate the reduction in fidelity caused by noisy entanglement swapping procedures, thus ensuring that the entanglement shared between the source-destination pair remains above the fidelity threshold. However, achieving a fidelity of 1 is unrealistic; therefore, we assume the initial EPR pairs shared among neighbouring nodes have a fidelity of 0.975. We set the efficiency figures for high-quality and low-quality nodes at $\eta_h = 0.999$ and $\eta_l = 0.8$, respectively. It is important to note that the typical fidelity required by quantum applications should exceed 0.5. Consequently, we establish the threshold fidelity required for quantum applications at 0.53 unless a study of its variations is conducted in the first aspect. However, applications running on end nodes may demand higher fidelity levels. Our choice of a lower fidelity threshold enables clearer differentiation of algorithmic performance and enhances the interpretability of the results. Notably, in the first aspect \cite{kumar2024routing}, algorithm performance may converge at higher fidelity thresholds. Additionally, we set $k = 10$ for both the $k$-shortest path and $kx$ path-selection approaches, providing a sufficient number of shortest paths from which to select. For the first aspect, the analysis consists of four macro scenarios, which tackle four different aspects: quantum network topology, low-quality efficiency sensitivity, efficiency awareness, and blocking probability. For the second aspect, the results are categorised into six categories, i.e., blocking probability, fairness, number of source-destination pairs, average degree, robustness, and range of efficiency figures. For each of the metrics, we evaluate the performance of all the path selection policies defined in Section~\ref{ssec:routing_policies}: shortest-path (SP) and knowledge-aware (KA), as baselines, and $k$-shortest path (KSP) and $kx$-path selection (with $x=0$ as $kx_0$ and $x=1$ as $kx_1$), as the proposed novel approaches.

\begin{figure*}[tb]
\centering
    \includegraphics[width=\columnwidth]{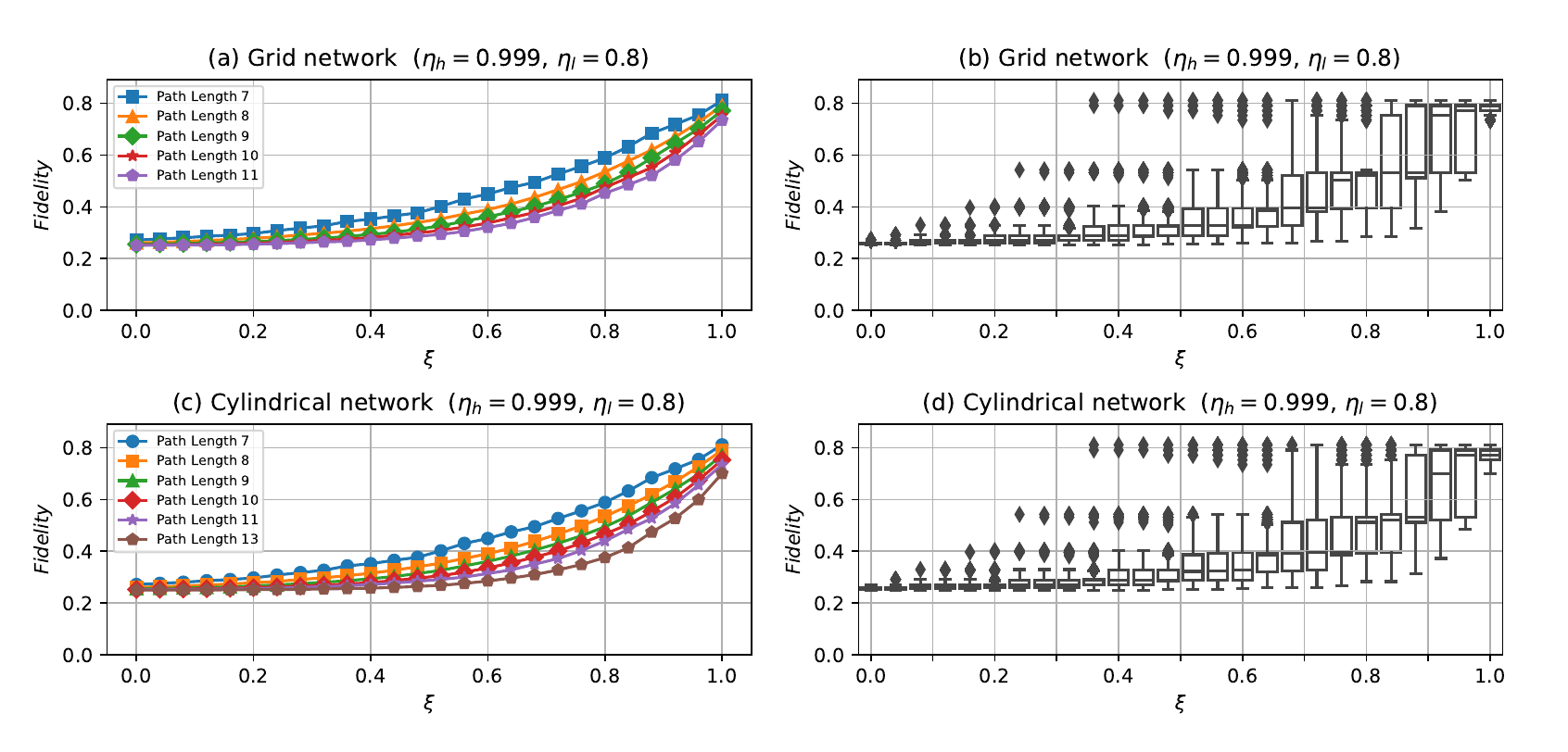}
    \caption{a) \& b) Fidelity vs $\xi$ for Grid network, c) \& d) Fidelity vs. $\xi$ for Cylindrical network}
    \label{fig:fid_vs_hq}
\end{figure*}

\section{Routing with heterogeneous quantum repeaters}\label{sec:performance_hetero_routing}

\begin{table}[h!]
\centering
\caption{Notation reference for chapter~\ref{ch:performance_evaluation}}
\label{tab:notation_table}
\begin{tabular}{cp{5cm}p{1.5cm}}
\toprule
Notation & Explanation & Default value \\
\midrule
$\xi$ & Fraction of high-quality nodes in the network &  \\
$\eta_h$ & Efficiency figure of a high-quality node & 0.999 \\
$\eta_l$ & Efficiency figure of a low-quality node & 0.8 \\
$F_{th} \ \text{or} \ \bar{F}$ & Threshold fidelity &  0.53\\
$F_p$ & End-to-end fidelity of the path &  \\
$F$ & Fidelity of initial EPR pairs shared & 0.975 \\
$L_p$ & Path length, i.e., number of nodes in the path &  \\
$\alpha$ & Waxman parameter that controls the proportion of shorter edges compared to longer &  0.85\\
$\beta$ & Waxman parameter that controls the edge density, i.e., average degree & 0.275 \\
$k$ & Algorithm parameter for number of candidate shortest paths & 10 \\
$\theta$ & Path establishment order &  \\
\bottomrule
\end{tabular}
\end{table}

In this aspect, the methodology described above is applied to both grid and cylindrical quantum networks as in Table~\ref{tab:topology_summary} and Section~\ref{sec:topology}.

\subsection{Quantum Network Topology Study}

As a recap, we set $\eta_h = 0.999$ and $\eta_l = 0.8$, and we use $f(\cdot) = 1$ (Eq.~\eqref{eq:f_sp_hetero}) and $\bar{F} = 0$, i.e., the path selection is independent of the efficiency figure of nodes and their fidelity.

Due to topology, the path length of grid vs. cylindrical networks is between 7 and 11 hops (or 13 only for the cylindrical case), with an average path length of 8.61 (grid) and 9 (cylindrical). As already introduced, paths are never blocked in a cylindrical network, while the blocking probability in the grid network, which has fewer resources (links), is found to be 0.28. This is a direct implication of more resources (i.e., additional links) in cylindrical networks. Importantly, it should be noted that the paths which are blocked in the grid network are redundant for any practical purposes of establishing end-to-end entanglement, and therefore, these are not considered in Fig.~\ref{fig:fid_vs_hq}.

In Fig.~\ref{fig:fid_vs_hq}a and \ref{fig:fid_vs_hq}c, we report the average fidelity when  $\xi$ increases from 0 to 1, i.e., when the network has an increasing fraction of HQ nodes.

As expected, all the curves increase with $\xi$, and shorter paths have a higher fidelity. We note that the presence of the 13-hop paths in the cylindrical case yields a slightly decreased average fidelity (0.3896 vs.\ 0.3972 for the grid case -- not shown in the plots).

The fidelity is also reported in Fig.~\ref{fig:fid_vs_hq}b and \ref{fig:fid_vs_hq}d as box plots. As can be seen, for $\xi < 0.32$ in the grid network and for $\xi < 0.36$ in the cylindrical network, fidelity is almost zero apart from outliers. Hence, it is of no use to upgrade 32\% in the grid network and 36\% in the cylindrical network of LQ nodes to the more expensive/sophisticated HQ nodes. The LQ nodes are observed to be bottlenecks in the performance for both grid and cylindrical networks. Fidelity is barely 0.5 even if $\xi$ is 0.8 in the grid network and 0.84 in the cylindrical network. However, a significant jump in fidelity, of about 50\%, is observed if one more LQ node is replaced with an HQ node for both grid and cylindrical networks. In conclusion, to have a useful impact on network performance from switching some of the LQ nodes to HQ nodes in a quantum network, $\xi$ should be very close to 1, i.e., in our simulations, at least 0.84 for the grid network and 0.88 for the cylindrical network.

\subsection{Low-quality efficiency sensitivity study}
In this case-study, we study the sensitivity of $\eta_l$ for a constant $\eta_h$ in the cylindrical core-based quantum network: while keeping $\eta_h = 0.999$ we sweep $\eta_l \in \{0.99, 0.8\}$.
Like in the previous scenario, we assume $f(\cdot) = 1$ and $\bar{F} = 0$.

\begin{figure}[tb]
    \centering
    \includegraphics[width=0.8\columnwidth]{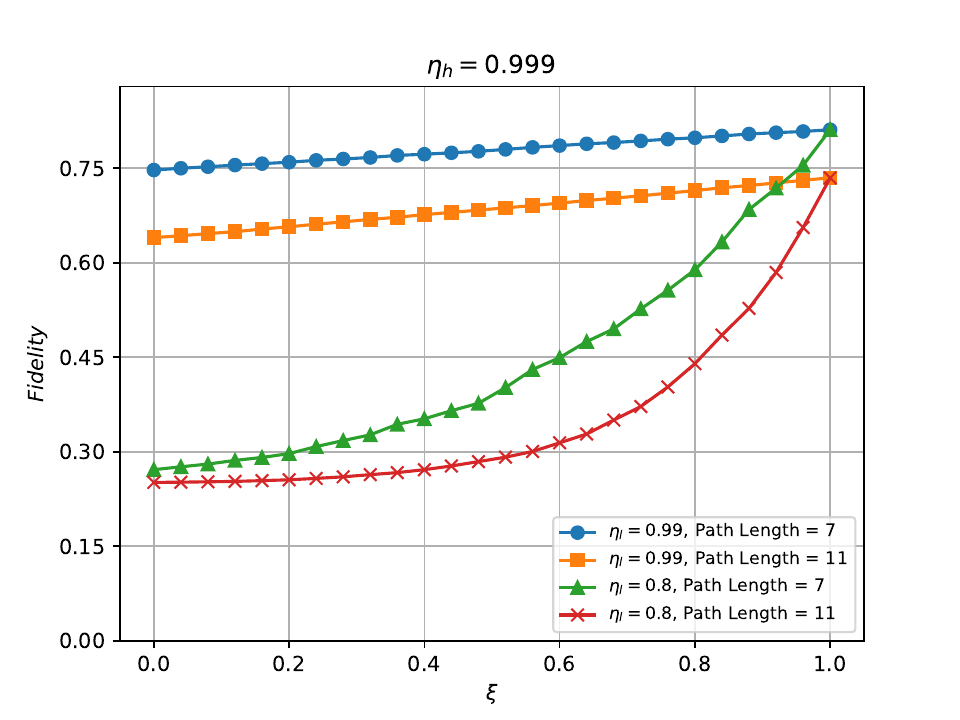}
    \caption{Fidelity vs. $\xi$}
    \label{fig:sub_fid_vs_hq}
\end{figure}

We show the fidelity in Fig.~\ref{fig:sub_fid_vs_hq}, only for paths of length 7 and 11, for better readability. As can be seen, a relatively higher $\eta_l$ can achieve considerably higher fidelities even if path length and $\xi$ are low, as shown by blue and orange squared points. Conversely, if $\eta_h$ is not high enough, then the network performance cannot be improved irrespective of path lengths and $\xi$ involved, as shown by green and red triangle points. In conclusion, for a network consisting of LQ nodes with a high efficiency figure, replacing some of the nodes with HQ ones is not effective.

\subsection{Efficiency Awareness Study}\label{sssec:noise study}

In this case-study, we study the efficiency awareness approach against the standard shortest-path approach in the cylindrical-core-based quantum network. We keep $\bar{F}=0$ but consider two possible mapping functions as follows:
\begin{equation}\label{eq:f_sp_hetero}
f_{\mathrm{shortest-path}} = 1 \ \forall \ \eta
\end{equation}
\begin{equation}\label{eq:f_ea_hetero}
f_{\mathrm{efficiency-aware}} = 
\begin{cases} 
100 & \text{if } \eta = \eta_l \\
1 & \text{if } \eta = \eta_h 
\end{cases}
\end{equation}
\begin{figure}[tb]
\centering
    \includegraphics[width=0.8\columnwidth]{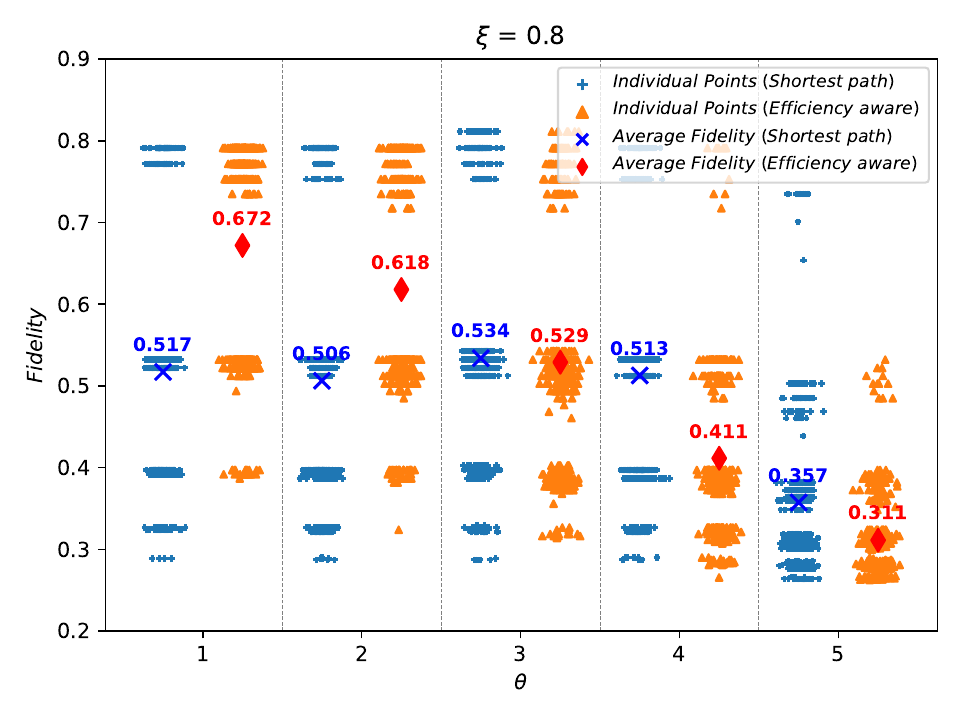}
    \caption{Fidelity vs. $\theta$ for $\xi$ = 0.8}
    \label{fig:fid_theta}
\end{figure}
In the standard shortest-path selection approach, the efficiency figures of quantum repeaters are not taken into account, leading to uniform weighting of all nodes. Conversely, in the efficiency-aware approach, nodes of lower quality are assigned disproportionately higher weights compared to those of higher quality.

The Fig.~\ref{fig:fid_theta} shows a scatter plot of the fidelity values for $\xi$ = 0.8 for the two mapping functions above. We consider the standard shortest-path case first. The study shows that for $\xi > 0.2$, the fidelity of the first four paths is found to be closer to each other with a considerably lower fidelity for the last path. This is clearly evident from Fig.~\ref{fig:fid_theta}
i.e., $\xi$ = 0.8, where the first four paths have an average fidelity of approximately 

0.51 while the last path has a fidelity much lower, i.e., 

0.31.

In the efficiency-aware scenario, there is a noticeable linear decrease in fidelity with increasing values of $\theta$. This trend is attributable to the consumption of high-quality (HQ) nodes by the initial paths, which subsequently leaves a predominance of lower-quality (LQ) nodes for later paths. Fig.~\ref{fig:fid_theta} illustrates significant improvements in higher fidelities for the efficiency-aware approach (depicted in orange) over the shortest-path approach (in blue), particularly at the initial stages of $\theta$. The contrast is starkly evident as there is a substantial increase in the number of paths served at higher fidelities for the efficiency-aware method compared to the shortest-path approach for $\theta = \{1, 2 \}$. However, for $\theta$ = 3, the average fidelity of the paths served by both approaches converges. Notably, for $\theta$ values of 4 and 5, the shortest-path approach outperforms the efficiency-aware method. This pattern indicates that the efficiency-aware approach prioritises enhancing the fidelity of initial paths, thereby affecting the overall distribution of fidelity in subsequent paths.

For the scenarios where $\xi = 0$ and $\xi = 1$, the fidelity distribution remains consistent across all values of $\theta$ for both shortest-path and efficiency-aware approaches, as anticipated. This consistency arises because, at $\xi = 0$, only low-quality (LQ) nodes are present, whereas at $\xi = 1$, only high-quality (HQ) nodes exist. In conclusion, the efficiency awareness technique can be used to strategically improve the fidelities of initial paths to manage workload in a quantum network.

\subsection{Blocking Probability Study}

\begin{figure}[tb]
    \centering
    \includegraphics[width=0.8\columnwidth]{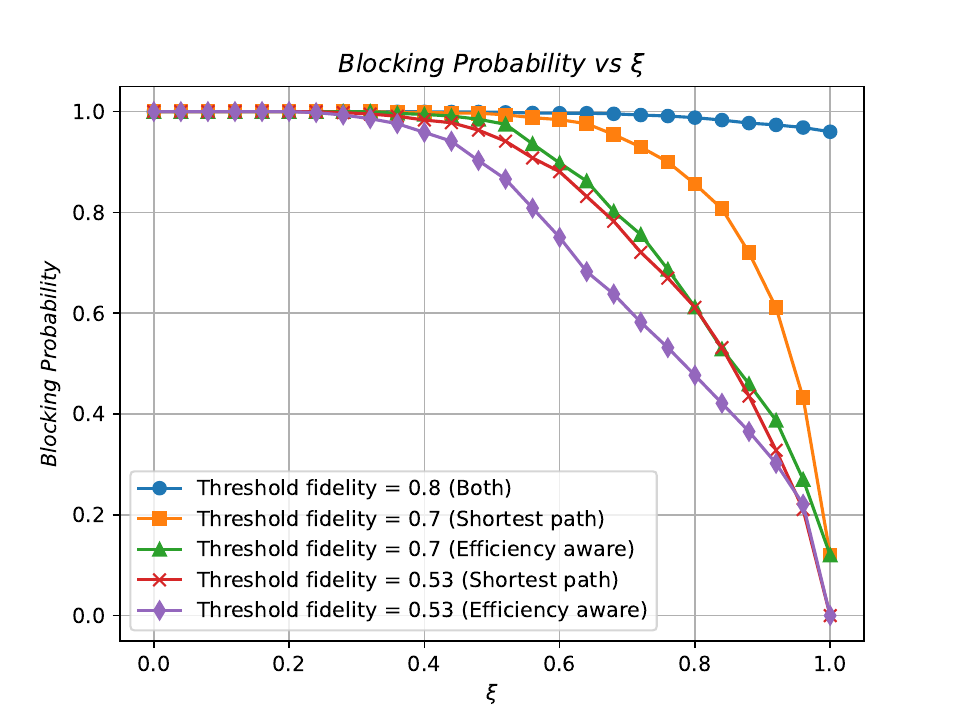}
    \caption{Blocking Probability vs. $\xi$}
    \label{fig:bp_vs_hq}
    \vspace{-0.5cm}
\end{figure}
\begin{figure}[tb]
    \centering
    \includegraphics[width=0.8\columnwidth]{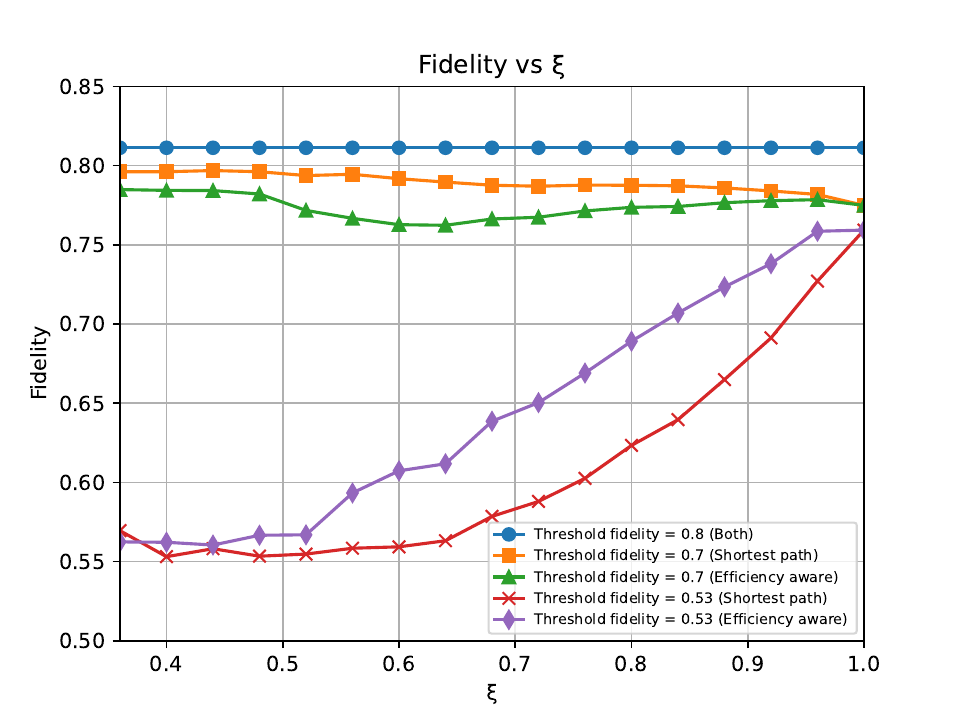}
    \caption{Fidelity vs. $\xi$ with fidelity threshold}
    \label{fig:fid_vs_xi_thres}
\end{figure}
In this final study, we introduce the threshold fidelity $\bar{F}$ into the simulation, which is set to 0.53, 0.7, and 0.8, respectively, in a cylindrical network topology. In addition to the methodology followed in the above study, we reject paths under specific threshold fidelities.
Note that values below 0.5 are too low for any practical use cases.

Fig.~\ref{fig:bp_vs_hq} illustrates the relationship between blocking probability and $\xi$ across various values of $\bar{F}$ with the corresponding fidelity delivered in Fig.~\ref{fig:fid_vs_xi_thres}. It is important to highlight that when $\xi < 0.35$, none of the paths achieved fidelity levels exceeding their respective thresholds. Consequently, these paths are designated as blocked and are subsequently excluded from the plot. Notably, a lower $\bar{F}$ corresponds to a reduced blocking probability, as fewer paths fall below the threshold value regardless of the approach employed.

In the comparative analysis depicted in Fig.~\ref{fig:bp_vs_hq}, it is evident that the efficiency-aware approach yields a markedly lower blocking probability than the shortest-path approach when the fidelity threshold, $\bar{F}$, is set at ${0.53, 0.7}$. As $\bar{F}$ approaches 0.8, however, the disparity between the two methods diminishes. This convergence in performance can be attributed to the high fidelity threshold nearing 0.8, at which point the blocking probability approaches unity. Consequently, at such elevated thresholds, the opportunity for the efficiency-aware approach to further improve performance becomes negligible.  

In conclusion, the efficiency awareness technique improves the network performance in terms of blocking probabilities, i.e., fewer paths would be blocked upon using the efficiency awareness technique.

\subsection{Conclusions and future works}
In the above aspect, insights into architecting and operating quantum networks with quantum repeaters characterised by different efficiency figures are provided. It is observed that incorporating the knowledge of the quality of the nodes into the routing process increases the network performance considerably in terms of fidelity boost to initial paths and lowering of blocked routing paths.  This technique is helpful not only in better network performance but also in motivating better traffic management by assigning requests to path orders depending upon their fidelity threshold. The major bottleneck in the performance of quantum networks lies with the low-quality nodes. Both their number and their efficiency figure play a major role in reduced network performance. For instance, the simulations show that a single replacement of a low-quality node with a high-quality node at the bottleneck point could increase network performance by 50\%.

Among the ongoing research activities in routing for quantum networks, this work introduces a perspective from an application and network deployment standpoint by considering mixed efficiency figures of quantum nodes. This opens the door for an immense amount of future work for further investigations, including the following. Investigation into an approach which considers both heterogeneous links and nodes would be a major upgrade to the current model. Also, it is apparent to investigate how much the performance of the quantum network would increase considering the entanglement purification as introduced in Section~\ref{ssec:purification}.

\section{Routing with end-to-end knowledge}

 \begin{figure}[tb]
 \centering
    \includegraphics[width=0.8\columnwidth]{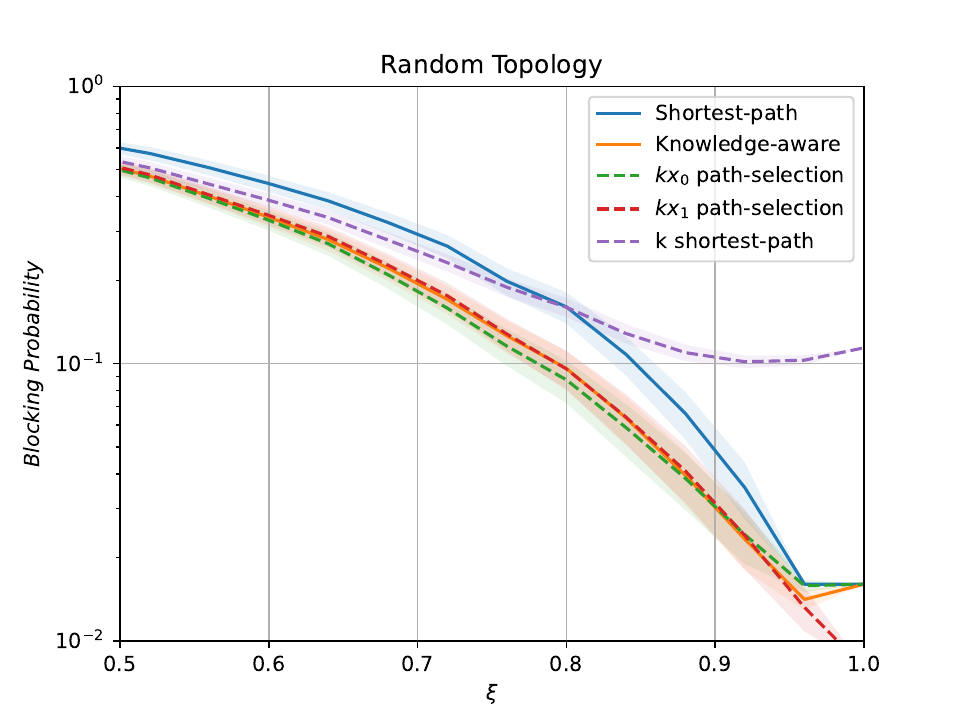}
    \caption{Blocking probability (BP) as a function of the fraction of high-quality nodes ($\xi$) in a randomised transport network. Parameters are set with a fidelity threshold ($F_{th}$) of 0.53 and a source-destination count ($n_{sd}$) of 5.}
    \label{fig:bp_wax_conf}
\end{figure}
\begin{figure}[tb]
    \centering
    \includegraphics[width=0.8\columnwidth]{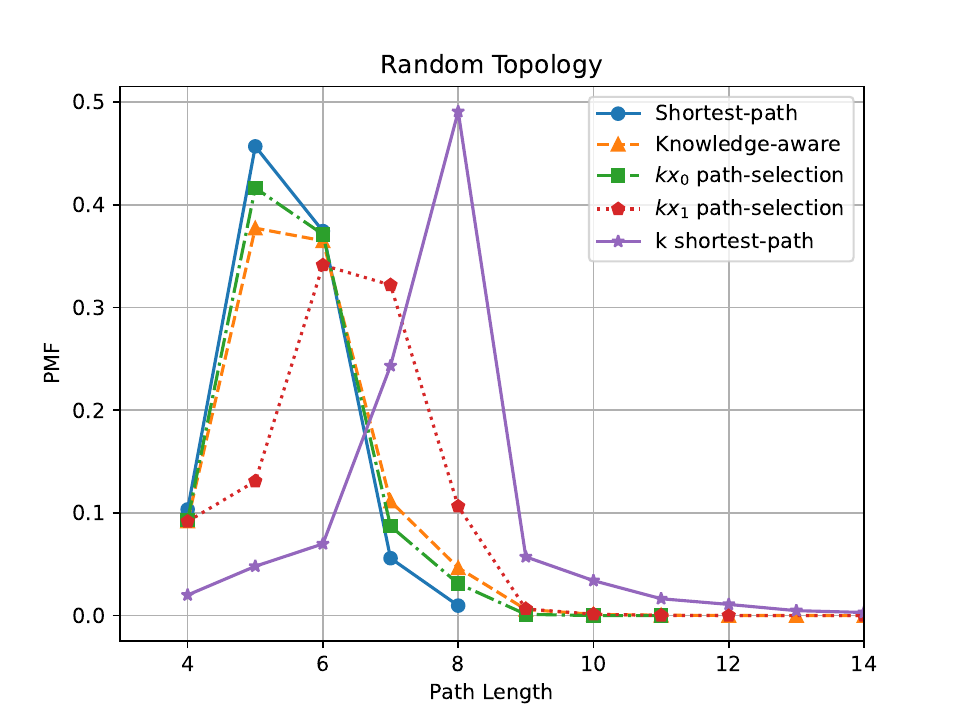}
    \caption{Probability mass function (PMF) of path length in a randomised transport network with a fidelity threshold ($F_{th}$) of 0.53 and a source-destination count ($n_{sd}$) of 5.}
    \label{fig:pl_wax}
\end{figure}

In this aspect, the methodology described in Section~\ref{sec:simulation} is applied to random and regular topologies with $n_{sd}$ = 5 as in Table~\ref{tab:topology_summary} and Section~\ref{sec:topology}. The parameters of the Waxman model are tuned to have a network of the same number of nodes and links as the regular topology, to compare the performance over the two networks on fair grounds.

\subsection{Blocking Probability}\label{ssec:bp}

\textit{\underline{Random Topology:}} Fig.~\ref{fig:bp_wax_conf} displays the trend of blocking probability versus the fraction of HQ repeaters ($\xi$) for random topology and the corresponding PMF of the path lengths in Fig.~\ref{fig:pl_wax}. The three evaluated approaches, i.e., $kx_{0}$, $kx_{1}$, and KA outperform the baseline SP approach, represented in a blue solid line. Although these three approaches demonstrate nearly identical blocking probabilities and PMF with the SP approach, the $kx_{0}$ approach is the most optimised. While KA requires the knowledge of network elements' characteristics to reduce the blocking probability, $kx_{0}$ achieves these results without it as described in Section~\ref{ssec:grey_white_box}. On the contrary, the KSP approach performed better than the shortest-path until $\xi \leq 0.8$ and became worse for $\xi \geq 0.8$. This phenomenon can be attributed to the fact that while the KSP method selects paths based on optimising end-to-end fidelity, it tends to choose routes with greater path lengths, as illustrated in Figure~\ref{fig:pl_wax}. These longer paths consequently consume more valuable network resources for served paths, eventually resulting in higher blocking probabilities. This behaviour is also the reason why the blocking probability with KSP \emph{increases} for high values of $\xi$, i.e., when the quality of the quantum repeaters increases. This may seem counterintuitive. Consider the extreme case of $\xi=1$, when all repeaters are of high quality. In this case, there would be a significant number of paths (of \emph{different} lengths) which are considered by KSP. KSP will pick the one with the minimum fidelity (above the threshold), which therefore will be one of the longest ones among the candidate set. This will consume more qubits (with respect to picking a shorter one), eventually leading to higher blocking probabilities.

\begin{figure}[tb]
    \centering
    \includegraphics[width=0.8\columnwidth]{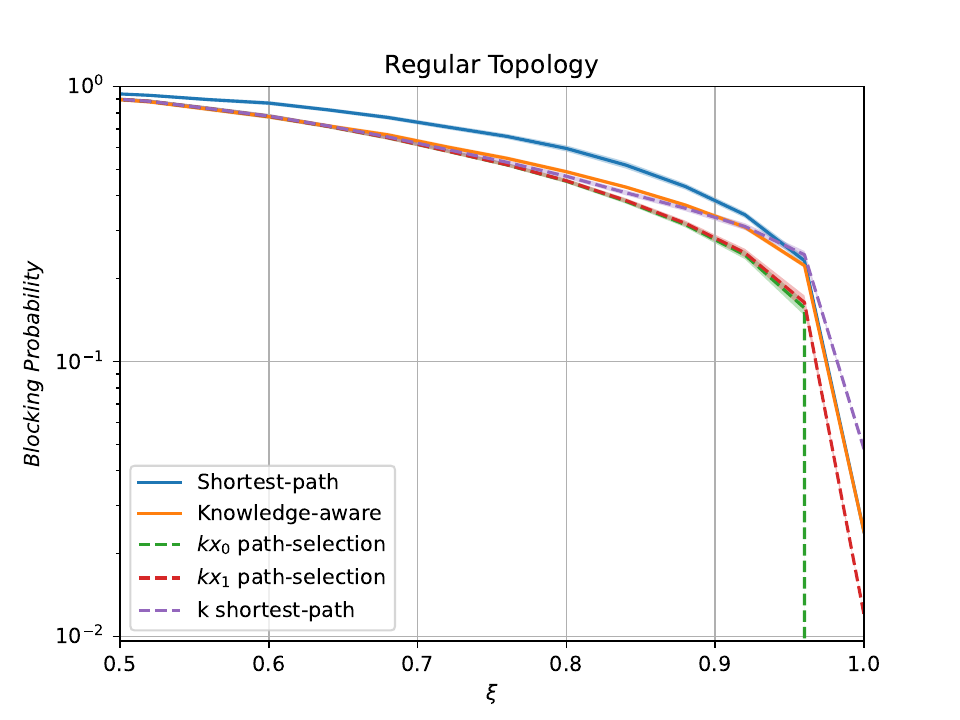}
    \vspace{-0.6cm}
    \caption{Blocking probability (BP) as a function of the fraction of high-quality nodes ($\xi$) in a regular transport network. Parameters are set with a fidelity threshold ($F_{th}$) of 0.53 and a source-destination count ($n_{sd}$) of 5.}
    \label{fig:bp_grid_conf}
\end{figure}
\begin{figure}[tb]
    \centering
    \includegraphics[width=0.8\columnwidth]{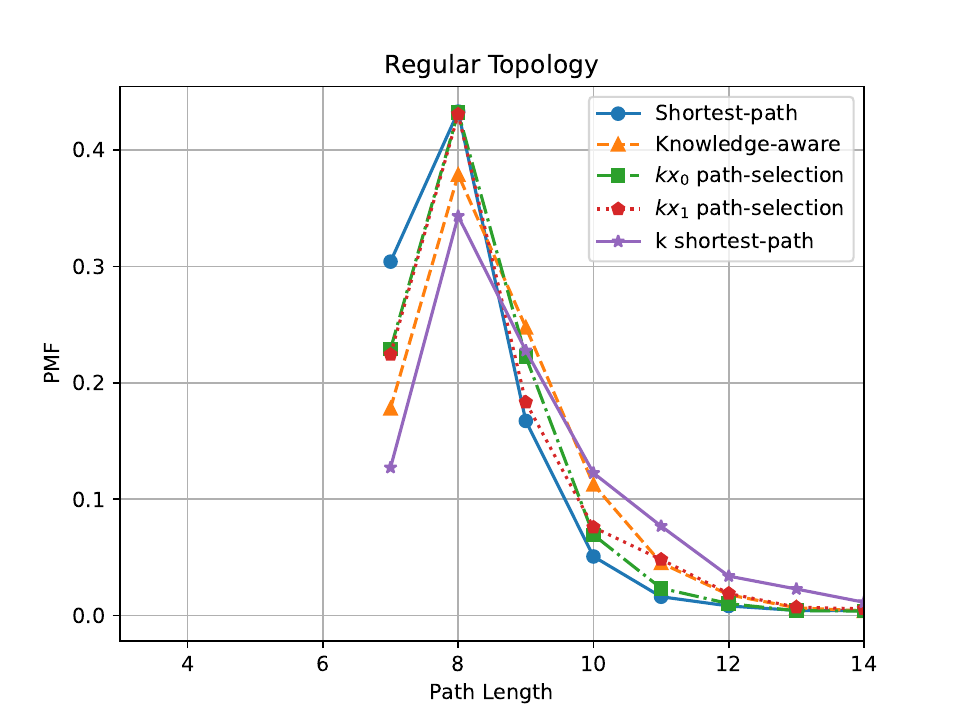}
    \caption{Probability mass function (PMF) of path length in a regular transport network with a fidelity threshold ($F_{th}$) of 0.53 and a source-destination count ($n_{sd}$) of 5.}
    \label{fig:pl_grid}
\end{figure}

\textit{\underline{Regular Topology:}}
Fig.~\ref{fig:bp_grid_conf} illustrates the trend in blocking probability as a function of $\xi$ for a regular topology, alongside the corresponding probability mass function (PMF) of path lengths depicted in Fig.~\ref{fig:pl_grid}. $kx_{0}$, $kx_{1}$, KA, and KSP consistently outperform SP, except with $\xi=1$ when SP and KA become identical due to the cost of all nodes being the same. The network's regular structure predominantly results in a path length of 8, which emerges as the most frequent outcome, although variations are observed depending on the strategic approach adopted by each method. Consistent with performances observed in random topology, the $kx_{0}$ strategy proves to be superior to KA and KSP, especially with higher values of $\xi$. We note that KSP in regular topology does not suffer from significant performance degradation as with a random topology: this is because, due to the network structure in regular topologies, KSP has less freedom to explore paths that are significantly longer than the shortest one (see Fig.~\ref{fig:pl_wax} vs.\ Fig.~\ref{fig:pl_grid}).

\begin{figure}[tb]
    \centering
    \includegraphics[width=0.8\columnwidth]{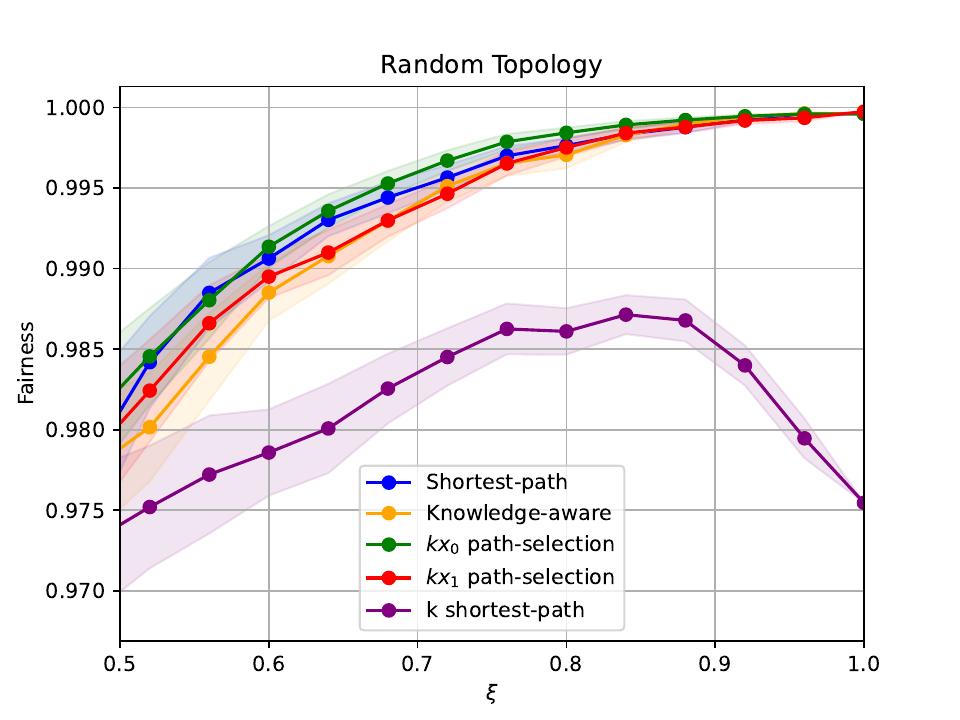}
    \caption{Fairness as a function of the fraction of high-quality nodes ($\xi$) in a randomised transport network, with a fidelity threshold ($F_{th}$) of 0.53 and a source-destination count ($n_{sd}$) of 5.}
    \label{fig:fair_wax_conf}
\end{figure}
\begin{landscape}
\begin{figure}[p]
    \centering
    \includegraphics[width=\linewidth]{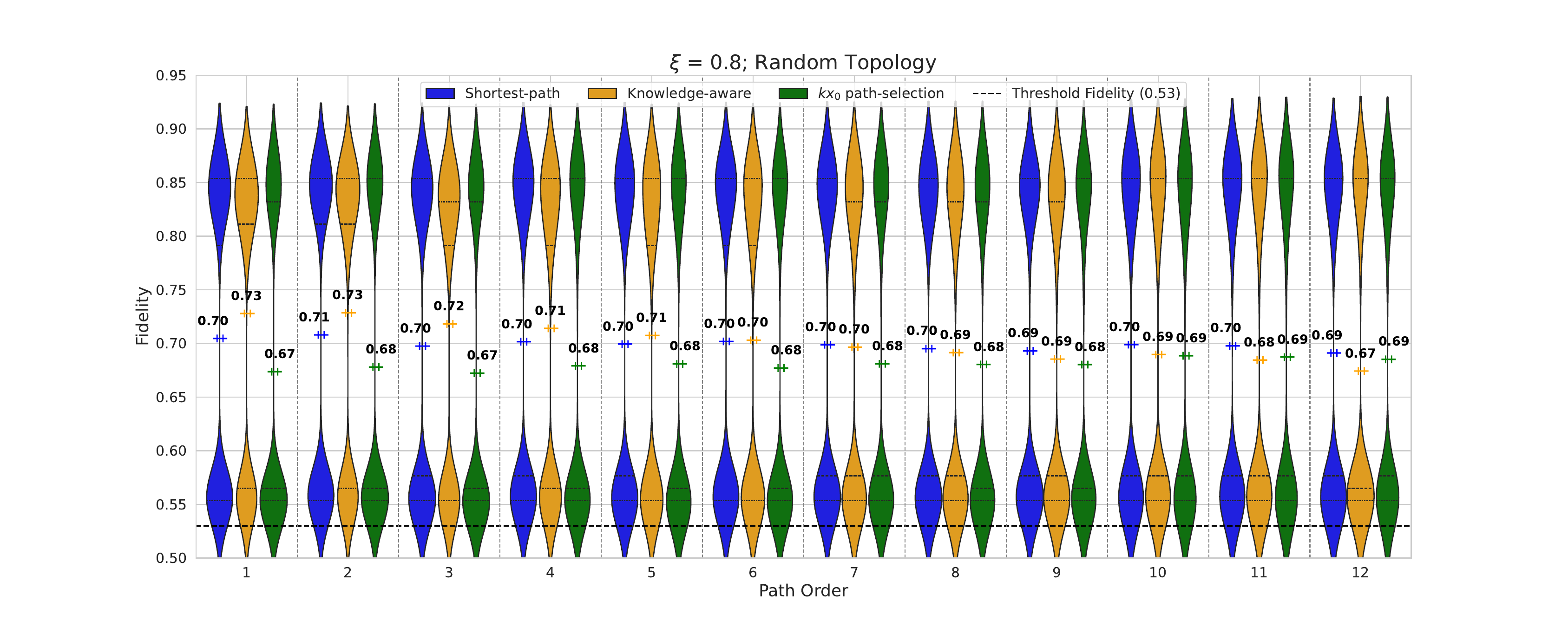}
    \caption{Fidelity as a function of path-order ($\theta$) in a randomised transport network, with a fidelity threshold ($F_{th}$) of 0.53 and a source-destination count ($n_{sd}$) of 12.}
    \label{fig:violin_wax_nsd12}
\end{figure}
\end{landscape}

\begin{table*}[h!]
\centering
\caption{Evaluation of Routing Performance. BP: Probability of a path being blocked, J: Jain's index for multiple source-destination serving, and $n_{sd}$: Number of source-destination pairs involved}
\label{tab:routing_evaluation}
\begin{adjustbox}{width=1\textwidth}
\begin{tabular}{lccc}
\toprule
\multirow{2}{*}{Property} & \multirow{2}{*}{Notation} & 
\multicolumn{2}{c}{Performance Comparison} \\
\cmidrule{3-4}
& & Random Topology & Regular Topology \\
\midrule
Blocking Probability & BP & $kx_0$ $\sim$ $kx_1$ $\sim$ KA $>$ SP $>$ KSP & $kx_0$ $\sim$ $kx_1$ $>$ KSP $\sim$ KA $>$ SP \\
Fairness & J & $kx_0$ $\sim$ SP $>$ $kx_1$ $\sim$ KA $>$ KSP & SP $>$ $kx_0$ $\sim$ $kx_1$ $>$ KSP $>>$ KA  \\

Number of SD pairs & $n_{sd}$ & $kx_0$ $>$ KA $\geq$ SP $>$ $kx_1$ $>$ KSP & N/A \\

\bottomrule
\end{tabular}
\end{adjustbox}
\end{table*}

Overall, the network topology exerts a significant influence on the performance outcomes of each approach. Blocking probability is notably higher in the regular topology compared to the random topology, attributed to the shifted path length distribution in the regular topology (see Fig.~\ref{fig:pl_wax} \& \ref{fig:pl_grid}). This effect is evident in the comparative slopes of the blocking probability curves, with the regular topology in Fig.~\ref{fig:bp_grid_conf} displaying a gentler slope than the random topology in Fig.~\ref{fig:bp_wax_conf}. The network's regular structure limits the available route options because it offers little variation in served path lengths. This results in a relatively smaller improvement in performance as compared to the case in random topology.

\begin{figure}[tb]
    \centering
    \includegraphics[width=0.8\columnwidth]{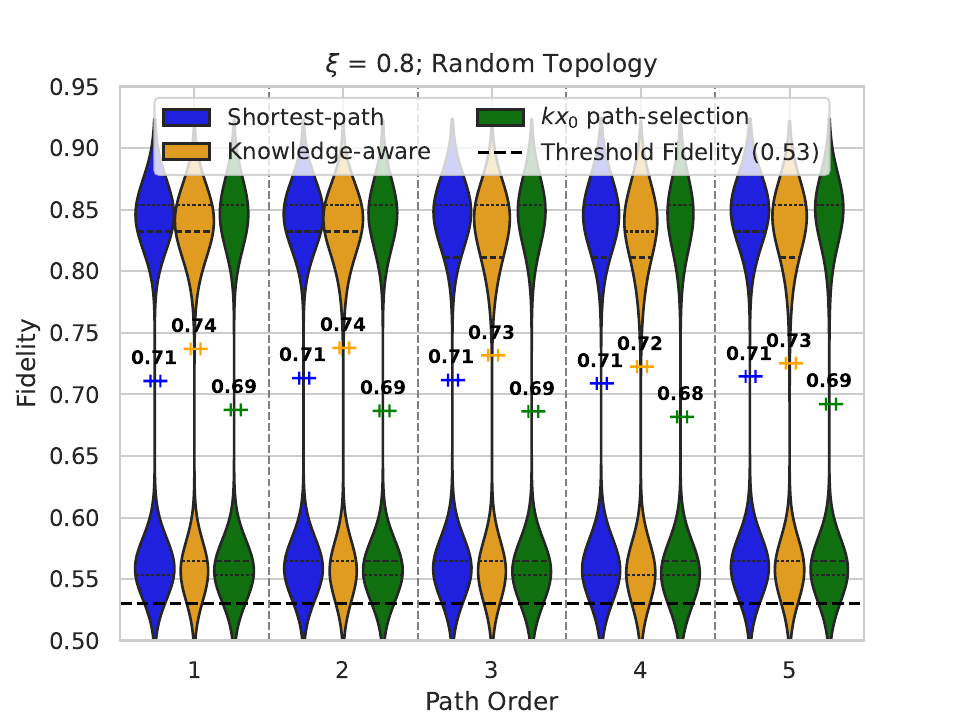}
    \caption{Fidelity as a function of path-order ($\theta$) in a randomised transport network, with a fidelity threshold ($F_{th}$) of 0.53 and a source-destination count ($n_{sd}$) of 5.}
    \label{fig:violin_wax}
\end{figure}
\begin{figure}[tb]
    \centering
    \includegraphics[width=0.8\columnwidth]{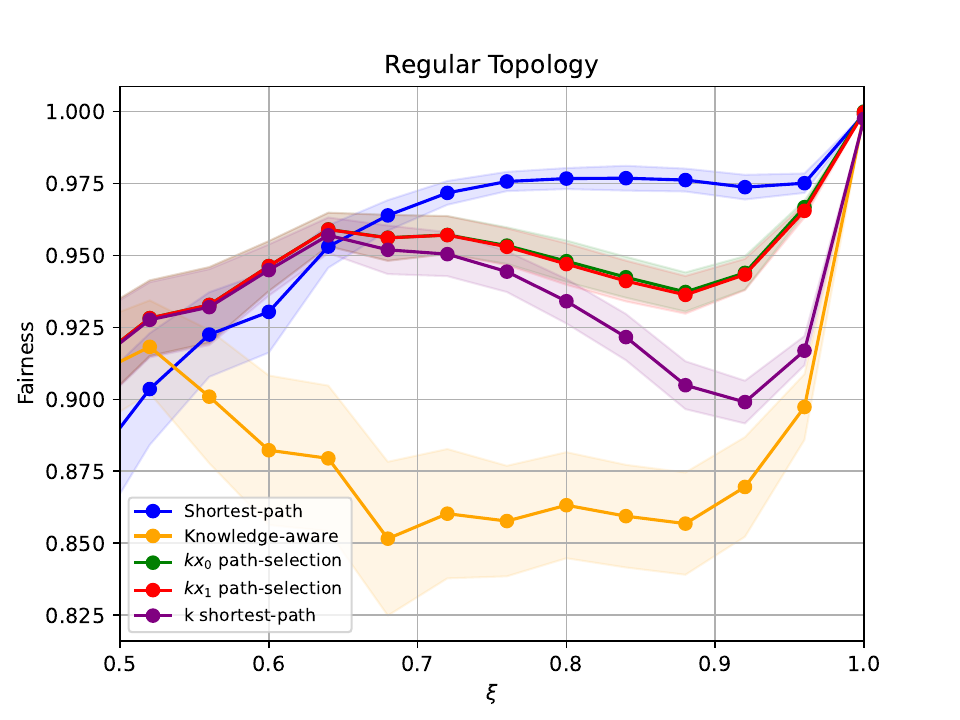}
    \caption{Fairness as a function of the fraction of high-quality nodes ($\xi$) in a regular transport network, with a fidelity threshold ($F_{th}$) of 0.53 and a source-destination count ($n_{sd}$) of 5.}
    \label{fig:fair_grid_conf}
\end{figure}
\begin{figure}[tb]
    \centering
    \includegraphics[width=0.8\columnwidth]{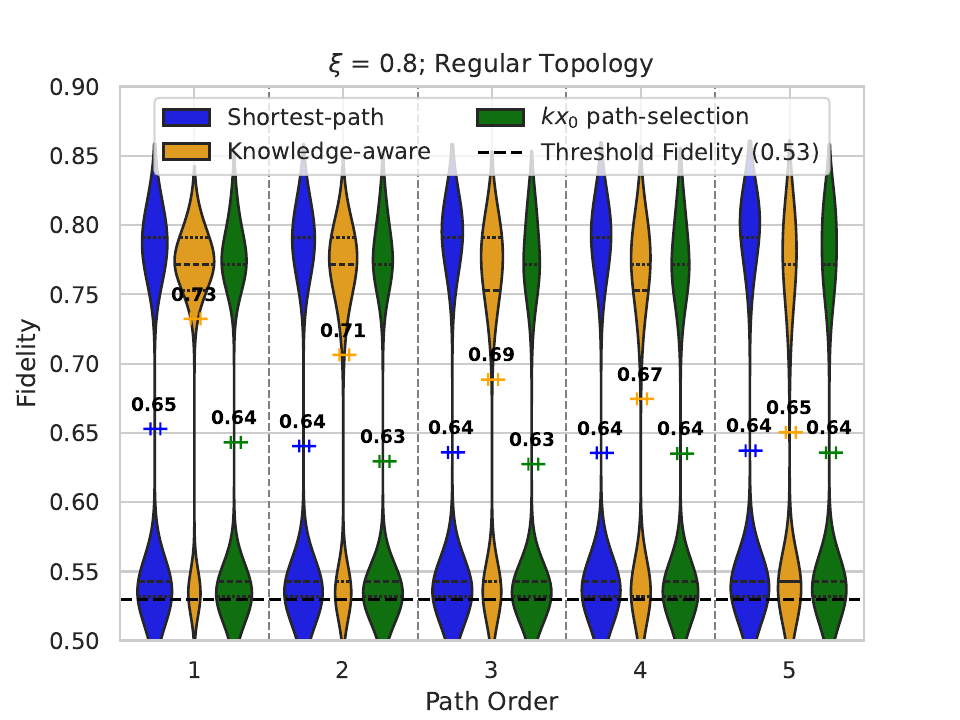}
    \caption{Fidelity as a function of path-order ($\theta$) in a regular transport network, with a fidelity threshold ($F_{th}$) of 0.53 and a source-destination count ($n_{sd}$) of 5.}
    \label{fig:violin_grid}
\end{figure}
\subsection{Fairness}\label{ssec:fairness}

\textit{\underline{Random Topology:}}
Fig.~\ref{fig:fair_wax_conf} presents the fairness among network users, quantified using Jain's index, in terms of the number of paths successfully delivered above the fidelity threshold of 0.53 for the random topology.

The majority of the approaches evaluated exhibit comparable levels of fairness and increase their fairness as $\xi$ increases, with the notable exception of the KSP approach, which demonstrates significantly lower fairness and decreases its fairness as $\xi$ increases. This means that, in general, these policies can fairly exploit the higher quality of repeaters. The reason for the drop of fairness experienced by KSP for high $\xi$ values is related to the corresponding increase of the blocking probabilities observed in Fig.~\ref{fig:bp_wax_conf}, i.e., KSP using unnecessarily long paths as shown in Fig.~\ref{fig:pl_wax} also for SD pairs that are topologically closer. In addition, those that are farther away (topologically) have thus a higher chance of being blocked, and this results in lower fairness values. Similarly, it is interesting to consider this result together with the performance in terms of blocking probability (Fig.~\ref{fig:bp_wax_conf}) even for slightly lower values of $\xi$, i.e., $\xi < 0.8$. Although KSP has better BP in that range than SP, it does serve the non-blocked SD pairs fairly. The $kx_{0}$ approach not only maintains equitable treatment among users but also performs the best (among the considered approaches) in terms of BP (see Fig.~\ref{fig:bp_wax_conf}). 

To better illustrate user fairness from another perspective, the fidelity distribution for each path order ($\theta$)—representing the distribution of served fidelity for requests in sequence (i.e., 1 to $n_{sd}$)—is presented. Fig.~\ref{fig:violin_wax_nsd12} ($n_{sd} = 12$) and Fig.~\ref{fig:violin_wax} ($n_{sd} = 5$) illustrate this fidelity distribution across path orders for various policies at $\xi = 0.8$, with the average fidelity served for each path order also highlighted. For presentation purposes, we only plot $kx_{0}$ path-selection from the `grey box' approaches, leaving behind the worse performing KSP and $kx_{1}$ path-selection as discussed above. It is to be noted that $\xi = 0.8$ is chosen for analysis as the corresponding blocking probabilities are lower enough to be considered for realistic network settings. Since the policies are designed to deliver just over the threshold fidelity of 0.53, achieving higher fidelity than the fidelity threshold would be an over-utilisation of resources and can be considered as not a fully optimised approach. The optimisation capability of the $kx_{0}$ path-selection approach is visible from Fig.~\ref{fig:violin_wax_nsd12}, where it delivers a lower average fidelity (yet still over the fidelity threshold) for each path order with almost equal value as compared to other approaches. This shows that the approach uses low-quality nodes if the fidelity threshold allows for it, which keeps the valuable resources to deliver the path for the upcoming path orders. The same is true for the SP approach, but the average fidelity values are more than those of the $kx_{0}$ path-selection approach, which means that the SP approach is not fully optimised. Contrary to this, KA has the worst case, where it has unnecessarily higher average fidelity than the fidelity threshold for initial path orders, and the average fidelity values gradually drop with path order. The same is true for $n_{sd} = 5$ as well in Fig.~\ref{fig:violin_wax}, but the effect is not clearly visible due to a smaller number of source-destination pairs involved.

\textit{\underline{Regular Topology:}}
Fig.~\ref{fig:fair_grid_conf} illustrates the fairness within the regular topology, while Fig.~\ref{fig:violin_grid} shows the corresponding fidelity distribution for each path order ($\theta$) in the regular topology at $\xi = 0.8$. The average fidelity served for each path order is highlighted, similar to the random topology case. While all approaches exhibit comparable levels of fairness at $\xi = 0.5$, the KA method demonstrates the least fairness as $\xi$ increases. This trend is attributable to the preferential consumption of high-quality nodes in the initial delivery paths, a phenomenon also evident from Fig.~\ref{fig:violin_grid}, where the average fidelity decreases almost linearly with increasing path order, represented in orange. Contrary to this, it is to be noted that KA did not exhibit the worst fairness in random topology as seen in Fig.~\ref{fig:fair_wax_conf} due to relatively shorter path length as seen in Fig.~\ref{fig:pl_wax} than regular topology in Fig.~\ref{fig:pl_grid}. The relatively longer path lengths, which arise from the inherent structure of the regular topology, enhance the effect of preferential consumption of high-quality nodes in the regular topology.

Although the blocking probability of the KA policy is closer to other policies, as in Fig.~\ref{fig:bp_grid_conf}, it does not serve the paths fairly as in Fig.~\ref{fig:fair_grid_conf}.

\begin{figure}[tb]
\centering

\begin{subfigure}[t]{0.49\textwidth}
    \centering
    \includegraphics[width=\linewidth]{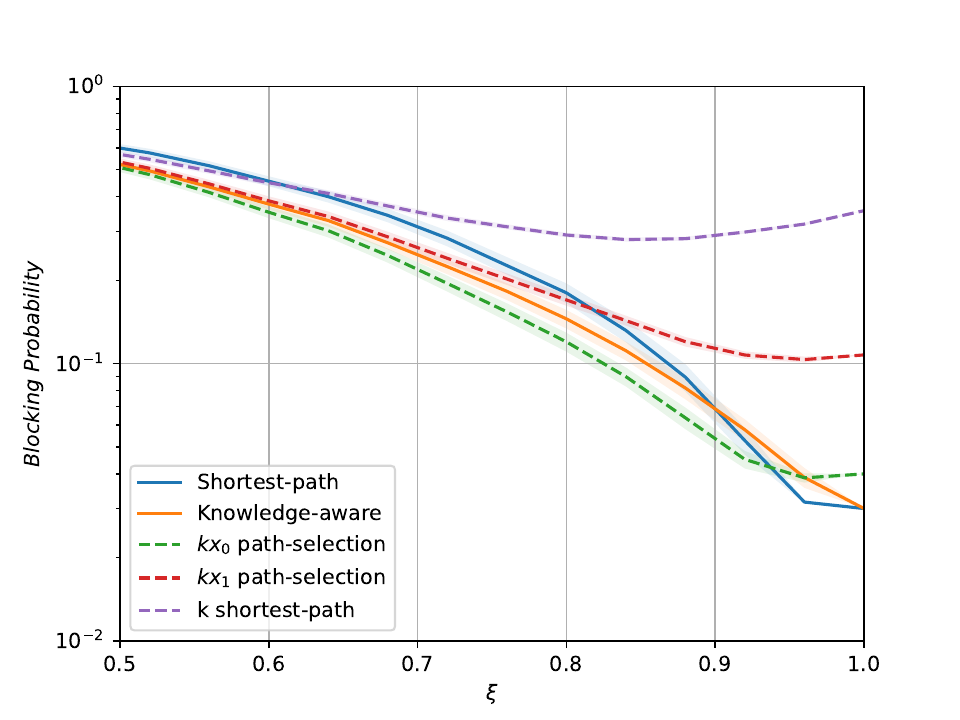}
    \caption{$n_{sd} = 8$}
    \label{fig:nsd8}
\end{subfigure}
\hfill
\begin{subfigure}[t]{0.49\textwidth}
    \centering
    \includegraphics[width=\linewidth]{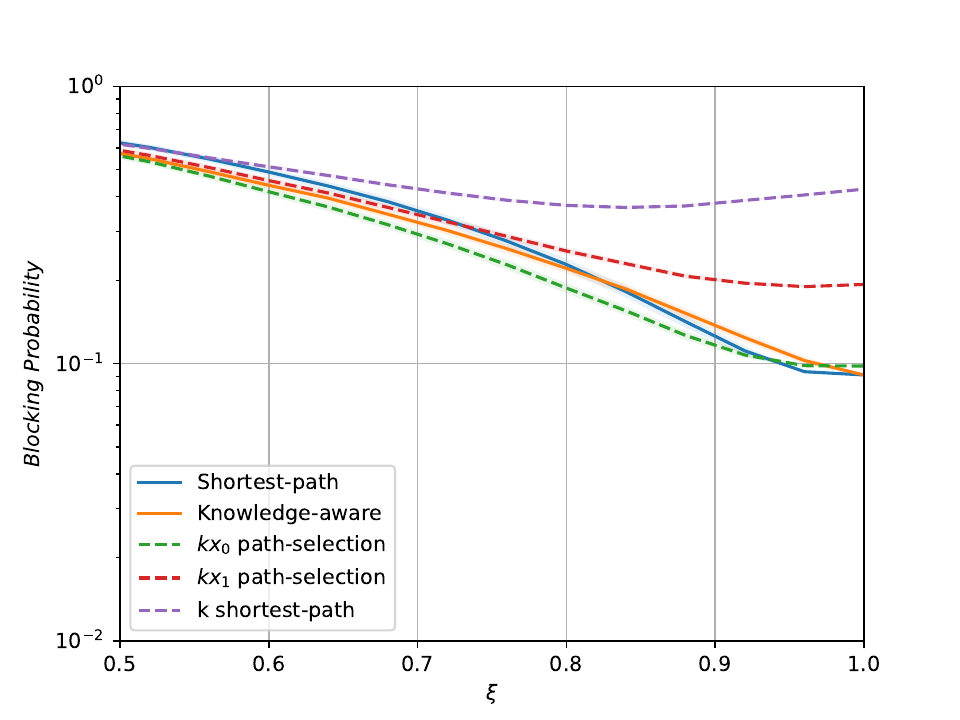}
    \caption{$n_{sd} = 10$}
    \label{fig:nsd10}
\end{subfigure}

\begin{subfigure}[t]{0.49\textwidth}
    \centering
    \includegraphics[width=\linewidth]{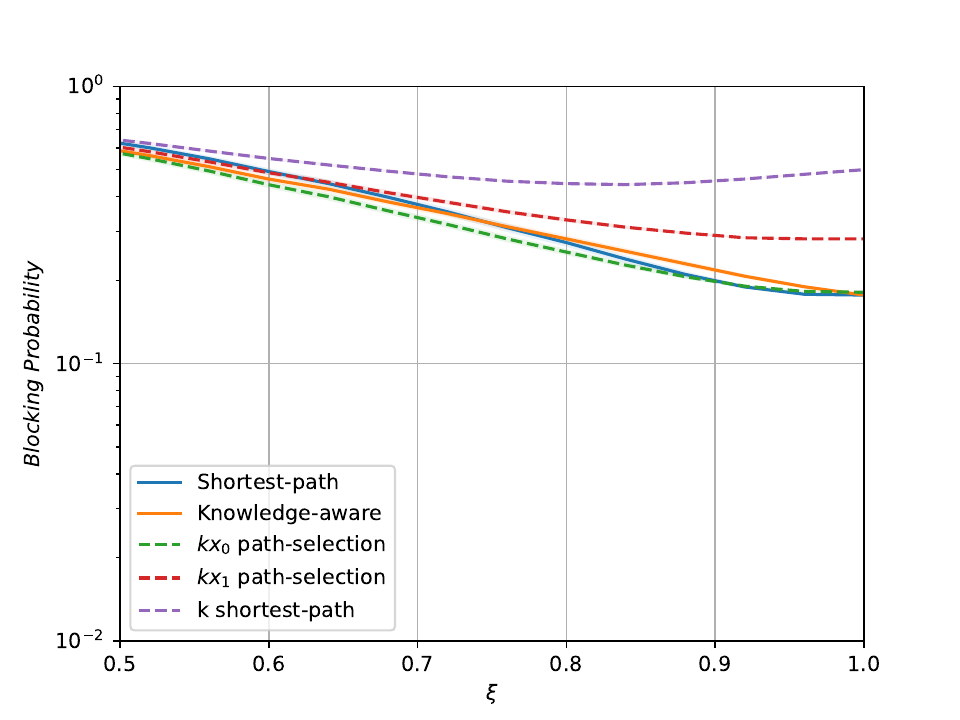}
    \caption{$n_{sd} = 12$}
    \label{fig:nsd12}
\end{subfigure}

\caption{Blocking probability (BP) as a function of the fraction of high-quality nodes ($\xi$) in a randomized transport network with a fidelity threshold ($F_{th}$) of 0.53, shown for source-destination counts ($n_{sd}$) of: a) 8, b) 10, and c) 12.}
\label{fig:load}
\end{figure}

An interesting observation in Fig.~\ref{fig:fair_grid_conf} is the relation between blocking probability and fairness of the approaches. For mid to slightly higher values of $\xi$, approaches excluding KA and SP sacrifice a small amount of fairness as seen in Fig.~\ref{fig:fair_grid_conf} to have a significant boost in the blocking probability performance as seen in Fig.~\ref{fig:bp_grid_conf}. It is to be noted that this behaviour is only limited to regular topology and arises from the lack of availability of different path lengths as observed in Fig.~\ref{fig:pl_grid} where only one peak is observed at path length of 8 while in case of random topology in Fig.~\ref{fig:pl_wax}, path length of 5 \& 6 are observed almost equally. Ultimately, the $kx_{0}$ approach achieves commendable fairness in regular topology, surpassed only by the SP method, which was observed to have significantly less performance in terms of blocking probability.

\subsection{Number of source-destination pairs}\label{ssec:nsd}
We now analyse the impact of varying the number of source-destination pairs. For this analysis, we exclusively consider a random topology with a fixed number of transport nodes and links, and we vary the number of SD pairs, considering $n_{\text{sd}} \in \{5, 8, 10, 12\}$.

Fig.~\ref{fig:bp_wax_conf} ($n_{sd} = 5$) and Fig.~\ref{fig:load} ($n_{sd} \in \{8, 10, 12\}$) collectively demonstrates the network's performance under varying number of source-destination pairs for the random topology, displaying the blocking probability as a function of $\xi$. With increasing $n_{sd}$, the blocking probability escalates due to the limited availability of shared entangled pairs at quantum repeaters. Among the methods assessed, the $kx_{0}$ approach consistently outperforms all others across the different numbers of source-destination pairs. Notably, the $kx_{1}$ and KSP approaches, which perform well under a lower number of source-destination pairs, fail to maintain superiority over the SP approach as the $n_{sd}$ increases, exhibiting a decline in performance. In the case of KSP, the BP begins to increase with \( \xi \) when \( \xi \) approaches values close to 1 (i.e., \( \xi \in [0.9, 1] \)) for higher numbers of source-destination pairs (\( n_{sd} \in \{ 8, 10, 12\} \)). This counterintuitive phenomenon arises from the KSP approach of selecting longer paths, which consumes valuable resources. This behaviour, explained in Section~\ref{ssec:bp}, is further intensified by the increased load resulting from a higher number of SD pairs operating within the network. This shows that the algorithm $kx_{i}$ is very sensitive to allowing paths even longer, only 1 hop more than the shortest possible, and this is more and more visible as the $n_{sd}$ increases. At $n_{sd} = 12$, the performance of SP, $kx_{0}$, and KA appears to converge, highlighting a uniform response to extreme network demands.

Table~\ref{tab:routing_evaluation} consolidates all results obtained thus far, providing a comprehensive summary of the performance comparison across the evaluated approaches, organised by the transport network topologies. The table includes the performance metrics discussed in Sections~\ref{ssec:bp}, \ref{ssec:fairness}, and \ref{ssec:nsd}, specifically blocking probability, fairness, and source-destination count.

\begin{figure}[tb]
\centering
\includegraphics[width=0.8\columnwidth]{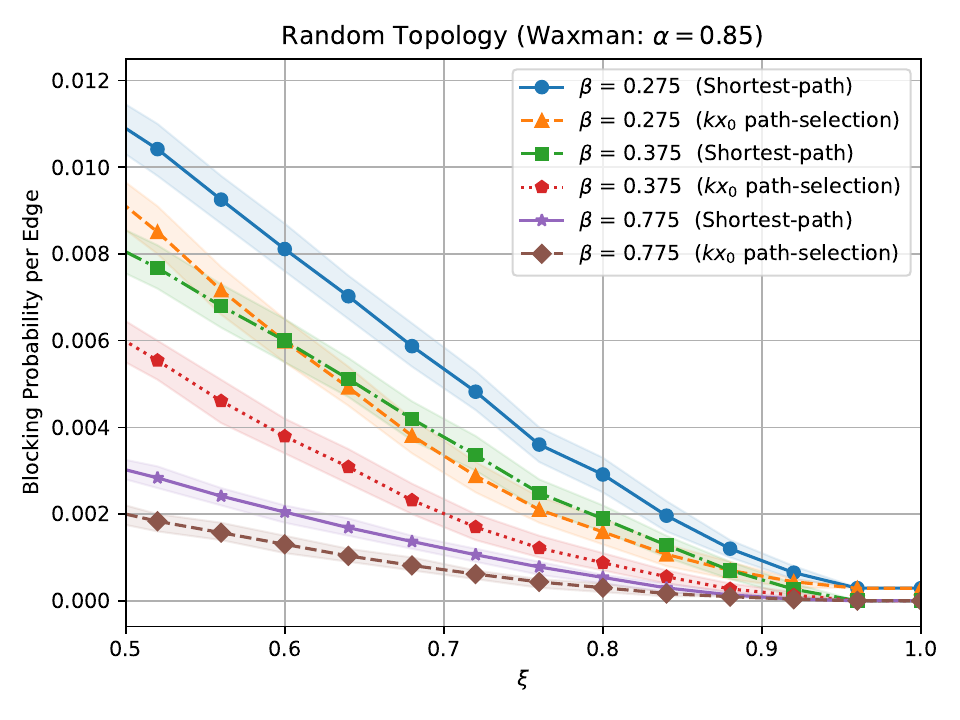}
    \caption{Blocking probability per edge (BP/E) as a function of the fraction of high-quality nodes ($\xi$) in a randomized transport network, with a fidelity threshold ($F_{th}$) of 0.53 and a source-destination count ($n_{sd}$) of 5, for varying Waxman parameter values ($\beta$).}
    \label{fig:bpe_wax_conf}
\end{figure}

\begin{figure}[tb]
    \centering
\includegraphics[width=0.8\columnwidth]{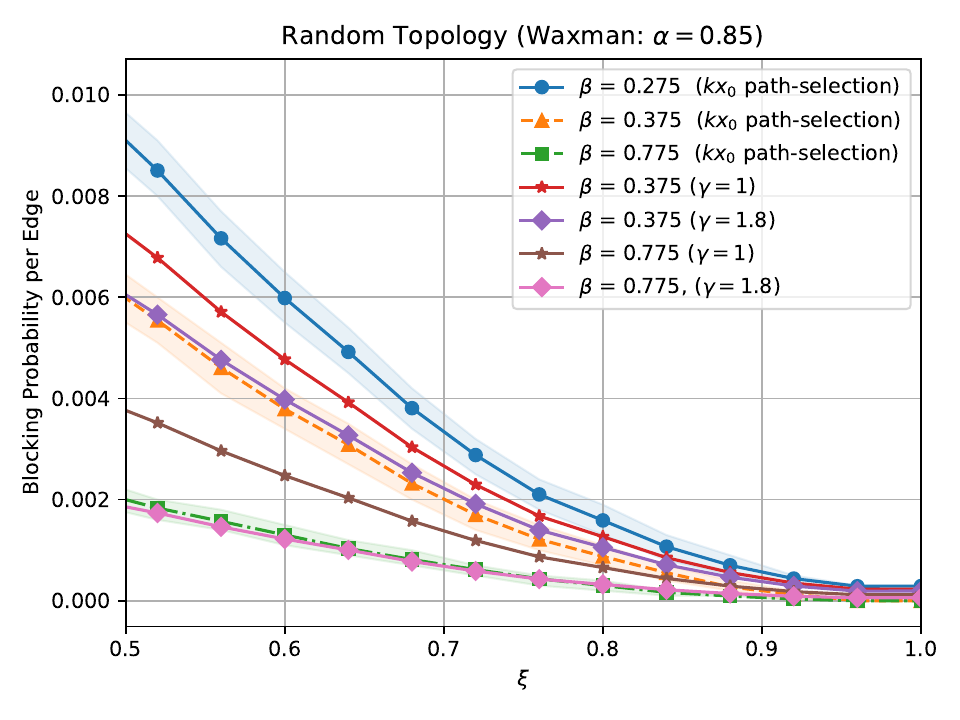}
    \caption{Blocking probability per edge (BP/E) as a function of the fraction of high-quality nodes ($\xi$) in a randomized transport network, with a fidelity threshold ($F_{th}$) of 0.53 and a source-destination count ($n_{sd}$) of 5, to model the performance of the $kx_0$ path-selection algorithm.}
    \label{fig:bpe_ref_wax_conf}
\end{figure}

\subsection{Average Degree}
So far, we have considered a random topology equivalent (in terms of number of nodes and links) to the reference regular topology. In this section, we investigate the effect of increasing the number of links (or average degree) in the network. To this end, we exploit the flexibility of the Waxman model in capturing the network density via the $\beta$ parameter. For clarity in our presentation, we focus on evaluating the performance of the most effective routing approach identified in our study ---the $kx_0$ path selection method--- against the traditional shortest-path method for reference. We consider, as before, $n_{sd} = 5$ only.

Fig.~\ref{fig:bpe_wax_conf} delineates the relationship between the BP/E and the parameter $\xi$ with a random topology. The figure clearly shows that as the network expands in complexity with an increased number of links, the $kx_0$ path selection method consistently outperforms the shortest path approach. Although not depicted in the figure, the $kx_0$ method also surpasses the other routing strategies under comparison.
We consider the BP/E instead of the BP because this

helps us better grasp how each new link addition individually affects the network. Specifically, we capture this effect by modelling the performance differentials between networks with varying average degrees ($\beta$) as follows

\begin{equation}\label{eq:bpe1}
    bpe_{\beta} = bpe_{\beta_{0}} \left( \frac{n_{e}^{\beta_{0}}}{n_{e}^{\beta}} \right)^{\gamma},
\end{equation}
where, $bpe_\beta$ is the BP/E for a random quantum network with Waxman parameters ($\alpha, \beta$) having $n_{e}^{\beta}$ number of links; $bpe_{\beta_{0}}$ is the BP/E for a random quantum network with Waxman parameters ($\alpha, \beta_{0}$) having $n_{e}^{\beta_{0}}$ number of links; $\gamma$ is a parameter that is selected to model the performance $bpe_\beta$ in terms of $bpe_{\beta_{0}}$.

As seen in Fig.~\ref{fig:bpe_ref_wax_conf}, the value of $\gamma = 1.8$ can accurately predict the performance improvement of the $kx_0$ path-selection approach in denser networks with $\beta = 0.375$ and $\beta = 0.775$ compared to the baseline case of $\beta = 0.275$, which entails a scalability that is slightly less than quadratic.
The results obtained in our simulated conditions show that, as the network scales, the performance of the $kx_{0}$ algorithm is given by

\begin{equation}\label{eq:bpe2}
    bpe_{\beta} = bpe_{\beta_{0}} \left( \frac{n_{e}^{\beta_{0}}}{n_{e}^{\beta}} \right)^{1.8}.
\end{equation}

\begin{figure*}
\centering
    \begin{subfigure}{0.52\textwidth}
        \includegraphics[width=0.95\linewidth]{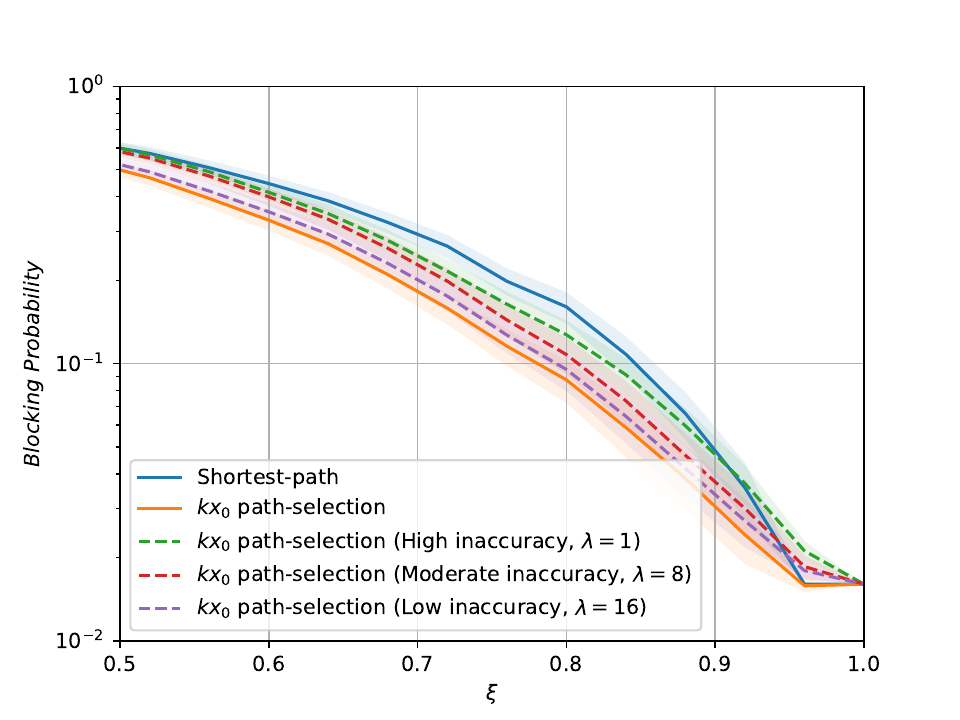}
        \caption{$n_{sd} = 5$}
    \end{subfigure}
    \hspace{-1cm}
    \begin{subfigure}{0.52\textwidth}
        \includegraphics[width=0.95\linewidth]{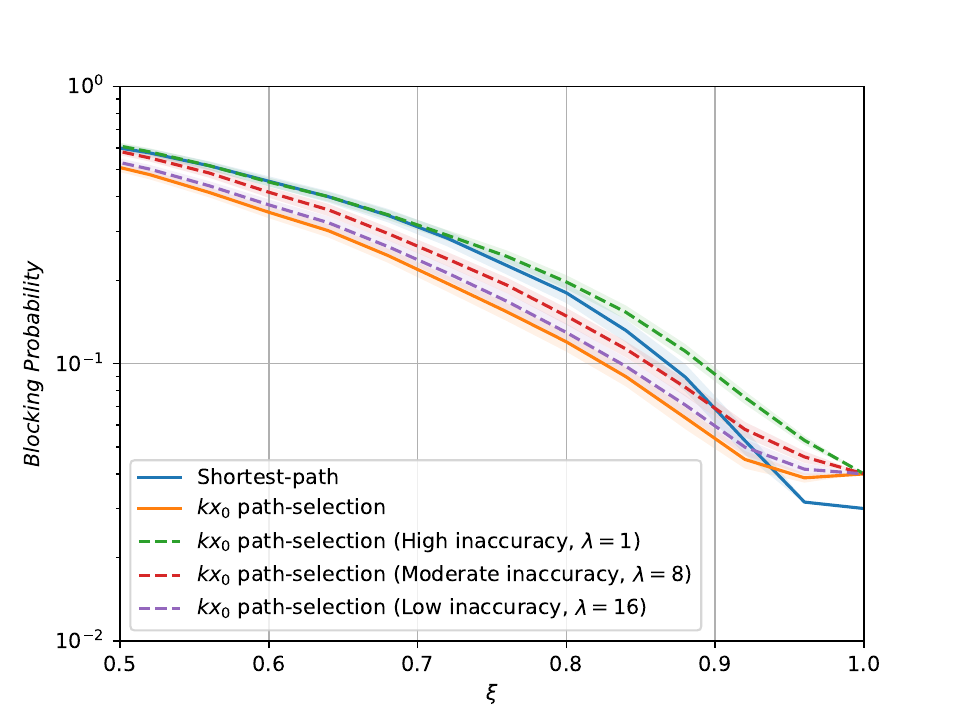}
        \caption{$n_{sd} = 8$}
    \end{subfigure}

    \caption{Robustness test: Blocking probability (BP) as a function of the fraction of high-quality nodes ($\xi$) in a randomised transport network, with a fidelity threshold ($F_{th}$) of 0.53 and a source-destination count ($n_{sd}$) of a) 5, b) 8. The robustness parameter ($\lambda$) models end-to-end fidelity inaccuracy, with $\lambda \in \{1, 8, 16\}$, corresponding to decreasing noise magnitude.}
    \label{fig:bp_robust_conf}
\end{figure*}

\begin{figure}[tb]
\centering
\includegraphics[width=0.8\columnwidth]{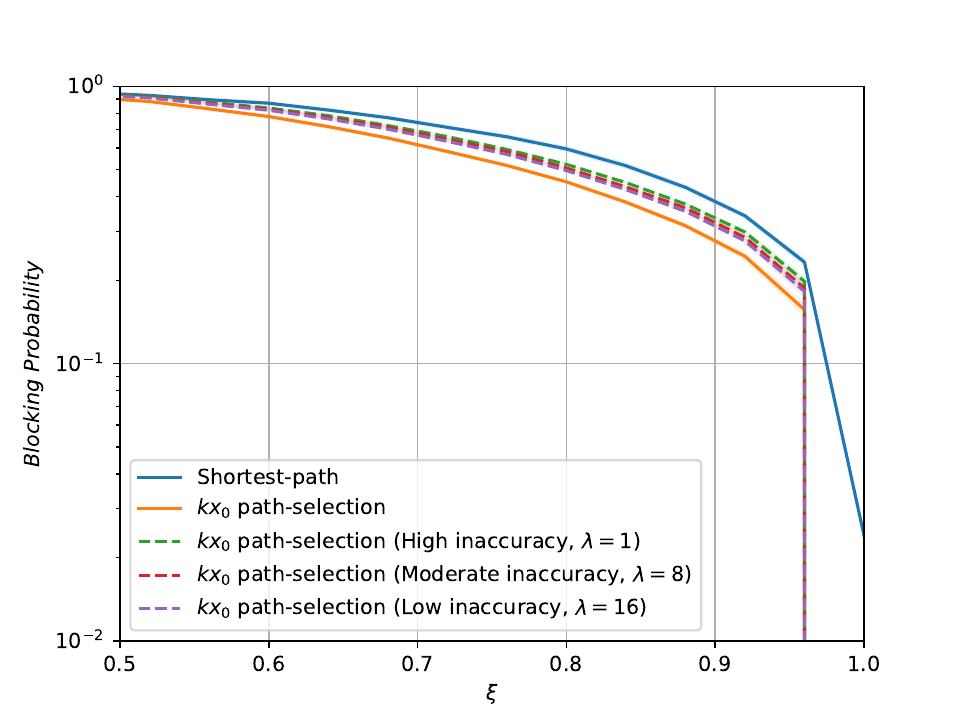}
    \caption{Robustness test: Blocking probability (BP) as a function of the fraction of high-quality nodes ($\xi$) in a regular transport network, with a fidelity threshold ($F_{th}$) of 0.53 and a source-destination count ($n_{sd} = 5$). The robustness parameter ($\lambda$) models end-to-end fidelity inaccuracy, with $\lambda \in \{1, 8, 16\}$, corresponding to decreasing noise magnitude.}
    \label{fig:bp_robust_conf_regular}
\end{figure}

\subsection{Robustness}
As described in Section~\ref{ssec:grey_white_box}, the end-to-end fidelity data acquired through sampling is central to the $kx_{0}$ path-selection approach. In this aspect, we test the robustness of the $kx_{0}$ approach by supplying \emph{inaccurate} end-to-end fidelity data to the algorithm. Specifically, we examine its performance on a random topology with $n_{sd} = \{ 5, 8\}$ and a regular topology with $n_{sd} = 5$. We sample the end-to-end fidelity inaccuracy from a normal distribution given by: \newline
\begin{equation}\label{eq:distribution}
Q = N \left( 0, \frac{|F_1 - F_2|}{\lambda}\right),
\end{equation}
where $F_1$ is the correct end-to-end fidelity of the path; $F_2$ is the end-to-end fidelity of the path by replacing an HQ repeater with an LQ one; $\lambda$ is a parameter modelling the end-to-end fidelity inaccuracy. Specifically, by varying $\lambda$, we modify the standard deviation of the noise added to the correct fidelity, thus controlling its magnitude. We classify these as follows: $\lambda = 1$ (High inaccuracy), $\lambda = 8$ (Moderate inaccuracy), and $\lambda = 16$ (Low inaccuracy).

Fig.~\ref{fig:bp_robust_conf}a for $n_{sd} = 5$ and Fig.~\ref{fig:bp_robust_conf}b for $n_{sd} = 8$ presents the outcomes of a robustness assessment for the $kx_{0}$ approach for $\lambda \in \{ 1, 8, 16\}$ in a random network. The blocking probability escalates as inaccuracies in the end-to-end fidelity data provided to the $kx_{0}$ approach increase or $\lambda$ decreases. For $n_{sd} = 5$, even with the introduction of high inaccuracy, the performance of $kx_{0}$ remains superior to that of the SP approach. The $kx_{0}$ continues to better SP under demanding condition i.e., with more SD pairs ($n_{sd} = 8$) in low and moderate inaccuracies cases. However, under more demanding conditions, i.e., when fidelity inaccuracies are high along with a higher number of source-destination pairs ($n_{sd} = 8$), the advantage of $kx_0$ over SP diminishes, as shown in Fig.~\ref{fig:bp_robust_conf}b. Similarly, Fig.~\ref{fig:bp_robust_conf_regular} shows the robustness results for the $kx_0$ approach in a regular network for $\lambda = \{ 1, 8, 16\}$. Consistent with the random topology results, $kx_0$ outperforms the SP approach even at high fidelity inaccuracies, although the performance differences between accuracy levels decrease. This reduced performance variation is due to the regular network structure, which constrains path diversity between SD pairs as compared to the random network.

Overall, these results suggest that the $kx_{0}$ approach demonstrates robust performance with inaccurate fidelity estimates, except under highly demanding conditions.

\begin{figure*}
\centering
    \begin{subfigure}{0.52\textwidth}
        \includegraphics[width=\linewidth]{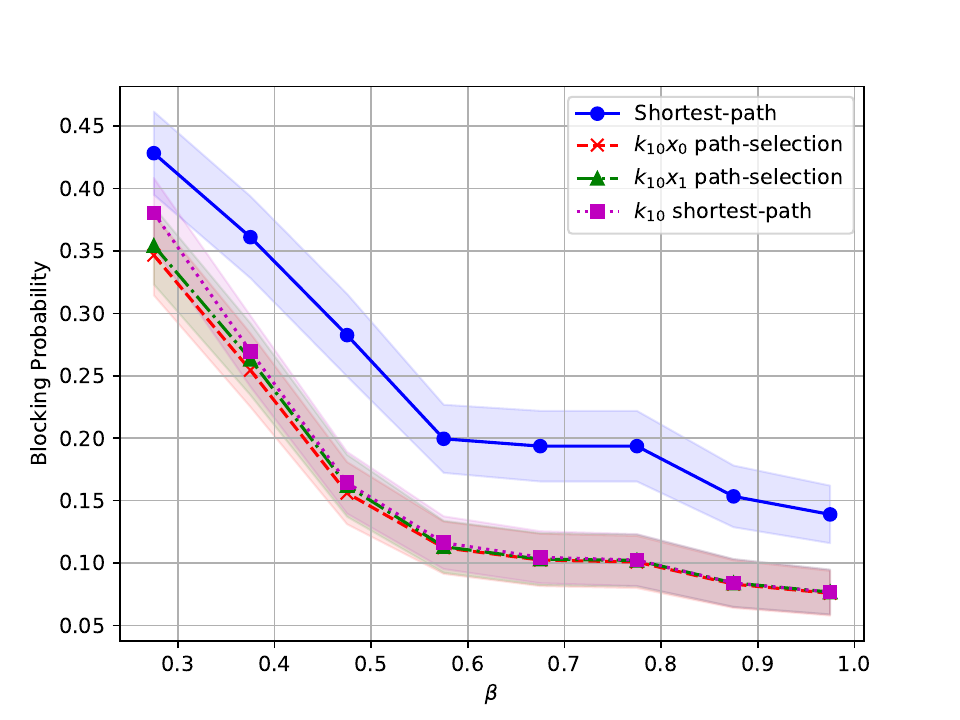}
        \caption{For $k = 10$}
        \label{fig:bp_wax_spec_conf10}
    \end{subfigure}
    \hspace{-1cm}
    \begin{subfigure}{0.52\textwidth}
        \includegraphics[width=\linewidth]{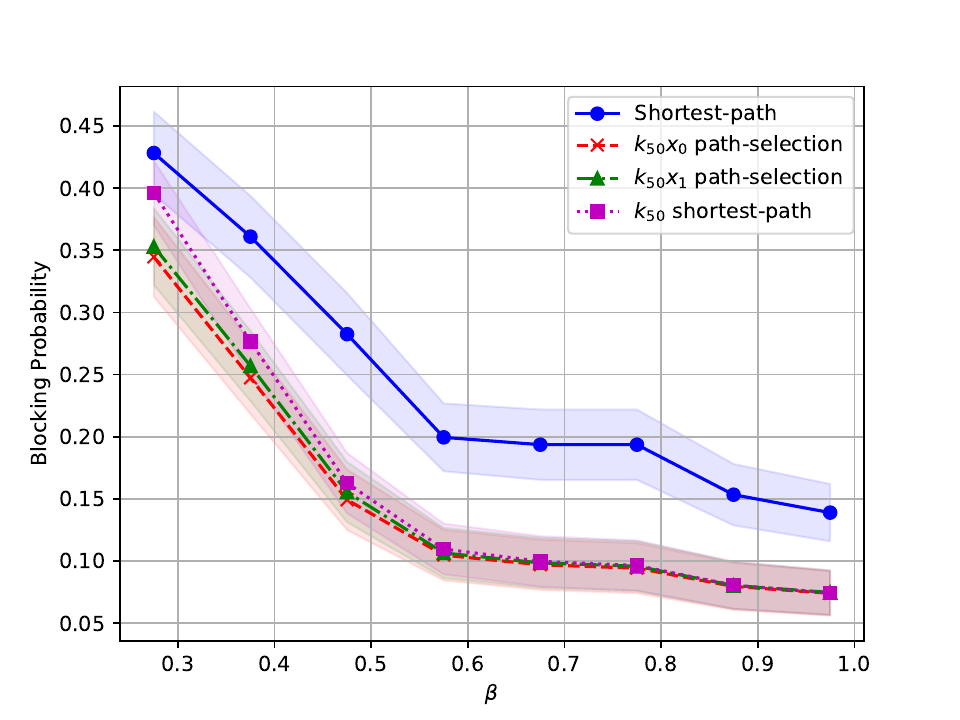}
        \caption{For $k = 50$}
        \label{fig:bp_wax_spec_conf50}
    \end{subfigure}

    \caption{Blocking probability vs. Average degree ($\beta$) with a fidelity threshold ($F_{th}$) = 0.53, number of source and destination ($n_{sd}$) = 5 in a random topology, and having a range of efficiency figures of quantum repeaters. Comparison of shortest-path with grey-box algorithms having algorithm configuration parameter: a) $k = 10$, b) $k = 50$.}
    \label{fig:bp_wax_spec_conf}
\end{figure*}

\subsection{Range of Efficiency Figures}
As a final case-study of this aspect, we test the performance of the considered routing algorithms by relaxing the constraint on two node classes and assigning instead efficiency figures that are randomly drawn in the range given by Eq.~\eqref{eq:range}, as detailed in Section~\ref{sec:topology}. We set $a = 10$ in Eq.~\eqref{eq:range} to attain practically useful performances of the routing policies. Setting a higher value of parameter $a$ would decrease the BP, even achieving BP~$=0$ for higher $\beta$ values. Conversely, setting a lower value would increase the blocking probabilities above 0.5, which can be assumed to be unrealistic for practical purposes. 

As a recap, the network configuration under consideration employs a random topology with $n_{sd} = 5$. In both the $k$-shortest path (with $k = 10$ as $k_{10}$) and $kx$-path selection (with $k = 10, \ x = 0$ as $k_{10}x_{0}$; $k = 10, \ x = 1$ as $k_{10}x_1$; $k = 50, \ x = 0$ as $k_{50}x_{0}$; $k = 50, \ x = 1$ as $k_{50}x_1$) approaches, we test the performance for $k = 10$ as well as $k = 50$. 

Fig.~\ref{fig:bp_wax_spec_conf} illustrates the correlation between the blocking probability and the Waxman parameter ($\beta$), with algorithm configurations, i.e., $k = 10$ (Fig.~\ref{fig:bp_wax_spec_conf}a) and $k = 50$ (Fig.~\ref{fig:bp_wax_spec_conf}b) comprising nodes of varying efficiency. In comparison to the SP, all the other policies yielded superior results, with $kx_{0}$ achieving top performance, as already observed in the case of two-class quality repeaters. 

With a lower $\beta$ value (i.e., a lower number of links), some differences emerge, as explained in the following. The parameter $k$, that is, the number of candidate shortest paths, has a minor influence on the performance, as illustrated in Fig.~\ref{fig:bp_wax_spec_conf}. For a higher value, i.e., $k = 50$, the performance of $kx_0$ and $kx_1$ is marginally better compared to a lower value, i.e., $k = 10$. In contrast, the performance of KSP slightly deteriorates when $k=50$. Although having additional shortest paths in the sample from which paths are selected does not seem beneficial for KSP. $kx_0$ and $kx_1$ benefit from the increased number of paths, even though the improvement is marginal. The reason is that KSP suffers from the same issue already discussed: the presence of additional paths leads to the selection of longer path lengths, which consumes valuable resources. Conversely, $kx_0$ efficiently utilises these additional paths without suffering from this drawback. 

However, as $\beta$ increases, leading to a richer diversity of potential paths, the differences in path selection among the methodologies diminish, and their performance begins to converge irrespective of parameter $k$. The greater path diversity, in fact, causes all methodologies to eventually choose the same paths, regardless of the specific parameter being optimised.

\subsection{Conclusions and future works}
In the above aspect of routing of this work, a grey-box approach to routing in quantum networks is introduced, lifting the conventional assumption that necessitates complete transparency and certainty regarding the characteristics of network components. This approach employs a sampling process to estimate the end-to-end fidelity of potential paths for source-destination pairs. Two new policies have been defined, which both assign paths on a first-come-first-serve basis: $k$ shortest-path (KSP), which assigns potentially longer/worse paths, provided that the minimum fidelity is met, to favour requests that may arrive later; and $kx$ path selection (evaluated with $x=0$ and $x=1$), which has a similar concern but only considers paths that are not $x$ hops longer than the minimum one. 
These policies with two alternatives have been compared via extensive simulations: shortest-path, which also does not require additional information about the network topology, and a knowledge-aware (KA) path selection that is assumed to have exact and precise knowledge about the noise figures of the quantum repeaters of the network.

A comprehensive comparison between regular and random topologies, equipped with the same resources, indicates that the performance of a routing algorithm in one topology does not necessarily translate to another. Performance is primarily influenced by the probability mass function of path lengths within each topology. For example, the performance of KA is better than that of KSP with a random topology, while the performance of KSP is better than that of KA with a regular topology. The proposed $kx_0$ path-selection approach surpasses all the other approaches evaluated in delivering end-to-end fidelity paths that exceed the fidelity threshold across various topologies and numbers of source-destination pairs employed, while maintaining a significant fairness among users. Robustness tests confirm that the $kx_0$ path-selection method is resilient to inaccuracies in the supplied end-to-end fidelity data, underscoring its practical applicability in real-world scenarios. From our simulations, it is evident that the blocking probability per edge associated with the $kx_0$ path-selection approach scales in a manner slightly less than quadratic, showcasing its potential scalability. Additionally, this approach demonstrates superior performance in scenarios involving a diverse range of node efficiency figures.

Despite the extensive analysis in this work, some research challenges remain open for possible future investigation. First, throughput, which depends on the stochastic failure/success of local link entanglement generation and entanglement swapping, could be studied together with the fidelity of end-to-end entanglements, which is the focus of this work. Second, a system-level simulation could be performed, not limited to the path selection, but also includes all the phases foreseen in a time-slot. Finally, it is possible to add priorities to the connections to provide users with a differentiated quality of service by appropriately considering them as part of the path selection algorithm.

\chapter{Calibration Aware Quantum Link Orchestration}\label{ch:calibration}

With this chapter, we begin the quest into the second scope of this work (see Section~\ref{sec:scope_objective}). As a reminder, we will be focusing on the third research question (see Section~\ref{sec:research_questions}) defined for this work in this chapter.

\section{Experimental motivation for calibration modelling}\label{sec:motivation_calibration}
As established in Section~\ref{sec:stages} and Section~\ref{sec:routing_forwarding_scheduling}, at the core of wide-area quantum networking lies \emph{entanglement distribution}. Entanglement is first generated between neighbouring nodes and then extended across the network via \emph{entanglement swapping} at intermediate nodes as seen in Chapter~\ref{ch:routing}. From an experimental perspective, recent milestones have demonstrated the feasibility of multi-hop and multinode operation across different platforms, including remote solid-state qubits linked into a multinode network \cite{pompili2021realization}, trapped-ion fiber links over hundreds of meters \cite{krutyanskiy2023entanglement}, telecom-compatible nanophotonic quantum memories \cite{knaut2024entanglement}, and city-scale automated entanglement distribution over deployed fiber \cite{craddock2024automated}. Entanglement-assisted functionality at the network level has also been demonstrated, e.g., teleportation between non-neighbouring nodes \cite{hermans2022qubit}. These advances underscore that entanglement distribution is achievable in practice, yet they also reveal a system bottleneck: \emph{calibration}.

In fiber-connected links, slow polarisation and spectral drift degrade the initial link fidelity over time, requiring periodic calibration during which links are temporarily unavailable. A representative dataset is shown in Fig.~\ref{fig:fid_vs_time}, derived from the trapped-ion fiber experiment of \cite{krutyanskiy2023entanglement}: the average transmitted fidelity falls from $\approx 1$ to $\approx 0.5$ over $\sim 10{,}000$~s ($\sim$166~min), after which a $\sim$2~min calibration restores the link. Such duty-cycling directly reduces end-to-end throughput, and crucially interacts with multi-path orchestration when links are shared among paths.  

The abstraction adopted here does not assume a specific microscopic origin of drift, but instead captures its operational effect on usable link quality under typical network operation. In this chapter, we especially account for this experimental constraint and formulate \emph{calibration-aware} entanglement distribution as a network optimisation problem, first for a linear quantum network chain and then extend it to any general quantum network.

\begin{figure}[!t]
\centering
\includegraphics[width=0.8\columnwidth]{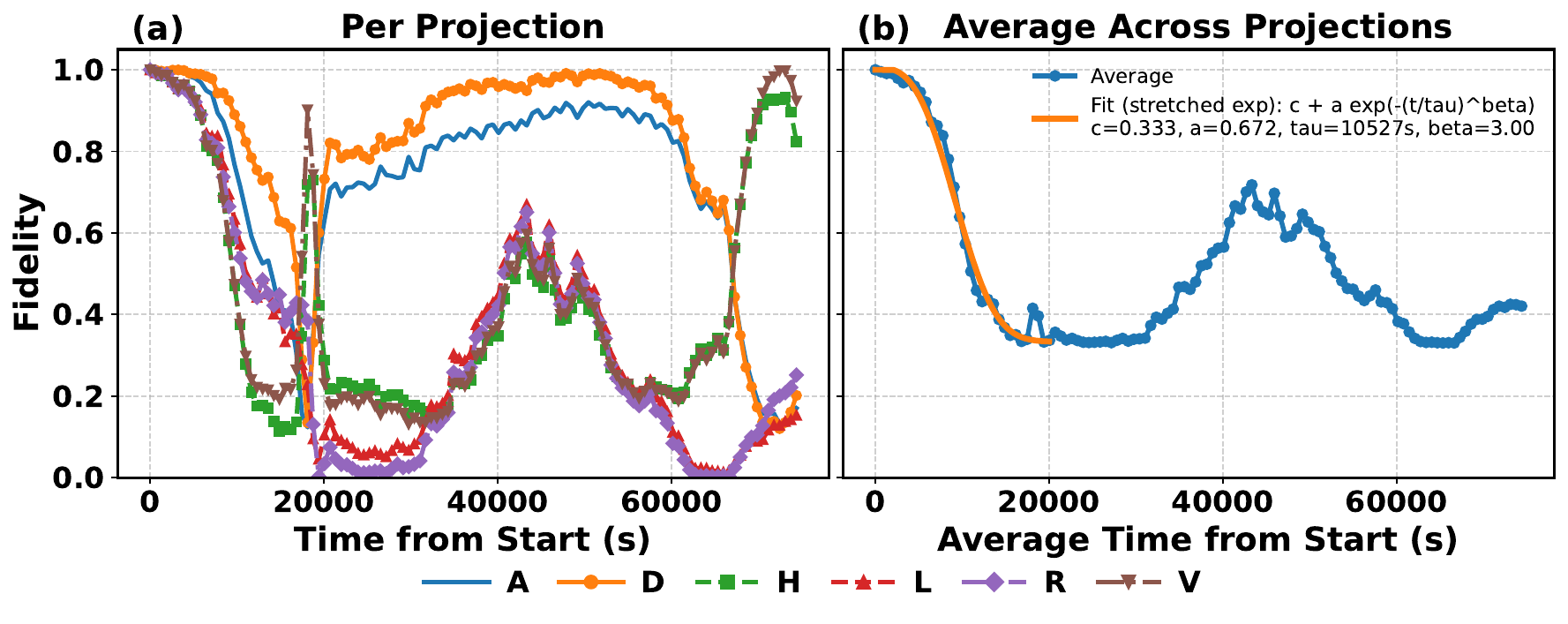}
\caption{Experimental polarisation/fidelity drift on the fiber link of \cite{krutyanskiy2023entanglement}. 
(a) Time evolution of polarisation states $\{A, D, H, L, R, V\}$. 
(b) The averaged fidelity decay over the activation period is well described by an exponential-family model (see fitted curve), motivating the exponential decay assumption adopted in the system model.}
\label{fig:fid_vs_time}
\end{figure}

\begin{table*}[!t]
\centering
\caption{Notation reference for chapter~\ref{ch:calibration}.}
\label{tab:notation}
\renewcommand{\arraystretch}{1.15}
\setlength{\tabcolsep}{6pt}
\begin{tabular}{p{0.22\textwidth} p{0.73\textwidth}}
\hline
\textbf{Symbol} & \textbf{Meaning} \\
\hline
$G=(V,E)$ & Quantum network graph with node set $V$ and link set $E$ \\
$e \in E$ & A physical quantum link (edge) \\
$\pi$ & A generic source to destination path \\
$N$ & Number of links in a linear chain (path length in links) \\
$i,j,k$ & Link indices; also used for partition sizes in QLO (see Theorem~\ref{theorem.chain}) \\
$\gamma \in \mathbb{N}$ & Optimization level (recursion depth) in QLO \\
$a_e$ & Activation phase duration of link $e$ \\
$c_e$ & Calibration phase duration of link $e$ (same for each link in this work) \\
$T$ & End-to-end throughput (entanglement delivery rate) \\
$C$ & Path dependent constant factor in throughput (success terms, attempt rate), e.g., $\mathcal{A}p^N q^{N-1}$ \\
$\mathcal{A}$ & Total attempt rate or number of attempts (model constant) \\
$T_s$ & Duration of one entanglement generation attempt \\
$x$ & Number of attempts in a fidelity window $(T_f)$, typically $x = T_f/T_s$ \\
$\tau_c$ & Calibration time (alternate notation for $c_e$ in Section~\ref{sec:system_model_calibration}) \\
$p,q$ & Success probabilities for external links and internal operations, respectively \\
$F_e$ & Initial fidelity of an entangled pair generated on link $e$ \\
$F_e^M$ & Best achievable initial fidelity on link $e$ (value as $a_e \to 0$), default numerical value $\approx 0.99$ \\
$F_e^{th}$ & Per link threshold on initial fidelity (minimum acceptable initial fidelity) \\
$\Gamma_e$ & Fidelity decay rate parameter for link $e$ (exponential drift model) \\
$F_{ete}$ & End-to-end fidelity after entanglement swapping along a path \\
$F_{ete}^{th}$ & End-to-end fidelity requirement (threshold) for a path \\
$p_1,p_2$ & Single qubit and two qubit operation success parameters in swapping model \\
$\eta$ & Bell state measurement efficiency parameter \\
$\mathcal{U}$ & Swapping operation factor, $\mathcal{U}=(p_1^2p_2)^{N-1}\left(\frac{4\eta^2-1}{3}\right)^{N-1}$ \\
$A$ & Drift model constant, $A = F_e^M - \tfrac{1}{4}$ \\
$L_e$ & Log fidelity variable, $L_e = \Gamma_e a_e = \ln\!\left(\frac{A}{F_e-\tfrac{1}{4}}\right)$ \\
$L_e^{th}$ & Per link log fidelity upper bound induced by $F_e^{th}$ \\
$\Omega$ & Equal activation optimum for a path (from end to end constraint), $\Omega_F$ (F-space) $\&$ $\Omega_L$ (L-space) \\
$\Omega_\pi$ & Optimal point for a specific path $\pi$ \\
$\mathcal{L}$ & End-to-end log fidelity budget for a chain, typically $\sum_{e\in\pi} L_e \le \mathcal{L}$ \\
$\mathcal{L}_\pi$ & Path specific budget for path $\pi$ in the network wide formulation \\
$K_e$ & Calibration term in $L$ space, typically $K_e=\Gamma_e c_e$ \\

\hline
\end{tabular}
\end{table*}

\section{System model}\label{sec:system_model_calibration}
Continuing with the trend of this work, we consider a general quantum network structure consisting of quantum repeaters, quantum links, quantum devices, and a network controller \cite{kumar2025routing, kumar2024routing}. The components of this quantum network consist of the necessary equipment, such as heralded sources of entanglement, detectors, and quantum memories, enabling execution of any network-level protocol. The quantum links are assumed to be optical fibers connecting neighbouring nodes, akin to the experimental setting in \cite{krutyanskiy2023entanglement}, and are continuously required to attempt the generation of an entangled pair, either for advanced or on-demand entanglement generation (see Section~\ref{sec:routing_forwarding_scheduling} and Fig.~\ref{fig:advanced_and_on_demand}) required for quantum routing. As discussed in Section~\ref{sec:motivation_calibration}, we model the degradation in the initial fidelity of a successfully generated entanglement pair along with a calibration period, as shown in Fig.~\ref{fig:fidelity_decay_ton}. 

\definition{In a quantum network, for a quantum link $e$, the link activation phase $(a_e)$ is the period during which it is actively generating entanglement.}
\definition{In a quantum network, for a quantum link $e$, the link calibration phase $(c_e)$ is the period dedicated to calibrating the quantum link for optimal entanglement generation.}

It is to be noted that the generation of the entangled pair through each link is inherently probabilistic. For a successful delivery of throughput, each of the links in the path connecting end users has to be successful within the fidelity period window as depicted in Fig.~\ref{fig:fidelity_decay_ton}.

\begin{figure}[!t]
\centering
\includegraphics[width=0.8\columnwidth]{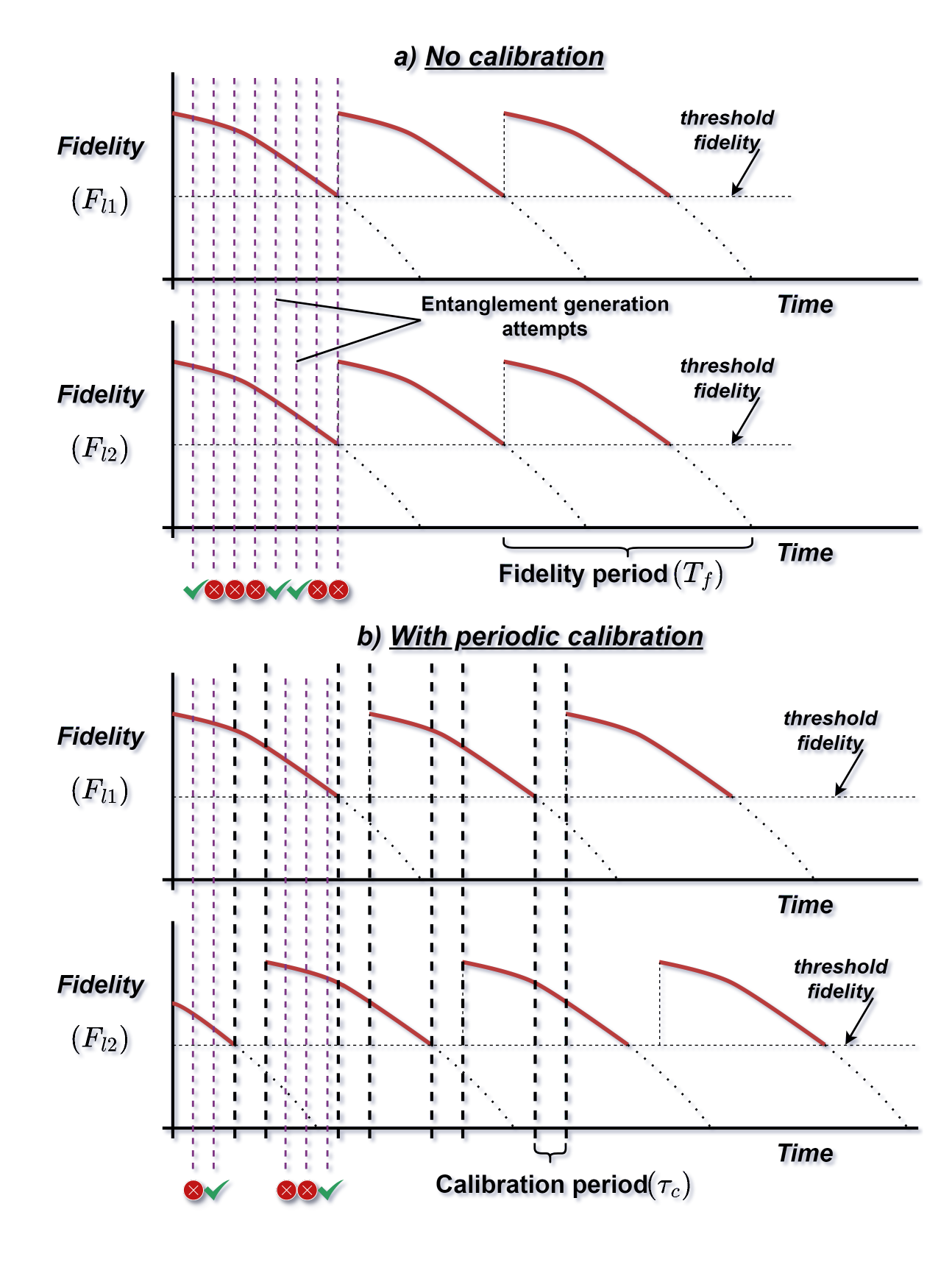}
\caption{An example of a two-link quantum network. Fidelity vs. Time for each link with \textbf{a) No calibration} and \textbf{b) With periodic calibration.} The initial fidelity of the generated entangled pair decays over a fidelity period $(T_f)$. A fixed calibration period $(\tau_c)$ is required to rejuvenate each link.}
    \label{fig:fidelity_decay_ton}
\end{figure}

\subsection{Calibration modelling}
We now define the calibration of the quantum network links and their types, and establish its model by accounting for it in the throughput between endpoints.
\definition{In a quantum network, the throughput $(T)$ is the rate at which quantum states can be successfully transmitted across the two endpoints.}

In a functioning quantum network, the calibration can be performed in two ways as follows:
\begin{enumerate}
    \item \textbf{Synchronised periodic calibration (SPC):} In SPC, all the quantum links are in sync to have overlapping activation and calibration phases.
    \item \textbf{Unsynchronised periodic calibration (USPC):} In USPC, the quantum links operate with non-overlapping activation and calibration phases.

\end{enumerate}
It is to be noted that calibration restores only the \emph{initial fidelity} of the entangled pairs and does not change the underlying attempt success probabilities $p$ and $q$. Thus, drift affects the quality of successful pairs but not the raw attempt rate.

Let a linear quantum network chain consisting of $N$ quantum links.
The probability of coinciding attempts in the no calibration case is:
\begin{equation}
    \mathcal{P}_{nc} = 1,
\end{equation}
which gives the throughput ($\mathcal{T} = \mathcal{A} p^{N}q^{N-1}.\mathcal{P}$) as \cite{pant2019routing}
\begin{equation}\label{eq:nc_throughput}
     \mathcal{T}_{nc} = \mathcal{A} p^N q^{N-1},
\end{equation}
where $\mathcal{A}$ = arbitrary number of total attempts; $p$ = success probability of external link; $q$ = success probability of internal link; $N$ = number of external links.

Similarly, for SPC, the throughput is
\begin{equation}\label{eq:spc_throughput}
    \mathcal{T}_{spc} = \mathcal{A} \ p^N q^{N-1} \left( \frac{x \ T_s}{x \ T_s + \tau_c} \right),
\end{equation}
where $x$ = number of attempts possible in fidelity period ($T_f$); $T_s$ = time taken by an attempt; $\tau_c$ = calibration period.

For USPC, 
    considering independent calibration of links, the probability of coinciding attempts in the USPC is simply
\begin{equation}
    \mathcal{P}_{uspc} = \left( \frac{x}{ x + \frac{\tau_c}{T_s}} \right)^{N},
\end{equation}
with a corresponding throughput of 
\begin{equation}\label{eq:uspc_throughput}
    \mathcal{T}_{uspc} = \mathcal{A} \ p^N q^{N-1}  \left( \frac{ x \ T_s}{x \ T_s + \tau_c} \right)^{N}.
\end{equation}

From Eq.~\eqref{eq:nc_throughput}, \eqref{eq:spc_throughput} \& \eqref{eq:uspc_throughput} we have the ratio of throughput in a linear quantum chain for the case of operation when there is no calibration, SPC, and USPC:
\begin{equation}
\mathcal{T}_{nc} : \mathcal{T}_{spc} : \mathcal{T}_{uspc} =
\begin{aligned}[t]  1 :
\left( \frac{x\,T_s}{x\,T_s + \tau_c}
\right) : \left(
    \frac{x\,T_s}{x\,T_s + \tau_c}
  \right)^{N}.
\end{aligned}
\end{equation}
From the above equation, it is clear that throughput decreases in synchronised periodic calibration with respect to the no calibration case by accounting for the calibration. In addition, the throughput then exponentially decreases with the number of links in the path involved in the case of unsynchronised periodic calibration with respect to synchronised periodic calibration.

\subsection{Throughput and fidelity modelling}
We now define the assumptions and model of throughput, initial fidelity of the EPR pair generated, and end-to-end fidelity.
\begin{definition}
    In a quantum network, end-to-end fidelity $(F_{ete})$ is defined as the fidelity of the long-distance entangled state produced between two endpoints through a series of entanglement swapping operations along the network path.
\end{definition}

\begin{assumption}\label{assumption.activationcalibrationratio}
    For a quantum network link $e$, the effective throughput $T$ is proportional to the ratio of activation to the total cycle time as in Eq.~\eqref{eq:spc_throughput} $$T \propto \left( \frac{a_e}{a_e + c_e} \right).$$
    where $a_e = x \ T_s$ is the activation phase of the link, and $c_e = \tau_c$ is the calibration phase of the link.
\end{assumption}
In this work, we will only consider the case of USPC with independent calibration. We leave behind the simpler case of SPC.

\begin{lemma}\label{lemma.throughputN}
    In a linear quantum network with $N$ links, the effective throughput across the network where each link $e$ is operating independently is given by (Eq.~\eqref{eq:uspc_throughput}) $$T = C \left( \frac{a_e}{a_e + c_e} \right)^N.$$
    where $C= \mathcal{A} \  p^Nq^{N-1}$ is some constant encompassing the successful establishment\footnote{The establishment of nodes and links here refers to the probabilistic success of entanglement swapping within nodes and the success of initial entanglement generation over links. While nodes and links are always physically present, it is this probabilistic establishment that makes them usable.} of nodes (with probability $q$) and links (with probability $p$) of the path.
\end{lemma}
\begin{proof}
    From assumption~\ref{assumption.activationcalibrationratio}, the throughput of a single link is proportional to the time it spends in the activation phase relative to the total cycle time. Extending this to $N$ links which are operating independently, the throughput is given by
\begin{equation}
    T = C \left( \frac{a_e}{a_e + c_e} \right)^N
\end{equation}
\end{proof}

\begin{remark}
    While we use $c_e$ for the calibration time of the link $e$ throughout this work, the calibration time for each link in the network is the same and equal to $c_e$. This is in line with the experimental execution of calibration, which requires precisely a fixed amount of time for the recalibration of the link.
\end{remark}

\begin{assumption}\label{assumption.initialfidelityexponential}
    The initial fidelity $F_e$ of the EPR pair generated via link $e$ decreases approximately exponentially in time $a_e (t)$ \cite{inesta2025entanglement}
    $$F_e =  \left( F_e^M - \frac{1}{4} \right) e^{- \Gamma_e a_e} + \frac{1}{4}$$ where $\Gamma$ is the exponential decay parameter and $F^M_e$ is the initial fidelity when $a_e \to 0. $
\end{assumption}

\begin{assumption}\label{assumption.endtoendfidelity}
    The end-to-end fidelity $F_{ete}$ decreases exponentially with the number of links $N$ \cite{dur1999quantum}
\begin{align}
F_{ete}(a_e) 
&= \frac{1}{4} \Bigg[ 1 
+ 3 \left( p_1^2 p_2 \right)^{N-1} 
\left( \frac{4 \eta^2 - 1}{3} \right)^{N - 1} \nonumber\\
&\quad \times \prod_{i=0}^{N-1} \left( \frac{4 F_i(a_e) - 1}{3} \right) \Bigg]
\label{eq:end_to_end_fidelity}
\end{align}

where $N$ = number of links; $F_i (a)$ = initial fidelity function of the link; $p_1$ = efficiency of single-qubit operation (Hadamard gate); $p_2$ = efficiency of two-qubit operation (CNOT gate); $\eta$ = efficiency of bell measurement.
\end{assumption}

\subsection{Constraints: Initial fidelity and end-to-end fidelity}
We now define the constraints on the initial fidelity and end-to-end fidelity.
\begin{lemma}\label{lemma.initialfidelityvariation}
    The activation phase $a_e$ is bounded by the initial fidelity threshold $F^{th}_e$ set for the EPR pair generated in link $e$. 
    $$a_e \leq \frac{1}{\Gamma_e} \ln{\left( \frac{F_e^M - \frac{1}{4}}{F_e^{th} - \frac{1}{4}} \right) }.$$
\end{lemma}

\begin{proof}
Follows from assumption~\ref{assumption.initialfidelityexponential} by setting  $F_e^m \geq F_e^{th}$.
\end{proof}
Lemma~\ref{lemma.initialfidelityvariation} gives the upper bound on the activation phase of the link $e$. At the upper bound, we define
\begin{equation}
    a_e = a_e^{th} = \frac{1}{\Gamma_e} \ln{\left( \frac{F_e^M - \frac{1}{4}}{F_e^{th} - \frac{1}{4}} \right) }
\end{equation}

The constraint on the fidelity of initially generated EPR pairs can be imposed to meet the expectations of purification protocols, which require certain fidelity of the initially generated EPR pairs to perform purification on the links.

\begin{lemma}\label{lemma.endtoendfidelityactivation}
    The activation phase $a_e$ is bounded by the end-to-end fidelity threshold $F_{ete}^{th}$ set for the EPR pair delivery along a path $$a_e \leq \frac{1}{\Gamma_e N} \ln{\left( 
 \frac{\left( \frac{4}{3} \right)^{N-1} \left( 
 F_e^M - \frac{1}{4}\right)^N \mathcal{U}}{F^{th}_{ete} - \frac{1}{4}}\right)}$$
 where $\mathcal{U} = \left( p_1^2 p_2 \right)^{N-1} \left( \frac{4 \ \eta^2 - 1}{3} \right)^{N - 1}$ defines the noisy operations during entanglement swapping.
\end{lemma}

\begin{proof}
Follows from assumption~\ref{assumption.initialfidelityexponential} and \ref{assumption.endtoendfidelity} by setting $F_{ete} \geq F^{th}_{ete}$ and assuming $F_i=F_e.$
\end{proof}
Note that here, the end-to-end fidelity expression assumes identical noisy-operation parameters $(p_1,p_2,\eta)$ and identical functional form of $F_e$ across all links. Heterogeneity considerations can be addressed by simply extending the notation. The end-to-end fidelity of EPR delivery is a general metric used in the evaluation of the quality of delivery. A constraint on it refers to a certain threshold end-to-end fidelity requested by users to run a quantum application on their end or certain threshold required by purification protocols on end-to-end similar to initial fidelity constraint.

\subsubsection{Fidelity and Log-fidelity variable spaces}
In the upcoming analysis, apart from fidelity, as will be seen, working in the Log-fidelity variable space is easier and suitable for analysis. Below is the relation between these two working spaces.

\begin{remark}
From assumption~\ref{assumption.endtoendfidelity}, we express initial fidelity in terms of the end-to-end fidelity threshold. Assume perfect operations, that is, $\mathcal{U}=1$ and same initial fidelity $F_i=F_e$ for each link we have
\begin{equation}
    F_e = \left( \frac{4}{3} \right)^{\left( \frac{1-N}{N} \right)} \left( F_{ete}^{th} - \frac{1}{4} \right)^{ \left( 1/N \right)}
\end{equation}
Let $L_e = \ln{\left( \frac{A}{F_e - 1/4} \right)}$ where $A= F_e^M - \frac{1}{4}$ then
\begin{equation}
    \boxed{F_e = \frac{1}{4} + \frac{A}{e^{L_e}}}
\end{equation}
The above equation is the relation between the fidelity variable space (F-space) $F_e$ and logarithm-fidelity variable space (L-space) $L_e$.
\end{remark}

\section{Optimal calibration in linear quantum chains}\label{sec:calibration_prelim}
In this section, we discuss the preliminaries of the problem under different levels of constraints in linear quantum chains under initial fidelity constraint in Section~\ref{ssec:initial_fidelity_const} and under end-to-end fidelity constraint in Section~\ref{ssec:ete_fidelity_cont}, which would enable the discussion of the link orchestration under constraints problem in Section~\ref{ssec:qlotheorem}.
\subsection{Linear quantum chains under initial fidelity constraint}\label{ssec:initial_fidelity_const}
\begin{theorem}[Optimal Throughput via Fidelity Threshold Selection]\label{theorem.optimalthroughput}
    Consider a quantum network chain comprising \( N \) links (that is, \( N+1 \) nodes), where each link \( e \) operates independently and sequentially undergoes an activation phase \( a_e \) and a calibration phase \( c_e \). Let \( F_i^{th} \) denote the threshold fidelity for initial entanglement generation on the \( i^{th} \) link. Then, the throughput \( T \) of the quantum network is maximised when each link is operated at its respective initial fidelity threshold
    $$F_e = F_e^{th} \ \text{or} \ L_e=L_e^{th}$$ From lemma~\ref{lemma.initialfidelityvariation}, in other words, the throughput $T$ is maximized when the activation period $a_e$ of links are operated at its respective thresholds. $$a_e = a_e^{th}.$$
\end{theorem}
\begin{proof}
    From Lemma~\ref{lemma.throughputN}, the throughput is strictly increasing in each activation duration $a_e$ (equivalently in each log-fidelity variable $L_e$) over the feasible region. Because the only restriction on $a_e$ is the upper bound imposed by the fidelity constraint, the maximum throughput is necessarily achieved at the largest feasible activation duration, namely $a_e^{th}$. No interior point can be optimal under a strictly increasing objective.
\end{proof}

\subsection{Linear quantum chains under end-to-end fidelity constraint}\label{ssec:ete_fidelity_cont}
\begin{theorem}[Uniform Activation-Fidelity Theorem]\label{theorem.equalactivation}
Consider a quantum network chain comprising $N$ links, where each link $e$ operates independently and sequentially undergoes an activation phase $a_e$ and a calibration phase $c_e$. Let $F_{ete}^{th}$ denote the end-to-end fidelity threshold for the network chain. Then, the throughput $T$ of the quantum network is maximised when all links are activated equally, and the fidelity of EPR generation in each link equals the optimal point $\Omega$ defined in Eq~\eqref{eq:omega_f} and \eqref{eq:omega_l}.
\end{theorem}

\vspace{-0.5em}
\begin{align}
\Omega_F(F_{ete}^{th}, N) = F_e 
&= \left( \frac{4}{3} \right)^{\left( \frac{1 - N}{N} \right)} 
\left( F_{ete}^{th} - \frac{1}{4} \right)^{1/N} 
&& \text{(F-space)}
\label{eq:omega_f}
\end{align}

\vspace{-0.5em}
\begin{align}
\Omega_L(F_{ete}^{th}, N) = L_e 
&= \left( \frac{N - 1}{N} \right) \ln \left( \frac{4}{3} \right) 
+ \ln(A) \nonumber\\
&\quad - \frac{1}{N} \ln \left( F_{ete}^{th} - \frac{1}{4} \right)
\quad \text{(L-space)}
\label{eq:omega_l}
\end{align}

\noindent where $A = F_e^M - \frac{1}{4}$.

\begin{proof}
    Refer to Appendix~\ref{sec:proof_th_6.2}.
\end{proof}
Throughout this work, we will mainly be working in L-space; we will omit the subscript $L$ and simply denote the optimal point as $\Omega$. To denote the optimal point for a particular path $\pi$, we would use it as the subscript instead, as $\Omega_{\pi}$.

\subsection{Quantum Link Orchestration (QLO) Theorem}\label{ssec:qlotheorem}
We now formally define the throughput optimisation problem accounting for the periodic calibration in a linear quantum chain as Problem~\ref{prob:problem1}. We then treat the problem analytically and provide its solution as Theorem~\ref{theorem.chain}.
\begin{problem}\label{prob:problem1}[Throughput optimisation]
For a given quantum network chain $G(E, V)$ (also known as path $\pi$) characterised by end-to-end fidelity threshold $F_{ete}^{th}$ (or optimal point $\Omega_{\pi}$ in L-space) where each link alternates between activation phase $a_e$ and calibration phase $c_e$ with each link $E$ having the initial fidelity thresholds $F_e^{th}$ (or log-fidelity threshold variable $L_e^{th} = \Gamma_e a_e^{th} = \ln{\left( \frac{F_e^M - \frac{1}{4}}{F_e^{th} - \frac{1}{4}} \right) }$ where $\Gamma_e$ is the decay rate parameter and $F_e^M$ is the initial fidelity when $a_e \to 0$ in L-space), the throughput optimization problem is to find the activation phases $a_e$ (or log-fidelity variable $L_e$ in L-space) for each link satisfying the individual initial fidelity threshold and the end-to-end fidelity threshold such that the throughput of end-to-end entanglement across the path $\pi$ is maximized.
\end{problem}

\begin{theorem}[Quantum Link Orchestration Theorem]\label{theorem.chain}
    Consider a quantum network chain (also known as path $\pi$) having optimal point $\Omega_{\pi}$ comprising $N$ links, where each link $e$ operates independently and sequentially undergoes an activation phase $a_e$ and a calibration phase $c_e$. Let $L_i^{th}$ denote the threshold fidelity for initial entanglement generation on the $i^{th}$ link and $F_{ete}^{th}$ denote the end-to-end fidelity threshold for the network chain. Then, the throughput $T$ of the quantum network is maximised by allocating links $L_i$ as follows
    \begin{enumerate}
        \item For $ L_i^{th} \leq \Omega_{\pi}: L_i = L_i^{th}$
        \item For $\Omega_{\pi} < L_i^{th}: L_i = \Omega_{\pi}$
        \item For $L_j^{th} \leq \Omega_{\pi} \ \& \ \Omega_{\pi} < L_k^{th}$
        \begin{enumerate}
            \item if $\sum L_j^{th} + \sum L_k^{th} \leq \mathcal{L}: L_j = L_j^{th} \ \& \ L_k = L_k^{th}$
            \item if $\sum L_j^{th} + \sum L_k^{th} > \mathcal{L}: L_j = L_j^{th}$ \ \& \ repeat the procedure for $L_k $ links with their corresponding initial fidelity thresholds $L_k^{th} $ and optimal point for the path $\Omega_{\pi}^{\gamma} = \frac{1}{k} \left[ (j+k) \; \Omega_{\pi} - \left( \sum L_j^{th}\right) \right]$. Here $\gamma \in \mathbb{N}$ denotes the level of optimisation.
        \end{enumerate}
    \end{enumerate}
    where $\mathcal{L} = \sum_{i=0}^{N-1} L_i = (N-1) \ln{\left( \frac{4}{3} \right)} + N \ln{A} - \ln{ \left( F_{ete}^{th} - \frac{1}{4} \right)} $
\end{theorem}
\begin{proof}
    According to theorem~\ref{theorem.optimalthroughput}, $L_i = L_i^{th}$. In addition, from theorem~\ref{theorem.equalactivation} (Eq.~\eqref{eq:sumLi}), the end-to-end fidelity defines the constraint that is, $\sum_{i=0}^{N-1} L_i \leq (N-1) \ln{\left( \frac{4}{3} \right)} + N \ln{A} - \ln{ \left( F_{ete}^{th} - \frac{1}{4} \right)}$. Let the individual thresholds on the initial fidelity be $L_i \leq L_i^{th}$. 
    Essentially, the subsection of N-dim L-hyperspace under consideration is defined by
    \begin{enumerate}
        \item Equation of optimality: $L_i = L_i^{th}$
        \item End-to-end fidelity constraint: $\sum_{i=0}^{N-1} L_i \leq (N-1) \ln{\left( \frac{4}{3} \right)} + N \ln{A} - \ln{ \left( F_{ete}^{th} - \frac{1}{4} \right)}$
        \item Individual initial fidelity threshold constraint: $L_i \leq L_i^{th}$
    \end{enumerate}
    Consider the following cases:
    \begin{enumerate}
        \item For $L_i^{th} \leq \Omega_{\pi}$, all the individual initial fidelity thresholds are less than the optimal point $\Omega_L$. In such a case, constraints are fulfilled trivially; hence, using the equation of optimality, the throughput is maximised for $L_i=L_i^{th}.$
        \item For $\Omega_{\pi}<L_i^{th}$, all the individual initial fidelity thresholds are more than the optimal point $\Omega_{\pi}.$ In such a case, operating at optimal using the equation of optimality would violate constraint 2. Hence, to satisfy constraint 2, a simple optimal point is used, that is, the throughput is maximised when $L_i=\Omega_{\pi}.$
        \item For $L_j^{th} \leq \Omega_{\pi} \ \& \ \Omega_{\pi} < L_k^{th}$, that is, $j$ of the individual initial fidelity thresholds are less than the optimal point $\Omega_{\pi}$, and $k$ of the individual initial fidelity thresholds are more than the optimal point $\Omega_{\pi}$. In such a case, constraint 2 can be satisfied or violated depending upon the individual initial fidelity thresholds. 
        \begin{enumerate}
            \item If constraint 2 is satisfied that is, $\sum L_j^{th} + \sum L_k^{th} \leq \mathcal{L} $ where $\mathcal{L} = \sum_{i=0}^{N-1} L_i = (N-1) \ln{\left( \frac{4}{3} \right)} + N \ln{A} - \ln{ \left( F_{ete}^{th} - \frac{1}{4} \right)} $, then similar to case 1, using the equation of optimality, the throughput is maximized for $L_j=L_j^{th} \ \& \ L_k=L_k^{th}.$
            \item If the constraint 2 is violated, that is, $\sum L_j^{th} + \sum L_k^{th} > \mathcal{L}$ where $\mathcal{L} = \sum_{i=0}^{N-1} L_i = (N-1) \ln{\left( \frac{4}{3} \right)} + N \ln{A} - \ln{ \left( F_{ete}^{th} - \frac{1}{4} \right)} $ then, $j$ of the links are operated using the equation of optimality that is, $L_j=L_j^{th}$ and hence the original $N$-dim $L$-hyper-space is reduced to $(N-j)$-dim $L$-hyper-space. The remaining $k$ of the links are evaluated again with an increased optimal point given by $\Omega_{\pi}^{1} = \frac{1}{k} \left[ (j+k) \Omega_{\pi} - \left( \sum L_j^{th}\right) \right]$ following the same procedure. See appendix~\ref{sec:optimal_point_proof} for proof of $\Omega^{\gamma+1}_{\pi} \geq \Omega^{\gamma}_{\pi}$. The recursion follows until all the links are optimally allocated at the level of optimisation $\gamma$ with optimal point $\Omega_{\pi}^{\gamma+1} = \frac{1}{k} \left[ 
 (j+k) \; \Omega_{\pi}^{\gamma} - \sum_j L_j^{th} \right].$
        \end{enumerate}
    \end{enumerate}
\end{proof}

\subsection{Recursive Threshold Allocation Algorithm}\label{ssec:allocation_algorithm}
The theorem above describes how the \emph{initial fidelity} $L_i$ on each link must be chosen relative to the global 
optimal point $\Omega_{\pi}$ to maximise throughput. The following pseudocode in Algorithm~\ref{alg:QLO} captures the procedure.

\begin{algorithm}[t]
\caption{Quantum Link Orchestration for a Path $\pi$}
\label{alg:QLO}
\begin{algorithmic}[1]
\REQUIRE Path $\pi$ with $N$ links, initial fidelity thresholds $\{L_i^{th}\}_{i=1}^N$,
end-to-end budget $\mathcal{L}$, base optimal point $\Omega_\pi$
\ENSURE Allocated link values $\{L_i\}_{i=1}^N$

\STATE Initialise optimisation level $\gamma \gets 0$
\STATE Set current optimal point $\Omega \gets \Omega_\pi$
\STATE Initialise fixed set $\mathcal{F} \gets \emptyset$
\STATE Initialise remaining set $\mathcal{R} \gets \{1,2,\ldots,N\}$
\FOR{$i=1$ \TO $N$}
    \STATE $L_i \gets 0$
\ENDFOR

\WHILE{$\mathcal{R} \neq \emptyset$}

    \STATE Identify $\mathcal{S} \gets \{ i \in \mathcal{R} \mid L_i^{th} \le \Omega \}$

    \FOR{$i \in \mathcal{S}$}
        \STATE $L_i \gets L_i^{th}$
        \STATE $\mathcal{F} \gets \mathcal{F} \cup \{i\}$
        \STATE $\mathcal{R} \gets \mathcal{R} \setminus \{i\}$
    \ENDFOR

    \IF{$\mathcal{R} = \emptyset$}
        \RETURN $\{L_i\}_{i=1}^N$
    \ENDIF

    \STATE Compute $S_F \gets \sum_{i \in \mathcal{F}} L_i$
    \STATE Compute $S_R^{th} \gets \sum_{i \in \mathcal{R}} L_i^{th}$

    \IF{$S_F + S_R^{th} \le \mathcal{L}$}
        \FOR{$i \in \mathcal{R}$}
            \STATE $L_i \gets L_i^{th}$
        \ENDFOR
        \RETURN $\{L_i\}_{i=1}^N$
    \ELSE
        \STATE $j \gets |\mathcal{F}|$ \hspace{1em} $k \gets |\mathcal{R}|$
        \STATE Update optimal point:
        \STATE $\Omega \gets \frac{1}{k}\Big( (j+k)\Omega_\pi - S_F \Big)$
        \STATE $\gamma \gets \gamma + 1$
    \ENDIF

\ENDWHILE

\end{algorithmic}
\end{algorithm}

\begin{remark}[Explanation of the Algorithmic Steps] \label{remark.explanation}

Algorithm~\ref{alg:QLO} provides a constructive implementation of the allocation rules derived in Theorem~\ref{theorem.chain} for a given quantum network path $\pi$ consisting of $N$ links. The goal of the algorithm is to determine the per-link allocation values $\{L_i\}$ that maximise the end-to-end throughput while respecting both the individual link fidelity thresholds and the global end-to-end fidelity constraint.

The algorithm proceeds iteratively by fixing links whose allocations are unambiguous under the current optimal point and recursively updating the optimal point whenever the global constraint becomes active.

\begin{enumerate}
\item \textbf{Initialisation:}
The algorithm starts by setting the optimisation level $\gamma = 0$ and initialising the current optimal point $\Omega$ to the base optimal point $\Omega_\pi$ of the path. Two sets are maintained throughout the execution:
(i) the fixed set $\mathcal{F}$, containing links whose allocations have been permanently determined, and
(ii) the remaining set $\mathcal{R}$, containing links that are still subject to optimisation.
Initially, $\mathcal{F}$ is empty, $\mathcal{R} = \{1,2,\ldots,N\}$, and all link allocations are set to zero.

\item \textbf{Identification of trivially allocatable links:}
At each iteration, the algorithm identifies the subset
\[
\mathcal{S} = \{ i \in \mathcal{R} \mid L_i^{th} \le \Omega \},
\]
that is, the links whose individual fidelity thresholds lie below the current optimal point. According to cases (1) and (3) of Theorem~\ref{theorem.chain}, these links can safely operate at their threshold values without violating the end-to-end constraint. Consequently, for each $i \in \mathcal{S}$, the allocation is fixed as $L_i = L_i^{th}$, and the link is moved from the remaining set $\mathcal{R}$ to the fixed set $\mathcal{F}$. If this operation exhausts the remaining set, the algorithm terminates.

\item \textbf{Feasibility check of remaining links}
If unallocated links remain, the algorithm evaluates whether the remaining links can also be assigned their individual thresholds without violating the global fidelity budget. To this end, it computes
\[
S_F = \sum_{i \in \mathcal{F}} L_i
\quad \text{and} \quad
S_R^{th} = \sum_{i \in \mathcal{R}} L_i^{th}.
\]
If the condition $S_F + S_R^{th} \le \mathcal{L}$ is satisfied, the end-to-end fidelity constraint holds even when all remaining links operate at their thresholds. In this case, corresponding to case (3a) of Theorem~\ref{theorem.chain}, the algorithm assigns $L_i = L_i^{th}$ for all $i \in \mathcal{R}$ and terminates.

\item \textbf{Optimal point update and recursion}
If instead $S_F + S_R^{th} > \mathcal{L}$, the global constraint would be violated by assigning all remaining links their thresholds. This corresponds to case (3b) of Theorem~\ref{theorem.chain}. At this stage, the allocations in $\mathcal{F}$ are frozen, and the optimisation problem is reduced to the remaining $k = |\mathcal{R}|$ links. Letting $j = |\mathcal{F}|$, the optimal point is updated as
\[
\Omega \gets \frac{1}{k}\left( (j+k)\Omega_\pi - S_F \right),
\]
which coincides with the recursive optimal point $\Omega_\pi^\gamma$ defined in Theorem~\ref{theorem.chain}. The optimisation level $\gamma$ is incremented, and the procedure is repeated on the reduced link set with the updated optimal point.

\item \textbf{Termination and optimality}
The algorithm terminates once all links have been allocated. The resulting allocation satisfies the individual threshold constraints $L_i \le L_i^{th}$, the global end-to-end fidelity constraint $\sum_i L_i \le \mathcal{L}$, and achieves throughput optimality as characterised by Theorem~\ref{theorem.chain}. As shown in Appendix~\ref{sec:optimal_point_proof}, the recursive update of the optimal point is non-decreasing, ensuring convergence in a finite number of iterations.
\end{enumerate}
\end{remark}

\subsection{Complexity and Optimization Depth}\label{ssec:complexity}
\begin{corollary}[Worst-Case Complexity]
\label{cor:qlo_complexity}
In the worst case, Algorithm~\ref{alg:QLO} runs in $O(N^2)$ time.
\end{corollary}
\begin{proof}
At each iteration, the algorithm scans the remaining set $\mathcal{R}$ to identify
$\mathcal{S}=\{i\in\mathcal{R}\mid L_i^{th}\le \Omega\}$, which costs $O(|\mathcal{R}|)$.
In the worst case, exactly one link is fixed per iteration, so $|\mathcal{R}|$ decreases
from $N$ to $1$. Hence the total scanning cost is
$\sum_{m=1}^{N} O(m)=O(N^2)$.
All additional per-iteration updates and sum computations are at most linear in $N$,
therefore the overall worst-case time complexity is $O(N^2)$.
\end{proof}

\begin{remark}[Optimisation Level]
\label{remark.optlevel}
Algorithm~\ref{alg:QLO} proceeds through a sequence of \emph{optimisation levels} indexed by
$\gamma \in \mathbb{N}$, where each level corresponds to an update of the effective optimal point
$\Omega_\pi^\gamma$ used to allocate the links that remain unresolved.

At the beginning of level $\gamma$, let $\mathcal{R}^\gamma$ denote the set of \emph{remaining} (unfixed) links,
and let $\mathcal{F}^\gamma$ denote the set of links already fixed at earlier levels. The current optimal point is
$\Omega_\pi^\gamma$ (with $\Omega_\pi^0 = \Omega_\pi$). The remaining links are partitioned as
\[
\mathcal{J}^\gamma \;=\; \{\, i \in \mathcal{R}^\gamma : L_i^{th} \le \Omega_\pi^\gamma \,\},
\qquad
\mathcal{K}^\gamma \;=\; \{\, i \in \mathcal{R}^\gamma : L_i^{th} > \Omega_\pi^\gamma \,\}.
\]
All links in $\mathcal{J}^\gamma$ are \emph{frozen} at their thresholds, i.e., $L_i = L_i^{th}$ for
$i \in \mathcal{J}^\gamma$, and moved from $\mathcal{R}^\gamma$ to $\mathcal{F}^\gamma$.

Next, the algorithm checks whether allocating the remaining links in $\mathcal{R}^\gamma$ at their thresholds
respects the end-to-end budget. Writing
\[
S_F^\gamma \;=\; \sum_{i \in \mathcal{F}^\gamma} L_i,
\qquad
S_R^{th,\gamma} \;=\; \sum_{i \in \mathcal{R}^\gamma} L_i^{th},
\]
if $S_F^\gamma + S_R^{th,\gamma} \le \mathcal{L}$ then the procedure terminates by assigning
$L_i = L_i^{th}$ for all $i \in \mathcal{R}^\gamma$.

Otherwise, the budget is violated, and the algorithm updates the optimal point for the unresolved links.
Let $j^\gamma = |\mathcal{F}^\gamma|$ and $k^\gamma = |\mathcal{R}^\gamma|$. The next-level optimal point is
\[
\Omega_\pi^{\gamma+1}
\;=\;
\frac{1}{k^\gamma}\Bigl( (j^\gamma + k^\gamma)\,\Omega_\pi \;-\; S_F^\gamma \Bigr),
\]
which matches the update rule in Theorem~\ref{theorem.chain}. The algorithm then continues to level
$\gamma+1$ on the reduced set of remaining links.

Since at least one link is fixed whenever $\mathcal{J}^\gamma \neq \emptyset$, the process terminates in a finite
number of levels, and the final allocation assigns each link either to its threshold $L_i^{th}$ or according to the
appropriate updated optimal point $\Omega_\pi^\gamma$.
\end{remark}

\begin{corollary}[Number of Orchestration Cases]
\label{cor:number-of-cases}
Consider the quantum network from Theorem~\ref{theorem.chain}, now extended to $P$ overlapping 
source-to-destination paths, and let $N$ be the total number of links. 
At the \emph{primary} (i.e., $\gamma=0$) level of optimization in Theorem~\ref{theorem.chain}, 
the total number of distinct ``cases'' or configurations that can arise across these $P$ paths is
\[
    P! \;\bigl(P + 1\bigr)^N.
\]
Furthermore, for every subsequent level of optimization $\gamma \ge 1$, the new configurations 
are counted by 
\[
1 \;+\; Q_{\gamma}!\;\bigl(Q_{\gamma} + 1\bigr)^{k_{\gamma}},
\]
where $Q_{\gamma}$ is the number of paths still requiring link-allocation decisions at level~$\gamma$, 
and $k_{\gamma}$ is the total number of links whose initial fidelities remain unallocated from the previous level.
\end{corollary}

\begin{proof}
We first prove the count for the primary level $\gamma = 0$. 

\begin{enumerate}
    \item[\textbf{(1)}] \textbf{Permutations of Path Thresholds.} 
    Since there are $P$ overlapping paths, each has its own end-to-end fidelity threshold. 
    We can order these $P$ thresholds in $P!$ ways (e.g., from largest to smallest or any other ordering). 

    \item[\textbf{(2)}] \textbf{Assigning $N$ Links to ``Segments.''} 
    Once a particular ordering of these $P$ thresholds is chosen, the real axis of possible 
    link-fidelity values are partitioned into $P+1$ segments:
    \begin{itemize}
        \item one below the smallest threshold, 
        \item one above the largest threshold, 
        \item and $P-1$ in-between segments (i.e., between successive thresholds in the sorted list).
    \end{itemize}
    Each of the $N$ links (with its initial fidelity threshold $L_i^{th}$) is then assigned to exactly one 
    of these $P+1$ segments. This yields $(P + 1)^N$ possible assignments.
\end{enumerate}

By the multiplication principle, the total number of distinct orchestration cases at level $\gamma=0$ is 
\[
   P! \;\times\; (P + 1)^N 
   \;=\; P!\;\bigl(P + 1\bigr)^N.
\]

For each \emph{subsequent} level of optimization $\gamma \ge 1$, Theorem~\ref{theorem.chain} 
(Case~3b) shows that a new “configuration” arises whenever the sum of thresholds 
(for those links that fell on both sides of $\Omega_{\pi}^{\gamma}$) still violates or just satisfies 
the end-to-end constraint. Effectively, we freeze the already-allocated links 
and focus on $Q_{\gamma}$ paths that remain active (i.e., have unallocated links), with $k_{\gamma}$ 
such links are left to be assigned. These paths and links follow the same combinatorial logic 
as the primary level but with \emph{fewer} total paths/links, yielding $Q_{\gamma}!\,(Q_{\gamma} + 1)^{k_{\gamma}}$ 
configurations. Additionally, there is exactly one extra configuration 
corresponding to the situation where all thresholds are forced into compliance by the 
end-to-end budget (an immediate termination). Hence, the total for level~$\gamma$ is
\[
   1 \;+\; Q_{\gamma}!\;\bigl(Q_{\gamma} + 1\bigr)^{k_{\gamma}}.
\]

This completes the count of orchestration cases across all levels $\gamma$.
\end{proof}

\subsection{Numerical examples}\label{ssec:two_link_example}

\subsubsection{Two-link quantum linear chain}
Let's take an example of a two-link linear quantum chain and visualise the calibration problem with initial fidelity and end-to-end fidelity constraints.

\begin{figure}[!t]
\centering
\includegraphics[width=0.8\columnwidth]{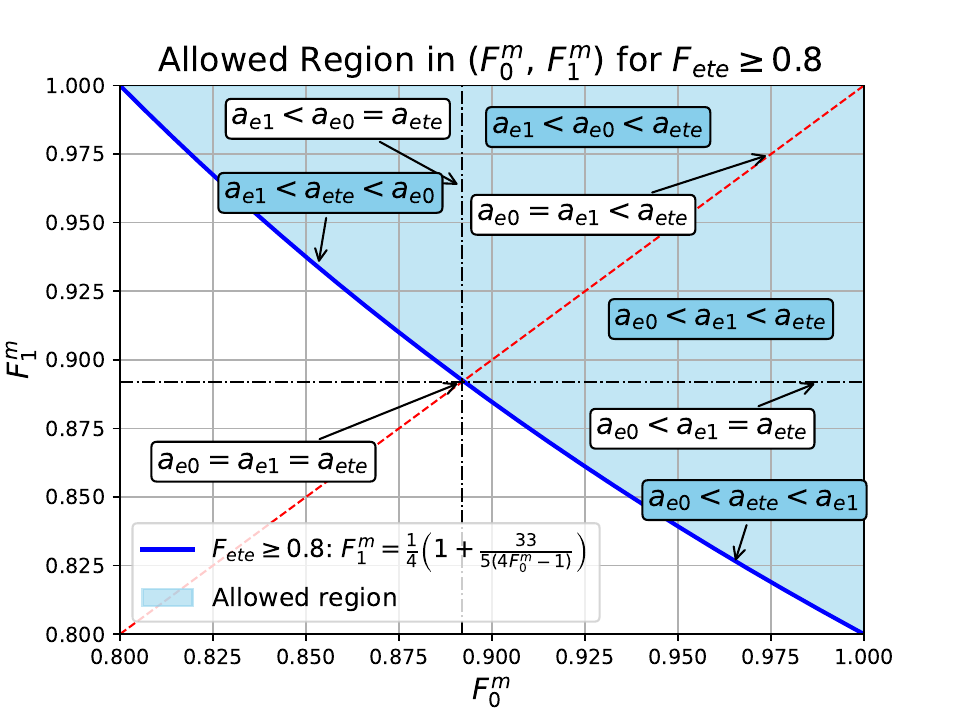}
\caption{Allowed region for a quantum linear chain of two links with end-to-end fidelity $F_{ete} \geq$ 0.8.}
    \label{fig:F0_vs_F1}
\end{figure}

\begin{figure}[!t]
\centering
\includegraphics[width=0.8\columnwidth]{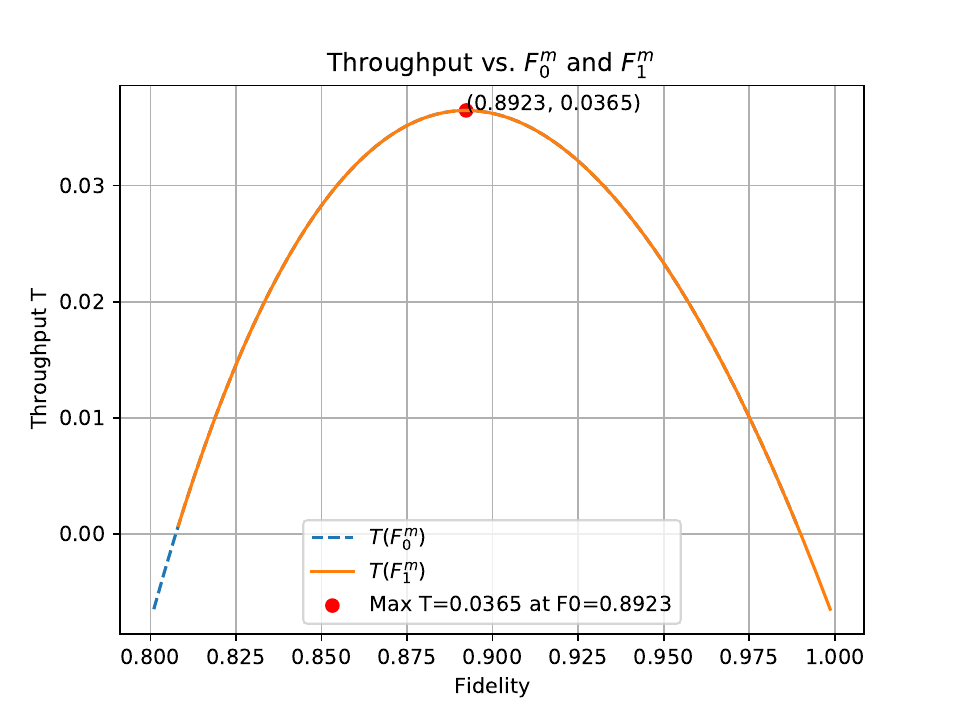}
\caption{Throughput variation at the boundary of the allowed region for a two-link network with an end-to-end fidelity threshold of 0.8. Constants: $C = 1, c_e = 1, F_e^M = 0.99, \Gamma_e = 0.6$.}
    \label{fig:maxT_F01}
\end{figure}

\begin{figure}[!t]
\centering
\includegraphics[width=0.8\columnwidth]{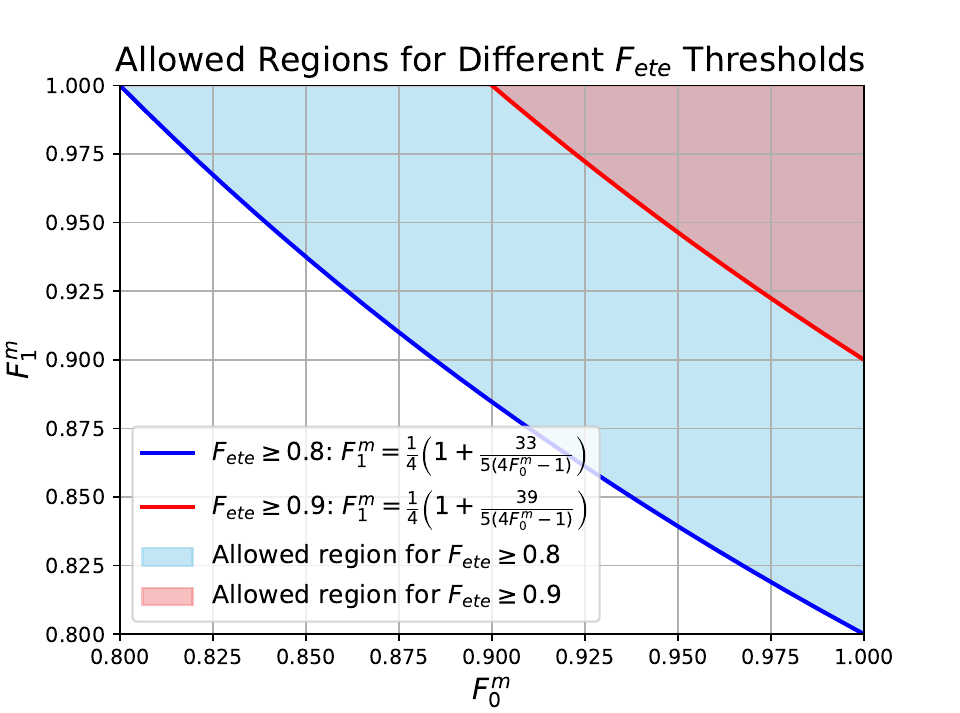}
\caption{Allowed region for a quantum linear chain of two links with end-to-end fidelities $F_{ete} \geq$ 0.8 and $F_{ete} \geq$ 0.9.}
\label{fig:F0_vs_F1_0.8_0.9}
\end{figure}

Utilising Eq.~\eqref{eq:end_to_end_fidelity} for a two-link chain, we have the end-to-end fidelity of EPR pair delivery in terms of the initial fidelity of EPR pair generation at each of the links as:
\begin{equation}
    F_{ete} = \frac{1}{4} \left[ 1 + 3 \left( \frac{4 F_0^m - 1}{3} \right) \left( \frac{4 F_1^m - 1}{3} \right) \right]
\end{equation}
Assuming an end-to-end fidelity threshold for the EPR pair delivery across the two-link linear quantum chain to be 0.8, Fig.~\eqref{fig:F0_vs_F1} depicts the allowed region for the selection of the initial fidelity of the generated EPR pair at each of the link which would meet the set end-to-end fidelity threshold. The allowed region in Fig.~\eqref{fig:F0_vs_F1} is divided into four distinct parts curated by $F_1^m=F_0^m$ (red dotted line) and $F_0^m=F_1^m=0.8923$ (black dotted lines). Point $\Omega_F = (F_0^m, F_1^m) =(0.8923, 0.8923)$, the central point on the graph, is the point of maximum throughput as seen in Fig.~\eqref{fig:maxT_F01}. Using the definitions of activation periods in lemma~\ref{lemma.initialfidelityvariation} and lemma~\ref{lemma.endtoendfidelityactivation} on this optimal point we have $a_{e0} = a_{e1} = a_{ete}$. Moving along the boundary of the allowed region introduces inequality among the activation periods with $a_{e1} < a_{ete} < a_{e0}$ along decreasing $F_0^m$ and $a_{e0} < a_{ete} < a_{e1}$ along increasing $F_0^m$. The lines separating parts in the allowed region work as a transition in the inequality of the activation period. For example, for a constant $F_0^m > 0.8923$, with increase in $F_1^m$, the inequality between the activation periods goes from $a_{e0} < a_{ete} < a_{e1}$ on the boundry to $a_{e0}<a_{e1}=a_{ete}$ at $F_1^m = 0.8923$. This is because of the inverse relation between activation period and fidelity. The increase in $F_1^m$ decreases $a_{e1}$ to make it equal to the $a_{ete}$. Further increasing $F_1^m$ at the same constant $F_0^m > 0.8923$ makes the activation period even smaller, which transitions to inequality of $a_{e0}< a_{e1} < a_{ete}$. Even further increasing $F_1^m$ leads to the transition line $F_1^m = F_0^m$ where the activation periods are related as $a_{e0}=a_{e1}<a_{ete}$. A slight increase now will transition the inequality to $a_{e1}< a_{e0}< a_{ete}$. The same observation can be done for any constant increasing $F_0^m$ with $F_1^m > 0.8923$. The size of the allowed region is dictated by the selection of the end-to-end fidelity threshold as shown in Fig.~\eqref{fig:F0_vs_F1_0.8_0.9}. A higher end-to-end fidelity threshold reduces the area of the allowed region from which the initial fidelity of the EPR pair generated for each link can be selected.

\begin{figure}[!t]
\centering
\begin{subfigure}[t]{0.49\textwidth}
    \centering
    \includegraphics[width=\linewidth]{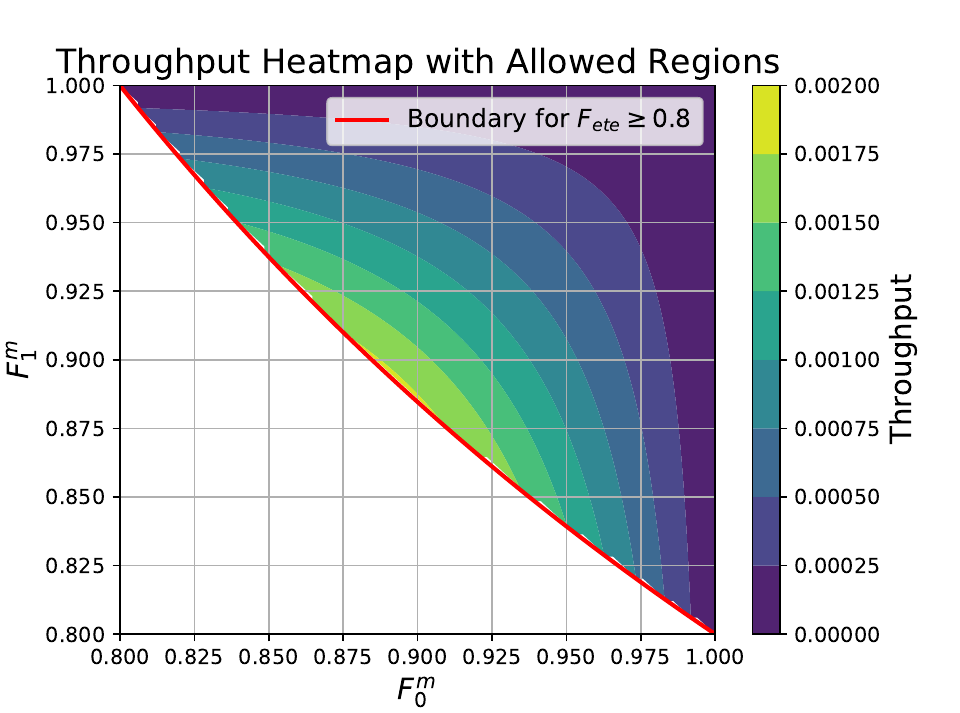}
    \caption{Throughput region, $F_{ete} \ge 0.8$.}
    \label{fig:heatT}
\end{subfigure}
\hfill
\begin{subfigure}[t]{0.49\textwidth}
    \centering
    \includegraphics[width=\linewidth]{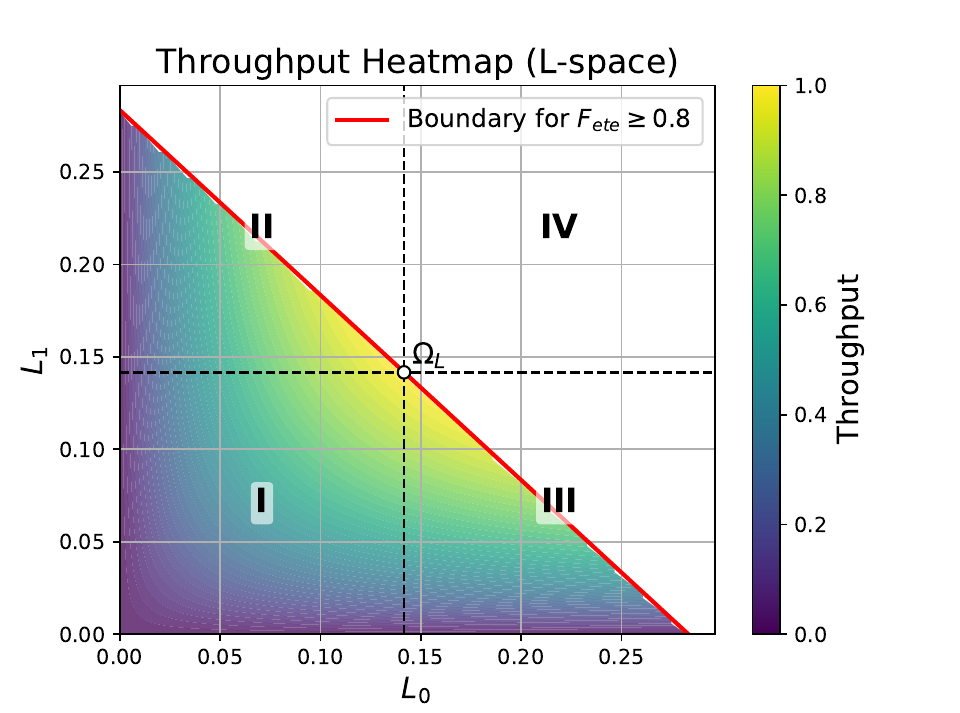}
    \caption{Throughput region in L-space, $F_{ete} \ge 0.8$.}
    \label{fig:heatT_Lspace}
\end{subfigure}
\caption{Heat maps of allowed throughput regions for a two-link quantum linear chain.}
\label{fig:heatT_combined}
\end{figure}

A more informative way of visualising allowed regions for a two-link network chain is done in Fig.~\eqref{fig:heatT_combined}, which gives the heat map of the throughput achieved for each operating point in the allowed region. Fig.~\eqref{fig:heatT} shows the normal fidelity variable space while Fig.~\eqref{fig:heatT_Lspace} shows the corresponding L-space. Similar to Fig.~\eqref{fig:F0_vs_F1}, the central point on the boundary is the point of maximum throughput $\Omega$ or $\Omega_L$(marked in Fig.~\eqref{fig:heatT_Lspace}). Given that we now have the heat map for the throughput visually present, let's write the constraint of our throughput optimisation problem and proceed visually towards tackling the problem. 
Using theorem~\ref{theorem.equalactivation}, for a two-link quantum linear chain, the allowed region in L-space is given by (see Eq~\ref{eq:sumLi} in appendix~\ref{sec:proof_th_6.2}):
\begin{equation}
   \text{C1:} \ L_0 + L_1 \leq \ln{ \left( A^2 \frac{80}{33} \right) }
\end{equation}
Remember that this equation is the manifestation of the end-to-end fidelity constraint in L-space.
Also, let the constraint on the initial fidelity of the generated EPR pair for each link:
\begin{equation}
\begin{aligned}
    \text{C2:} \ L_0 &\leq L_0^{th} \\
    \text{C3:} \ L_1 &\leq L_1^{th}
\end{aligned}
\end{equation}
Now let's consider cases corresponding to each quadrant of the Fig.~\ref{fig:heatT_Lspace}.
\begin{itemize}
    \item \textbf{Case $I$:} For $L_0^{th} \leq \Omega_L$ and $L_1^{th} \leq \Omega_L$, Constraint C1 is always satisfied, so according to theorem~\ref{theorem.optimalthroughput}, the throughput is maximised while operating at the thresholds, i.e., choose $L_0 = L_0^{th}$ and $L_1 = L_1^{th}$.
    \item \textbf{Case $II$:} For $L_0^{th} < \Omega_L$ and $\Omega_L < L_1^{th}$, as seen in the Fig.~\ref{fig:heatT_Lspace}, the area under this quadrant consists of allowed as well as not-allowed region (i.e., satisfy or violate Constraint C1). Hence, we have the following two sub-cases:
    \begin{enumerate}
        \item \textbf{Sub-case I:} If it lies in the allowed region, then the Constraint C1 is satisfied by definition. Then similar to case $I$ choose $L_0 = L_0^{th}$ and $L_1 = L_1^{th}.$
        \item \textbf{Sub-case II:} If it does not lie in the allowed region, Constraint C1 is violated. Then to maximize the throughput choose $L_0 = L_0^{th}$ while $L_1$ on the boundary of the allowed region at $L_0^{th}.$
    \end{enumerate}
    \item \textbf{Case $III$:} For $\Omega_L < L_0^{th}$ and $L_1^{th} < \Omega_L$, variable reverse case of case $II$.

    \item \textbf{Case $IV:$} For $\Omega_L < L_0^{th}$ and $\Omega_L < L_1^{th}$, this quadrant is entirely outside the allowed region, and it violates Constraint C1. To maximise the throughput, we simply choose $L_0 = L_1 = \Omega_L$, which is the point of highest throughput.
\end{itemize}

\section{Calibration aware orchestration in general networks}
Having established the optimal solution for calibration-aware orchestration in a linear quantum network chain, we now explore the problem in a general network. To this end, we first extend QLO to a general quantum network by utilising the \emph{minimum selection rule} which serves as a baseline heuristic (Section~\ref{ssec:minimum_selection_rule}). We then propose a greedy quantum link orchestration heuristic for the calibration problem in general quantum networks (Section~\ref{ssec:greedy_qlo}).

\subsection{Minimum selection rule benchmark}\label{ssec:minimum_selection_rule}

\begin{proposition}[General overlapping-link calibration via minimum selection]
\label{prop:general-overlap}
Consider a quantum network with edge set $E$ and a set of $P$ source--destination pairs, 
each admitting at least one path $\pi_{sd}$. By Theorem~\ref{theorem.chain}, 
each path $\pi_{sd}$ admits an individually optimal allocation 
$L_{\{\pi_{sd},e\}}$ for every link $e \in \pi_{sd}$.

For any link $e \in E$, let $\Gamma(e)$ denote the set of paths that traverse $e$.
A feasible benchmark allocation for shared-link calibration is obtained by assigning
\[
    L_e \;=\; \min_{\pi_{sd} \in \Gamma(e)}\; L_{\{\pi_{sd}, e\}},
    \qquad \forall e \in E \text{ with } \Gamma(e) \neq \varnothing.
\]
Links that are not used by any path incur no overlap constraint.
\end{proposition}

\begin{algorithm}[!t]
\caption{Minimum selection rule benchmark for shared-link calibration}
\label{alg:multi-path}
\begin{algorithmic}[1]
\REQUIRE Edge set $E$ with per-link thresholds $L_e^{th}$; 
a set of source--destination pairs with paths $\{\pi_{sd}\}$; 
per-path QLO allocations $L_{\{\pi_{sd},e\}}$ for each $e \in \pi_{sd}$
\ENSURE Global link assignment $\{L_e : e \in E\}$

\FORALL{$e \in E$}
  \STATE $\Gamma(e) \gets \{\pi_{sd} \mid e \in \pi_{sd}\}$
  \IF{$\Gamma(e) = \varnothing$}
    \STATE $L_e \gets L_e^{th}$
  \ELSE
    \STATE $L_e \gets \min_{\pi_{sd} \in \Gamma(e)} \; L_{\{\pi_{sd},e\}}$
  \ENDIF
\ENDFOR
\RETURN $\{L_e : e \in E\}$
\end{algorithmic}
\end{algorithm}

\begin{remark}[Explanation of the minimum selection rule]
\label{remark:multi-path-explanation}
Algorithm~\ref{alg:multi-path} assigns each link $e$ in the network a calibration threshold
$L_e$ such that no path $\pi_{sd}$ is forced beyond its individually optimal per-link
allocation $L_{\{\pi_{sd},e\}}$.

For each link $e$, the algorithm first identifies the set $\Gamma(e)$ of paths that traverse
that link. If $e$ is unused by all paths, no overlap constraint arises and the link is
assigned its default threshold $L_e^{th}$. If $e$ is shared by one or more paths, the
algorithm assigns
\[
L_e = \min_{\pi_{sd} \in \Gamma(e)} L_{\{\pi_{sd},e\}},
\]
thereby ensuring that no path is forced to operate beyond its feasible region.

Because this assignment is performed independently for each link, the procedure resolves
overlapping-link constraints without iteration or coordination across paths. While this
rule does not yield a globally optimal solution for general quantum networks, it provides
a simple, deterministic, and computationally inexpensive benchmark against which more
sophisticated orchestration heuristics can be evaluated.
\end{remark}

\begin{corollary}[Time complexity of Algorithm~\ref{alg:multi-path}]
\label{cor:multi-path-complexity}
Let $|E|$ be the number of links in the network and let $\mathcal{P}$ denote the set of
candidate paths. For each path $\pi \in \mathcal{P}$, let $|\pi|$ be its hop count, and define the
total path length
\[
M \;=\; \sum_{\pi \in \mathcal{P}} |\pi|.
\]
If each path is stored as a list of links, then Algorithm~\ref{alg:multi-path} can be implemented
in $O(M)$ time by scanning all path-link incidences once and maintaining, for each link $e$,
the running minimum of $L_{\{\pi,e\}}$ over all paths $\pi$ that use $e$.

In a naive implementation that, for each link $e$, checks all paths for membership using a
linear scan of the path list, the cost becomes $O(|E|\cdot M)$ in the worst case. If one further
bounds $|\pi| \le L_{\max}$ for all $\pi \in \mathcal{P}$ and $|\mathcal{P}| = P$, then $M \le P L_{\max}$,
yielding the upper bound $O(|E|\cdot P \cdot L_{\max})$.
\end{corollary}

\begin{remark}[Overall complexity with overlapping paths]
\label{remark:overall-complexity}
When extending the linear-chain orchestration of Theorem~\ref{theorem.chain} to a general
quantum network with $P$ source--destination paths, the overall procedure naturally decomposes
into two stages.

\begin{enumerate}
    \item \textbf{Per-path orchestration (Stage A).}
    Theorem~\ref{theorem.chain} (implemented via Algorithm~\ref{alg:QLO}) is applied
    independently to each path $\pi_{sd}$. Let $N_\pi$ denote an upper bound on the number
    of links in any path. By Corollary~\ref{cor:qlo_complexity}, the worst-case complexity per
    path is $O(N_\pi^2)$. Summing over all $P$ paths yields a total cost of
    \[
        O\bigl(P \cdot N_\pi^2\bigr).
    \]

    \item \textbf{Merging overlapping links (Stage B).}
    After computing the individual per-path allocations $L_{\{\pi_{sd},e\}}$, shared links
    are harmonised using Algorithm~\ref{alg:multi-path}. As discussed in
    Corollary~\ref{cor:multi-path-complexity}, a naive implementation that determines, for
    each link $e \in E$, the set $\Gamma(e)$ of paths using that link and computes the minimum
    allocation incurs a worst-case cost of $O(E \cdot P)$.
\end{enumerate}

Combining both stages, the overall running time is
\[
    O\bigl(P \cdot N_\pi^2\bigr) \;+\; O\bigl(E \cdot P\bigr).
\]
The first term captures the cost of independent per-path orchestration, while the second term
accounts for resolving shared-link constraints across paths. In networks with many links but
relatively short paths, the overlap resolution step may dominate; conversely, when paths are
long, the per-path orchestration cost can be the leading factor.
\end{remark}

\begin{corollary}[Reduction of orchestration complexity with $l$ shared links]
\label{cor:reduced-complexity-l-shared-links}
Let $l$ denote the number of links in a quantum network that are shared by more than one
source--destination path, as per Proposition~\ref{prop:general-overlap}. After independently
computing the per-path optimal allocations via Theorem~\ref{theorem.chain}, the global
orchestration problem reduces to resolving calibration constraints on only these $l$ shared
links.
\end{corollary}

\begin{proof}
From Corollary~\ref{cor:number-of-cases}, the general orchestration problem over $P$ overlapping
paths and $N$ links admits a large combinatorial configuration space, arising from the joint
assignment of link thresholds across paths. However, by Proposition~\ref{prop:general-overlap},
all links that are used by exactly one path retain their individually optimal allocations and
introduce no coupling with other paths.

Consequently, only the $l$ links that are shared by multiple paths require coordination.
For each such link $e$, the minimum selection rule assigns a unique calibration threshold
\[
L_e = \min_{\pi_{sd} \in \Gamma(e)} L_{\{\pi_{sd},e\}},
\]
independently of the order in which other links are processed. Thus, the original global
orchestration problem collapses to $l$ independent selection operations, one per shared link,
with no residual combinatorial coupling between them.

This represents a drastic reduction in orchestration complexity: instead of exploring a
high-dimensional configuration space over all $N$ links, coordination is required only on
the $l \ll N$ shared links, and even there the solution is uniquely determined.
\end{proof}

\subsection{Greedy quantum link orchestration}
\label{ssec:greedy_qlo}

While the minimum selection rule provides a simple and conservative benchmark for resolving
overlapping-link constraints, it does not exploit the calibration budget freed on each path when
shared links are reduced below their path-optimal values. To address this limitation, we propose
\emph{Greedy Quantum Link Orchestration} (GRO), a heuristic that redistributes the freed budget
locally on each path while preserving feasibility across all paths.

The heuristic proceeds in two stages. First, each source--destination path is solved
independently using Theorem~\ref{theorem.chain}, yielding provisional per-path link allocations
$L_{e,\pi}^\star$. Second, conflicts on shared links are resolved by assigning each shared link
the minimum of its per-path optimal values, as in the minimum selection rule. The budget freed
by this conservative assignment is then reallocated greedily on each affected path, but only
over links that are not shared with other paths. In this way, GRO can improve throughput relative
to the benchmark while avoiding new inter-path dependencies.

\paragraph{Shared links and prioritisation.}
Let $S_e$ denote the set of paths that traverse link $e$. GRO operates on the set of shared links
\[
S \;=\; \{\, e \in E \mid |S_e| \ge 2 \,\}.
\]
To prioritise which shared links to process first, each $e \in S$ is assigned a weight
\[
w(e) \;=\;
\bigl( \max_{\pi \in S_e} L_{e,\pi}^\star
      - \min_{\pi \in S_e} L_{e,\pi}^\star \bigr)\,|S_e|.
\]
This \emph{harm index} captures both the disagreement between paths on the preferred calibration
level of $e$ and the number of paths affected by that disagreement. Shared links are processed
in descending order of $w(e)$.

\paragraph{Greedy reallocation.}
When a shared link $e$ is fixed to
\[
L_e = \min_{\pi \in S_e} L_{e,\pi}^\star,
\]
each affected path $\pi \in S_e$ incurs a reduction in its allocated budget relative to its
independent optimum. Let $\mathcal{L}_\pi$ denote the end-to-end budget of path $\pi$ implied by
the fidelity constraint. The residual budget
\[
\Delta_\pi \;=\; \mathcal{L}_\pi - \sum_{e' \in \pi} L_{e'}
\]
is then reallocated over the \emph{unshared} links of $\pi$, following the same principle as the
linear-chain optimal solution, with per-link caps given by $L_{e',\pi}^\star$. Since these links
are exclusive to $\pi$, this local reallocation does not interfere with any other path. The pseudocode for the greedy quantum link orchestration is given as Algorithm~\ref{alg:gro}.

\begin{remark}[Explanation of the GRO Algorithmic Steps]
\label{remark:gro-explanation}

Algorithm~\ref{alg:gro} implements the \emph{Greedy Quantum Link Orchestration (GRO)} heuristic.

The algorithm proceeds in a sequence of conceptually distinct steps:

\begin{enumerate}
\item \textbf{Path–link incidence construction.}
For each physical link $e \in E$, the algorithm first constructs the set
\[
S_e = \{\pi \in \mathcal{P} \mid e \in \pi\},
\]
which records all source--destination paths that traverse $e$. This identifies whether a link
is exclusive to a single path or shared by multiple paths, and thus whether it induces coupling
constraints.

\item \textbf{Initialization of non-shared links.}
Links that belong to exactly one path ($|S_e|=1$) are immediately assigned their independently
optimal per-path allocation $L_{e,\pi}^\star$. Since these links are not shared, this assignment
cannot violate any other path constraint and requires no coordination.

\item \textbf{Identification and prioritisation of shared links.}
The algorithm collects all shared links into the set
\[
S = \{ e \in E \mid |S_e| \ge 2 \}.
\]
Each shared link $e$ is assigned a \emph{harm index}
\[
w(e) =
\bigl( \max_{\pi \in S_e} L_{e,\pi}^\star
      - \min_{\pi \in S_e} L_{e,\pi}^\star \bigr)\,|S_e|,
\]
which quantifies both the disagreement among paths on the preferred calibration of $e$ and the
number of paths affected. Shared links are processed in descending order of $w(e)$ so that the
links inducing the largest aggregate loss are resolved first.

\item \textbf{Conservative resolution of shared links.}
For each shared link $e \in S$, the algorithm assigns
\[
L_e = \min_{\pi \in S_e} L_{e,\pi}^\star,
\]
which is identical to the minimum selection rule. This guarantees feasibility for all paths
using $e$, since no path is forced beyond its individually optimal allocation.

\item \textbf{Greedy local reallocation on affected paths.}
Fixing a shared link may reduce the total allocation available to the paths that traverse it.
For each affected path $\pi \in S_e$, the algorithm computes the residual budget
\[
\Delta_\pi = \mathcal{L}_\pi - \sum_{e' \in \pi} L_{e'}.
\]
This residual is then redistributed greedily over the \emph{unshared} links of $\pi$, i.e.,
\[
\pi_{\mathrm{un}} = \{ e' \in \pi \mid |S_{e'}| = 1 \},
\]
subject to the per-link caps $L_{e',\pi}^\star$ and the end-to-end constraint
$\sum_{e' \in \pi} L_{e'} \le \mathcal{L}_\pi$. Because these links are exclusive to $\pi$, this
reallocation cannot interfere with any other path.

\item \textbf{Termination.}
After all shared links have been processed, every link has a well-defined allocation. The
resulting assignment $\{L_e\}_{e \in E}$ is feasible for all paths and weakly dominates the
minimum selection benchmark by recovering locally usable calibration budget whenever possible.
\end{enumerate}

\end{remark}

\paragraph{Complexity and scope.}
Let $m = |S|$ be the number of shared links, and let $M=\sum_{\pi\in\mathcal{P}}|\pi|$ denote the
total path length. Sorting the shared links costs $O(m \log m)$. Constructing the sets $S_e$ and
performing the local reallocations can be implemented in time proportional to $O(M)$ with
appropriate bookkeeping. Hence, the overall complexity is
\[
O(m \log m + M).
\]
Although GRO does not guarantee global optimality, it improves upon the minimum selection
benchmark by attempting to recover locally usable calibration budget while maintaining feasibility
across all paths.

\begin{algorithm}[!t]
\caption{Greedy Quantum Link Orchestration (GRO) Heuristic}
\label{alg:gro}
\begin{algorithmic}[1]
\REQUIRE Candidate paths $\mathcal{P}$; for each $\pi \in \mathcal{P}$ its budget $\mathcal{L}_\pi$;
per-path optima $L_{e,\pi}^\star$ for all $\pi \in \mathcal{P}$ and all $e \in \pi$
\ENSURE Global link assignment $\{L_e : e \in E\}$

\STATE For each link $e \in E$, compute $S_e \gets \{\pi \in \mathcal{P} \mid e \in \pi\}$
\STATE Initialise $L_e$ as undefined for all $e \in E$

\STATE \textit{/* Initialise non-shared links at their independent optima */} 
\FOR{$e \in E$}
  \IF{$|S_e| = 1$}
    \STATE Let $S_e = \{\pi\}$; set $L_e \gets L_{e,\pi}^\star$
  \ENDIF
\ENDFOR

\STATE Build shared-link list $S \gets \{\, e \in E : |S_e| \ge 2 \,\}$
\FOR{$e \in S$}
  \STATE $w(e) \gets \Big(\max_{\pi\in S_e} L_{e,\pi}^\star - \min_{\pi\in S_e} L_{e,\pi}^\star\Big)\,|S_e|$
\ENDFOR
\STATE Sort $S$ in descending order of $w(e)$

\FOR{each link $e$ in $S$}
  \STATE \textit{/* Resolve the shared-link conflict conservatively */} 
  \STATE $L_e \gets \min_{\pi \in S_e} \; L_{e,\pi}^\star$

  \FOR{$\pi \in S_e$}
    \STATE $\pi_{\text{un}} \gets \{e' \in \pi \mid |S_{e'}| = 1\}$
    \STATE $\Delta_\pi \gets \mathcal{L}_\pi - \sum_{e' \in \pi} L_{e'}$
    \STATE Reallocate $\Delta_\pi$ over links in $\pi_{\text{un}}$ up to caps $L_{e',\pi}^\star$,
    while maintaining $\sum_{e' \in \pi} L_{e'} \le \mathcal{L}_\pi$
  \ENDFOR
\ENDFOR

\RETURN $\{L_e\}_{e \in E}$
\end{algorithmic}
\end{algorithm}

\paragraph{Simulation-based performance evaluation.}
The theoretical analysis of the GRO heuristic establishes feasibility and the qualitative
improvement over the minimum selection benchmark, but it does not quantify how close
GRO comes to the true optimum for realistic network instances.
A natural evaluation methodology proceeds as follows.

\textit{Network model and approaches.}
Consider a grid-based quantum network in which $n$ source-destination (SD) pairs
are drawn uniformly at random over the graph, with each pair assigned an individual
end-to-end fidelity threshold $F_{ete,\pi}^{th}$ drawn uniformly from a specified
fidelity range. For each value of $n$ and each random SD-pair combination, five
approaches can be evaluated and compared (summarised in tab.~\ref{tab:comparison}):
\begin{table}[t]
\centering
\caption{Comparison of calibration strategies}
\label{tab:comparison}
\setlength{\tabcolsep}{3pt}
\renewcommand{\arraystretch}{1.1}

\begin{tabular}{>{\centering\arraybackslash}p{0.13\columnwidth}
                >{\centering\arraybackslash}p{0.23\columnwidth}
                >{\centering\arraybackslash}p{0.23\columnwidth}
                >{\centering\arraybackslash}p{0.24\columnwidth}}
\hline
\textbf{Method} &
\makecell[c]{\textbf{Local thresholds} \\ $(F^{th}_e)$} &
\makecell[c]{\textbf{E2E thresholds} \\ $(F^{th}_{ete})$} &
\makecell[c]{\textbf{Shared link} \\ \textbf{treatment}} \\
\hline
\textbf{REF} & $\checkmark$ & $\times$ & $\times$ \\
\textbf{QLO} & $\checkmark$ & $\checkmark$ & $\times$ \\
\textbf{MIN/GRO} & $\checkmark$ & $\checkmark$ & $\checkmark$ \\
\textbf{NUM} & $\checkmark$ & $\checkmark$ & $\checkmark$ \\
\hline
\end{tabular}
\end{table}

\begin{enumerate}
    \item \textbf{REF} (Reference baseline): each link is operated at its individual
    hardware threshold $L_e^{th}$, with no end-to-end fidelity constraint enforced.
    This is an infeasible upper bound on throughput serving as a reference ceiling.

    \item \textbf{QLO} (per-path, independent): Theorem~\ref{theorem.chain} is applied
    independently to each path, with no coordination on shared links.
    In a general network, this assignment is physically infeasible since a shared link
    cannot simultaneously satisfy the conflicting optimal values assigned to it by
    different paths. The resulting throughput therefore constitutes a second infeasible
    upper bound, tighter than REF but still unachievable.

    \item \textbf{MIN} (minimum selection rule): the benchmark of
    Section~\ref{ssec:minimum_selection_rule}, which clamps each shared link to the
    minimum of the per-path QLO values and performs no further budget recovery.
    This is feasible by construction and serves as the conservative lower bound
    among the constrained approaches.

    \item \textbf{GRO} (greedy quantum link orchestration): Algorithm~\ref{alg:gro},
    which extends MIN by redistributing freed budget over unshared links on each
    affected path. GRO is feasible and weakly dominates MIN.

    \item \textbf{NUM} (numerical optimisation): a local numerical solver, such as
    sequential quadratic programming (SLSQP), initialised from the GRO solution as a
    warm start and applied to the original objective $\max \sum_\pi T_\pi$
    subject to the linear fidelity-budget constraints.
    NUM finds a local optimum of the true problem and provides a practical upper bound
    on the achievable constrained throughput, against which GRO can be compared.
\end{enumerate}
The principal metric is \emph{observed throughput}: the sum of throughputs over paths
whose delivered end-to-end fidelity satisfies the per-path threshold
$F_{ete,\pi}^{th}$, averaged over the random SD-pair combinations.

\textit{Expected ordering and interpretation.}
The five approaches admit the ordering
\[
    T^{\mathrm{REF}} \;\ge\; T^{\mathrm{QLO}} \;\ge\; T^{\mathrm{NUM}}
    \;\ge\; T^{\mathrm{GRO}} \;\ge\; T^{\mathrm{MIN}},
\]
where the first two inequalities reflect the progressive tightening of feasibility
constraints and the last two reflect the improvement of recovery over no recovery.
The practically relevant quantity is the GRO--NUM gap
$T^{\mathrm{NUM}} - T^{\mathrm{GRO}}$, which measures the throughput that GRO leaves
unrealised relative to the best locally achievable constrained solution.

\textit{Dependence on network size.}
As $n$ increases, the fraction of shared links grows and the unshared headroom
$H_\pi = \sum_{e \in \pi_{\mathrm{un}}} r_e$ available for GRO's greedy reallocation
shrinks. Consequently, the residual freed budget that GRO cannot redistribute increases
with $n$, and the GRO--NUM gap is expected to widen with network size.
At small $n$, where most links are unshared, GRO is expected to be essentially
optimal and near-indistinguishable from NUM. At larger $n$, the joint optimisation
performed by NUM provides a measurable advantage, quantifying the regime in which
a more sophisticated solver, would be warranted.

\chapter{Programmable Quantum Repeater Nodes}
\label{ch:programmable_repeaters}

With this chapter, we begin the quest into the third scope of this work (see Section~\ref{sec:scope_objective}). As a reminder, we will be focusing on the fourth research question (see Section~\ref{sec:research_questions}) defined for this work in this chapter.

\section{Motivation for programmability in quantum networks}
\label{sec:prog_motivation}

Scalability and flexibility in classical communication networks have been largely enabled by programmability. The introduction of software-defined networking (SDN) and protocol-independent data planes made it possible to decouple network control from hardware implementation, allowing centralised controllers to dynamically configure network behaviour \cite{mckeown2008openflow, kreutz2014software, bosshart2014p4}. In contrast, most existing quantum network proposals implicitly assume fixed-function repeater nodes, where the set of local quantum operations is hard-coded at design time.

This lack of programmability poses a fundamental limitation for the quantum internet. Quantum network protocols such as entanglement swapping, purification, and routing require nodes to perform different local operations depending on protocol stage, network conditions, and control-plane decisions. Moreover, realistic quantum hardware is subject to drift, noise, and device-dependent imperfections, which necessitate adaptive control and in situ calibration.

Motivated by these considerations, this chapter introduces a programmable abstraction for quantum repeater nodes\cite{kumar2026nvcenter}. The central idea is to expose the physically implementable local operations of a repeater node through an instruction-set interface that can be driven by a classical controller. This enables flexible deployment of quantum network protocols without modifying node hardware, and creates a bridge between network-layer control and physical-layer implementation.

\section{NV center based quantum repeater architecture}
\label{sec:nv_architecture}

NV-centers in diamond have emerged as a powerful solid-state platform for quantum information
processing and quantum networking. Over the past two decades, extensive research has established their capabilities for
entanglement distribution, quantum memory, and multi-qubit control. Bernien et al. \cite{bernien2013heralded} achieved heralded entanglement between NV-centers separated by three meters by simultaneous measurement of two photons, each of which is entangled with an NV-center,
and within the same group, Hensen et al. \cite{Hensen2015LoopholeFreeBell} reported a loophole-free violation of Bell’s inequality with electron spins
separated by 1.3 km. Kalb et al. \cite{kalb2017entanglement} further demonstrated entanglement distillation between solid-state quantum
network nodes.

Beyond experiments, theoretical proposals have outlined how NV-centers can serve as quantum repeaters combining
entanglement generation, storage, purification, and swapping \cite{childress2005fault, nemoto2016photonic}. These architectures envision entanglement swapping
via Bell-state measurements on pairs of qubits entangled with neighbouring nodes.
These results provide the foundation for viewing NV-centers as building blocks of large-scale entanglement distribution
systems.

Complementary to experiments distributing entanglement, a broad range of studies have characterised the capabilities achievable within a single NV-center, forming the relevant operational regime for the NV-center based quantum repeater architecture, which are as follows:
\begin{enumerate}
    \item \emph{Optical Initialisation and Readout}: Spin-selective optical pumping and spin-dependent fluorescence contrast enable electron spin initialisation and measurement \cite{jelezko2006single}.
    \item  \emph{Single-Qubit Control}: Coherent microwave-driven Rabi oscillations have demonstrated high-fidelity rotations of
the electron spin \cite{dobrovitski2013quantum}.
\item  \emph{Electron-Nuclear Coupling}: Hyperfine interactions facilitate coherent control and entanglement between the
electron spin and nearby nuclear spins \cite{childress2006coherent}.
\item \emph{Quantum Memory}: Nuclear spins have been used as long-lived quantum memories, maintaining coherence for
timescales up to minutes \cite{maurer2012room}.
\item  \emph{Error Correction and Repetitive Readout}: Neumann et al. \cite{neumann2010single} demonstrated repetitive readout of an electron
spin via CNOT-like gates to nuclear spins. Waldherr et al. \cite{waldherr2014quantum} further achieved quantum error correction within a
small register.
\item \emph{Multi-Qubit Registers}: The realisation of the universal control of nuclear spin registers
up to ten qubits \cite{bradley2019ten, abobeih2018one, taminiau2014universal}.
\end{enumerate}
These studies suggest that a single NV-center can serve as a programmable quantum processor with initialisation, manipulation, and readout of the electron spin, conditional logic between nuclear and electronic spins, and limited entanglement operations within the local register.
However, such systems do not inherently enable entanglement swapping between remote nodes without photonic interfaces.

Architecturally, field-programmable spin array designs aim to realise reconfigurable spin-based processors by tuning local parameters across NV arrays \cite{wang2023field}. At the network level, routing and entanglement distribution have been explored via centralised, software-defined-network-style control planes with global link-state knowledge \cite{pant2019routing} and via fully distributed, decentralised protocols \cite{chakraborty2019distributed}. Beyond these approaches, quantum-native control architectures have also been proposed, where the control plane itself is placed in superposition. For instance, \cite{caleffi2025quantum} introduces an entanglement-defined controller (EDC) that manages a quantum control plane with superposed network addresses, enabling quantum-native routing. While their architecture introduces superposition at the network layer, our framework keeps the controller classical and introduces superposition within the node, that is, in the nuclear-spin register, to achieve programmability, diagnostics, and LCU-type formulations.

\subsection{Quantum repeater versus quantum router}

Till now, we have only utilised the term quantum repeater in the context of quantum networking and quantum internet. We have also used the term routing only as a process done by a quantum repeater. However, the term \emph{quantum router} in addition to \emph{quantum repeater} does emerge in the literature, and they both are used to describe intermediate network elements, and are sometimes employed interchangeably in the literature. However, these terms correspond to distinct abstraction layers and serve different conceptual roles within a quantum network architecture. Clarifying this distinction is important for correctly interpreting the node model introduced in this chapter and in general awareness of the terminology used in the field while for relating it to the routing and orchestration mechanisms discussed earlier in the work.

A \emph{quantum repeater} denotes a physical network node as discussed in earlier chapters, whose primary purpose is to enable multi-hop entanglement distribution in the presence of loss, noise, and decoherence. At a functional level, a repeater is characterised by its ability to generate entanglement with neighbouring nodes, store quantum states in local memory, and perform local quantum operations such as entanglement swapping, purification, and measurement. These operations are inherently local and concern the manipulation of quantum states at a single network vertex. As such, the notion of a quantum repeater captures the \emph{capabilities} of a node, rather than its role in end-to-end path selection or network-wide coordination.

A \emph{quantum router}, by contrast, is a network-layer abstraction that builds upon repeater functionality by incorporating address-based forwarding and path selection. Quantum routers determine how quantum states or entangled resources are forwarded across a network toward a specified destination, coordinating multiple links and managing the interaction between classical and quantum communication planes. This role is analogous to routing in classical communication networks and requires additional control-plane mechanisms, such as destination awareness, forwarding state, and scheduling across competing demands. An explicit formulation of this perspective is provided in \cite{huberman2020quantum}, where routing functionality is introduced on top of entanglement-enabled links to support flexible, multi-destination quantum networking.

The focus of this chapter is on the \emph{quantum repeater node} rather than on routing functionality. Each programmable NV-center node is modelled as a quantum repeater equipped with local programmability: it performs entanglement swapping and controlled quantum operations on locally stored qubits under the supervision of a classical controller, but it does not make routing or forwarding decisions. Instead, the architecture developed here provides the physical and logical building blocks from which quantum routers can be constructed when combined with higher-layer routing, forwarding, and scheduling mechanisms, as discussed in Chapters~4--6.

Specifically, a single repeater node is modelled as comprising one optically interfaced electron spin and a register of nuclear spins. The electron spin acts as the data qubit on which network-relevant quantum operations are performed, while the nuclear-spin register functions as a local control program that selects and configures these operations. A classical controller issues instruction vectors that determine the preparation of the nuclear register and the execution of the corresponding electron-spin operations. This separation reflects the conceptual distinction between data-plane and control-plane functionality familiar from classical networking, but is realised here entirely within a single quantum device rather than through distinct hardware modules.

By adopting this abstraction, the work isolates node-level quantum programmability from network-level routing concerns. This separation is consistent with the cross-layer design philosophy developed throughout the work: routing algorithms, calibration-aware link operation, and node architecture are designed to interact through well-defined interfaces, while remaining analytically and conceptually separable.

\section{Instruction set architecture for quantum repeaters}
\label{sec:isa}
\subsection{ISA Abstraction and System Model}
\label{sec:isa_model}

Consider an NV-center node with one electron spin $E_0$ and $n$ nuclear spins $N_1, \dots, N_n$, which together form a control register. 
Then the Hilbert space is given by
\begin{equation}
    \mathcal{H} = \mathcal{H}_{E_{0}} \otimes \left( \bigotimes_{i=1}^{n} \mathcal{H}_{N_i} \right),
\end{equation}
with $\dim (\mathcal{H}) = 2^{1+n}$. Considering a generalised model of a node having many electron spins, here, $0$ signifies the index of the electron spin involved.

\begin{figure}[tb]
\centering
\includegraphics[width=0.8\columnwidth ]{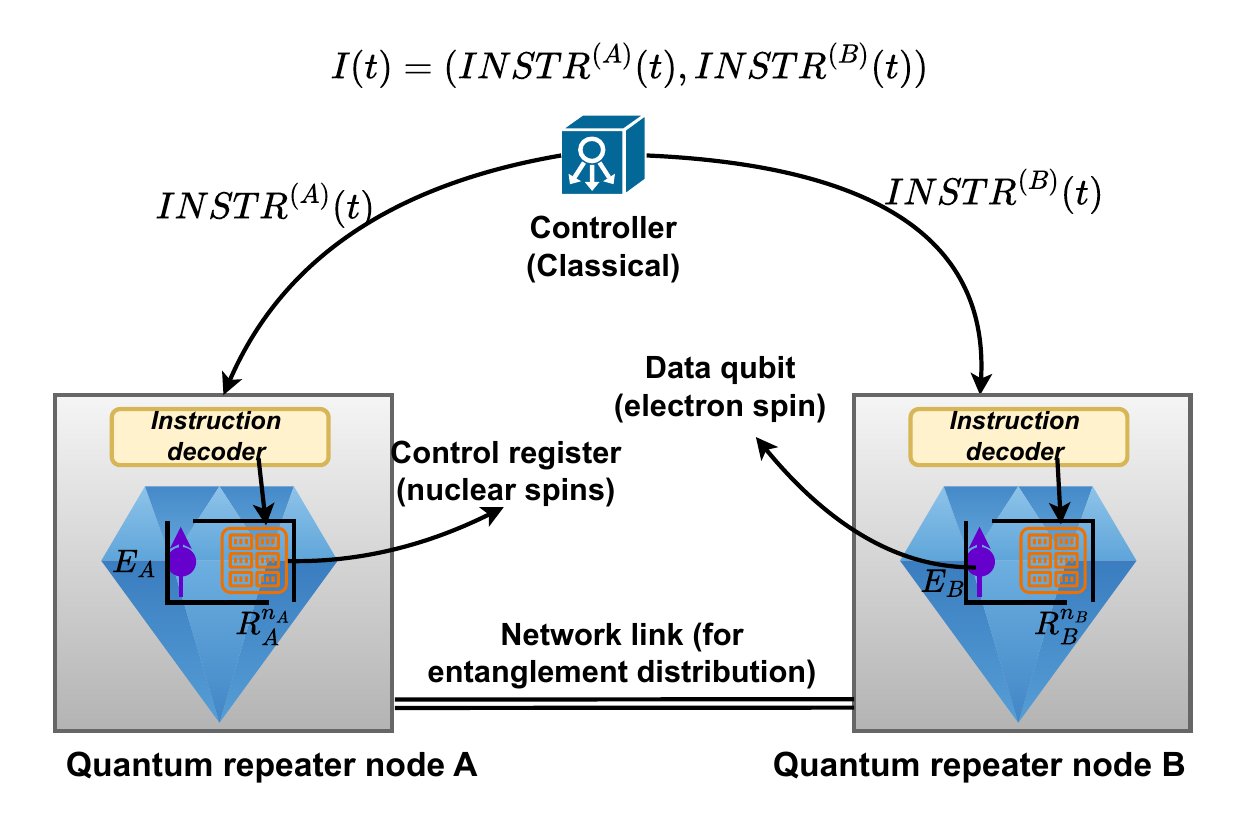}
\caption{Two-node example. A classical controller broadcasts the instruction vector \(\mathbf{I}(t)\). Each node decodes it into nuclear register preparation and microwave (MW) or radio-frequency (RF) pulse sequences that realise the selected electron spin operation.}
\label{fig:isa}
\end{figure}

A quantum network comprises $M$ such nodes, each controlled by a centralised classical controller, as shown in Fig.~\ref{fig:isa} for a two-node example. 
At any time-step $t$, the controller broadcasts an instruction vector
\begin{equation}
\mathbf{I}(t) = \big(\INSTR^{(1)}(t), \dots, \INSTR^{(M)}(t)\big),
\end{equation}
where each element $\INSTR^{(m)}(t)$ specifies the operation to be executed by node $m$. The \emph{controller-driven} execution model implicitly assumes a time-slotted
schedule: each instruction round \(t\) corresponds to a network-wide control slot during which all nodes complete the operations specified by \(\mathbf{I}(t)\) before proceeding to the next round. While asynchronous
execution is theoretically possible, synchronisation ensures that all spin operations complete within the coherence time of the most short-lived qubit involved in the protocol, thereby preserving end-to-end or protocol fidelity. This time-slotted abstraction aligns with the architecture advocated in \cite{beauchamp2025modular}, where slotted control
is shown to be beneficial for near-term quantum networks.
Each node contains a \emph{local decoder} that interprets its incoming instruction and configures both the nuclear register and the control pulse sequence required to enact the instruction.

\paragraph*{Instruction format}
Let $n$ denote the number of nuclear control qubits assigned to electron spin qubit $E_0$, so the register exposes $2^n$ addresses. Each instruction sent to an electron spin has the structure
\begin{equation}
\begin{aligned}
\INSTR^{(0)} = \big(
&\textsf{OPCODE},\;
 \textsf{PARAMS},\\
&\textsf{PATTERN}\subseteq\{0,\dots,2^n\!-\!1\},\;
 \textsf{MODE}
\big),
\end{aligned}
\label{eq:instr_def}
\end{equation}
where:
\begin{itemize}
\item \textsf{OPCODE} selects a primitive gate, e.g.\ $X$, $Z$, $H$, $R_y(\theta)$, $\text{CNOT}(E\!\to\!N_j)$, or $\text{MEASURE}$;
\item \textsf{PARAMS} carries any continuous parameters such as rotation angles or phases, and 
or addressing information for multi-qubit gates, such as the identifiers of control and target qubits (\texttt{control = $E_i$}, \texttt{target = $E_j$});
\item \textsf{PATTERN} identifies the subset of nuclear configurations that enable the operation;
\item \textsf{MODE}\,$\in\{\textsf{deterministic},\textsf{coherent}\}$ specifies the mode of programmability of nuclear register.
\end{itemize}

When \textsf{MODE} is \textsf{deterministic}, the local decoder initialises the nuclear register into one configuration from the set defined by \textsf{PATTERN}, producing a deterministic operation. 
When \textsf{MODE} is \textsf{coherent}, the decoder prepares a coherent superposition over the nuclear register configurations, allowing the corresponding electron spin operations to occur in superposition within the same execution cycle.

Physically, these logical instructions are realised through local microwave (MW) and radio-frequency (RF) control fields. 
Each node includes classical electronics that generate the MW pulses required for manipulating the electron spin and the RF pulses used to drive nuclear-spin transitions. 
The decoder configures these pulse generators according to $\INSTR^{(0)}(t)$ so that the required initialisation and controlled electron operation are performed. 
In this way, the \emph{instruction-set abstraction} encapsulates both the logical operation and its hardware realisation through MW/RF control at each node.

\paragraph*{Execution semantics}
An instruction acting in \textsf{coherent} mode realizes a conditional unitary of the form \cite{nielsen1997programmable}
\begin{equation}
U^{(0)} = \sum_{a_i\in \textsf{PATTERN}} |a_i\rangle\!\langle a_i|_N \otimes U_{E_0}(a_i),
\end{equation}
where $U_{E_0}(a_i)$ corresponds to the primitive selected by \textsf{OPCODE}$,$ parameterized by \textsf{PARAMS}. In \textsf{deterministic} mode the same instruction specializes to a single selected branch \(U^{(0)} = \ket{a}\!\bra{a}_N \otimes U_E(a)\) for some \(a \in \textsf{PATTERN}\).

The global network evolution at time-step $t$ is then
\begin{equation}
U_{\text{network}}(t) = \bigotimes_{m=1}^M U^{(m)}_{\INSTR^{(m)}(t)},
\end{equation}
and a complete protocol of $T$ instruction rounds is represented as
\begin{equation}
U_{\text{protocol}} = U_{\text{network}}(T)\cdots U_{\text{network}}(1).
\end{equation}

While the instruction format presented here is specifically designed for NV center-based quantum repeaters with the architecture described in Section~\ref{sec:nv_architecture}, the core abstraction principles are broadly applicable. The instruction set distinguishes between:
\begin{itemize}
    \item \textbf{Architecture-independent operations:} entanglement generation, swapping, measurement (applicable to any quantum repeater)
    \item \textbf{Architecture-specific details:} register addressing, timing parameters, control mechanisms (require adaptation for different platforms)
\end{itemize}
For other quantum repeater architectures (trapped ions, quantum dots, atomic ensembles), the same high-level operation types would be supported, but specific encoding and timing parameters would differ based on hardware characteristics. The abstraction layer (Section~\ref{sec:controller_driven}) is designed to accommodate such variations through architecture-specific backend implementations.
\subsection{Deterministic Register Control}
\label{sec:deterministic_control}

In the deterministic mode defined by the ISA, the nuclear-spin register serves as a classical control program while the electron spin functions as the data qubit. 
The controller initialises the nuclear register into a specific configuration selected by the instruction’s \textsf{PATTERN} field (one of the allowed basis states), and the corresponding electron spin operation is applied according to the \textsf{OPCODE} and \textsf{PARAMS}. 
This realises classical programmability, where each nuclear configuration \emph{deterministically} selects a single operation on the electron spin.

The required conditional gates, such as controlled rotations or CNOT operations, are implemented through the same MW and RF pulse sequences described in Section~\ref{sec:isa}. 
To illustrate, consider a node consisting of two nuclear spins ($N_1$, $N_2$) and one electron spin ($E_0$).  
Before instruction execution, the joint state is initialised as
\begin{equation}
\ket{\Psi_{\text{initial}}} =
\ket{0}_{N_1} \ket{0}_{N_2} \otimes
\bigl( \alpha \ket{0}_{E_0} + \beta \ket{1}_{E_0} \bigr),
\label{eq:deterministic_initial}
\end{equation}
where $\ket{0}_{N_1} \ket{0}_{N_2}$ denotes one of the four basis configurations of the nuclear register.  
The action taken by the node depends deterministically on this configuration, as summarised below.

\begin{figure}[tb]
    \centering
    \includegraphics[width=0.8\columnwidth]{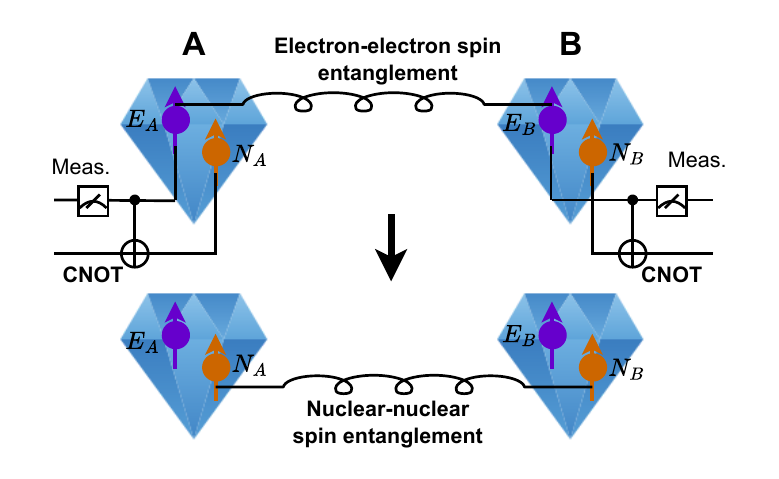}
    \caption{Entanglement transfer from electron-electron spin to nuclear-nuclear spin in neighbouring NV-center nodes.}
    \label{fig:NV_entanglement_transfer}
\end{figure}

\begin{enumerate}
    \item \textbf{$N_1=0, N_2=0$ (Idle)}
    
    No operation is applied to the electron spin:
    \begin{equation}
    \ket{\Psi_{\text{after}}} = 
    \ket{0}_{N_1}\ket{0}_{N_2} \otimes 
    \bigl( \alpha \ket{0}_{E_0} + \beta \ket{1}_{E_0} \bigr).
    \end{equation}

    \item \textbf{$N_1=0, N_2=1$ ($X$ gate)}
    
    The decoder triggers a bit-flip operation on the electron:
    \begin{equation}
    \ket{\Psi_{\text{after}}} = 
    \ket{0}_{N_1}\ket{1}_{N_2} \otimes 
    \bigl( \alpha \ket{1}_{E_0} + \beta \ket{0}_{E_0} \bigr).
    \end{equation}
    
    \item \textbf{$N_1=1, N_2=0$ ($R_y(\theta)$ rotation)}
    
    A rotation about the $y$-axis is applied to the electron spin:
    \begin{equation}
    \ket{\Psi_{\text{after}}} = 
    \ket{1}_{N_1}\ket{0}_{N_2} \otimes 
    \Bigl( \cos\frac{\theta}{2}\ket{0}_{E_0} + \sin\frac{\theta}{2}\ket{1}_{E_0} \Bigr),
    \end{equation}
    corresponding to a change of measurement basis or purification step.

    \item \textbf{$N_1=1, N_2=1$ ($\text{CNOT}_{{E_0}\rightarrow A}$)}
    
    For this configuration, the controller instructs a CNOT operation with the electron spin as control and an ancillary nuclear spin $A$ (say) as target.  
    The initial state, including $A$, is
    \begin{equation}
    \ket{\Psi_{\text{initial}}} =
    \ket{0}_{N_1}\ket{0}_{N_2} \otimes 
    \bigl( \alpha \ket{0}_{E_0} + \beta \ket{1}_{E_0} \bigr) \otimes \ket{0}_A.
    \end{equation}
    After applying $\text{CNOT}_{{E_0}\rightarrow A}$,
    \begin{equation}
    \ket{\Psi_{\text{after}}} =
    \ket{1}_{N_1}\ket{1}_{N_2} \otimes 
    \bigl( \alpha \ket{0}_{E_0}\ket{0}_A + \beta \ket{1}_{E_0}\ket{1}_A \bigr),
    \end{equation}
    transferring the electron’s entanglement to the ancillary nuclear spin for long-lived storage as illustrated in Fig.~\ref{fig:NV_entanglement_transfer} and discussed below. 
\end{enumerate}
Figure~\ref{fig:NV_entanglement_transfer} schematically illustrates the corresponding procedure in a realistic network setting, in which entanglement initially established between electron spins is coherently transferred to nuclear spins. In a distributed scenario, the protocol proceeds as follows.

Consider two neighbouring NV centres, denoted by nodes A and B, whose electron spins are prepared in a maximally entangled Bell state as shown in Fig.~\ref{fig:NV_entanglement_transfer}:
\begin{equation}
    \ket{\Psi_{E_A E_B}} = \frac{1}{\sqrt{2}}\bigl( \ket{0}\ket{0} + \ket{1}\ket{1} \bigr)_{E_A E_B}.
\end{equation}
Assuming that the nuclear spins at both nodes are initialised in the state $\ket{0}$, the combined electron–nuclear system is described by
\begin{equation}
    \ket{\Psi_{\text{initial}}}
    = \frac{1}{\sqrt{2}}\bigl( \ket{0}_{E_A}\ket{0}_{E_B} + \ket{1}_{E_A}\ket{1}_{E_B} \bigr)
    \otimes \ket{0}_{N_A}\ket{0}_{N_B}.
\end{equation}

Next, local controlled-NOT (CNOT) gates are applied independently at both nodes, using the electron spins as control qubits and the nuclear spins as target qubits, i.e.,
\begin{equation}
    \mathrm{CNOT}_{E \rightarrow N} \qquad \text{(applied at nodes A and B)}.
\end{equation}
As a result, the joint state of the system evolves to
\begin{equation}
    \ket{\Psi_{\text{after CNOT}}}
    = \frac{1}{\sqrt{2}}\bigl(
    \ket{0}_{E_A}\ket{0}_{N_A} \otimes \ket{0}_{E_B}\ket{0}_{N_B}
    + \ket{1}_{E_A}\ket{1}_{N_A} \otimes \ket{1}_{E_B}\ket{1}_{N_B}
    \bigr).
\end{equation}

Finally, a projective measurement\footnote{Even in the absence of explicit measurement on the electron spins, the nuclear spins remain entangled. The measurement is included here solely for clarity of presentation.} is performed on the electron spins. Conditioning on the measurement outcome, the nuclear spins are left in an entangled Bell state:
\begin{equation}
    \ket{\Psi_{N_A N_B}} = \frac{1}{\sqrt{2}}\bigl( \ket{0}_{N_A}\ket{0}_{N_B} + \ket{1}_{N_A}\ket{1}_{N_B} \bigr).
\end{equation}

In the discussion above, this transfer process is emulated locally by applying an entangling gate between the electron spin and a nuclear spin under the control condition $N_1 = 1$ and $N_2 = 1$.

\begin{figure}[tb]
    \centering
    \includegraphics[width=0.65\columnwidth]{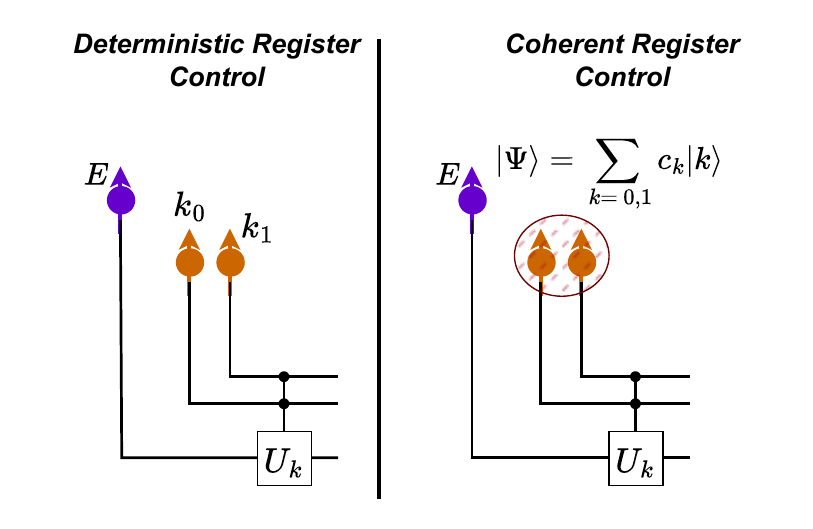}
    \caption{Programmability modes. Deterministic mode initialises the nuclear register in a basis state that selects one operation. Coherent mode prepares and reads the register in a rotated basis to enact a linear combination of operations on the electron spin.}
    \label{fig:deterministic_coherent_control}
\end{figure}

\subsection{Coherent Register Control}
\label{sec:coherent_control}

Building upon the deterministic case in Section~\ref{sec:deterministic_control}, we now extend to \emph{coherent register control}, where the nuclear spins are prepared in superposition rather than fixed classical configurations. In this mode, the nuclear register acts as a quantum address register, coherently selecting among multiple electron-spin side unitaries within a single execution cycle.

\paragraph*{Superposed nuclear register} Consider again two nuclear spins $N_1$ and $N_2$ prepared independently in arbitrary superpositions:
\[
|\Psi_{N_1 N_2}\rangle = 
\big(\alpha_0 |0\rangle + \alpha_1 |1\rangle\big)_{N_1}
\otimes
\big(\beta_0 |0\rangle + \beta_1 |1\rangle\big)_{N_2},
\]
which can be written explicitly as
\begin{equation}\label{eq:N1N2_superposition}
|\Psi_{N_1 N_2}\rangle
= \alpha_0 \beta_0 |00\rangle
+ \alpha_0 \beta_1 |01\rangle
+ \alpha_1 \beta_0 |10\rangle
+ \alpha_1 \beta_1 |11\rangle,
\end{equation}

where the amplitudes $\alpha_i \beta_j$, $\forall \ i, \ j \in \{0, 1\}$ represent the probability amplitudes of selecting nuclear bits in basis $|ij\rangle$. Now, each computational basis configuration $\ket{N_1 N_2}$ corresponds to a deterministic operation from Section~\ref{sec:deterministic_control}. We represent these collectively by a controlled operation \cite{nielsen1997programmable}

\begin{equation}\label{eq:U_router_alt}
    U_{\text{repeater}} = |00\rangle\!\langle 00| \otimes I_{E_0} + |01\rangle\!\langle 01| \otimes U_1 + |10\rangle\!\langle 10| \otimes U_2 + |11\rangle\!\langle 11| \otimes U_3,
\end{equation}
which can be written compactly as follows,
\begin{equation}\label{eq:U_router}
U_{\mathrm{repeater}} = 
\sum_{i,j \in \{0,1\}} |ij\rangle\!\langle ij| \otimes U_{ij},
\end{equation}
where
\begin{equation}\label{eq:basis_unitaries}
U_{00}=I, \quad
U_{01}=X, \quad
U_{10}=R_y(\theta), \quad
U_{11}=\text{CNOT}_{E_{0}\to A}.
\end{equation}
Let the initial state of the electron spin be:
\begin{align}
|\psi_{E_0}\rangle = \gamma_0 |0\rangle + \gamma_1 |1\rangle,
\end{align}
The joint initial state of the electron and the two nuclear spins is:
\begin{equation}
|\Psi_{\text{initial}}\rangle = |\Psi_{N_1 N_2} \rangle \otimes |\psi_{E_0}\rangle.
\end{equation}
Using Eq.~\eqref{eq:N1N2_superposition}:
\begin{equation}
     |\Psi_{\text{initial}}\rangle = (\alpha_0 \beta_0 |00\rangle + \alpha_0 \beta_1 |01\rangle + \alpha_1 \beta_0 |10\rangle + \alpha_1 \beta_1 |11\rangle) \otimes |\psi_{E_0}\rangle 
\end{equation}
Using Eq.~\eqref{eq:U_router_alt}, on applying the controlled operation:
\begin{equation}
|\Psi_{\text{after}} \rangle = U_{\text{repeater}} | \Psi_{\text{initial}} \rangle
\end{equation}

\begin{align}\label{eq:entangled_router}
\boxed{%
\begin{aligned}
|\Psi_{\text{after}}\rangle = {} & \alpha_0 \beta_0 |00\rangle \otimes I_{E_0} |\psi_{E_0}\rangle \\
&+ \alpha_0 \beta_1 |01\rangle \otimes X |\psi_{E_0}\rangle \\
&+ \alpha_1 \beta_0 |10\rangle \otimes R_y(\theta) |\psi_{E_0}\rangle \\
&+ \alpha_1 \beta_1 |11\rangle \otimes \text{CNOT}_{E \to A} |\psi_{E_0}\rangle
\end{aligned}}
\end{align}

Equation~\eqref{eq:entangled_router} shows that the repeater is now placed in a coherent superposition of performing different electron-side operations, weighted by the amplitudes of the nuclear register. This is the key distinction from deterministic control, where only one branch is active per run.

\paragraph*{Generalized formulation}
For an $n$-qubit nuclear register with $K=2^n$ configurations $\{\ket{k}\}$ and associated unitaries $\{U_k\}$, the general controlled operation is (Eq.~\eqref{eq:U_router})
\begin{equation}\label{eq:U_r_generalised}
U_{\mathrm{repeater}} = \sum_{k=0}^{K-1} \ketbra{k} \otimes U_k,
\end{equation}
acting on a register prepared in superposition
\begin{equation}
\ket{\psi_N} = \sum_{k=0}^{K-1} c_k \ket{k},
\qquad \sum_k |c_k|^2 = 1.
\end{equation}
The post-operation state is then

\begin{equation}\label{eq:qaddr_general}
U_{\mathrm{repeater}}\big(\ket{\psi_N}\otimes\ket{\psi_{E_0}}\big)
= \sum_{k=0}^{K-1} c_k\, \ket{k}\otimes U_k\ket{\psi_{E_0}}.
\end{equation}
Measuring the nuclear register in the computational basis collapses the system to one branch $U_k$, reproducing classical programmability. The advantage of coherent control emerges only when the nuclear register is measured in a rotated basis, producing interference between the electron-side operations as shown in Section~\ref{sec:diagnostics}.

\section{Controller-driven quantum network operation}
\label{sec:controller_driven}

One of the most common protocols in quantum networks is entanglement purification. In this section, we showcase the typical implementation of the ISA by using the example of the BBPSSW purification protocol~\cite{bennett1996purification}.

\begin{figure}[tb]
    \centering
    \includegraphics[width=0.8\columnwidth]{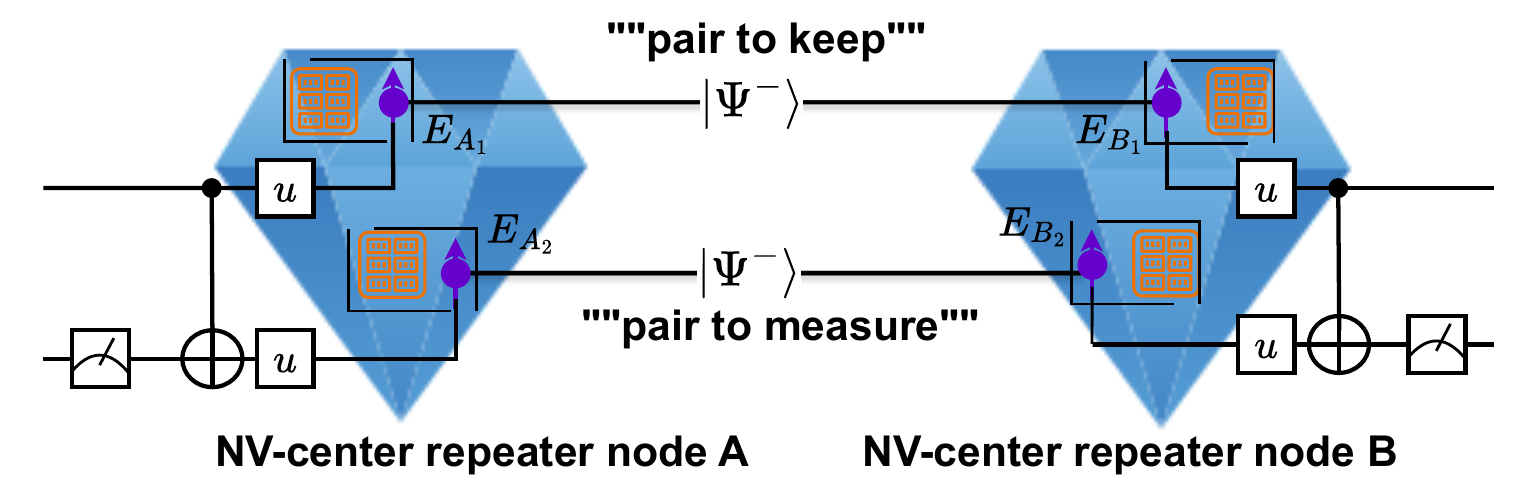}
    \caption{An example: BBPSSW purification protocol in NV-center nodes.}
    \label{fig:isa_purification}
\end{figure}

Consider two NV-center nodes, $A$ and $B$, that each hold two entangled electron-spin pairs shared between them through a noisy quantum channel, as shown in Fig.~\ref{fig:isa_purification}:
\begin{itemize}
    \item Pair~1: $(E_{A1}, E_{B1})$: the pair to keep.
    \item Pair~2: $(E_{A2}, E_{B2})$: the pair to measure.
\end{itemize}

The controller sends out a vector of structured instructions to both nodes to execute the Bennett \emph{et al.} purification protocol. Each instruction to electron-spin $E_e$ at node $m$ follows the format
\[
\INSTR^{(m, e)} = \big(\textsf{OPCODE},\; \textsf{PARAMS},\; \textsf{PATTERN},\; \textsf{MODE}\big),
\]
where in this case \textsf{MODE}=\textsf{deterministic}. The local decoder at each node interprets the received instruction to configure the relevant electron and nuclear spins and to generate the corresponding MW/RF pulse sequence as summarised in table~\ref{tab:BBPSSW_Uk}. The protocol proceeds as follows:
\newcommand{\colI}{0.14\columnwidth}
\newcommand{\colO}{0.20\columnwidth}
\def\colD{0.58\columnwidth}

\begin{table}[t]
\centering
\caption{Instruction set \(U_k\) sufficient to implement the BBPSSW protocol.}
\label{tab:BBPSSW_Uk}
\renewcommand{\arraystretch}{1.1}
\setlength{\tabcolsep}{4pt}
\footnotesize
\begin{tabular}{p{\colI} p{\colO} p{\colD}}
\hline
\textbf{Index} (\textsf{PATTERN})
& \textbf{Operation} (\textsf{OPCODE})
& \textbf{Description} \\
\hline
000 & \(I\)                & Idle or no operation \\
001 & \(X\)                & Bit flip (Pauli \(X\)) \\
010 & \(Y\)                & Bit and phase flip (Pauli \(Y\)) \\
011 & \(Z\)                & Phase flip (Pauli \(Z\)) \\
100 & \(\mathrm{CNOT}\)    & Two-qubit entangling gate \\
101 & \(\mathrm{MEASURE}\) & Projective measurement in the computational basis \\
\hline
\end{tabular}
\end{table}

\begin{enumerate}
    \item \textbf{Bilateral twirling to form rotationally symmetric Werner states.}
    The controller issues
    \[
    \mathbf{I}(t) = \Big(\INSTR^{(A,e)}(t),\,\INSTR^{(B,e)}(t)\Big),
    \]
\[
\begin{aligned}
\Rightarrow\;
\mathbf{I}(t)
&= \Big(
   \big(\INSTR^{(A, 0)}(t),\, \INSTR^{(A, 1)}(t)\big), \\
&\qquad
   \big(\INSTR^{(B, 0)}(t),\, \INSTR^{(B, 1)}(t)\big)
  \Big).
\end{aligned}
\]

    where each instruction has
\begin{equation}
\begin{aligned}
\textsf{OPCODE} &= \textsf{rand}\{I, X, Y, Z\} = u (say), \\
\textsf{PARAMS} &= \emptyset, \\
\textsf{PATTERN}   &= 000, \\
\textsf{MODE}   &= \textsf{deterministic}.
\end{aligned}
\label{eq:pauli_fields}
\end{equation}

    Both nodes apply the same random single-qubit rotation $u$ to their qubits, effectively depolarising each pair into a rotationally symmetric Werner state.

    \item \textbf{Bilateral CNOT from “pair to keep” to “pair to measure”.}
    The controller issues
    \[
\begin{aligned}
\INSTR^{(A)}(t{+}1)
&= \big(
   \textsf{CNOT},\;
   \textsf{PARAMS}=(E_{A1}\!\to\!E_{A2}), \\
&\qquad
   100,\;
   \textsf{deterministic}
  \big).
\end{aligned}
\]

and similarly for node $B$ as $E_{B1}\!\to\!E_{B2}$

so that both nodes perform a CNOT with the keep-pair qubit as control and the measure-pair qubit as target.

    \item \textbf{Local measurement of target qubits.}
    The controller issues
\[
\begin{aligned}
\INSTR^{(A)}(t{+}2)
&= \big(
   \textsf{MEASURE},\;
   \textsf{PARAMS}=(E_{A2}, \\
&\qquad\textsf{basis}=Z),\;
   101,\;
   \textsf{deterministic}
  \big).
\end{aligned}
\]

    and similarly for node $B$ on $E_{B2}$. Each node measures its target qubit in the computational basis and sends the classical outcome to the other node.

    \item \textbf{Parity check and sifting.}
    If the measurement outcomes are identical (even parity), the control pair $(E_{A1}, E_{B1})$ is kept; otherwise, it is discarded. This sequence is iterated over surviving pairs to progressively raise their fidelity.
\end{enumerate}

The instruction set required for this protocol thus consists of the primitives
\[
\{\textsf{PAULI},\, \textsf{CNOT},\, \textsf{MEASURE}\},
\]
which correspond to the essential elements of local node operation under controller supervision.

Let the corresponding local quantum operation at node $m$ be denoted by $\mathcal{E}^{(m)}_{\textsf{OPCODE}}$. For $\textsf{PAULI}$ and $\textsf{CNOT}$, the operation is unitary:
\begin{equation}
\mathcal{E}^{(m)}_{\textsf{OPCODE}}(\rho)
= U^{(m)}_{\textsf{OPCODE}}\rho\,U^{(m)\dagger}_{\textsf{OPCODE}},
\label{eq:Emk_def_new}
\end{equation}
where $U^{(m)}_{\textsf{OPCODE}} \in \{I,X,Y,Z,\mathrm{CNOT}\}$ acts on the addressed qubit(s).

For $\textsf{MEASURE}$ in the computational basis,
\begin{equation}
\begin{aligned}
\mathcal{E}^{(m)}_{\textsf{MEASURE}}(\rho)
&= \sum_{b\in\{0,1\}}
   \big(\Pi_b\,\rho\,\Pi_b\big)
   \otimes \ket{b}\!\bra{b}_{C_m},\\
\Pi_0 &= \ket{0}\!\bra{0},\quad
\Pi_1 = \ket{1}\!\bra{1},
\end{aligned}
\label{eq:Em5_def_new}
\end{equation}
optionally recording the classical outcome in a register $C_m$. Here, $\rho$ denotes the quantum state at node $m$, and $\Pi_0$ and $\Pi_1$ are the projectors onto the computational basis states $\ket{0}$ and $\ket{1}$, respectively. 
The sum over $b$ accounts for both possible outcomes of the measurement.
The term $\Pi_b\,\rho\,\Pi_b$ represents the \emph{state collapse} associated with outcome $b$, while the tensor product with $\ket{b}\!\bra{b}_{C_m}$ optionally records the classical result into a register $C_m$ associated with node $m$.
This record allows the measurement outcome to be used in later classical-control steps of a protocol (e.g., parity checks in BBPSSW).

\begin{center}
\fbox{%
\parbox{0.96\columnwidth}{%
\textbf{Note:}
In Sections~\ref{sec:deterministic_control} and~\ref{sec:coherent_control}, all instructions corresponded to unitaries. Here, the \textsf{MEASURE} operation is non-unitary and irreversible. 
Accordingly, the model extends from a set of unitaries $U_k$ to a set of completely positive trace-preserving (CPTP) maps $\mathcal{E}_{\textsf{OPCODE}}$, encompassing both reversible and measurement-based actions under the same instruction framework.
}}
\end{center}

Having established the controller-driven operation model, we now explore how the programmable ISA framework enables two critical practical capabilities: 
\begin{enumerate}
    \item diagnostic, calibration, and fidelity witnessing procedures (Section~\ref{sec:diagnostics}), which leverage the ISA's instruction-based control to characterize and optimize node performance, and
    \item scalability to multiple electron spins per node (Section~\ref{ssec:gen_espin}), which exploits the ISA's modularity to increase network throughput while maintaining programmability.
\end{enumerate}

\section{Diagnostics, calibration, and fidelity witnessing}
\label{sec:diagnostics}

The instruction set architecture introduced in Sections~\ref{sec:isa} and \ref{sec:controller_driven} provides not only a programming abstraction for quantum repeater operations, but also a foundation for network diagnostics, calibration, and fidelity witnessing. Specifically:

\begin{enumerate}
    \item \textbf{Diagnostics:} The deterministic register control instructions (Section~\ref{sec:deterministic_control}) enable systematic probing of quantum states without disrupting network operation, allowing real-time detection of hardware degradation or calibration drift.
    
    \item \textbf{Calibration:} The coherent register control operations (Section~\ref{sec:coherent_control}) can be orchestrated to perform in-situ recalibration sequences, adjusting control parameters based on measured performance metrics without taking repeaters offline.
    
    \item \textbf{Fidelity witnessing:} The measurement instructions combined with classical feedback enable implementation of fidelity witness protocols that provide real-time quality metrics on probed state to the network layer (connecting to the calibration-aware routing discussed in Chapter~\ref{ch:calibration}).
\end{enumerate}

This unified view demonstrates how programmability at the hardware layer directly supports the network-level requirements identified in earlier chapters.

A key advantage of coherent register control is its ability to enable interference-based diagnostics. By measuring the nuclear register in a rotated basis, it is possible to extract overlaps between different electron-side operations applied to the same input state.

Let the nuclear register be projected onto a general readout state
\begin{equation}
\ket{\phi} = \sum_{k=0}^{K-1} d_k \ket{k},
\qquad \sum_k |d_k|^2 = 1.
\end{equation}
Projecting Eq.~\eqref{eq:qaddr_general} gives the (unnormalized) electron state
\begin{equation}\label{eq:router_action}
\begin{aligned}
\big(\bra{\phi}\otimes I\big)\, U_{\mathrm{r}}\big(\ket{\psi_N}\otimes\ket{\psi_{E_0}}\big)
&= \Big( \sum_{k=0}^{K-1} d^{\!*}_k \bra{k} \Big) \\
& \times \Big(\sum_{k=0}^{K-1} c_k\, \ket{k} \otimes U_k \ket{\psi_{E_0}}\Big),
\end{aligned}
\end{equation}
\begin{equation}\label{eq:Kraus-operator}
\Rightarrow (\bra{\phi}\otimes I)\,U_{\mathrm{repeater}}(\ket{\psi_N}\otimes\ket{\psi_{E_0}})
= \Big(\sum_{k=0}^{K-1} d_k^{\!*}c_k\,U_k\Big)\ket{\psi_{E_0}}.
\end{equation}
Then, taking the square of the probability amplitude, the corresponding success probability of having outcome $|\phi \rangle$ is given by:
\begin{equation}
p_\phi = \big\|(\sum_k d_k^{\!*}c_k\,U_k)\ket{\psi_{E_0}}\big\|^2.
\end{equation}
Note that the electron state $|\psi_E \rangle$ is not normalised. Hence, the conditional action on the electron is a \emph{linear combination} of the branch unitaries with complex weights $d_k^{\!*}c_k$, and the probability of obtaining that action depends on the interference between branches. Let us take an example of two unitaries to show how to obtain fidelity witnessing.
\paragraph*{Example: Two-branch interference}
Consider a single nuclear qubit with unitaries $\{U_0,U_1\}$ and program state $\ket{\psi_N}=(\ket{0}+\ket{1})/\sqrt{2}$. After applying $U_{\mathrm{repeater}}$,
\begin{equation}\label{eq:two-branch-after}
\begin{aligned}
\ket{\Psi_{\mathrm{after}}}
&= U_{\mathrm{repeater}}\big(\ket{\psi_N}\otimes\ket{\psi_{E_0}}\big) \\
&= \tfrac{1}{\sqrt{2}}\Big( \ket{0}\otimes U_0\ket{\psi_{E_0}}
   + \ket{1}\otimes U_1\ket{\psi_{E_0}} \Big).
\end{aligned}
\end{equation}
Measuring the nuclear spin in the $X$ basis,
\(
\ket{\phi_\pm}=(\ket{0}\pm\ket{1})/\sqrt{2},
\)
yields the unnormalized electron states

\begin{align}
(\bra{\phi_\pm}\otimes I)\ket{\Psi_{\mathrm{after}}}
&= \frac{1}{\sqrt{2}}
   \Big(\tfrac{1}{\sqrt{2}}\bra{0}\;\pm\;\tfrac{1}{\sqrt{2}}\bra{1}\Big) \notag \\
&\quad \times
   \Big(\ket{0}\otimes U_0\ket{\psi_{E_0}}
   + \ket{1}\otimes U_1\ket{\psi_{E_0}}\Big), \\
\Rightarrow\;
(\bra{\phi_\pm}\otimes I)\ket{\Psi_{\mathrm{after}}}
&= \tfrac{1}{2}\Big(U_0\ket{\psi_{E_0}} \;\pm\; U_1\ket{\psi_{E_0}}\Big), 
\end{align}
\begin{align}
 \Rightarrow \ket{\psi_{E_0}^{(\pm)}}&= (\bra{\phi_\pm}\otimes I)\ket{\Psi_{\mathrm{after}}} = \frac{1}{2}\Big(U_0 \;\pm\; U_1\Big) \ket{\psi_{E_0}}.
\end{align}
Hence, the corresponding success probability is:

\begin{align}
p_\pm
&= \left\langle \psi_{E_0} \left| \tfrac{1}{2}(U_0^\dagger \pm U_1^\dagger)\, 
                      \tfrac{1}{2}(U_0 \pm U_1) \right| \psi_{E_0} \right\rangle, \notag \\
p_\pm
&= \tfrac{1}{4} \Big( \langle \psi_{E_0} | U_0^\dagger U_0 | \psi_{E_0} \rangle
   + \langle \psi_{E_0} | U_1^\dagger U_1 | \psi_{E_0} \rangle \notag \\
&\quad \pm \langle \psi_{E_0} | U_0^\dagger U_1 | \psi_{E_0} \rangle
   \pm \langle \psi_{E_0} | U_1^\dagger U_0 | \psi_{E_0} \rangle \Big). \label{eq:p_pm_expanded}
\end{align}

Since $U^{\dagger} U =I$ so $\langle \psi_{E_0} | U_0^{\dagger} U_0 | \psi_{E_0} \rangle = \langle \psi_{E_0} | U_1^{\dagger} U_1 | \psi_{E_0} \rangle = \langle \psi_{E_0} | I | \psi_{E_0} \rangle = 1$. 

And since $\langle \psi_{E_0} | U_0^\dagger U_1 | \psi_{E_0} \rangle$ is hermitian conjugate of $\langle \psi_{E_0} | U_1^\dagger U_0 | \psi_{E_0} \rangle$ i.e., $\langle \psi_{E_0} | U_0^\dagger U_1 | \psi_{E_0} \rangle$ = $\langle \psi_{E_0} | U_1^\dagger U_0 | \psi_{E_0} \rangle^{\dagger}$. Hence $\langle \psi_{E_0} | U_0^{\dagger} U_1 | \psi_{E_0} \rangle + \langle \psi_{E_0} | U_1^{\dagger} U_0 | \psi_{E_0} \rangle = 2 \ \mathrm{Re} \langle \psi_{E_0} | U_0^{\dagger} U_1 | \psi_{E_0} \rangle$ which gives
\begin{align}
\Rightarrow p_\pm &= \tfrac{1}{4} \Big( 1 + 1 \pm 2\,\mathrm{Re}\,\langle \psi_{E_0} | U_0^\dagger U_1 | \psi_{E_0} \rangle \Big),
\end{align}

\begin{equation}\label{eq:p_pm_basic}
\Rightarrow p_\pm = \tfrac{1}{2}\big(1 \pm \mathrm{Re}\langle\psi_{E_0}|U_0^\dagger U_1|\psi_{E_0}\rangle\big).
\end{equation}

And the corresponding (normalized) electron states after measurement of nuclear qubit is as follows:
\begin{equation}\label{eq:electronpm_map}
\ket{\psi_{E_0}^{(\pm)}} = \frac{\big(U_0 \pm U_1\big)\ket{\psi_{E_0}}}{\sqrt{\langle \psi_{E_0} | (U_0^{\dagger} \pm U_1^{\dagger}) \big(U_0 \pm U_1\big) | \psi_{E_0} \rangle}},
\end{equation}
\begin{equation}\label{eq:estate_overlap_nophase}
\Rightarrow \ket{\psi_{E_0}^{(\pm)}} \;\propto\; \big(U_0 \pm U_1\big)\ket{\psi_{E_0}}.
\end{equation}
From Eq.~\eqref{eq:p_pm_basic} and Eq.~\eqref{eq:estate_overlap_nophase}, we know the information about the interference term (the mapped state after the measurement and its probability) from the application of the two branches of unitary. The interference term $\mathrm{Re}\langle\psi_{E_0}|U_0^\dagger U_1|\psi_{E_0}\rangle$ encodes the overlap between the two implemented unitaries, that is, the information inaccessible in deterministic control.

\paragraph*{Phase scanning and fidelity witnessing}
If, instead, the nuclear register is prepared with a relative phase:
\begin{equation}
\ket{\psi_N(\varphi)} \;=\; \frac{\ket{0}+e^{i\varphi}\ket{1}}{\sqrt{2}} ,
\end{equation}
while we keep the same $X$-basis measurement. The analogue of \eqref{eq:two-branch-after} after inclusion of a relative phase is
\begin{equation}
\ket{\Psi_{\mathrm{after}}(\varphi)} \;=\; \frac{1}{\sqrt{2}}\Big( \ket{0}\otimes U_0\ket{\psi_{E_0}} \;+\; e^{i\varphi}\ket{1}\otimes U_1\ket{\psi_{E_0}} \Big).
\end{equation}
Similar to the two-branch interference case, repeating the steps we have:
\begin{align}
\Rightarrow\;
\ket{\psi_{E_0}^{(\pm)} (\varphi)}
&= (\bra{\phi_\pm}\otimes I)\ket{\Psi_{\mathrm{after}}(\varphi)} \notag \\
&= \tfrac{1}{2}\Big(U_0\ket{\psi_{E_0}}
   \;\pm\; e^{i\varphi} U_1\ket{\psi_{E_0}}\Big),
\label{eq:psi_e_pm}
\end{align}

\begin{align}
\Rightarrow \ket{\psi_{E_0}^{(\pm)} (\varphi)}
&= \frac{1}{2}\Big(U_0 \;\pm\; e^{i\varphi} U_1\Big) \ket{\psi_{E_0}}.
\end{align}
Similarly, the success probability is:
\begin{equation}
p_\pm (\varphi)
= \left\langle \psi_{E_0} (\varphi) \left| \tfrac{1}{2}(U_0^\dagger \pm e^{i \varphi} U_1^\dagger)\, \tfrac{1}{2}(U_0 \pm e^{-i\varphi} U_1) \right| \psi_{E_0} (\varphi) \right\rangle,
\end{equation}
\begin{align}
\Rightarrow\;
p_\pm (\varphi)
&= \tfrac{1}{4} \Big( \langle \psi_{E_0} | U_0^\dagger U_0 | \psi_{E_0} \rangle
   + \langle \psi_{E_0} | U_1^\dagger U_1 | \psi_{E_0} \rangle \notag \\
&\quad \pm \langle \psi_{E_0} | e^{i \varphi}\, U_0^\dagger U_1 | \psi_{E_0} \rangle
   \pm \langle \psi_{E_0} | e^{-i\varphi}\, U_1^\dagger U_0 | \psi_{E_0} \rangle \Big).
\label{eq:p_pm_phi}
\end{align}

Using the similar facts about the hermitian operators as above, we have:
\begin{align}
\Rightarrow\;
p_\pm(\varphi)
&= \tfrac{1}{4}\Big(
2
\;\pm\; e^{i\varphi}\bra{\psi_{E_0}} U_0^\dagger U_1 \ket{\psi_{E_0}}
\notag \\
&\hspace{2.5em}
\;\pm\; e^{-i\varphi}\bra{\psi_{E_0}} U_1^\dagger U_0 \ket{\psi_{E_0}}
\Big),
\label{eq:p_pm_phi_compact}
\end{align}

\begin{equation}\label{eq:p_pm_phi_final}
p_\pm(\varphi)
= \tfrac{1}{2}\big(1 \pm \mathrm{Re}[e^{i\varphi}\langle\psi_{E_0}|U_0^\dagger U_1|\psi_{E_0}\rangle]\big).
\end{equation}
In addition to the basic case without $\varphi$ discussed above, the phase $\varphi$ in the program state rotates the accessed quadrature. By sweeping $\varphi$, one can read out both real and imaginary parts of $\bra{\psi_{E_0}} U_0^\dagger U_1 \ket{\psi_{E_0}}$ in a single coherent experiment without needing two separate runs with stabilised relative phase using the classical deterministic-register control approach.

\paragraph*{\textbf{\underline{Extracting out the complete state with two-phase settings.}}} Using $e^{i\varphi} = \cos{\varphi}+i \sin{\varphi}$ in Eq~\eqref{eq:p_pm_phi_final} we have:

\begin{equation}
\begin{aligned}
\Rightarrow\;
p_{\pm}(\varphi)
&= \tfrac{1}{2}\Big(1 \ \pm \mathrm{Re}\Big[
     \cos{\varphi}\,\bra{\psi_{E_0}} U_0^{\dagger} U_1 \ket{\psi_{E_0}} \\
&\qquad\quad + i \sin{\varphi}\,\bra{\psi_{E_0}} U_0^{\dagger} U_1 \ket{\psi_{E_0}}
   \Big]\Big).
\end{aligned}
\label{eq:p_pm_phi_reim}
\end{equation}

Taking just $p_+$ measurements with $\varphi_1=0$ and $\varphi_2=-\pi/2$. Let $a = \bra{\psi_{E_0}} U_0^{\dagger} U_1 \ket{\psi_{E_0}}$.
\begin{equation}
\begin{aligned}
\text{For} \ \varphi_1 &= 0 \\
\Rightarrow\quad
p_+^{(1)} &= \tfrac12\big(1+\mathrm{Re}[\,1\cdot a + i\cdot 0\cdot a\,]\big) \\
&= \tfrac12\big(1+\mathrm{Re}\,a\big).
\end{aligned}
\label{eq:phi1_case}
\end{equation}

\begin{equation}
\begin{aligned}
\text{For} \ \varphi_2 &= -\tfrac{\pi}{2} \\
\Rightarrow\quad
p_+^{(2)} &= \tfrac12\big(1+\mathrm{Re}[\,0\cdot a + i\cdot(-1)\cdot a\,]\big) \\
&= \tfrac12\big(1+\mathrm{Im}\,a\big).
\end{aligned}
\label{eq:phi2_case}
\end{equation}

Inverting these to recover the overlap amplitude:
\begin{equation}
\mathrm{Re}\,a \;=\; 2\,p_+^{(1)} - 1,
\qquad
\mathrm{Im}\,a \;=\; 2\,p_+^{(2)} - 1.
\end{equation}
Hence, the state is:
\begin{equation}\label{eq:amplitude_calculated}
    \Rightarrow a = \bra{\psi_E} U_0^{\dagger} U_1 \ket{\psi_E} = (2 p_+^{(1)} - 1) + i\, (2 p_+^{(2)} - 1).
\end{equation}

\paragraph*{Fidelity between the two applied unitaries}
Let's say we have two implemented unitaries $U_0$ and $U_1$ acting on the same electron input $\ket{\psi_{E_0}}$, and we want to quantify how close their outputs are. The relevant \emph{state fidelity} is as follows:
\begin{equation}
F_{\mathrm{state}}(\psi_{E_0};U_0,U_1) =
\big|\bra{\psi_{E_0}}U_0^\dagger U_1\ket{\psi_{E_0}}\big|^2,
\label{eq:Fstate-def}
\end{equation}
which is straightforward from the eq~\eqref{eq:amplitude_calculated}:
\begin{equation}
    F_{\mathrm{state}}(\psi_{E_0};U_0,U_1) = (2 p_+^{(1)} - 1)^2 + (2 p_+^{(2)} - 1)^2.
\end{equation}

\paragraph*{Interpretation}
Coherent register control thus enables the NV-center node to act as an \emph{interferometric fidelity witness}, comparing the action of two or more operations on a chosen input state without full process tomography. The nuclear register serves as a coherent selector that allows interference between different branches of control, providing in situ calibration and diagnostic capability not possible under deterministic, classical programmability.

\begin{center}
\fbox{%
\parbox{0.96\columnwidth}{%
\textbf{Note:}
By this approach, the repeater does not give direct access to the
\emph{global} gate fidelity
\(F(U_0,U_1)=|\mathrm{Tr}(U_0^\dagger U_1)|^2/d^2\),
which would require full process tomography. Instead, by preparing a chosen
input state \(\ket{\psi_{E_0}}\), the repeater yields the state-dependent overlap
\(\langle\psi_{E_0}|U_0^\dagger U_1|\psi_{E_0}\rangle\). This serves as a
\emph{fidelity witness}: it certifies the closeness of \(U_0\) and \(U_1\) on the probed
state, providing partial diagnostic information without reconstructing the
complete channels.
}}
\end{center}
This capability is also useful for \emph{in situ calibration}: the applied pulse or pulse sequence can be tuned so that the implemented operation better matches the intended unitary. In addition, having a coherent repeater enables situations where the node can coherently evaluate or compare multiple unitaries before committing to one, something that is impossible under classical register control. In the classical case, once the nuclear register is prepared in a definite basis state, the corresponding unitary is fixed, and no further adjustment or comparison is possible.

Such diagnostic capabilities have no classical analogue and represent a uniquely quantum benefit of the proposed programmability framework.

\section{Generalisation to multiple electron spins per node}\label{ssec:gen_espin}

Building on the single-electron ISA model, this section generalizes the abstraction to quantum repeaters with multiple electron spins per node. This extension is directly motivated by the diagnostics and calibration requirements discussed in Section~\ref{sec:diagnostics}: multiple registers enable \emph{concurrent} operation and calibration, where some electron spins serve entanglement requests while others undergo maintenance procedures.

The ISA formulation so far has assumed a single electron spin denoted by $E_0$ and $n$ nuclear spins per node, as described in Sections~\ref{sec:deterministic_control} and~\ref{sec:coherent_control}. 
We now generalise to the case of $E_e$ electron spins per node, where each electron $E_e$ is paired with a local nuclear-spin register $\{N_{e,1}, \dots, N_{e,r}\}$ that provides its control. 
The ratio $r$ thus denotes the number of nuclear spins per electron.

The Hilbert space for a single electron–nuclear cluster is
\begin{equation}
\mathcal{H}_{\mathrm{cluster}}^{(e)} 
= \mathcal{H}_{E_e} \otimes \bigotimes_{i=1}^{r} \mathcal{H}_{N_{e,i}},
\end{equation}
where $\mathcal{H}_{E_e} \cong \mathbb{C}^2$ and $\mathcal{H}_{N_{e,i}} \cong \mathbb{C}^2$. 
The node-level Hilbert space is therefore
\begin{equation}
\mathcal{H}_{\mathrm{node}} = 
\bigotimes_{e=1}^{E} \mathcal{H}_{\mathrm{cluster}}^{(e)}, 
\qquad
\dim(\mathcal{H}_{\mathrm{node}}) = 2^{E(1+r)}.
\end{equation}

For each electron $E_e$, the control register spans a computational basis $\{\ket{k}_{N_e}\}_{k=0}^{2^r-1}$, defining the repeater operator corresponding to Eq.~\eqref{eq:U_r_generalised}:
\begin{equation}
U_{\mathrm{repeater}}^{(e)} = 
\sum_{k=0}^{2^r-1} \ketbra{k}_{N_e} \otimes U^{(e)}_k,
\end{equation}
where each $U^{(e)}_k$ is a local unitary (or joint electron–nuclear gate) determined by the controller-issued instruction $\INSTR^{(e)}$. 
Assuming independent control over each electron spin, the node-level evolution is
\begin{equation}
U_{\mathrm{node}} = 
\bigotimes_{e=1}^{E} U_{\mathrm{repeater}}^{(e)}.
\end{equation}

Each instruction $\INSTR^{(e)}$ takes the same structured form as defined in Eq.~\eqref{eq:instr_def}.

For a network with $M$ nodes, each containing $E_m$ electron spins and $r_m$ nuclear spins per electron, the global instruction vector at round $t$ is
\begin{equation}
\mathbf{I}(t) = 
\big\{\,\INSTR^{(m,e)}(t)\ \big|\ 
m=1,\dots,M,\ e=1,\dots,E_m\,\big\},
\end{equation}
and the corresponding evolution is
\begin{equation}
U_{\mathrm{network}}(t) = 
\bigotimes_{m=1}^{M} \ 
\bigotimes_{e=1}^{E_m} 
U^{(m,e)}_{\mathrm{r}}.
\end{equation}

Higher $E$ enables greater parallelism but increases the controller’s instruction-issuing overhead. 
Larger $r$ broadens the instruction set but lengthens nuclear re-initialisation cycles, potentially limiting protocol repetition rates as evaluated in Section~\ref{ssec:throughput_model}.

Current diamond platforms have demonstrated multi-electron operation within a single crystal up to two NV electron spins, with coherent dipolar coupling and room-temperature entanglement \cite{dolde2013room}.
On the memory side, a fully controlled ten-qubit (one electron + nine $^{13}$C nuclei) register has been operated as a processor with minute-scale quantum memory \cite{bradley2019ten}. Beyond that, isotopically engineered samples have been identified and individually addressed much larger neighbourhoods of nuclear spins around a single NV; for example, \cite{abobeih2022fault} reports control of dozens of $^{13}$C spins and uses a subset to realise a fault-tolerant logical qubit.
Taken together, these results motivate our modelling choice of one electron spin with a multi-qubit nuclear register as the baseline node today, while the ISA generalisation to $E_{e}>1$ electron spins per node should be read as a forward-looking architectural extension that becomes relevant as multi-electron spin NV modules mature \cite{taminiau2014universal}.

\subsection{Performance Model: Throughput vs. Nuclear Re-initialization}\label{ssec:throughput_model}

In the time-slotted model (Sec.~\ref{sec:isa}), each instruction round $t$ first re-initialises the nuclear register, followed by local MW/RF control on the electron spin and nuclear register and optional readout. A simple node-level performance metric can be defined as the per-electron spin \emph{round throughput} $R$ (rounds/s), given by:
\begin{equation}
R \;\approx\; \frac{1}{t_{\mathrm{MW}} + t_{\mathrm{RF}} + t_{\mathrm{meas}} + t_{\mathrm{class}} + t_{\mathrm{reinit}}(r)},
\label{eq:throughput}
\end{equation}
where $t_{\mathrm{MW}}$ and $t_{\mathrm{RF}}$ denote aggregated microwave and radio-frequency pulse times needed to execute the corresponding round (e.g., $X$, $Y$, and $\text{CNOT}$), $t_{\mathrm{meas}}$ captures any qubit readout invoked by the instruction, $t_{\mathrm{class}}$ is any classical-controller latency folded into the slot, and $t_{\mathrm{reinit}}(r) = \mathrm{\tau_{reset}} \ . \ r$ is the nuclear register re-initialisation time for a register of size $r$ nuclear spins with per-nuclear spin reset time of $\mathrm{\tau_{reset}}$. So, $t_\mathrm{reinit}$ grows with $r$ (e.g., linearly if spins are reset sequentially), reflecting the cost of programmability.

If a node contains $E_e$ electron spins operating in parallel independently (Sec.~\ref{ssec:gen_espin}), then the idealised node round throughput simply scales as
\begin{equation}
R_{\mathrm{node}} \;\approx\; E_e \, R,
\end{equation}
subject to controller and crosstalk constraints between spins. For multi-round protocols such as BBPSSW purification (Sec.~\ref{sec:controller_driven}), the time to complete one round is simply $1/R$.

\paragraph*{Assumptions}
The model is intentionally minimal. We treat $t_{\mathrm{MW}}, t_{\mathrm{RF}}, t_{\mathrm{meas}}, t_{\mathrm{class}}$ as configuration-dependent constants for a given protocol step; the only variable we sweep is $t_{\mathrm{reinit}}(r)$ to highlight the programmability-rate trade-off. Coherent control adds a small fixed overhead (phase preparation and rotated-basis readout), which can be absorbed into $t_{\mathrm{MW}}{+}t_{\mathrm{RF}}{+}t_{\mathrm{meas}}$.

\begin{figure}[t]
  \centering
  \includegraphics[width=0.8\columnwidth]{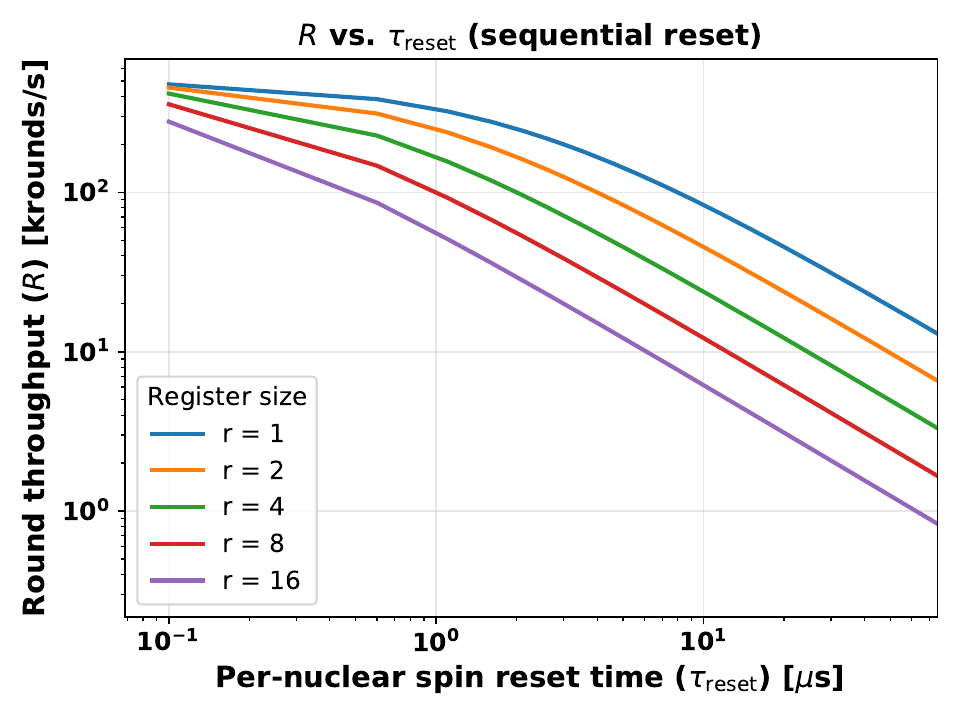}
  \caption{Per-electron round throughput $R$ vs Per-nuclear spin reset time $\tau_{\mathrm{reset}}$ for register sizes $r\in\{1,2,4,8, 16\}$ with fixed overheads ($t_{\mathrm{MW}}{+}t_{\mathrm{RF}}{+}t_{\mathrm{meas}}{+}t_{\mathrm{class}}$). Larger $r$ expands the ISA address space (more programmable operation choices) but increases re-initialisation time.}
  \label{fig:throughput_vs_reinit}
\end{figure}

Fig.~\ref{fig:throughput_vs_reinit} shows that as $\mathrm{\tau_{reset}}$ grows, throughput falls hyperbolically, with steeper degradation for larger $r$. Thus, beyond a register size that saturates the required instruction set, increasing $r$ yields diminishing returns: it broadens programmability while slowing repetition rate. The condition $2^r \ge p$ ensures that all $p$ distinct operations required by a protocol can be mapped to unique register states. This provides insights into how to choose $r$ that balances programmability against protocol latency in a slotted controller-driven network.

\section{Relation to Kraus operators and linear combination of unitaries}
\label{sec:lcu_kraus}

Equation~\eqref{eq:Kraus-operator} can be recast as a general operator acting on the electron spin subsystem:
\begin{equation}
K_{\phi,\psi_N} = 
\sum_{k=0}^{K-1} d_k^{\!*} c_k\, U_k,
\end{equation}
where $\ket{\psi_N} = \sum_k c_k \ket{k}$ is the prepared nuclear-register state and $\ket{\phi} = \sum_k d_k \ket{k}$ is the measurement outcome. 

Upon obtaining outcome $\ket{\phi}$, the electron spin state undergoes the transformation
\begin{equation}
\begin{aligned}
\ket{\psi_{E_0}} 
&\longmapsto 
\frac{K_{\phi,\psi_N}\ket{\psi_{E_0}}}
{\big\|K_{\phi,\psi_N}\ket{\psi_{E_0}}\big\|},\\
p_\phi 
&= \big\|K_{\phi,\psi_N}\ket{\psi_{E_0}}\big\|^2,
\end{aligned}
\label{eq:kraus_map}
\end{equation}
where $p_\phi$ is the probability of the outcome.

For a general mixed input state of electron spin $\rho_E$, this becomes $p_\phi = \mathrm{Tr}[K_{\phi,\psi_N}\rho_E K_{\phi,\psi_N}^\dagger]$. The above is a generalised version of the two-branch example Eq.~\eqref{eq:electronpm_map}, 
where $K_{\phi,\psi_N}^{\pm} = \tfrac{1}{2}(U_0 \pm U_1)$.

More generally, an orthonormal measurement $\{\ket{\phi_j}\}_j$ on the nuclear register, with $\ket{\phi_j} = \sum_k d_{j,k}\ket{k}$, defines a \emph{quantum instrument} on the electron spin  (Eq.~\eqref{eq:Kraus-operator}):
\begin{equation}
\begin{gathered}
K_j = \bra{\phi_j} U_{\mathrm{repeater}} \ket{\psi_N}
    = \sum_k d_{j,k}^{\!*} c_k\, U_k, \\
\mathcal{E}_j(\rho_E) = K_j \rho_E K_j^\dagger,
\end{gathered}
\label{eq:kraus_operator_final}
\end{equation}

where each $\mathcal{E}_j$ is a completely positive (CP), trace-non-increasing map with probability $p_j = \mathrm{Tr}[K_j\rho_E K_j^\dagger]$. 
The overall transformation $\sum_j \mathcal{E}_j$ is completely positive and trace preserving (CPTP).

Equation~\eqref{eq:kraus_operator_final} directly implements a \emph{Linear Combination of Unitaries (LCU)}:
\begin{equation}
K_j = \sum_k \alpha_{j,k} U_k, 
\qquad \alpha_{j,k} = d_{j,k}^{\!*} c_k,
\end{equation}
linking the nuclear-register superposition coefficients to the effective combination weights. 
This formalism connects node-level programmability to established techniques in Hamiltonian simulation and quantum channel decomposition~\cite{childs2012hamiltonian,chakraborty2024implementing}.

A related mathematical structure appears in the recent ``quantum path'' framework, or \emph{quantum SWITCH}, proposed in \cite{caleffi2023beyond} that places the order of two channels in coherent superposition. Our ISA instead uses a local nuclear register to coherently select among electron spin-side unitaries within a node. Extending the ISA to noisy CPTP operations and order superposition is interesting future work.
\chapter{Cross Layer Insights and Design Principles}\label{ch:crosslayer}

The preceding chapters developed three complementary building blocks for making quantum networks operational under realistic constraints: network layer routing under limited knowledge and heterogeneity (Chapters~\ref{ch:routing} and~\ref{ch:performance_evaluation}), link layer operation with calibration overhead and drift (Chapter~\ref{ch:calibration}), and node level programmability via an instruction set abstraction (Chapter~\ref{ch:programmable_repeaters}). While each component can be studied in isolation, the core message of this thesis is that quantum networks will only scale when these layers are designed and operated jointly.

This chapter aims to distill cross-layer insights that emerge when the three contributions of this work are viewed as a single system. The goal is twofold. First, to clarify how assumptions made at one layer silently constrain the design space at another layer. Second, to provide design principles that are directly actionable for near-term deployments and testbeds, where the dominant limitations are not only loss and noise, but also drift, partial knowledge, heterogeneous hardware, and constrained control.

\section{Interplay between routing and calibration}
\label{sec:interplay_routing_calibration}

A classical intuition is that routing selects paths, while link maintenance is a background function that keeps links available. In quantum networks, this separation is fragile because link quality degrades with operational usage, elapsed time, or both, and calibration restores an operational budget depleted by entanglement generation. Calibration, as observed in this work, therefore behaves like a consumable resource that competes with entanglement generation for time and link availability.

\subsection{Routing must internalise operationally evolving link quality}
The calibration model in Chapter~\ref{ch:calibration} implies that an ``available'' link is not a binary state. Instead, the link offers a fidelity profile that evolves over its operational window as a function of elapsed time, usage intensity, or both. Consequently, a route that is feasible at the moment of selection can become infeasible during execution if the schedule of entanglement generation, swapping, and classical signalling pushes operations beyond the link's usable degradation budget. This has two immediate implications.

First, routing metrics should incorporate an operational notion of feasibility, not only an abstract end-to-end fidelity computed from static parameters. A path selection rule that ignores drift will systematically overestimate success, leading to higher blocking probability at the network level.

Second, path choice changes the global calibration burden. Selecting a path that repeatedly uses a subset of links concentrates operational stress and accelerates the depletion of their calibration budget, increasing the frequency of recalibration on those links. In contrast, distributing traffic across multiple high-quality regions of the topology can reduce the peak calibration load, even if some selected paths are longer.

\subsection{Calibration policies reshape the effective topology}
Calibration converts a physical topology into an effective time-expanded topology. Over an operational horizon, each physical edge alternates between activation and calibration phases, which means that the network layer effectively sees a graph whose edge set changes over time. Two routes with identical hop count can therefore have very different success probabilities depending on where each link is within its activation cycle.

A useful design principle is as follows:

\paragraph{Design principle 1:} Treat calibration as a first-class scheduling constraint and expose it to routing through a coarse state variable, for example ``freshness'' of each link, rather than hiding it behind static link quality parameters.

This principle does not require full white box knowledge. Even in grey box settings, a controller can maintain a small amount of state derived from end-to-end estimates, for example, recent success rate or inferred fidelity window, and use it to bias path selection away from links near the end of their activation phase.

\subsection{Routing and calibration form a coupled optimisation problem}
Chapter~\ref{ch:calibration} shows that even for a linear chain, optimally allocating activation time across links under fidelity constraints is a nontrivial optimisation problem. In general networks, the coupling becomes stronger because paths share links. From a cross-layer viewpoint, the network-wide objective is no longer simply ``find disjoint paths'' or ``calibrate each link optimally''. It is now to:

\begin{itemize}
\item choose which paths to instantiate,
\item choose when to instantiate them,
\item choose how to allocate activation time and recalibration effort to the links they share.
\end{itemize}

This suggests a decomposition that aligns with how near-term networks can realistically be controlled.

\paragraph{Design principle 2:} Use of a hierarchical control. Let routing decide candidate paths using robust end-to-end metrics, then let link orchestration allocate calibration and activation time locally or per shared subset, and finally let node-level control realise the required operations.

This matches the thesis architecture: grey box or knowledge-aware routing generates candidate paths, quantum link orchestration resolves shared resources, and the programmable node abstraction executes the required operations via a central controller.

\subsection{Operational fairness is calibration fairness}
Chapter~\ref{ch:performance_evaluation} emphasises fairness across competing source destination demands. Calibration introduces a second notion of fairness: fairness in how calibration overhead is distributed across the infrastructure. A routing policy can appear fair in terms of accepted requests while still being unfair to the hardware, by repeatedly stressing a small set of links and forcing frequent recalibration cycles.

\paragraph{Design principle 3:} Include calibration cost as a regulariser in routing decisions. When multiple feasible paths exist, prefer the path that reduces marginal calibration load, even if its static fidelity estimate is slightly worse.

In practice, this can be implemented as a small penalty term that increases when a link has recently been calibrated or is close to its drift limit, encouraging diversification of traffic.

\section{Hardware constraints shaping network protocols}
\label{sec:hardware_constraints_protocols}

Quantum networking protocols are often presented as abstract sequences of entanglement generation, swapping, purification, and classical communication. Chapter~\ref{ch:programmable_repeaters} makes explicit that the feasibility and cost of these operations depend on the node architecture and its control interface. Once node-level constraints are acknowledged, several protocol-level assumptions must be revisited.

\subsection{Programmability changes what is ``a protocol''}
A fixed function repeater implements a predetermined set of operations. A programmable repeater exposes a library of operations that can be selected by a controller based on network conditions. This shifts the boundary between protocol design and protocol execution.

\paragraph{Design principle 4:} Separate protocol intent from protocol realisation. Protocols should be specified at the network layer in terms of the abstract quantum transformations to be achieved and the associated resource constraints, without committing to a specific sequence of hardware-level operations. The concrete execution of these protocols is then realised at the node level by mapping the abstract intent onto an instruction set architecture (ISA) that exposes the available physical operations and can evolve as hardware capabilities change.

Under this separation, introducing a new purification variant or an alternative entanglement swapping strategy does not require redesigning the entire repeater firmware. Instead, new protocols can be realised by composing or extending existing ISA primitives, with the controller selecting the appropriate instruction sequences based on current network conditions.

\subsection{Resource constraints translate into protocol choices}
NV center based nodes highlight concrete constraints: limited electron spins, finite nuclear register size, and reset overhead that grows with register capacity. This creates a direct trade-off between expressive programmability and raw throughput.

\paragraph{Design principle 5:} Treat programmability as a tunable resource. For a given deployment, select an operating point that balances required protocol diversity against throughput and reset overhead.

A practical example is the choice of how many distinct operations the node must support concurrently. Early networks may benefit from a smaller ISA that prioritises stable execution and faster cycling, while later networks may expand the ISA to support richer diagnostics, adaptive purification, or coherent control features.

\subsection{Diagnostics and witnessing as network primitives}
Classical networks rely on extensive telemetry. Quantum networks cannot freely measure quantum states without destroying them. Chapter~\ref{ch:programmable_repeaters} proposes diagnostics and fidelity witnessing mechanisms that are compatible with programmable hardware.

\paragraph{Design principle 6:} Build observability into the architecture. Provide explicit diagnostic operations in the ISA, and reserve network capacity for periodic inference of link and node health.

This supports both calibration and routing: routing benefits from improved end-to-end performance estimates, and calibration benefits from faster detection of drift regimes.

\section{Scalability and near-term feasibility}
\label{sec:scalability_nearterm}

Scalability in quantum networks is not only about adding more nodes. It is about maintaining operability when knowledge is incomplete, when parameters drift, and when heterogeneous upgrades occur.

\subsection{Partial knowledge as a scalability requirement}
The grey box approach shows that requiring full link-level knowledge may not be robust. As networks grow, the measurement and maintenance of detailed parameters becomes increasingly difficult and stale.

\paragraph{Design principle 7:} Optimise for decision quality under uncertainty, not for optimality under perfect knowledge. Prefer routing objectives that remain stable when performance estimates are noisy or outdated.

This provides a concrete roadmap for near-term deployments: start with topology and end-to-end measurements, then incorporate richer information only when it demonstrably improves outcomes.

\subsection{Incremental upgrades must be exploited, not ignored}
Heterogeneous repeater efficiencies are not an exception; they are the expected deployment path. Hardware will improve unevenly across a network.

\paragraph{Design principle 8:} Design routing and orchestration to benefit from partial upgrades. Even coarse awareness of node classes can deliver significant gains in fidelity and blocking probability, and it provides a clear incentive structure for incremental deployment.

\subsection{Complexity must be controlled at every layer}
Even when an optimisation problem is solvable analytically in special cases, general networks require heuristics, approximations, and decomposition. The thesis contributions point to a pragmatic posture.

\paragraph{Design principle 9:} Use layered approximations with explicit complexity budgets. For example,
\begin{enumerate}
\item limit candidate paths using a small $k$ shortest path set,
\item apply link orchestration only on the subgraph induced by accepted paths,

\item perform local node-level compilation and scheduling at programmable repeater nodes, rather than global pulse- or circuit-level synthesis.

\end{enumerate}

This ensures that control remains tractable as the network grows.

\section{Comparison with classical networking paradigms}
\label{sec:comparison_classical}

The quantum internet is often described by analogy to the classical internet, yet the analogy breaks in specific, instructive ways. This section summarises the most important correspondences and mismatches that inform design.

\subsection{Routing versus forwarding}
Classical routing selects a path, forwarding executes packet-by-packet delivery, and congestion control adapts to loss and delay. In quantum networks, forwarding is an active quantum operation, entanglement swapping, and it consumes the resource it creates. No packet can be buffered and resent without cost because of no cloning.

\paragraph{Insight:} Path establishment is closer to circuit setup than packet forwarding, but the circuit itself is probabilistic and perishable.

\subsection{Telemetry versus destructive measurement}
Classical networks instrument everything. Quantum networks cannot measure their internal state freely. This shifts the burden to indirect inference, occasional sacrificial measurements, and hardware-embedded diagnostics.

\paragraph{Insight:} Observability must be co-designed with the hardware interface. Rather than unrestricted measurement, ISA should provide diagnostic primitives whose impact on the quantum state is limited, predictable, and explicitly accounted for.

\subsection{Software-defined networking and controller-driven quantum networks}
The controller-driven model in Chapter~\ref{ch:programmable_repeaters} is motivated by software-defined networking in that it centralises policy and enables programmability. The key difference is that the controller must respect quantum constraints, including destructive measurement, coherence limits, and probabilistic outcomes.

\paragraph{Insight:} The quantum analogue of SDN requires a tighter integration between the control plane and hardware primitives. The ISA provides a concrete interface at which this integration can occur without hard-coding entire protocols into devices.

\subsection{Summary of cross-layer principles}
The cross-layer principles of this thesis can be summarised succinctly.

\begin{enumerate}
\item Routing, calibration, and node control are coupled. Any design that assumes independence will fail under drift and heterogeneity.
\item Partial knowledge is not a weakness; it is the realistic operating regime. Robust policies should be preferred over fragile optimal ones.
\item Calibration is a resource. Its cost must be accounted for explicitly in routing and scheduling decisions.
\item Programmability is the mechanism that allows network layer intent to reach the hardware without redesigning devices for every protocol change.
\item Scalability requires hierarchical control and bounded complexity at each layer.
\end{enumerate}

These principles provide the conceptual bridge from the technical results of Chapters~\ref{ch:routing} to~\ref{ch:programmable_repeaters} to the research outlook and future directions in Chapter~\ref{ch:conclusion}.

\chapter{Conclusions and Future Research Directions}
\label{ch:conclusion}

This thesis set out to address a central gap in the design of quantum networks: the lack of a coherent framework that connects abstract network-level protocols with the operational realities of quantum hardware. Beyond studying routing, link operation, and node design as independent problems, the insights developed in the Chapter~\ref{ch:crosslayer} argue that quantum networks must be understood and engineered as tightly coupled, cross-layer systems.

The contributions of this thesis can be grouped into three main themes. First, the formulation and analysis of routing problems under partial knowledge, heterogeneity, and fidelity constraints. Second, the modelling and optimisation of quantum link calibration under drift and finite activation windows. Third, the abstraction of programmable quantum repeater nodes through an instruction set architecture that bridges physical control and network-level intent. Together, these components provide a unified perspective on how quantum networks can be made operable, scalable, and adaptable under realistic assumptions.

This chapter concludes the thesis by summarising the main contributions and outlining promising directions for future research. The discussion focuses not on incremental extensions, but on research avenues that naturally emerge once quantum networks are viewed as software-defined, calibration-aware, and hardware-aware systems.

\section{Summary of contributions}
\label{sec:summary_contributions}

\subsection{Routing under partial knowledge and heterogeneity}

Chapters~\ref{ch:routing} and~\ref{ch:performance_evaluation} addressed the problem of routing in quantum networks where network knowledge is incomplete, noisy, or unavailable. Instead of assuming full knowledge of network components, the proposed models rely on grey box and end-to-end performance indicators that are realistically observable in near term systems.

The key contribution lies in demonstrating that routing decisions based on partial knowledge can still achieve strong performance guarantees, both in terms of blocking probability and delivered end-to-end fidelity, when compared to idealised full knowledge baselines. The analysis further showed that modest awareness of heterogeneity, for example, through repeater classes or coarse efficiency parameters, yields significant performance gains without incurring prohibitive monitoring.

These results question the implicit assumption that accurate, fine-grained link state information is a prerequisite for effective quantum routing, and instead support a design philosophy where robustness to uncertainty is prioritised.

\subsection{Calibration aware link operation}

Chapter~\ref{ch:calibration} introduced a model for quantum link operation that explicitly accounts for calibration overhead, drift, and finite activation windows. By treating calibration as a consumable resource that directly impacts throughput and fidelity, the chapter reframed link operation as an optimisation problem rather than a background maintenance task.

Analytical results for simple topologies, combined with a benchmark and heuristic for more general cases, showed that naive or static calibration policies can significantly underperform compared to adaptive strategies that balance usable activation across links. Importantly, the analysis revealed that calibration decisions at the link level propagate upward, shaping which network level paths are feasible and sustainable over time.

This contribution establishes calibration as a first-class concern in quantum network design and highlights the necessity of exposing calibration state, at least in coarse form, to higher layers.

\subsection{Programmable repeater nodes and instruction set abstraction}

Chapter~\ref{ch:programmable_repeaters} focused on the node level, proposing an instruction set architecture for programmable NV center based quantum repeaters. The chapter demonstrated how common network-relevant operations such as entanglement swapping, purification, diagnostics, and fidelity witnessing can be expressed as composable instructions executed under classical control.

By separating protocol intent from physical implementation, the proposed abstraction enables flexibility, extensibility, and hardware evolution without requiring protocol redesign. The instruction set model also clarifies the boundary between fast local control and slower network-level orchestration, providing a concrete interface at which cross-layer coordination can occur.

Considering these contributions together, the node-level abstraction complements the routing and calibration models by providing a realistic execution substrate for network decisions.

\section{Implications for quantum network design}
\label{sec:implications_design}

A unifying insight of this thesis is that many challenges attributed to quantum mechanics alone, such as fragility or limited scalability, are in fact consequences of architectural mismatches between layers. When routing ignores calibration, when calibration ignores hardware constraints, or when hardware hides essential state from the controller, performance degrades in systematic and avoidable ways.

The results of this work suggest several high-level implications.

First, quantum network protocols should be designed with explicit awareness of time. Fidelity is not only a function of path length, but also of when operations occur relative to calibration and drift cycles.

Second, programmability is not a luxury feature but a scalability enabler. As networks grow and hardware diversifies, fixed-function repeaters become a bottleneck for evolution and experimentation.

Third, partial knowledge could be embraced as the default operating regime. Systems that require perfect information to function well will not scale beyond small testbeds.

These implications point toward a future quantum internet that resembles a controlled, software-defined infrastructure rather than a static assembly of quantum links.

\section{Future research directions}
\label{sec:future_directions}

The framework developed in this thesis opens several promising avenues for future research. The directions outlined below are not exhaustive, but represent natural extensions that build directly on the concepts and models introduced here.

\subsection{Joint routing and calibration optimisation}

While this thesis analysed routing and calibration in complementary but distinct layers, a fully integrated optimisation remains an open problem. Future work could formulate joint objectives that simultaneously select paths, schedule entanglement attempts, and allocate calibration effort across the network.

Such formulations would likely require decomposition techniques, approximate solvers, or learning assisted heuristics to remain tractable. Exploring these approaches would further clarify the trade-offs between optimality, complexity, and robustness.

\subsection{Learning assisted network control}

The grey box routing paradigm naturally aligns with learning based methods. Reinforcement learning or Bayesian inference techniques could be used to infer link quality, drift regimes, or node behaviour from end-to-end observations, without explicit calibration data.

An important research question is how to combine learning with the strong structural constraints imposed by quantum mechanics, such as no cloning and destructive measurement. Hybrid approaches that integrate analytical models with data-driven adaptation could be promising.

\subsection{Extending the instruction set abstraction}

The instruction set architecture proposed in Chapter~\ref{ch:programmable_repeaters} provides a baseline abstraction for programmable repeaters. Future work needs to extend this abstraction to include richer control primitives, error analysis and its mitigation strategies, or multi-qubit operations as hardware capabilities evolve.

Another direction is the standardisation of such instruction sets across platforms, enabling heterogeneous networks composed of different physical technologies to interoperate at the control plane level.

\subsection{Experimental validation and testbed integration}

While the models and abstractions in this thesis are grounded in experimentally motivated parameters, full validation requires deployment in real or emulated testbeds. Integrating the proposed routing, calibration, and control mechanisms into experimental platforms would provide valuable feedback on model fidelity, control latency, and observability assumptions.

Such experiments would also help identify which aspects of the architecture are most sensitive to noise, drift, and hardware imperfections.

\subsection{Application-driven protocol design}

Finally, future work should consider how specific applications, such as distributed sensing, clock synchronisation, or blind quantum computation, map onto the cross-layer framework developed here. Different applications impose different fidelity, latency, and reliability requirements, which in turn shape routing, calibration, and control decisions.

Understanding these mappings will be essential for transitioning from general-purpose quantum networking research to application-specific deployments.

\section{Closing remarks}
\label{sec:closing_remarks}

The quantum internet will not emerge from isolated advances in hardware, protocols, or theory alone. It will emerge from architectures that reconcile the probabilistic, fragile nature of quantum information with the need for scalable, adaptable, and controllable networks.

This thesis has argued that such reconciliation is possible when routing, calibration, and node design are treated as interdependent components of a single system. By embracing partial knowledge, exposing hardware constraints, and leveraging programmability, quantum networks can be designed to operate robustly even under realistic limitations.

It is hoped that the ideas and models presented here will contribute to the development of quantum networks that are not only theoretically sound but also practically deployable.

\cleardoublepage
\appendix

\chapter*{Appendices}
\addcontentsline{toc}{chapter}{Appendices}

\cleardoublepage
\chapter{Appendix A: Proofs and Derivations}

\section{Proof of Theorem~\ref{theorem.equalactivation}}\label{sec:proof_th_6.2}
\begin{proof}
    From assumption~\ref{assumption.endtoendfidelity} with $\mathcal{U} = 1$ we have
\begin{equation}
    F_{ete} (a_e) = \frac{1}{4} \left[ 1 + 3 \prod_{i=0}^{N-1} \left( \frac{4 \ F_i (a_e) - 1}{3} \right) \right]
\end{equation}
Let the quantum network require an end-to-end fidelity threshold of $F_{ete}^{th}.$ So
\begin{equation}
    F_{ete}^{th} \leq F_{ete} (a_e) = \frac{1}{4} \left[ 1 + 3\prod_{i=0}^{N-1} \left( \frac{4 \ F_i (a_e) - 1}{3} \right) \right]
\end{equation}
\begin{equation}
    F_{ete}^{th} \leq \frac{1}{4}\left[ 1 + 3\prod_{i=0}^{N-1} \left( \frac{4 \ F_i (a_e) - 1}{3} \right) \right]
\end{equation}
\begin{equation}
    \Rightarrow 3^N \left( \frac{4 F_{ete}^{th} -1}{3} \right) \leq \prod_{i=0}^{N-1} \left(4 F_i -1\right)
\end{equation}
\begin{equation}
    \Rightarrow \frac{4^N}{3^N} \left( \frac{3}{4 F_{ete}^{th} -1} \right)  \geq \frac{ 1}{ \prod_{i=0}^{N-1} \left( F_i - \frac{1}{4} \right)}
\end{equation}
multiply by $\left( F_e^M - \frac{1}{4} \right)^N$ on both sides
\begin{equation}
    \Rightarrow \frac{4^N}{3^N} \left( \frac{3}{4 F_{ete}^{th} -1} \right) \left( F_e^M - \frac{1}{4} \right)^N \geq \frac{ \left( F_e^M - \frac{1}{4} \right)^N}{ \prod_{i=0}^{N-1} \left( F_i - \frac{1}{4} \right)}
\end{equation}
Taking logarithms on both sides and using the definition of $L_i = \ln{\left( \frac{ F_e^M - \frac{1}{4}}{F_i - \frac{1}{4}} \right)}$ and $A =  F_e^M - \frac{1}{4}$, we have
\begin{equation}
    \ln{\left[ \frac{4^{N-1}}{3^{N-1}} \frac{\left( F_e^M - \frac{1}{4}\right)^N}{F_{ete}^{th} - \frac{1}{4}} \right] } \geq \sum_{i=0}^{N-1} L_i
\end{equation}
\begin{equation}\label{eq:sumLi}
    \boxed{\sum_{i=0}^{N-1} L_i \leq (N-1) \ln{\left( \frac{4}{3} \right)} + N \ln{A} - \ln{ \left( F_{ete}^{th} - \frac{1}{4} \right)}    }
\end{equation}
In the above equation, the LHS denotes the contribution from all the links with an upper bound given by the RHS. Considering equal contribution from each link, we have 
\begin{equation}
    N L_e \leq (N-1) \ln{\left( \frac{4}{3} \right)} + N \ln{A} - \ln{ \left( F_{ete}^{th} - \frac{1}{4} \right)} 
\end{equation}
The maximum possible value for $L_e$ is for the equality, hence
\begin{equation}\label{eq:omegaL}
    \boxed{\Omega_L = L_e = \left( \frac{N-1}{N} \right) \ln{ \left( \frac{4}{3} \right)} + \ln{(A)}-\frac{1}{N} \ln{\left( F_{ete}^{th} - \frac{1}{4}\right)}}
\end{equation}
\end{proof}

\section{Optimal point at $\gamma+1$ layer of optimization is greater than or equal to at $\gamma$ layer}\label{sec:optimal_point_proof}
\begin{proof} 
From theorem~\ref{theorem.chain}, the recurrence relation between the optimal point is
\begin{equation}
    \Omega_{\pi}^{\gamma+1} = \frac{1}{k} \left[ 
 (j+k) \; \Omega_{\pi}^{\gamma} - \sum_j L_j^{th} \right]
\end{equation}
\begin{equation}\label{eq.recurrence_optimal_reduced}
    \Rightarrow k \; \Omega_{\pi}^{\gamma+1}=(j+k) \;\Omega_{\pi}^{\gamma} - \sum_j L_j^{th}
\end{equation}
Now since $L_j^{th} \leq \Omega_{\pi}^{\gamma}$, let $L_j^{th} = \Omega_{\pi}^{\gamma} - \epsilon_j$ where $\epsilon_j \geq 0$ dictates individual closeness of initial fidelity thresholds of each link to the optimal point of the path $\pi$.

From eq~\ref{eq.recurrence_optimal_reduced} we get,
\begin{equation}
    k \; \Omega_{\pi}^{\gamma+1} = (j+k) \; \Omega_{\pi}^{\gamma} - \sum_j \left( \Omega_{\pi}^{\gamma} - \epsilon_j \right)
\end{equation}
\begin{equation}
    \Rightarrow k \; \Omega_{\pi}^{\gamma+1} = (j+k) \; \Omega_{\pi}^{\gamma} - j \; \Omega_{\pi}^{\gamma} + \sum_j \epsilon_j
\end{equation}
\begin{equation}
    \Rightarrow \Omega_{\pi}^{\gamma+1} = \Omega_{\pi}^{\gamma} + \frac{1}{k} \left( \sum_j \epsilon_j \right)
\end{equation}
\begin{equation}
\boxed{
    \Omega_{\pi}^{\gamma+1} \geq \Omega_{\pi}^{\gamma} }
\end{equation}
\end{proof}

\section{Throughput as a family of hyperboloids}
From lemma~\ref{lemma.throughputN}, the throughput for a linear quantum chain network with heterogeneous initial fidelity generation is given by
    \begin{equation}\label{eq:throughputLgeneral}
        T = C \prod_{i=0}^{N-1} \left( \frac{L_i}{L_i + K} \right)
    \end{equation} where we transform the variables to L-space with $L_i = \ln{ \left( \frac{F_e^M - \frac{1}{4}}{F_i - \frac{1}{4}} \right)}$, $K=\Gamma_e c_e$.
    
Considering the two–link case, we obtain
\begin{equation}
    T = C \left( \frac{L_0}{L_0+K} \right)\left( \frac{L_1}{L_1+K} \right).
\end{equation}
Rearranging this expression, we can write
\begin{equation}\label{eq:hyperbola}
    \boxed{L_0L_1 \Bigl( 1 - \frac{T}{C} \Bigr) - \frac{T K}{C}(L_0+L_1) - \frac{K^2T}{C} = 0.}
\end{equation}
Thus, for fixed $T$, $C$, and $K$, the above equation defines a conic section in the $(L_0,L_1)$–plane. Its quadratic form in $L_0$ and $L_1$ has discriminant
\[
B^2-4AC=\left(1-\frac{T}{C}\right)^2,
\]
so that, letting
\[
\alpha \triangleq 1-\frac{T}{C},
\]
we have $B^2-4AC=\alpha^2$. In our case, with $0<T<1$ and $C=1$, it follows that $\alpha>0$. Hence, the throughput curves for the two–link case form a family of hyperbolas. By extension, the general $N$–link throughput defines a family of hyperboloids in $N$–dimensional space.

\section{Maximum throughput is the shortest distance from the optimal point $\Omega_L$}
\begin{proof} 
From eq~\ref{eq:hyperbola} for the two–link throughput model, we have
\begin{equation}
    L_0L_1 \Bigl( 1 - \frac{T}{C} \Bigr) - \frac{T K}{C}(L_0+L_1) - \frac{K^2T}{C} = 0
\end{equation}
defines a hyperbola in the $(L_0,L_1)$–plane. In the symmetric (or balanced) case we set
\[
L_0=L_1=L,
\]
so that the hyperbola reduces to
\begin{equation}
    L^2\Bigl( 1 - \frac{T}{C} \Bigr) - \frac{2TK}{C} \,L - \frac{K^2T}{C} = 0.
\end{equation}

From eq~\ref{eq:omegaL}, the optimal point for the two-link case is as follows
\begin{equation}
\boxed{\Omega_L = L_e = \frac{1}{2}\ln\left(\frac{4}{3}\right) + \ln(A) - \frac{1}{2}\ln\left(F_{ete}^{th} - \frac{1}{4}\right),}
\end{equation}
Then the corresponding optimal operating point in $(L_0,L_1)$–space is 
\[
P = (\Omega_L,\Omega_L).
\]

The throughput at any symmetric operating point is
\begin{equation}
T = C \left(\frac{L}{L+K}\right)^2.
\end{equation}
In particular, when $L=\Omega_L$, the maximum throughput is
\begin{equation}
T^* = C \left(\frac{\Omega_L}{\Omega_L+K}\right)^2.
\end{equation}

The Euclidean distance between an arbitrary point $(L_0,L_1)$ on the hyperbola and the optimal point $P=(\Omega_L,\Omega_L)$ is
\begin{equation}
d = \sqrt{(L_0-\Omega_L)^2 + (L_1-\Omega_L)^2}.
\end{equation}
In the symmetric case ($L_0=L_1=L$), this becomes
\begin{equation}
d = \sqrt{2}\,|L-\Omega_L|.
\end{equation}
Clearly, the minimum distance is achieved when
\begin{equation}
\boxed{L = \Omega_L}
\end{equation}

\end{proof} 

\cleardoublepage
\phantomsection
\small
\bibliographystyle{unsrt}
\bibliography{thesis}

\end{document}